\documentclass[a4paper,twoside,11pt]{book}
\nonstopmode
\usepackage[english]{babel}
\usepackage[T1]{fontenc}
\usepackage[ansinew]{inputenc}
\usepackage{geometry}
\usepackage{color} 
\usepackage{url} 
\PassOptionsToPackage{hyphens}{url}
\usepackage{lmodern}
\usepackage{graphicx}
\usepackage[titletoc]{appendix}
\usepackage{amsmath} \numberwithin{equation}{section}
\usepackage{amsthm}
\usepackage{amsfonts}
\usepackage{amssymb}
\usepackage{graphics}
\usepackage{color}
\usepackage{cases}

\theoremstyle{plain}
\newtheorem{theorem}{Theorem}[section]
\newtheorem{lemma}[theorem]{Lemma}
\newtheorem{corollary}[theorem]{Corollary}
\newtheorem{question}[theorem]{Question} 
\newtheorem{remark}[theorem]{Remark}
\theoremstyle{definition}
\newtheorem{definition}[theorem]{Definition}
    
\newtheorem{example}[theorem]{Example}
\newtheorem{exercise}{E}[section]

 \DeclareMathOperator{\tr}{Tr}
\definecolor{orange}{rgb}{1,0.5,0}

\usepackage{enumerate}

\newcommand\C{\mathbb C}         
\newcommand\R{\mathbb R}         
\newcommand\Z{\mathbb Z}         
\newcommand\N{\mathbb N}         
\newcommand\Ha{\mathbb H}       
\newcommand\D{\mathbb D}

\renewcommand\Im{\operatorname{Im}}
\renewcommand\Re{\operatorname{Re}}

      \newcommand{\eps}{\varepsilon}

\newcommand{\Landauo}{{\scriptstyle\mathcal{O}}}
\newcommand{\diam}{\operatorname{diam}}

\newcommand{\Pro}{\mathbb P}
\newcommand{\supp}{\operatorname{supp}} 
\newcommand{\E}{\mathbb{E}}

\allowdisplaybreaks[1]

\usepackage{tikz}
\usetikzlibrary{arrows}
\usepackage{background}
\usepackage{float}

\usepackage{sectsty}

\chapternumberfont{\Large} 
\chaptertitlefont{\Large}

\date{}

\usepackage{hyperref} 
\hypersetup{pdfpagemode={UseOutlines},bookmarksopen=true,
   bookmarksopenlevel=0,bookmarksnumbered=true,hypertexnames=false,
   colorlinks,linkcolor={blue},citecolor={blue},urlcolor={red},
   pdfstartview={FitV},breaklinks=true}
	
\begin{document}
\parindent 0pt 

\setcounter{section}{0}

\backgroundsetup{
scale=1,
angle=0,
opacity=1,
contents={\begin{tikzpicture}[remember picture,overlay]
 \path [left color = blue,middle color = white, right color = white] (current page.south west)rectangle (current page.north east);   
\end{tikzpicture}}
}

 \title{\vspace{-0.5cm}\textbf{{\Large A quantum invitation to probability theory}}\\[3.5cm]
\begin{tikzpicture}
\node[label={[label distance=0.5cm]90:\large \textit{Markov process}}] (A) at (-3,-2) [circle,shade,draw,top color=green] {$M_t$};
\node[label={[label distance=0.8cm]90:\large \textit{Quantum process}}] (B) at (3,-2) [fill=red!50] {$X_t$};
\node (C) at (9,-2) [fill=red!50]   {$X_t$};
\node[label={[label distance=0.3cm]270:\large \textit{Loewner chain}}] (D) at (9,-6) [circle,shading=ball, color=white] {$F_t$};
\node[radius=3mm,label={[label distance=0.1cm]268:\large \textit{Growing graphs}}] (E1) at (4,-8) [circle, fill=red] {};
\node[radius=3mm] (E2) at (4,-7.2) [circle, fill=red] {};
\node[radius=3mm] (E3) at (4,-6.4) [circle, fill=red] {};
\node[radius=3mm] (E4) at (4.8,-8) [circle, fill=red] {};
\node[radius=3mm] (E5) at (4.8,-7.2) [circle, fill=red] {};
\node[radius=3mm] (E6) at (4.8,-6.4) [circle, fill=red] {};
\coordinate[label={[label distance=-2mm]270: $\rightarrow$}] (AAA) at (2.8,-7.2);
\coordinate (BBB) at (2.8,-6.4);
\node[radius=3mm] (F1) at (1.8,-8) [circle, fill=red] {};
\node[radius=3mm] (F2) at (1.8,-7.2) [circle, fill=red] {};
\node[radius=3mm] (F3) at (1.8,-6.4) [circle, fill=red] {};
\node[radius=2mm] (H1) at (-3,-5) [circle, fill=white] {};
\node[radius=2mm] (H2) at (-3.8,-6) [circle, fill=white] {};
\node[radius=2mm,label={[label distance=0.2cm]270:\large \textit{Markov chain}}] (H3) at (-3,-7) [circle, fill=white] {};
\node[radius=2mm] (H4) at (-1.8,-5.5) [circle, fill=white] {};
\node[radius=2mm] (H5) at (-2,-7) [circle, fill=white] {};
\draw[->, very thick] (A) to[bend right=20]node[above] {$\mathbb{E}[\cdot|\mathcal{F}_t]$} (B);
\draw[->, very thick] (B) to node[above]{\footnotesize \textit{monotonically/freely}} node[below]{\footnotesize \textit{independent increments}}(C);
\draw[->, very thick] (C) to node[sloped,rotate=90]{\large $\displaystyle\varphi\left(\left(z-X_t\right)^{-1}\right)=\int_\R \frac{\mu_t(dx)}{z-x}$\hspace{1.4cm}${}$} (D) ;
\draw[->, very thick] (C) to (D);
\draw[->, very thick] (B) to[bend right=20] (BBB);
\draw[->, very thick] (B) to[bend right=20] (A);
\draw[-] (E1) to (E2);
\draw[-] (E2) to (E3);
\draw[-] (E1) to (E4);
\draw[-] (E2) to (E5);
\draw[-] (E3) to (E6);
\draw[-] (F1) to (F2);
\draw[-] (F2) to (F3);
\draw[->, thick] (A) to  (H1);
\draw[->, thick] (H1) to[bend right=20]node[left] {\small$0.7$} (H2);
\draw[->, thick] (H2) to[bend right=20]node[left] {\small$0.9$} (H3);
\draw[->, thick] (H1) to[bend left=20]node[above]  {\small$0.3$} (H4);
\draw[->, thick, thick] (H2) to[bend right=20]node[above] {\small$0.1$} (H4);
\draw[->, thick] (H5) to[bend right=20]node[right] {\small$0.1$} (H4);
\draw[->, thick] (H5) to[bend right=0]node[above]  {\small$0.9$} (H3);
\path[->, thick] (H4)   edge [loop right] node {\small$1.0$} ();
\end{tikzpicture}\\[3.5cm]
 \textbf{{\large Sebastian Schlei{\ss}inger}}
}

 \maketitle

\thispagestyle{empty}

\backgroundsetup{
scale=1,
angle=0,
opacity=1,
contents={\begin{tikzpicture}[remember picture,overlay]
 \path [left color = white,middle color = white, right color = white] (current page.south west)rectangle (current page.north east);   
\end{tikzpicture}}
}


\vspace*{\fill}
\textbf{\textit{\today}}

\tableofcontents

 \newpage

\textbf{\underline{Notation}}\\

\begin{center}
\begin{tabular}{l|l}
$2^M$& power set of the set $M$\\[2mm]
$\N$& $=\{1,2,...\}$, set of natural numbers\\[2mm]
$\N_0$& $=\{0,1,2,...\}$, set of non-negative integers\\[2mm]
$\R$& set of real numbers\\[2mm]
$\C$& set of complex numbers\\[2mm]
$\Ha$& $ = \{z\in\C\,|\, \Im(z)>0\}$, upper half-plane\\[2mm]
$\Ha^-$ & $= \{z\in\C\,|\, \Im(z)<0\}$, lower half-plane\\[2mm]
$\Pro$& probability measure\\[2mm]
$\mathbb{E}[X]$& classical expectation\\[2mm]
$\varphi(X)$& quantum expectation\\[2mm]
$\mathcal{B}(M)$& Borel $\sigma$-algebra of a metric space $M$\\[2mm]
$\mathcal{P}(\R)$&  set of all Borel probability measures on $\R$\\[2mm]
$\mathcal{P}_c(\R)$&  set of all $\mu\in\mathcal{P}(\R)$ with compact support\\[2mm]
$\varphi_{\mu}, \varphi_X$& characteristic function\\[2mm]
$\mathcal{N}(\mu,\sigma^2)$& normal distribution with mean $\mu$ and variance $\sigma^2$\\[2mm]
$A(\mu,\sigma^2)$& arcsine distribution with mean $\mu$ and variance $\sigma^2$\\[2mm]
$W(\mu,\sigma^2)$& semicircle distribution with mean $\mu$ and variance $\sigma^2$\\[2mm]
$G_\mu$& Cauchy transform\\[2mm]
$F_\mu$& $F$-transform\\[2mm]
$R_\mu$& $R$-transform\\[2mm]
$B_\mu$& Boolean transform\\[2mm]
$\mathcal{H}_\mu$&Hilbert transform\\[2mm]
$*$& classical (tensor) convolution / free product\\[2mm]
$\rhd$& monotone convolution / comb product\\[2mm]
$\lhd$& anti-monotone convolution\\[2mm]
$\boxplus$& free convolution\\[2mm]
$\uplus$& Boolean convolution\\[2mm]
$a^{\star n}$&the $n$th power of $a$ with respect to the operation $\star$, $a^{\star n}=a\star \ldots \star a$\\[2mm]
$\textbf{1}_A(x)$& the function $x\mapsto 1$ if $x\in A$, $x\mapsto 0$ if $x\not\in A$\\[2mm]
``a.s.''& almost surely\\[2mm]
``\textit{iid}'' & independent and identically distributed\\[2mm]
``variance $\sigma^2$''& always implies that $\sigma$ is the standard deviation, i.e.\ $\sigma\geq 0$
\end{tabular}
\end{center}

\chapter{Introduction}

These notes try to shine some light on the relations between the following three subjects:

\begin{center}
\begin{tikzpicture}
\node (A) at (0,0) [fill=red!20] {Quantum probability theory};
\node (B) at (-3,-2) [fill=red!20] {Classical probability theory};
\node (C) at (3,-2) [fill=red!20]   {${}$\qquad Complex analysis\qquad ${}$};
\draw[<->, very thick] (A) to (B);
\draw[<->, very thick] (A) to (C);
\draw[<->, very thick] (B) to (C);
\end{tikzpicture}\\[0cm]
\end{center}

The complex analyst might enjoy how holomorphic functions, in particular conformal mappings, appear as actors on the stage of quantum random variables, while the probabilist might discover that quantum probability offers a powerful framework that is far away from being used only in quantum mechanics. \\

Minimal required knowledge: any two of the three topics \emph{complex analysis}, \emph{classical probability theory}, \emph{Hilbert spaces}.\\

\underline{\textbf{Quantum probability spaces}}\\

According to M. Gromov, speaking of ``randomness'' in fields such as philosophy, psychology, natural evolution, and human history is completely meaningless:\footnote{Lecture on ``Probability, symmetry, linearity'', Institut des Hautes \'{E}tudes Scientifiques (IH\'{E}S).}

 \begin{center}``My point is, whenever you speak ``random'' [...] you have to have in mind a mathematical model. [...]
 Everything else is... I don't know what it is.''\end{center}

We adopt this point of view and immediately dive into mathematical models of randomness. We look for a mathematical model that is able to explain a number of observations \[x_1,...,x_N,\] which are real numbers as outputs of a black box function of reality. We might think of repeated observations from a physical experiment. A mathematical model of randomness should give us a probability distribution $\mu$ on $\R$, which we might use to make statistical predictions on some future observations.\\

\textbf{Classical probability theory.} In the 1930s, A. Kolmogorov developed the definition of a probability space as a triple 
\[ (\Omega, \mathcal{F}, \mathbb{P}), \] 
where $\Omega$ is a sample space, $\mathcal{F}$ is a set of events, which is a $\sigma$-algebra consisting of subsets of $\Omega$, and $\mathbb{P}$ is a function from $ \mathcal{F}$ into the interval $[0,1]$, yielding probabilities of events. \\
We often think of $\Omega$ as a mysterious, large set of outcomes working in the background.

In order to model the real numbers $x_1,...,x_N$, we would use a random variable on the probability space, which is simply a measurable function $X:\Omega \to \R$. Such a quantity works like a coordinate that projects $\Omega$ into $\R$ and it induces a probability measure $\mu$ on $\R$ by 
\[\mu(A)=\Pro[X\in A],\quad \text{for any Borel subset $A\subset \R$.  }\]

\textbf{Quantum mechanics.} Assume that $x_1, x_2,..., x_n$ are measurements of some property of a particle, e.g.\ the position or momentum of an electron. In quantum mechanics we also have a  mathematical model for these measurements yielding a probability distribution on $\R$, instead of a single point prediction. Let $H$ be a Hilbert space with inner product $\left<\cdot,\cdot\right>$  (we assume it is linear in the second argument) and let $X:H\to H$ be a (possibly unbounded) self-adjoint operator. Furthermore, let $\Psi \in H$ be a unit vector, the state of the quantum system. Then there is a unique probability distribution $\mu$, the spectral distribution of $X$,  defined  as  
\[\mu(A) = \left<\Psi, E_X(A) \Psi\right>,\quad \text{for any Borel subset $A\subset \R$},\]
where $E_X$ is the associated projection-valued measure.\\

Both models give a prediction $\mu$ for the observed values $x_1,...,x_n$. The differences of the models become clear when 
there are at least two measurements $x_1,..., x_n$ and $y_1,..., y_n$ of a different kind, which we model by adding a second random variable $Y:\Omega \to \R$ or a second operator $Y:H\to H$ respectively. Both frameworks also put a meaning to composed variables such as 
\[X+Y, \quad X\cdot Y, \quad X\cdot Y\cdot X+Y^2.\]
While  the product in the classical case (pointwise multiplication) is commutative, i.e.\ $X\cdot Y=Y\cdot X$, the product of operators on $H$ (the composition) is not commutative in general.  \\

Fortunately, \emph{quantum probability theory} (or \emph{noncommutative probability theory}) offers a framework for a more general probability theory which contains both models as special cases. 
In both cases, we have an algebra $\mathcal{A}$ of (bounded) random variables and an expectation functional $\varphi:\mathcal{A}\to\R$ with 
$\varphi(X)=\mathbb{E}[X]$ in the classical case and $\varphi(X)=\left<\Psi, X \Psi\right>$ in the quantum case. The distribution $\mu$ of $X$ can be defined via $\varphi(X^n) = \int_\R x^n \mu(dx)$ in both cases.\\ 
In quantum probability theory, one defines an abstract quantum probability space as an algebra $\mathcal{A}$ (with more or less additional structures), together with an expectation functional $\varphi:\mathcal{A} \to \R$.\\
The important notion of \emph{independence} of random variables can now still be defined in this framework. Interestingly, there are now five different possible definitions of independence (and here the theory splits into five branches). In particular, there are 
\begin{itemize}
\item five central limit theorems for independent and identically distributed quantum random variables,
\item five Poisson limit theorems,
	\item five classes of quantum stochastic processes with independent increments, etc.
\end{itemize}

\vspace{2mm}
 
\underline{\textbf{Complex analysis}}\\

Instead of real algebras we will rather consider complex algebras $\mathcal{A}$ with a linear mapping $\varphi:\mathcal{A} \to \C$. 
This is more helpful and more elegant even if we are only interested in real random variables. In fact, this trick is applied already in classical probability theory. Consider the characteristic function or Fourier transform of a real random variable $X$, given by
\[ \varphi_X(t) = \int_\R e^{itX(\omega)} {\rm d}\Pro(\omega), \quad t\in\R.\]
This function is simply the expected value of the complex random variable $e^{itX}$, i.e.\ $\varphi_X(t) = \mathbb{E}[e^{itX}].$
It plays an important role in classical probability theory, as it encodes not only distributions as functions from $\R$ to $\C$, but also the weak convergence of distributions, which corresponds to pointwise convergence of the characteristic functions. Furthermore, the independence of random variables can be simply expressed by higher dimensional characteristic functions.
In some parts of quantum probability theory, the role of $\varphi_X$ is replaced by the Cauchy transform
\[ G_X(z) = \mathbb{E}\left[\left(z-X\right)^{-1}\right], \quad \Im(z)>0. \]
 $G_X$ is a holomorphic function on the upper half-plane  (and maps it into the lower half-plane).
A complex analyst could ask the following fun question. How does $X$ have to be distributed such that $G_X$ is injective on $\Ha$, and thus maps $\Ha$ conformally onto a simply connected subdomain of the lower half-plane? Due to the Riemann mapping theorem, we know that there are many conformal mappings on $\Ha$, so maybe some of them are indeed Cauchy transforms? 
 It turns out that the answer to this question has in fact a deeper meaning and can be formulated via quantum probability theory: these distributions are precisely the distributions appearing in additive processes with monotonically independent increments. \\
Furthermore, the evolution $t\mapsto \mu_t$ of the distributions in such processes are Loewner chains, which are a standard tool in the theory of conformal mappings and describe a decreasing family of simply connected domains in $\C$.\\

\underline{\textbf{Outline}}

\begin{itemize}
	\item Sections 2 and 3: classical random variables, independence, central limit theorem, Markov processes.
	\item Sections 4--7: quantum random variables, independence, central limit theorems, additive processes. 
	\item Sections 8 and 9: Loewner chains and distributions of monotone and free additive processes.
	\item Sections 10-13: a selection of applications. 
\end{itemize}

\part[Theory]{Theory} 

\chapter{Classical probability theory}\label{sec_2}

In this section, we recall several basic notions and theorems from classical probability theory. Roughly speaking, it is the road from independent random variables to the central limit theorem. \\
Standard references for classical probability theory are the books \cite{bill} and \cite{kall}. 
We also refer to \cite{Bil99}, where the convergence of probability measures is treated comprehensively.

\section{The basic notions}

\begin{definition}[Classical probability space] A probability space is a triple $(\Omega, \mathcal{F}, \mathbb P)$, where $\Omega$ is a non-empty set (sample space), $\mathcal{F}\subset 2^{\Omega}$ is a $\sigma$-algebra (the set of events), and $\mathbb P: \mathcal{F}\to[0,1]$ is a probability measure.
\end{definition}

We recall that a set $\mathcal{F}\subset 2^{\Omega}$ is called $\sigma$-algebra if the following conditions are satisfied:
\begin{itemize}
	\item $\Omega\in \mathcal{F}$,
	\item if $A\in \mathcal{F}$, then $\Omega \setminus A \in \mathcal{F}$,
	\item if $A_1,A_2,... \in \mathcal{F}$, then $\cup_{n=1}^\infty A_n \in \mathcal{F}$.
\end{itemize}
A function $\mathbb P: \mathcal{F}\to[0,1]$ is called probability measure if 
\begin{itemize}
	\item $\mathbb P(\Omega)=1$ and
	\item if $A_1,A_2,... \in \mathcal{F}$ are pairwise disjoint, then $\mathbb P(\cup_{n=1}^\infty A_n)=\sum_{n=1}^\infty \mathbb P(A_n)$.

\end{itemize}

The $\sigma$-algebra $\mathcal{F}$ encodes the observable information of our random model.  

\begin{example}In the trivial case $\Omega=\{\emptyset, \Omega\}$, we only observe one event, namely $\Omega$, and we are completely blind to any further details (the elements of $\Omega$). \\
If $\mathcal{F}=2^\Omega$, then every subset of $\Omega$ is an event and we have the largest possible amount of information. This is the usual choice in case $\Omega$ is finite. If $\Omega=\{\omega_1,...,\omega_n\}$ and $\mathcal{F}=2^\Omega$, then $\Pro$ is completely determined by the probabilities $\Pro(\{\omega_1\}), ..., \Pro(\{\omega_n\})$.\hfill $\blacksquare$
\end{example}

\begin{example}Let us think of a fair dice with six faces. We could choose 
\[\Omega = \{1,2,3,4,5,6\}, \quad \mathcal{F}=2^\Omega, \quad \text{and} \quad \Pro(\{n\})=\frac1{6}, \quad n=1,...,6.\]
Compare this with the probability space
\[(\Omega,\quad \mathcal{F}_2=\{\emptyset,\{3\}, \{6\}, \{3,6\}, \Omega\setminus\{3\}, \Omega\setminus\{6\}, \Omega\setminus\{3,6\},\Omega\},\quad \Pro|_{\mathcal{F}_2}).\] Nothing changed, except for the smaller $\sigma$-algebra $\mathcal{F}_2$. We can still observe an event like ``The outcome can be divided by $3$'', i.e.\ $\{3,6\}$, but not anymore ``The outcome is a prime number'', i.e.\ $\{2,3,5\}$. In this second case, we can imagine that someone removed the numbers from the faces $1,2,4,5$. They are still possible outcomes of a random experiment, but they have become indistinguishable to us.\hfill $\blacksquare$
\end{example}

For any subset $M\subset 2^\Omega$, one can construct the $\sigma$-algebra $\mathcal{F}$ generated by $M$, i.e.\ $\mathcal{F}$ is defined as  the smallest $\sigma$-algebra that contains $M$, in other words
\[\mathcal{F}=\bigcap_{\substack{\text{$\sigma$-algebra $\mathcal{A}$}\\ M\subset \mathcal{A}}} \mathcal{A}.\]

\begin{example}If $\Omega$ is a metric space and $\mathcal{F}$ is the $\sigma$-algebra generated by all open subsets of $\Omega$, then 
$\mathcal{F}=\mathcal{B}(\Omega)$ is called the \emph{Borel $\sigma$-algebra} of $\Omega$ and a probability measure $\mathbb{P}:\mathcal{B}(\Omega)\to[0,1]$ is called a \emph{Borel probability measure}.\\
For example, endow $[0,1]$ with the usual Euclidean metric. Then there exists a unique probability measure $\lambda$ on $\mathcal{B}([0,1])$, the Lebesgue measure, such that $\lambda([a,b])=\lambda((a,b))=b-a$ for all intervals, see \cite[Sections 2 and 3]{bill}.\hfill $\blacksquare$
\end{example}

\begin{definition}[Random variable]
Let $M$ be a metric space endowed with its Borel $\sigma$-algebra. Then a measurable function $X:\Omega\to M$ is called an \emph{$M$-valued random variable}. 
\begin{itemize}
	\item  The push-forward of $\Pro$ with respect to $X$ yields the Borel probability measure $\mu(A)=\Pro(X\in A)$, which is called the \emph{distribution} of $X$. (We write $\Pro(X\in A)$ short for $\Pro(\{\omega\in \Omega\,|\, X(\omega)\in A\})$.)
	\item Likewise, the pullback 	$\sigma(X) = \{X^{-1}(A)\,|\, A \in \mathcal{B}(M)\}$ of the  Borel $\sigma$-algebra gives a $\sigma$-algebra on $\Omega$, consisting of all the information encoded by $X$.
\end{itemize}
Mostly, we will deal with \emph{real-valued random variables}, i.e.\ $S=\R$. If not stated otherwise, we will always assume that a random variable is real-valued.\\
However, in our study of real-valued random variables, we will also encounter the cases $S=\C$ and $S=\R^n$. In all these cases, we use the usual Euclidean metric.\\

Let $X$ be a real- or complex-valued random variable. 
Provided the integral exists, the expectation $\mathbb{E}[X]$ is defined by \[\mathbb{E}[X] = \int_\Omega X(\omega) d\mathbb{P}(\omega).\]
If the corresponding integrals exist, then 
\begin{itemize}
	\item $\mathbb{E}[X^n]$ is called the $n$-th moment of $X$ and
	\item $Var(X)=\mathbb{E}[(X-\mathbb{E}[X])^2]=\mathbb{E}[X^2]-(\mathbb{E}[X])^2$ is called the variance of $X$.
\end{itemize}
 If $\mu$ is the distribution of $X$, then $\mathbb{E}[X]$ is also called the \emph{mean}, $\mathbb{E}[(X-\mathbb{E}[X])^2]$ the \emph{variance}, and  $\sqrt{\mathbb{E}[(X-\mathbb{E}[X])^2]}$ the \emph{standard deviation of $\mu$}.
\end{definition}

\begin{definition}We denote by $\mathcal{P}(\R)$ the set of 
all Borel probability measures on $\R$. The support $\supp(\mu)$ 
of $\mu$ is defined as $\supp(\mu)=\{x\in\R\,|\, x\in U\subset\R,\; U\; \text{open} \Longrightarrow \mu(U)>0\}$. The support is always a closed subset of $\R$. We denote by $\mathcal{P}_c(\R)$ the set of all $\mu\in\mathcal{P}(\R)$ with compact support.
\end{definition}

\begin{example}The most important element of $\mathcal{P}(\R)$, and the most important from all distributions, is the Gaussian normal distribution $\mathcal{N}(c,\sigma^2)$, with mean $c$ and variance $\sigma^2$, given by 
\[\mathcal{N}(c,\sigma^2)(A) =  \frac1{\sigma\sqrt{2\pi}}\int_A e^{-\frac1{2}(\frac{x-c}{\sigma})^2} dx,\]
where $A\subset \R$ is a Borel subset.\hfill $\blacksquare$
\end{example}

 \begin{figure}[H]
 \begin{center}
 \includegraphics[width=0.9\textwidth]{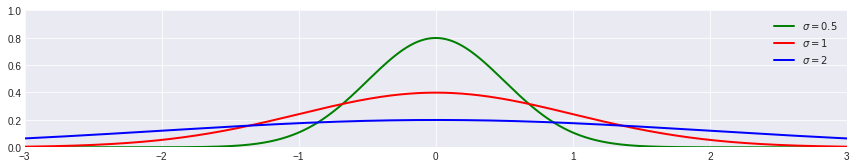}
 \caption{Densities of the Gaussian normal distribution.}
 \end{center}
 \end{figure}

\begin{remark}\label{Hilbert_space}Let $(\Omega, \mathcal{F}, \Pro)$ be a probability space. Call two complex-valued random variables equivalent if they differ only on sets of measure $0$. Fix some $p\geq1$. The space $L^p(\Omega, \mathcal{F}, \Pro)$ consists of all equivalence classes $[X]$ of random variables $X$ with $\mathbb{E}[|X|^p]<\infty$. Note that this definition does not depend on the choice of representative. This space becomes a complex Banach space with the norm $\|[X]\|_p=\mathbb{E}[|X|^p]$, see \cite[Chapter 1]{kall}.   \\
For $p=2$, the norm is induced by an inner product and $L^2(\Omega, \mathcal{F}, \Pro)$ becomes a complex Hilbert space via
\[\left<X,Y\right> := \mathbb{E}[X\overline{Y}]=\int_\Omega X(\omega)\overline{Y}(\omega) {\rm d}\Pro(\omega). \] 
\end{remark}

\begin{remark}On $\Omega=[0,1]$, the Lebesgue measure also yields a probability space\\ $([0,1], \mathcal{F}_\lambda, \lambda)$, where $\mathcal{F}_\lambda$ consists of all Lebesgue measurable subsets of $[0,1]$ (the completion of the Borel $\sigma$-algebra $\mathcal{B}([0,1])$). We have 
\[ \mathcal{B}([0,1]) \subsetneq \mathcal{F}_\lambda \subsetneq 2^{[0,1]}. \]

In probability theory, speaking of the ``uniform distribution'' on $[0,1]$ usually refers to \\
$([0,1], \mathcal{B}([0,1]), \lambda)$, and we also use the Borel $\sigma$-algebra to define random variables rather than $\mathcal{F}_\lambda$. The reason is that we obtain more random variables in this way. Compare this to the notion of Lebesgue measurable functions in analysis. A function $f:[0,1]\to[0,1]$ is called Lebesgue measurable if it is measurable as a function from $([0,1],  \mathcal{F}_\lambda)$ to $([0,1],  \mathcal{B}([0,1]))$. \\
Furthermore, it should be noted that the proof of $\mathcal{F}_\lambda \subsetneq 2^{[0,1]}$, i.e.\ the existence of a non-Lebesgue measurable subset of $[0,1]$, due to Vitali, requires the axiom of choice, see \cite[p.45]{bill}. One might look at a different kind of  mathematics by excluding the axiom of choice (and keeping the Zermelo-Fraenkel set theory) and adding the axiom that all subsets of $\R$ are Lebesgue measurable (the Solovay model). 
\end{remark}

\section{Independence}

Let us fix a probability space $(\Omega, \mathcal{F}, \mathbb P)$. 
If we are interested in only one random variable $X$ on $\Omega$, this whole setup might seem too cumbersome as we could simply look at the push-forward measure of $X$ on $\R$. This changes when we regard at least two random variables $X$ and $Y$ on $\Omega$, whose behavior might be interlocked in more or less complicated ways. A simple measure to analyze their interaction is the covariance.

 \begin{definition}[Covariance] Let $X$ and $Y$ be square-integrable random variables. Then the \emph{covariance $cov(X,Y)$} is defined as
\[ cov(X,Y) = \mathbb{E}[(X-\mathbb{E}[X])(Y-\mathbb{E}[Y])]=\mathbb{E}[XY]-\mathbb{E}[X]\mathbb{E}[Y].\] 
If $cov(X,Y)=0$, $X$ and $Y$ are called \emph{uncorrelated}. Under the assumption of positive variances, the correlation coefficient $\rho(X,Y)$ is defined as \[\rho(X,Y)=\frac{cov(X,Y)}{\sigma(X)\sigma(Y)}\in [-1,1].\]
\end{definition}     

The fact that $\rho(X,Y)\in[-1,1]$ follows from the Cauchy-Schwarz inequality $|cov(X,Y)|^2\leq \sigma^2(X)\sigma^2(Y)$ (recall the Hilbert space from Remark \ref{Hilbert_space}).\footnote{In the case of complex-valued random variables, we would define 
$cov(X,Y)=\mathbb{E}[(X-\mathbb{E}[X])\overline{(Y-\mathbb{E}[Y])}]$.}\\

The notion of independence of random variables is of paramount importance in probability theory. (See the next section for the implication ``independent $\Rightarrow$ uncorrelated''.)

\begin{definition}[Independence of events and $\sigma$-algebras] Let $(A_j)_{j\in J}\subset \mathcal{F}$ be a family of events with a non-empty index set $J$. Then $(A_j)_{j\in J}$ is called \emph{independent} if for any distinct $j_1,...,j_n \in J$, we have \[\Pro\left(\bigcap_{k=1}^n A_{j_k}\right) = \prod_{k=1}^n \Pro(A_{j_k}).\]
A family $(\mathcal{A}_j)_{j\in J}$ of subsets $\mathcal{A}_j\subset \mathcal{F}$ is called \emph{independent} if any family 
$(A_j)_{j\in J}$ with $A_j \in \mathcal{A}_j$ is independent.
\end{definition}

\begin{definition}[Independence of random variables] Random variables $X_1,...,X_n$ are called \emph{independent} if the $\sigma$-algebras $\sigma(X_1),...,\sigma(X_n)$ are independent. \\Equivalently,  the joint distribution of $X_1,...,X_n$ is the product distribution, i.e.\ \[\Pro[X_1\leq x_1 \wedge ... \wedge X_n\leq x_n] = \Pro[X_1\leq x_1] \cdots  \Pro[X_n\leq x_n]\] for all $x_1,...,x_n\in \R$.
An infinite sequence $X_1,X_2,...$ of random variables is called independent if every finite collection $X_{k_1},..., X_{k_n}$, for $n$ distinct indices $k_1,...,k_n$, is independent. 
\end{definition}

If $X_1, X_2,...$ are independent and identically distributed random variables, we will simply call them ``\textit{iid}'' random variables. 

\section{The characteristic function}

\begin{definition}[Characteristic function] Let $X$ be a random variable. The function 
\[\varphi_X(t):\R\to\C, \quad \varphi_X(t)=\mathbb{E}[e^{itX}]\] is called the \emph{characteristic function} of $X$. If $\mu$ is the distribution of $X$, we also write $\varphi_\mu$ instead of $\varphi_X$.
\end{definition}

The characteristic function is also called the \emph{Fourier transform}. In contrast to other transforms like 
\[\text{$\mathbb{E}[e^{-tX}]$ (Laplace transform)\quad or\quad $\mathcal{H}_\mu(t) := 
\lim_{\eps\downarrow 0} \frac {1}{\pi }\int_{|t-x|> \eps}\frac {1}{t-x}\,\mu(dx)$ (Hilbert transform)}
\]
the Fourier transform always exists for all $t\in\R$. We have the following simple properties.

\begin{theorem}Let $X$ be a random variable. Then
 $|\varphi_X(t)|\leq 1$ for all $t\in\R$, $t\mapsto \varphi_X(t)$ is continuous, and $\varphi_{X}(-t)=\overline{\varphi_X(t)}$ for all $t\in\R$.\\
Furthermore, for any $n\in\N$, complex numbers $c_1,...,c_n$ and real numbers $t_1,...,t_n$,
\begin{equation}\label{charara}\sum_{1\leq k,l\leq n}  c_k \overline{c_l} \varphi_X(t_k-t_l)\geq 0.\end{equation}
\end{theorem}
\begin{proof}Let $\mu$ be the distribution of $X$. Clearly, $|\varphi_X(t)|\leq \int_\R |e^{itx}|\mu(dx)=1$ and 
$\varphi_{X}(-t)=\int_\R e^{-itx}\mu(dx)=\int_\R \overline{e^{itx}}\mu(dx)=\overline{\int_\R e^{itx}\mu(dx)}=\overline{\varphi_X(t)}$ for all $t\in\R$. Furthermore 
\[|\varphi_X(t)-\varphi_X(s)|=\left|\int_\R e^{itx}-e^{isx}\mu(dx)\right|\leq \int_\R |e^{i(t-s)x}-1|\mu(dx).\] As $|e^{i(t-s)x}-1|\leq 2$, the dominated convergence theorem implies that the last integral converges to $0$ as $s\to t$ for fixed $t$. (In fact, we even have uniform continuity of $\varphi_X(t)$).\\

Finally, 
\begin{eqnarray*}&&\sum_{1\leq k,l\leq n}  c_k \overline{c_l} \varphi_X(t_k-t_l) = 
\sum_{1\leq k,l\leq n}  c_k \overline{c_l} \mathbb{E}[e^{i(t_k-t_l)X}] \\
&=&  \mathbb{E}\left[\sum_{k=1}^n c_k e^{it_kX}  \cdot \overline{\sum_{l=1}^n c_l e^{it_lX}}\right] =  \mathbb{E}\left[\left|\sum_{k=1}^n c_k e^{it_kX}\right|^2 \right] \geq 0.\end{eqnarray*}

\end{proof}

\begin{remark}Bochner's theorem states that any continuous function $\varphi:\R\to \C$ with $\varphi(0)=1$ and property \eqref{charara} is in fact a characteristic function, see \cite[Section 1.4.3]{Rud62}.
\end{remark}

\begin{example}\label{ex_1}We calculate $\varphi_\mu$ for some distributions $\mu$.
\begin{itemize} 
\item[(a)] If $\mu=\delta_{x_0}$, then $\varphi_\mu(t)=e^{itx_0}$. More generally, if $\mu = p\delta_1 + (1-p)\delta_{-1}$ for some $p\in[0,1]$, then $\varphi_\mu(t)=\cos(t) + i(2p-1)\sin(t)$. The image curve is a (possibly degenerate) ellipse.
\item[(b)] Let $\mu$ be the Poisson distribution $\mu(\{k\})=\frac{\lambda^k e^{-\lambda}}{k!}$, $k\in\N_0$, where $\lambda>0$ is the first moment as well as the variance of $\mu$. We have 
\[ \varphi_{\mu}(t) = e^{-\lambda} \sum_{k=0}^\infty \frac{\lambda^k}{k!} e^{itk}= 
e^{-\lambda} \sum_{k=0}^\infty \frac{(\lambda e^{it})^k}{k!}  = e^{\lambda (e^{it}-1)}. \]
\item[(c)] Let $\mu$ be the uniform distribution on $[a,b]$, $a<b$. Then $\varphi_\mu(0) = 1$ and 
$\varphi_\mu(t) = \frac{e^{itb}-e^{ita}}{it(b-a)}$ for $t\not=0$.
\item[(d)]  If $\mu=\mathcal{N}(0,\sigma^2)$, then 
\[\varphi_\mu(t)= e^{-\sigma^2t^2/2}.\]
This can be shown by looking at $\frac{d}{dt}\varphi_\mu(t)$. Via exchanging differentiation and integration, a calculation shows that 
\[\frac{d}{dt}\varphi_\mu(t)=-\sigma^2 t\varphi_\mu(t), \quad \varphi_\mu(0)=1.\] 
The unique solution to this initial value problem is given by $\varphi_\mu(t)=e^{-\sigma^2 t^2/2}$.
\end{itemize}\hfill $\blacksquare$
\end{example}

The characteristic function really encodes the distribution $\mu$ completely and we can recover $\mu$ from $\varphi_\mu$ by an inversion formula. 

\begin{theorem}Let $\mu\in \mathcal{P}(\R)$.
\begin{itemize}
\item[(a)] We have the inversion formula 
	\[ \mu((a,b))+\frac{1}{2}\mu(\{\alpha\})+\frac{1}{2}\mu(\{\beta\}) = \frac1{2\pi}\lim_{T\to\infty} \int_{-T}^T \frac{e^{-ita}-e^{-itb}}{it}\varphi_\mu(t) dt \]
	for all $a< b$. 
	In particular, if $\varphi_{\mu}=\varphi_{\nu}$ for two probability measures $\mu$ and $\nu$, then $\mu=\nu$.
	\item[(b)](Inverse Fourier transform) Assume that $\int_\R |\varphi_\mu(t)|dt<\infty$.
	Then $\mu$ is absolutely continuous with respect to the Lebesgue measure $\lambda$ and its density $f$ is given by 
	\[ f(x) = \frac1{2\pi}\int_\R  e^{-itx}\varphi_\mu(t)dt, \quad x\in\R. \]
\end{itemize}
\end{theorem}
\begin{proof}${}$
\begin{itemize}
\item[(a)]
As $\int_{a}^b e^{-ity}dy = \frac{e^{-ita}-e^{-itb}}{it}$, the right integral can be written, by Fubini's theorem, as 
\[F(a,b,T)  = \int_{[-T,T]\times [a,b]} e^{-ity}\varphi_\mu(t)dydt.\]
We have, again by Fubini's theorem:
 \begin{eqnarray*}
F(a,b,T) &=& 
\int_{[-T,T]\times [a,b]\times \R}e^{-ity}e^{itx}dydt\mu(dx) = 
\int_\R\left(\int_{[-T,T]} \int_{[a,b]} e^{-it(y-x)}  dydt \right)\mu(dx)\\
&=&
\int_\R\left(\int_{[-T,T]} \frac{1}{it}(e^{-it(a-x)}-e^{-it(b-x)})   dt \right)\mu(dx).
\end{eqnarray*}
Put $f(a,b,T,x)=\int_{-T}^T  \frac{1}{it}(e^{-it(a-x)}-e^{-it(b-x)})dt$. As sine is an odd and cosine an even function we have 
 \begin{eqnarray*}f(a,b,T,x) &=&  \int_{-T}^T  \frac{\sin(-t(a-x))}{t}- \frac{\sin(-t(b-x))}{t}dt \\&=& 
2\int_{0}^T  \frac{\sin(t(x-a))}{t}- \frac{\sin(t(x-b))}{t}dt.\end{eqnarray*}

Now, \[\lim_{T\to\infty}\int_0^T \frac{\sin(ct)}{t}= 
\begin{cases}
\frac{\pi}{2} & \text{if $c>0$,}\\
0 & \text{if $c=0$,}\\
-\frac{\pi}{2} & \text{if $c<0$.}\\
\end{cases}\]
Consequently, 
\[\lim_{T\to\infty}f(a,b,T,x)= 
\begin{cases}
\pi & \text{if $x=a$ or $x=b$,}\\
0 & \text{if $x<a$ or $x>b$,}\\
2\pi & \text{if $x\in(a,b)$.}\\
\end{cases}\]
As $\sup\{|f(a,b,t,x)| \,|\, x\in\R, T\geq 0\}<\infty$, we can use the dominated convergence theorem to obtain
 \begin{eqnarray*}
\lim_{T\to\infty}\frac1{2\pi}F(a,b,T) &=& 
\lim_{T\to\infty}\frac1{2\pi} \int_\R f(a,b,T,x) \mu(dx) = 
\frac1{2\pi} \int_\R \lim_{T\to\infty} f(a,b,T,x) \mu(dx) = \\
&=&
\mu((a,b)) +\frac1{2}\mu(\{a\})+ \frac1{2}\mu(\{b\}).
\end{eqnarray*}
Assume that $\varphi_\mu = \varphi_\nu$. The set $S$ of all $x$ with $\mu(\{x\})>0$ or $\nu(\{x\})>0$ is at most countably infinite, see Exercise \ref{atmdjfgh}. For all $x_0,x \in\R\setminus S$, $x_0<x$, we have $\mu((x_0,x)) = \nu((x_0,x))$ by the inversion formula. With $x_0\to-\infty$ we get $\mu((-\infty,x)) = \nu((-\infty,x))$. If $x\in S$, then we can write 
$(-\infty,x)=\cup_{n\in\N}(-\infty,x_n)$ for an increasing sequence $(x_n)\subset \R\setminus S$ and we obtain 
\begin{eqnarray*} \mu((-\infty,x)) &=& \mu(\cup_{n\in\N}(-\infty,x_n)) = \lim_{n\to\infty}\mu((-\infty,x_n))\\
&=& \lim_{n\to\infty}\nu((-\infty,x_n))= \nu(\cup_{n\in\N}(-\infty,x_n)) = \nu((-\infty,x)).\end{eqnarray*}
 Now we can show that $\mu(I)=\nu(I)$ for all open and closed intervals $I$ and thus $\mu=\nu$.
\item[(b)] The assumptions allow us to define the function $f(x) = \frac1{2\pi} \int_\R e^{-itx}\varphi_\mu(t)dt$. For $a<b$ we have
 \begin{eqnarray*}\int_a^b f(x)dx &=&
\frac1{2\pi}\int_a^b\int_\R  e^{-itx} \varphi_\mu(t)dtdx=
\frac1{2\pi}\int_\R \varphi_\mu(t) \int_a^b  e^{-itx} dx dt\\
&=& \frac1{2\pi}\int_\R \varphi_\mu(t) \frac{e^{-ita}-e^{-itb}}{it} dt = \frac1{2\pi}\lim_{T\to\infty}\int_{-T}^T \varphi_\mu(t) \frac{e^{-ita}-e^{-itb}}{it} dt\\
&=& \mu((a,b))+\frac1{2}\mu(\{a,b\}).
\end{eqnarray*}
As $\int_a^b f(x)dx$ varies continuously with respect to $a$ and $b$, we have $\mu(\{a\})=\mu(\{b\})=0$. Thus $\mu((a,b))=\int_a^b f(x)dx$ and we concldue that $\mu(A) = \int_A f(x)dx$ for any Borel subset $A\subset \R$. 
\end{itemize} 
\end{proof}

\begin{theorem}\label{inv_moments_classicl}Let $X$ be a random variable with $\mathbb{E}[|X|^n]<\infty$. Then 
$\varphi_X$ is $n$ times differentiable at every $t\in \R$ with 
\[ \frac{d^n}{(dt)^n}\varphi_X(t) = \mathbb{E}[e^{itX}(iX)^n].  \quad \text{In particular,}\quad 
 \mathbb{E}[X^n] = (-i)^n\frac{d^n}{(dt)^n}\varphi_X(0).\]
\end{theorem}
\begin{proof}
The statement holds trivially for $n=0$. So we prove the statement by induction starting at $n=0$. Let $\mu$ be the distribution of $X$. 
Now assume that $ \frac{d^n}{(dt)^n}\varphi_X(t) = \mathbb{E}[e^{itX}(iX)^n]$ holds for some $n\in \N$ and all $t\in\R$, and that 
$\mathbb{E}[|X|^{n+1}]<\infty$. Then 
\begin{eqnarray*}
&&\lim_{h\to0}\frac{\frac{d^n}{(dt)^n}\varphi_X(t+h)-\frac{d^n}{(dt)^n}\varphi_X(t)}{h}=
\lim_{h\to0}(\mathbb{E}[(e^{i(t+h)X}(iX)^n]-\mathbb{E}[e^{itX}(iX)^n])/h \\
&=&\lim_{h\to0}\int_\R \frac{e^{i(t+h)x}(ix)^n - e^{itx}(ix)^n}{h} \mu(dx)=
\int_\R (ix)^n \lim_{h\to0} \frac{e^{i(t+h)x} -  e^{itx}}{h} \mu(dx)\\
&=&\int_\R (ix)^{n+1} e^{itx} \mu(dx) = \mathbb{E}[e^{itX}(iX)^{n+1}].
\end{eqnarray*}
The exchange of the limit and the integral is justified by the estimate\footnote{Note that $|\frac{\sin(h)}{h}|\leq 1$ and $|\frac{\cos(h)-1}{h}|\leq 1$ for all $h\in \R\setminus\{0\}$. Thus $|\frac{\sin(hx)}{h}|\leq |x|$ and $|\frac{\cos(hx)-1}{h}|\leq |x|$ for all $h,x\in \R\setminus\{0\}$. Hence
$\left|\frac{e^{ihx} -  1}{h}\right|\leq \left|\frac{\cos(hx) -  1}{h}\right| + \left|\frac{\sin(hx)}{h}\right|\leq 2|x|$.} 
\[\left|(ix)^n\frac{e^{i(t+h)x} -  e^{itx}}{h}\right| = |x|^n \left|\frac{e^{ihx} -  1}{h}\right|\leq 2|x|^{n+1},\]
the fact that $\int_\R |x|^{n+1} \mu(dx) < \infty$, and by the dominated convergence theorem.\\
\end{proof}

\begin{remark}Many further relations between $\varphi_\mu$ and $\mu$ are known, e.g.\ concerning the regularity of $\varphi_\mu$ at $t=0$, analytic extension to the whole complex plane $\C$, or the limit behavior of $\varphi_\mu(t)$ as $|t|\to \infty$, see \cite{BS00}.
\end{remark}

\begin{theorem}\label{independence_classical}Let $X_1,...,X_n$ be random variables on a common probability space. 
Then $X_1,...,X_n$ are independent if and only if for any bounded measurable functions $f_1,...,f_n:\R\to\C$, we have
\begin{equation}\label{class_inde} \mathbb{E}[f_1(X_1)\cdots f_n(X_n)] =  \mathbb{E}[f_1(X_1)]\cdots  \mathbb{E}[f_n(X_n)]. \end{equation}
If all $X_1,...,X_n$ are bounded, then they are independent if and only if 
\begin{equation}\label{class_inde2} \mathbb{E}[X_1^{k_1}\cdots X_n^{k_n}] =  \mathbb{E}[X_1^{k_1}]\cdots  \mathbb{E}[X_n^{k_n}] \end{equation}
for all $k_1,...,k_n\in \N_0$.
\end{theorem}
\begin{proof}
If $X_1,...,X_n$ are independent, then the joint distribution of $(X_1,...,X_n)$ is the product distribution. 
Let $\mu_1,...,\mu_n$ be the corresponding distributions. Then Fubini's theorem yields
\[ \mathbb{E}[f_1(X_1)\cdots f_n(X_n)]=\int_\R f_1(x_1) \mu_1(dx_1) \cdots \int_\R f_n(x_n) \mu_n(dx_n)= \mathbb{E}[f_1(X_1)]\cdots  \mathbb{E}[f_n(X_n)].\]
If $X_1,...,X_n$ are bounded, then this holds also for $f_k(x)=x^{m_k}$, $m_k\in\N_0$.\\

Now consider arbitrary random variables $X_1,...,X_n$ on a common probability space. The multivariate characteristic function $\varphi_{(X_1,...,X_n)}$ is defined by \[\varphi_{(X_1,...,X_n)}: \R^n\to\C, \quad 
\varphi_{(X_1,...,X_n)}(t_1,...,t_n)=\mathbb{E}[e^{i(t_1X_1+...+t_nX_n)}].\]
 If $Y_1,...,Y_N$ are random variables with $\varphi_{(X_1,...,X_n)}=\varphi_{(Y_1,...,Y_n)}$, then the distribution of $(X_1,...,X_n)$ and $(Y_1,...,Y_n)$ are identical. This can be shown as in the univariate case. The inversion formula for the multivariate characteristic function can be found in \cite[Theorem 10.6.1]{Muk11}.\\
If $X_1,...,X_n$ are independent, then $\varphi_{(X_1,...,X_n)}(t_1,...,t_n) = \varphi_{X_1}(t_1)\cdots\varphi_{X_n}(t_n)$ for all $t_1,...,t_n\in\R$.\\
Now assume that the random variables $X_1,...,X_n$ satisfy \eqref{class_inde}. For $f_k(x)=e^{it_k x}$  we obtain that $\varphi_{(X_1,...,X_n)}(t_1,...,t_n) = 
\varphi_{X_1}(t_1)\cdots\varphi_{X_n}(t_n)$ for all $t_1,...,t_n\in\R$. Hence, $X_1,...,X_n$ are independent. \\

Now assume that $X_1,...,X_n$ are bounded and satisfy \eqref{class_inde2}. We find $M>0$ such that the support of the distribution of $X_k$ is contained in $[-M,M]$ for all $k=1,...,n$. Then $X_1,...,X_n$ satisfy \eqref{class_inde} for any polynomials $f_1,...,f_n:\R\to\C$. By approximation of arbitrary bounded measurable functions on $[-M,M]$ by polynomials, we obtain that \eqref{class_inde} also holds for bounded measurable functions.
\end{proof}

Theorem \ref{independence_classical} is important for the goal of defining independence in noncommutative probability theory. We see that the independence of bounded random variables can be expressed by an algebraic property involving only the expectation $\mathbb{E}$ and products of random variables.

\begin{theorem}\label{subindependence_classical}
Let $X_1,...,X_n$ be random variables on a common probability space. If $X_1,...,X_n$ are independent, then 
\begin{equation}\label{class_subinde} \mathbb{E}[e^{it (X_1+...+X_n)}] =  \mathbb{E}[e^{it X_1}]\cdots  \mathbb{E}[e^{it X_n}] \quad 
\text{for all} \quad t\in\R. \end{equation}
Furthermore, if two square-integrable random variables $X_1,X_2$ satisfy \eqref{class_subinde}, then they are uncorrelated.
\end{theorem}
\begin{proof}
The first statement follows directly from Theorem \ref{independence_classical}, equation \eqref{class_inde}.\\
Let $X_1,X_2$ be square-integrable random variables which satisfy \eqref{class_subinde}. Then Theorem \eqref{inv_moments_classicl} implies 
 \begin{eqnarray*} 
&&\mathbb{E}[X_1^2]+2\mathbb{E}[X_1X_2]+ \mathbb{E}[X_2^2]=
\mathbb{E}[(X_1+X_2)^2] = -\frac{d^2}{(dt)^2}\varphi_{X_1+X_2}(0)\\&=&-\frac{d^2}{(dt)^2}(\varphi_{X_1}\cdot \varphi_{X_2})(0) = -\left(\frac{d^2}{(dt)^2}\varphi_{X_1}(0)+2\frac{d}{dt}\varphi_{X_1}(0)\frac{d}{dt}\varphi_{X_2}(0)+\frac{d^2}{(dt)^2}\varphi_{X_2}(0)\right)\\
&=& \mathbb{E}[X_1^2]+2\mathbb{E}[X_1]\mathbb{E}[X_2]+\mathbb{E}[X_2^2],
\end{eqnarray*}
 and we see that $\mathbb{E}[X_1X_2]=\mathbb{E}[X_1]\mathbb{E}[X_2]$.
\end{proof}

Two random variables satisfying \eqref{class_subinde} are called \emph{subindependent}. Thus we have 
\[\text{$X,Y$ independent $\Longrightarrow$ $X,Y$ subindependent }\]
\[\text{ $\Longrightarrow$ $X,Y$ uncorrelated (if $X,Y$ are square-integrable).}\]

In Exercise \ref{ind_sub_cor}, we see that these implications cannot be reversed.

\begin{definition}If $X$ and $Y$ are subindependent, the distribution $\alpha$ of $X+Y$ only depends on the distributions $\mu$ and $\nu$ of $X$ and $Y$. The probability measure $\mu * \nu :=\alpha$ is called the (classical) additive convolution of $\mu$ and $\nu$.
\end{definition}

Finally, we turn to the weak convergence of probability measures.

\begin{definition}
Let $\mu$ and $\mu_1,\mu_2,... \in \mathcal{P}(\R)$. We say that 
$\mu_n$ converges weakly to $\mu$ if 
\[ \int_\R f(x) \mu_n(dx) \to  \int_\R f(x) \mu(dx) \]
for every bounded and continuous $f:\R\to\C$.
For random variables $X,X_1,X_2,...$, we say that $X_n$ converges in distribution to $X$ if the distribution of $X_n$ converges weakly to the distribution of $X$.
\end{definition}

The random variables do not need to be defined on the same probability space. 
Again, we can use the characteristic function to translate the notion of weak convergence into a very simple condition.

\begin{theorem}[L\'{e}vy's continuity theorem]
Let $\mu, \mu_1,\mu_2,...\in\mathcal{P}(\R)$. Then $\mu_n\to \mu$ weakly if and only if $\varphi_{\mu_n}(t)\to \varphi_\mu(t)$ for all $t\in\R$. 
\end{theorem}
\begin{proof}See \cite[Theorem 4.3]{kall}.
\end{proof}

\begin{remark}
The weak convergence turns $\mathcal{P}(\R)$ into a topological space. The topology is in fact induced by a metric. Define the L\'{e}vy distance for $\mu,\nu\in\mathcal{P}(\R)$ by
 \[d_{\text{L\'{e}vy}}(\mu, \nu)=\inf\{\delta>0\,|\, \mu((-\infty, x-\delta])-\delta \leq \mu((-\infty, x]) \leq \mu((-\infty, x+\delta])+\delta 
\quad \text{for all $x\in\R$}\}.\]
Then $\mu_n \to \mu$ if and only if $d_{\text{L\'{e}vy}}(\mu_n, \mu)\to0$, see \cite[Section 7]{Bil99}.
\end{remark}

\section{Central limit theorem}

Assume that the expectation $c=\mathbb{E}[X]\in\R$ of a random variable exists. The law of large numbers states that the arithmetic average of independent samples of $X$, the sample mean, converges to $c$ as the number of the samples tends to $\infty$. This clarifies in which sense we should \emph{expect} $c$ if $X$ is our model for some random numbers. 

\begin{theorem}[Strong law of large numbers]  If $X_1, X_2, ...$ are \textit{iid} random variables 
with finite mean $c=\mathbb{E}[X_k]$, then $(X_1+...+X_n)/n$ converges with probability 1 to $c$.
\end{theorem}
\begin{proof}
See \cite[Theorem 3.23]{kall}.
\end{proof}

In particular, the distribution $\mu_n$ of $(X_1+...+X_n)/n$ 
converges weakly to the distribution $\delta_c$. This convergence is further refined by the famous central limit theorem.
We first need a small auxiliary lemma.

\begin{lemma}\label{help_clt}Let $z_1,...,z_n,w_1,...,w_n$ be complex numbers with $|z_k|\leq 1, |w_k|\leq 1$ for all $k$. Then 
\[\left|\prod_{k=1}^n z_k -  \prod_{k=1}^n w_k \right| \leq \sum_{k=1}^n |z_k-w_k|.\]
\end{lemma}
\begin{proof}This is proved by induction. For $n=1$ we have equality. So assume the statement holds for some $n\in \N$. 
Then 
 \begin{eqnarray*}&&\left|\prod_{k=1}^{n+1} z_k -  \prod_{k=1}^{n+1} w_k \right| =  \left|(z_{n+1}-w_{n+1})\left(\prod_{k=1}^{n} z_k\right) + w_{n+1} \left(\prod_{k=1}^n z_k -  \prod_{k=1}^n w_k \right) \right|\\
&\leq& |z_{n+1}-w_{n+1}| +  \left|\prod_{k=1}^n z_k -  \prod_{k=1}^n w_k \right| \leq \sum_{k=1}^{n+1} |z_k-w_k|.
 \end{eqnarray*}
\end{proof}

\begin{theorem}[Central limit theorem] 
 If $X_1, X_2, ...$ are \textit{iid} random variables 
with finite mean $c$ and finite positive variance $\sigma^2$, then \[S_n:=(X_1+...+X_n-nc)/(\sigma \sqrt{n})\] converges in distribution to the normal distribution $\mathcal{N}(0,1)$.
\end{theorem}
\begin{proof}
Put $Y_k=(X_k-c)/\sigma$ and let $\varphi$ be the characteristic function of $Y_k$, which does not depend on $k$ as all $Y_k$ are \textit{iid}.
By Theorem \ref{subindependence_classical} we have 
\[\varphi_{S_n}(t) = \varphi(t/\sqrt{n})^n.\]
Fix $t\in\R\setminus\{0\}$. Then $\varphi(t/\sqrt{n})=1-\frac{t^2}{2n}+\Landauo(1/n)$ by Theorem \ref{inv_moments_classicl}. Assume that $n$ is so large that $|\varphi(t/\sqrt{n})|\leq 1$ and $|1-\frac{t^2}{2n}|\leq 1$. Then Lemma \ref{help_clt} implies
 \begin{equation*} \left|  \left(1-\frac{t^2}{2n}+\Landauo(1/n)\right)^n - \left(1-\frac{t^2}{2n}\right)^n  \right| 
\leq n\Landauo(1/n)\to 0
\end{equation*} 
as $n\to \infty.$ As $\left(1-\frac{t^2}{2n}\right)^n\to e^{-t^2/2}$, we have $\varphi_{S_n}(t) \to e^{-t^2/2}$ as $n\to\infty$,
which is the characteristic function of the normal distribution, see Example \ref{ex_1} (d). L\'{e}vy's continuity theorem implies that $S_n$ converges in distribution to $\mathcal{N}(0,1)$.
\end{proof}

\begin{example}\label{ex_classroom}
Imagine there are $N$ lectures being held at a university in $N$ different classrooms. Also, assume that the numbers of students attending these lectures are all positive. (Rumor has it that there have been math professors teaching in front of $0$ students, talking to the blackboard as they would also with audience of positive size.)
We may assume that the height of the students are \textit{iid} random variables. So, if we go to each classroom and calculate the average  height, we obtain $N$ numbers whose histogram will have the shape of a normal distribution by the central limit theorem. 
The following histograms are derived from simulations with $N=100$ and $N=100,000$ classrooms respectively. We put 100 students in each classroom and model the height by independent random variables with a uniform distribution between $1.5$ and $1.7$ (meters).
  \begin{figure}[H]
 \begin{center}
 \includegraphics[width=0.9\textwidth]{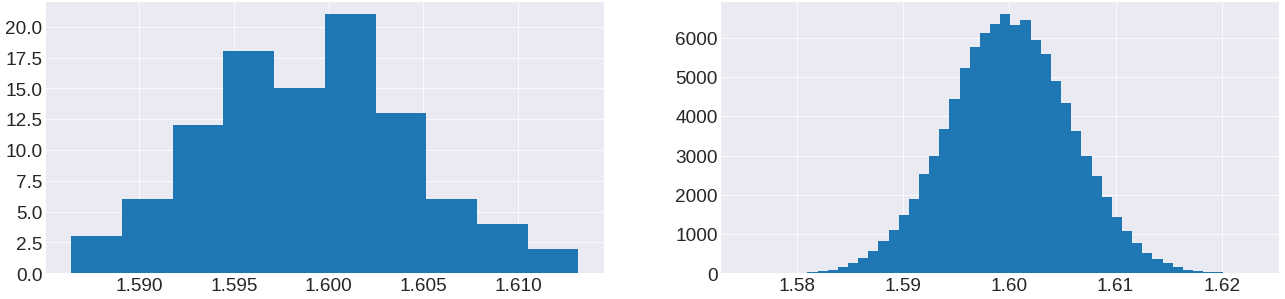}
 \caption{Histograms for the average heights.}
 \end{center}
 \end{figure}
\vspace{-8mm}\hfill $\blacksquare$
\end{example}

The mean of $X$ is $1.6$ and the variance is calculated as $\sigma^2=\frac{(1.7-1.5)^2}{12}$. Then we have $\frac{X_1+...+X_n-n\cdot 1.6}{\sigma\sqrt{n}} \to \mathcal{N}(0,1)$ and thus 
\[
\text{Distribution of } \; \frac{X_1+...+X_n}{n} \approx \mathcal{N}(1.6,\sigma^2 /n).
\]

\textit{The normal distribution everywhere:}\\[-1mm]

The normal distribution appears in many statistical models and real-world samplings. The central limit theorem provides an explanation  if we are dealing with random numbers that arise as the sum of many independent noise terms. \\ 
The assumption in the previous example that the height of a full-grown human is a uniform distribution was artificial. It turns out that the height of (only male or only female) full-grown humans is rather also a normal distribution. One could explain it as follows. 
The human height depends on genetic (e.g.\ male/female) and environmental factors. Assume that the genetic factors don't play any role for a fixed population (let's say all females of a population that has evolved on an isolated island for thousands of years). 
Say an individual of this population is full-grown at $18$ years. Then the size $X$ can be seen as a random variable which is the sum of all 
$18\cdot 12=216$ height gains per month: \[X = G_1 + ... + G_{216}.\] The assumption that all $G_k$'s are \textit{iid} (with finite mean and variance) might appear too strict when we think of the difference between age $3$ months and age $15$ years. But by simplifying things and assuming that they are indeed \textit{iid}, we see why the distribution of $X$ comes close to a normal distribution. \\

\textit{The normal distribution not everywhere:}\\[-1mm]

At the same time, not all random numbers around us are normally distributed. 
Let us go back to the example of measuring student heights in classrooms. Instead of collecting the $N$ arithmetic averages of the classroom heights, we could instead write down all $N$ maximum heights, medians, the 0.75-quantiles, etc.\\
How are these numbers distributed? It is not always possible to give nice analytic characterizations. In case of the maximum (or minimum), there are well-known limit theorems available (\emph{extreme value theory}). Of course, for practical purposes, simulations might be sufficient to gain some reasonable insight. Here is the result of the simulation from Example \ref{ex_classroom} when we replace the average by the maximum height:

 \begin{figure}[h]
 \begin{center}
 \includegraphics[width=0.9\textwidth]{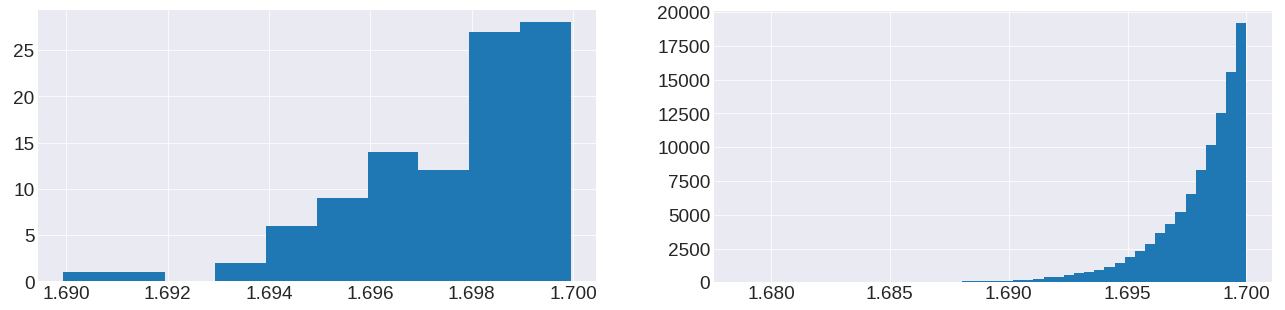}
 \caption{Histograms for the largest heights.}
 \end{center}
 \end{figure}

Apparently, there should be a limit distribution of the maximum with a nice exponential shape. In case of our uniform distribution, this limit can be calculated quite easily. For $a\in \R$ and $b>0$, let $\mu_{Wei}(a,b)$ be the distribution given by the density $be^{b(x-a)}\textbf{1}_{(-\infty,a]}(x)$ (a shifted Weibull distribution).
Then, see Exercise \ref{exercise_classroom}, 
$n\cdot \frac{\max\{X_1,...,X_n\}-1.7}{0.2} \to \mu_{Wei}(0,1)$. Thus
\[
\text{Distribution of } \; \max\{X_1,...,X_n\} \approx  \mu_{Wei}(1.7, n/0.2).
\]

\begin{remark}We end this section by noting that the Poisson distribution $\Pro[X=k]=\frac{\lambda^k e^{-\lambda}}{k!}$, $k\in\N_0$, is also an important limit distribution. It appears as the limit of a binomial distribution as follows (see \cite[Theorem 4.7]{kall} for a more general Poisson limit theorem).\\

For $n\in\N$, let $(X_{k,n})_{1\leq k\leq n}$ be \textit{iid} random variables with $\Pro[X_{k,n} = 1] = 1-\Pro[X_{k,n} = 0]=c_n$ and assume that $nc_n \to c>0$ as $n\to\infty$. 
We can determine the limit distribution of $X_{1,n}+...+X_{n,n}$ as in the proof of the central limit theorem.
We have $\psi_{X_{k,n}}(t)=1-c_n + c_n e^{it}$. Let $\psi_n(t)$ be the characteristic function of $X_{1,n}+...+X_{n,n}$. Then 
\[\psi_n(t) = (1-c_n + c_n e^{it})^n =
 \left(1+\frac{nc_n e^{it}-nc_n}{n}\right)^n = 
\left(1+\frac{c e^{it}-c}{n} + \Landauo(1/n)\right)^n.  \]
By Lemma \ref{help_clt}, $\lim_{n\to \infty} \psi_n(t) = e^{c(e^{it}-1)}$, which is the characteristic function of the Poisson distribution with mean $c$. L\'{e}vy's continuity theorem implies that $X_{1,n}+...+X_{n,n}$ converges in distribution to the Poisson distribution with mean $c$.
\end{remark}

\section{Conditional expectation}

Let $X$ be a random variable on $P_1=(\Omega, \mathcal{F}, \Pro)$.
If $\mathcal{G}\subset \mathcal{F}$ is a sub-$\sigma$-algebra, we can pass to the probability space $P_2=(\Omega, \mathcal{G}, \Pro)$. We can think of this process as simplifying the model $P_1$, for the information encoded by $\mathcal{F}$ is now reduced to that 
of $\mathcal{G}$. Can we also ``simplify'' $X$ to a $\mathcal{G}$-measurable random variable $\hat{X}$? If $X$ is integrable, we would like to have
\[ \int_A  \hat{X}(\omega) {\rm d}\Pro(\omega) = \int_A X(\omega) {\rm d}\Pro(\omega) \]  
for all $A\in \mathcal{G}$. Indeed, this is possible and 
$\hat{X}=\mathbb{E}[X|\mathcal{G}]$ is called the \emph{conditional expectation} with respect to the $\sigma$-algebra $\mathcal{G}$.

\begin{theorem}\label{condi} There exists an almost surely unique linear operator $\mathbb{E}[\cdot|\mathcal{G}]: L^1(P_1)\to L^1(P_2)$ such that 
\begin{equation}\label{cond_exp} \mathbb{E}[\mathbb{E}[X|\mathcal{G}] \cdot \textbf{1}_A] =\mathbb{E}[X \cdot \textbf{1}_A] 
\end{equation}
for all $X\in L^1(P_1)$ and $A\in \mathcal{G}$.
Furthermore, the following properties hold:
\begin{itemize}
	\item[(a)] ($L^1$-contractivity) $\mathbb{E}[|\mathbb{E}[X|\mathcal{G}]|] \leq \mathbb{E}[|X|]$ for all $X\in L^1(P_1)$.
	\item[(b)] (Positivity) If $X\in L^1(P_1)$ and $X\geq 0$, then $\mathbb{E}[X|\mathcal{G}]\geq 0$ a.s.
	\item[(c)] ($L^1(P_2)$-linearity) If $X\in L^1(P_1)$ and $Y\in L^1(P_2)$, then
	$\mathbb{E}[XY|\mathcal{G}]=	Y\cdot \mathbb{E}[X|\mathcal{G}]$	a.s.
\end{itemize}
\end{theorem}
\begin{proof}
First, assume that $X\in L^2(P_1)$. Then $L^2(P_2)$ is a closed linear subspace of $L^2(P_1)$ and we can define 
$\mathbb{E}[X|\mathcal{G}]$ as the projection of $X$ to $L^2(P_2)$. Then $\mathbb{E}[X|\mathcal{G}]$ 
is a $\mathcal{G}$-measurable random variable on $\Omega$, uniquely defined almost everywhere.\\
We have 
$\left<\mathbb{E}[X|\mathcal{G}], f\right> = \left<X, f\right>$ 
for all $f\in L^2(P_2)$, which implies \eqref{cond_exp}.\\ 
For $A=\{\omega\in\Omega\,|\, \mathbb{E}[X|\mathcal{G}](\omega)\geq0\}$, we get
\[  \mathbb{E}[|\mathbb{E}[X|\mathcal{G}]|] = 
 \mathbb{E}[\mathbb{E}[X|\mathcal{G}]\textbf{1}_A]-
 \mathbb{E}[\mathbb{E}[X|\mathcal{G}]\textbf{1}_{\Omega\setminus A}]  = 
\mathbb{E}[X\textbf{1}_A]-
 \mathbb{E}[X\textbf{1}_{\Omega\setminus A}]\leq \mathbb{E}[|X|],
 \]
which is property (a). It follows that the mapping 
$\mathbb{E}[\cdot|\mathcal{G}]$ is uniformly $L^1(P_1)$-continuous on $L^2(P_1)$. As $L^2(P_1)$ is dense in $L^1(P_1)$, we can extend 
$\mathbb{E}[\cdot|\mathcal{G}]$ uniquely to a linear and continuous mapping on $L^1(P_1)$ and (a) holds for all $X\in L^1(P_1)$.\\
Assume that $X\in L^1(P_1)$ with $X\geq 0$. Put $A=\{\omega\in\Omega\,|\, \mathbb{E}[X|\mathcal{G}](\omega)<0\}$. 
Then 
\[  \mathbb{E}[\mathbb{E}[X|\mathcal{G}] \cdot \textbf{1}_A] =\mathbb{E}[X \cdot \textbf{1}_A] \geq  0. \]
This implies $\Pro(A)=0$ and thus $\mathbb{E}[X|\mathcal{G}]\geq0$ a.s., which proves (b).\\
Finally, let $X\in L^2(P_1)$ and $Y\in L^2(P_2)$. We have 
$\left<\mathbb{E}[XY|\mathcal{G}], f\right> = \left<XY, f\right>$ and also 
\[\left<Y\mathbb{E}[X|\mathcal{G}], f\right> = \left<\mathbb{E}[X|\mathcal{G}], \overline{Y}f\right> =  \left<X, \overline{Y}f\right> = \left<XY, f\right>\] for all $f\in L^2(P_2)$. This implies that 
$\mathbb{E}[XY|\mathcal{G}]=Y\mathbb{E}[X|\mathcal{G}]$ a.s. The general case $X\in L^1(P_1)$, $Y\in L^1(P_2)$ follows by approximation.
\end{proof}

\begin{example}The extreme cases are $\mathcal{G}=\mathcal{F}$ and $\mathcal{G}=\{\emptyset, \Omega\}$ for which we obtain 
\[ \text{$\mathbb{E}[X|\mathcal{F}]=X$ a.s. \quad and \quad  
$\mathbb{E}[X|\{\emptyset, \Omega\}]=\mathbb{E}[X]$ a.s.}\]\hfill $\blacksquare$
\end{example}

\begin{example}\label{cond_B}Let $B\in \mathcal{F}$ with $p:=\Pro(B)\in(0,1)$. Another simple case is $\mathcal{G}=\sigma(B)=\{\emptyset, B, \Omega\setminus B, \Omega\}$. We have 
\[ \text{$\mathbb{E}[X|\mathcal{G}]=p^{-1}\mathbb{E}[X\textbf{1}_B]\textbf{1}_B + (1-p)^{-1}\mathbb{E}[X\textbf{1}_{\Omega\setminus B}]\textbf{1}_{\Omega\setminus B}$ \quad a.s.}\]
Note again the similarity to the projection in vector spaces: $\textbf{1}_B/\sqrt{p} =\textbf{1}_B/\|\textbf{1}_B\| :=v$ and  $\textbf{1}_{\Omega\setminus B}/\sqrt{1-p} =\textbf{1}_{\Omega\setminus B}/\|\textbf{1}_{\Omega\setminus B}\| :=w$ are vectors of norm $1$ and 
\[\text{$\mathbb{E}[X|\mathcal{G}]=\mathbb{E}[Xv]v + \mathbb{E}[Xw]w = 
\left< X, w\right>v + \left< X, w\right>w $  \quad a.s.}\]\hfill $\blacksquare$
\end{example}

If $X,Y$ are two random variables with $X\in L^1$, then we define 
\[\mathbb{E}[X|Y]=\mathbb{E}[X|\sigma(Y)].\] 

More generally, if $(Y_j)_{j\in J}$ is a family of random variables, then we define 
\[\mathbb{E}[X|(Y_j)_{j\in J}]=\mathbb{E}[X|\sigma((Y_j)_{j\in J})].\] 

The conditional probability of an event $A\in \mathcal{F}$, given a sub-$\sigma$-algebra $\mathcal{G}\subset \mathcal{F}$, is defined as 
\[\Pro[A|\mathcal{G}]=\mathbb{E}[\textbf{1}_A|\mathcal{G}].\]  
Then $\Pro[A|\mathcal{G}]$ is a random variable with $0\leq \Pro[A|\mathcal{G}] \leq 1$ a.s.\ and, for all $B\in \mathcal{G}$, 
\[ \mathbb{E}[\Pro[A|\mathcal{G}]\cdot \textbf{1}_B] = 
\mathbb{E}[\textbf{1}_A\cdot \textbf{1}_B] = \Pro(A\cap B).\] 

\begin{example}For $X=\textbf{1}_A$ in Example \ref{cond_B} we obtain 
\[ \text{$\Pro[A|\mathcal{G}]=
\frac{\mathbb{E}[\textbf{1}_A\textbf{1}_B]\textbf{1}_B}{\Pro(B)} + 
\frac{\mathbb{E}[\textbf{1}_A\textbf{1}_{\Omega\setminus B}]\textbf{1}_{\Omega\setminus B}}{\Pro(\Omega\setminus B)}=
\frac{\Pro(A\cap B)\textbf{1}_B}{\Pro(B)} + 
\frac{\Pro(A\cap (\Omega\setminus B))\textbf{1}_{\Omega\setminus B}}{\Pro(\Omega\setminus B)}$ \quad a.s.}\]\hfill $\blacksquare$
\end{example}

For two events $A,B$ with $\Pro(B)>0$, \[\Pro(A|B):=\frac{\Pro(A\cap B)}{\Pro(B)}\]
is called the \emph{conditional probability} of $A$ given $B$.

\begin{definition}Let $\mathcal{F}_1, ..., \mathcal{F}_n, \mathcal{G}$ be sub-$\sigma$-algebras of $\mathcal{F}$. Then $\mathcal{F}_1, ..., \mathcal{F}_n$ are called 
\emph{conditionally independent given $\mathcal{G}$} if 
\[ \Pro[\cap_{k=1}^n A_k|\mathcal{G}] = \prod_{k=1}^n 
\Pro[A_k|\mathcal{G}]  \quad \text{a.s.}, \quad A_k\in \mathcal{F}_k.\]
\end{definition}

\begin{example}Let $A_1,A_2\in\mathcal{F}$ and  
$\mathcal{F}_1=\sigma(A_1), \mathcal{F}_2=\sigma(A_2)$. We choose $\mathcal{G}$ as in Example \ref{cond_B}. Then 
 \begin{eqnarray*}\Pro[A_1\cap A_2|\mathcal{G}]&=&
\frac{\Pro(A_1\cap A_2\cap B)\textbf{1}_B}{\Pro(B)} + 
\frac{\Pro(A_1\cap A_2\cap (\Omega\setminus B))\textbf{1}_{\Omega\setminus B}}{\Pro(\Omega\setminus B)},\\
\Pro[A_1|\mathcal{G}]\cdot \Pro[A_2|\mathcal{G}]&=&
\frac{\Pro(A_1\cap  B)\Pro(A_2\cap  B)\textbf{1}_B}{\Pro(B)^2} + 
\frac{\Pro(A_1\cap (\Omega\setminus B))\Pro(A_2\cap (\Omega\setminus B))\textbf{1}_{\Omega\setminus B}}{\Pro(\Omega\setminus B)^2}, \end{eqnarray*}
and $\mathcal{F}_1$ and $\mathcal{F}_2$ are conditionally independent given $\mathcal{G}$ if and only if 
\[ \Pro(A_1\cap A_2|B) = \Pro(A_1|B) \cdot \Pro(A_2|B). \]\hfill $\blacksquare$
\end{example}

\newpage

\section{Exercises}

\begin{exercise}\label{atmdjfgh}Let $\mu$ be a probability measure on $\R$. Show that $S=\{x\in\R \,|\, \mu(\{x\})>0\}$ is at most countably infinite.
\end{exercise}

\begin{exercise}\label{umdhg}Let $X$ be a random variable with values in $\N_0$ and $\mathbb{E}[X]<\infty$. Show that 
\[\mathbb{E}[X] = \sum_{n=0}^\infty\Pro(X>n).\]
\end{exercise}

\begin{exercise}Let $\mu\in \mathcal{P}(\R)$ and $a\in \R$. Show:
\[ \mu(\{a\}) = \lim_{T\to\infty} \frac1{2T}\int_{-T}^T e^{-iat}\varphi_{\mu}(t)dt. \]
\end{exercise}

\begin{exercise}[Injective and non-injective characteristic functions]${}$
\begin{itemize}
	\item[(a)] Show: If $\mu$ is symmetric, i.e.\ $\mu(A)=\mu(-A)$ for every Borel subset $A\subset\R$, then $t\mapsto \varphi_\mu(t)$ is not injective.
	\item[(b)] Let $\lambda>0$ and let $\mu$ be the exponential distribution defined by the density $\lambda e^{-\lambda x}$, $x\geq0$. Show that 
$\varphi_\mu$ is an injective function. How does the image $\varphi_\mu(\R)$ look like?
\end{itemize}
\end{exercise}

\begin{exercise}\label{cdwhjkdsffg}Compute the characteristic function of the random variable $X$.
\begin{itemize}
	\item[(a)] $X$ has a (centered) Cauchy distribution given by the density 
	$\frac1{\pi} \frac{\gamma}{x^2+\gamma^2}$ with scale $\gamma>0$. \\
	(Hint: consider a complex integral along the path 
	$\Gamma_R=[-R,R] \cup [R,R+iR]\cup [R+iR,-R+iR]\cup [-R+iR, -R].$)
	\item[(b)] Let $X_1,X_2,...$ be \textit{iid} random variables with $\Pro[X_1=-1]=\Pro[X_1=1]=\frac1{2}$ and put $X=\frac1{2}+\sum_{k=1}^\infty \frac{X_k}{3^k}$. 
\end{itemize} 
\end{exercise}

\begin{exercise}\label{var_adds}Let $X,Y$ be independent, square-integrable random variables.\\
 Show that $Var(X+Y) = Var(X) + Var(Y)$. 
\end{exercise}

\begin{exercise}\label{ind_sub_cor}Recall Theorem \ref{subindependence_classical}.
\begin{itemize}
	\item[(a)] Construct random variables $X,Y$ which are subindependent but not independent. 
	\item[(b)] Construct random variables $X,Y$ which are uncorrelated but not subindependent.
\end{itemize} 
\end{exercise}

\begin{exercise}Let $\mu$ and $\nu$ be probability measures. 
Show that Parseval's identity holds:
\[\int_\R e^{-its}\varphi_{\mu}(t) \nu(dt) = \int_\R \varphi_{\nu}(t-s) \mu(dt) \quad \text{for all $s\in\R$}.\]

\end{exercise}

\begin{exercise}Consider the metric space $M=(\mathcal{P}(\R),d_{\text{L\'{e}vy}})$.
\begin{itemize}
	\item[(a)] Is $M$ sequentially compact? (Does every sequence $(\mu_n)_{n\in\N}\subset\mathcal{P}(\R)$ have a convergent subsequence?)
	\item[(b)] Is $M$ connected?
\end{itemize}
\end{exercise}

\begin{exercise}\label{exercise_classroom}Consider \textit{iid} random variables $X_1,X_2,...$ with a uniform distribution on $[-1,0]$. Calculate the limit distribution of 
$Y_n=\max\{X_1,...,X_n\}\cdot n$ as $n\to\infty$.
\end{exercise}

\newpage

\chapter{A crash course on Markov processes}

A stochastic process is simply a family $(X_t)_{t\in T}$ of random variables on a common probability space $(\Omega, \mathcal{F}, \mathbb{P})$  for a non-empty index set $T\subset \R$. It can also be seen as a random function via the \emph{sample paths}
\[ t\mapsto X_t(\omega).\]

For $t\in T$, we denote by $\sigma((X_s)_{s\leq t})$ the $\sigma$-algebra generated by the set $\cup_{s\in T,  s\leq t} \sigma(X_s)$, which encodes all information described by the stochastic process up to time $t$.

\begin{remark}
A stochastic process $(X_t)_{t\in T}$ often comes together with a filtration $(\mathcal{F}_t)_{t\in T}$, which is a family of $\sigma$-subalgebras of $\mathcal{F}$ such that $\mathcal{F}_s\subset \mathcal{F}_t$ whenever $s\leq t$. It describes an increasing amount or history of information. $(X_t)$ is called \emph{adapted} to $(\mathcal{F}_t)$ if $X_t$ is $\mathcal{F}_t$-measurable for every $t\in T$, which means that $X_t$ cannot see into the future of our available information. Every stochastic process is adapted to its natural filtration $\mathcal{F}_t=\sigma((X_s)_{s\leq t})$. For our purposes, it will be enough to consider this filtration only.
\end{remark}

The distributions of all $X_t$ might be interdependent as complicated as one might wish. If we require that all $X_t$ are independent, we end up with processes that are much too simple. So another property is needed to define a class of stochastic processes which are both tameable and interesting enough. Markov processes turn out to have that dream property.

\section{Markov Processes}

\begin{definition}Let $T=[0,\infty)$ or $T=\N_0$ and let $S\subset \R$ be a non-empty Borel subset, the state space. An $S$-valued stochastic process $(X_t)_{t\in T}$ on $(\Omega, \mathcal{F}, \mathbb{P})$ 
is called a \emph{Markov process} if, for all $s,t\in T$ with $s\leq t$, $\sigma((X_\tau)_{\tau\leq s})$ and $\sigma(X_t)$ are conditionally independent given $X_s$.\\ 
In the case $T=\N_0$, a Markov process is also called a \emph{Markov chain}.
\end{definition}

\begin{remark}The conditional independence in the definition of the Markov property can also be stated as follows, see e.g.\ \cite[Proposition 2.3]{CD17}:\\
For all $s,t\in T$ with $s\leq t$ and every bounded and Borel measurable $f:\R\to\C$, we have
\begin{equation}\label{uaa} \mathbb{E}[f(X_{t})|(X_\tau)_{\tau\leq s})] = \mathbb{E}[f(X_{t})|X_s] \quad a.s. \end{equation}
\end{remark} 

A Markov process is a stochastic process where, given the present state, the future is independent of the past.
The expectation of some property of $X_t$, i.e.\ $f(X_t)$, conditioned on the whole history of the process up to time $s$ is equal to the expectation conditioned on knowing the process only at the time $s$.

\begin{definition}
A \emph{probability kernel} $k$ on $(S, \mathcal{B}(S))$ is a map $k:S\times\mathcal{B}(S)\to[0,1]$ such that
\begin{itemize}
\item[(i)]
$B \mapsto k(x,B)$ is a probability measure for each $x\in S$,
\item[(ii)]
$x\mapsto k(x,B)$ is a measurable function for each $B\in\mathcal{B}(S)$.
\end{itemize}
For two probability kernels $k_1$ and $k_2$ we can define its composition
$$
(k\star k_2)(x, B) = \int_S k_1(x,{\rm d}y) k_2(y,B)\qquad \mbox{ for }x\in S, B\in\mathcal{B}(S).
$$
\end{definition}

Two $S$-valued random variables $X,Y$ produce a kernel $k$ such that $\Pro[Y\in B| X] = k(X, B)$ almost surely, see \cite[Theorem 5.3]{kall}.\\
A Markov process thus produces a family $k_{s,t}$ of probability kernels, called \emph{transition kernels}, where $s,t\in T$ with $s\leq t$, such that 
\[k_{s,t}(X_s,B) = \mathbb{P}[X_t \in B|X_s]=\mathbb{P}[X_t \in B|\{X_\tau\,|\, \tau\leq s\}]\]
 almost surely, $B\subset \R$. 

\begin{lemma}[Chapman-Kolmogorov relation] Let $s,t,u\in T$ with $s\leq t\leq u$. Then 
\begin{equation}\label{CKR}k_{s,u} = k_{s,t} \star k_{t,u}.\end{equation}
\end{lemma}
\begin{proof}
See \cite[Corollary 7.3]{kall}.
\end{proof}

\begin{example}\label{finite_MP} The simplest example of a Markov process is the case $T=\N_0$ and the state space $S$ is finite: $S=\{x_1,...,x_N\}$. The transition kernel can now be represented by a transition matrix. 
 For $s,t \in \N_0$ with $s \leq t$, we define the $N\times N$-matrix $P_{s,t}$ by 
\begin{equation*}
P_{s,t} = (p_{j,k,s,t})_{1\leq j,k \leq N}\quad \text{with} \quad p_{j,k,s,t} = \Pro[X_t = x_j|X_s= x_k] \quad \text{a.s.}\end{equation*}
(If $\Pro[X_s= x_k]=0$, we let $\Pro[X_t = x_j|X_s= x_k]$ be arbitrary probabilities that sum up to $1$, such that $P_{s,t}$ is a kernel.) The product $\star$ now simply becomes the matrix product and the 
Chapman Kolmogorov relation reads as 
\[P_{s,u} = P_{t,u} \cdot P_{s,t}.\]\hfill $\blacksquare$
\end{example}

On the one hand, the Chapman Kolmogorov relation is simply a consistency condition for the transition kernels of a Markov process. On the other hand, a family of kernels satisfying this relation, together with an initial distribution, already determine a Markov process completely.

\begin{theorem}\label{existence_Markov}
Let $T=[0,\infty)$ or $T=\N_0$. Let $S\subset \R$ be a non-empty Borel subset and let $\nu$ be a probability measure on $S$. Furthermore, let $k_{s,t}$ be a family of transition kernels on $S$, $s,t\in T$ with $s\leq t$, which satisfies \eqref{CKR}. Then there exists a Markov process 
$(X_t)_{t\in T}$ on $S$ with transition kernels $k_{s,t}$ and initial distribution $\nu$, i.e.\ $\nu$ is the distribution of $X_0$.
\end{theorem} 
\begin{proof}See \cite[Theorem 7.4]{kall}.
\end{proof}

\section{Time-homogeneous Markov processes}

\begin{definition} A Markov process with transition kernels $k_{s,t}$ is called \emph{time-homogeneous} if $k_{s,t}=k_{0,t-s}$ for all $0\leq s\leq t \in T$.
\end{definition}

A time-homogeneous Markov chain, $T=\N_0$, is uniquely determined by the initial distribution $\nu$ and by $k_{0,1}$, as $k_{s,t}=k_{0,t-s}=k_{0,1}\star ... \star k_{0,1}$.

\begin{example}[Random walk on $\Z$] The (Bernoulli) random walk on $\Z$ is the Markov chain $(X_n)_{n\in\N_0}$ on $S=\Z$ with initial distribution $\nu = \delta_0$ and $k_{0,1}$ is given by $k_{0,1}(m,\{m+1\})=k_{0,1}(m,\{m-1\})=\frac1{2}$, i.e.\
 $\Pro[X_{n+1}=m+1|X_n = m]=\Pro[X_{n+1}=m-1|X_n = m]=\frac1{2}$ a.s.\hfill $\blacksquare$
\end{example}

For the rest of this section we consider a time-homogeneous Markov chain $(X_n)_{n\in\N_0}$ on a finite state space $S=\{x_1,...,x_N\}$. We will represent a distribution $\mu$ on $S$ simply as the vector $(\mu\{x_1\},...,\mu(\{x_N\}))^T$.\\

 With the notation of Example \ref{finite_MP}, we then have $P_{s,t}=P_{0,t-s}=P^{t-s}$ with $P=P_{0,1}$.
Thus the Markov chain is uniquely determined by its initial distribution $v\in\R^N$ and the \emph{transition matrix} $P=(P_{j,k})_{1\leq j,k\leq N}$. The distribution of $X_n$ is thus given by 
$P^n v$.

\begin{example}Consider the finite state space $S = \{1,2\}$ with initial distribution $\delta_1$ and the transition matrices
 \[ P_1 = \begin{pmatrix}
	0 & 1 \\ 1 & 0\end{pmatrix},\quad P_2 = \begin{pmatrix}
	\frac1{2} & 0\\ \frac1{2} & 1\end{pmatrix}, \quad P_3 = \begin{pmatrix}
	\frac1{2} & \frac1{2}\\ \frac1{2} & \frac1{2}\end{pmatrix}.\]
	A Markov process $(X_n)_{n\in\N_0}$ in the first case simply switches between $1$ and $2$, i.e.\ $\Pro[X_n=1]=1$ if $n$ is even and 
	$\Pro[X_n=1]=0$ if $n$ is odd.\\
	In the second case, the state $2$ is ``absorbing'' and we have 	$\Pro[X_n=1] = \frac1{2^n}$ for all $n\in\N_0$.\\
	In the third case, $\Pro[X_n=1] = \Pro[X_n=2]=\frac1{2}$ for all $n\in\N$.\hfill $\blacksquare$
\end{example}

We now consider the question whether $X_n$ converges in distribution as $n\to\infty$, which is equivalent to the existence of the limit 
\[ \lim_{n\to\infty} P^nv. \] 
From the previous example we see that the limit does not exist in general. 

In the following we will abuse notation and write $P^n_{j,k}$ for the $(j,k)$-element of the matrix $P^n$.

\begin{definition}The transition matrix $P$ is called \emph{irreducible} if for all states $x_j,x_k\in S$,  there exists $m\in \N$ such that the $P^m_{j,k}$ is positive, i.e.\ the probability of getting from state $x_k$ to $x_j$ in $m$ steps is positive.\\
A distribution $\pi\in \R^N$ on $S$ is called \emph{stationary} if $P\pi = \pi$. 
\end{definition}

\begin{lemma}Let $h:S \to \R$ such that 
\[ h(x_k) = \sum_{j=1}^N P_{j,k} h(x_j) \quad \text{for all $k=1,...,N$. \quad ($h$ is also called harmonic.)}\]
If $P$ is irreducible, then $h$ is constant.
\end{lemma}
\begin{proof}As $S$ is finite, $h$ attains its maximum at some $x_{k_0}\in S$. Let $x_{j_0}\in S$ be some state with $P_{j_0,k_0}>0$. Assume that $h(x_{j_0})< h(x_{k_0})$. Then 
\begin{eqnarray*}
h(x_{k_0}) = \sum_{j=1}^N P_{j,k_0} h(x_j) 
 < \sum_{j=1}^N P_{j,k_0} h(x_{k_0}) = h(x_{k_0}),
\end{eqnarray*}
a contradiction and thus $h(x_{j_0})= h(x_{k_0})$.\\
Now let $x_{j_0}\in S$ be any state. As $P$ is irreducible, there exists a path $x_{k_0}, x_{j_m},..., x_{j_1}, x_{j_0}$ such that $P_{j_m, k_0}, P_{j_{m-1}, j_m}, ..., P_{j_0,j_1}$ are all positive. We now conclude inductively that $h(x_{k_0})=h(x_{j_m}) = ... = h(x_{j_0})$. Hence, $h$ is constant.
\end{proof}

\begin{lemma}\label{atm_stat}
If $P$ is irreducible, then $P$ has at most one stationary 
distribution $\pi$.
\end{lemma}
\begin{proof}The previous lemma can also be written as follows. 
If $w\in \R^N$ with $w^TP = w^T$, then $w^T=(c,...,c)$ for some $c\in \R$, or $(P^T-I)w = 0$ implies $w=(c,...,c)^T$. Thus the dimension of the kernel of $P^T-I$ is equal to $1$ and its rank is equal to $N-1$. As $P^T-I$ and $(P^T-I)^T=P-I$ have the same rank, we conclude that $P-I$ has rank $N-1$ and the dimension of its kernel is equal to $1$. Hence, if $\pi$ and $\pi'$ are stationary distributions, then $(P-I)\pi = (P-I)\pi'=0$ and either $\pi=c\cdot \pi'$ or $\pi'=c\cdot \pi$ for some $c\in\R$. But this implies $\pi=\pi'$.
\end{proof}

For $x\in S$ we define the random (hitting) time $\tau(x) = \min \{n\geq 0\,|\, X_n = x\}$ and the first return time 
$\tau^+(x) = \min \{n>0 \,|\, X_n = x\}$.\\

If the initial distribution is equal to $\delta_{x}$ for some $x\in S$, we will denote the probability by $\Pro_x$ and the expectation by $\mathbb{E}_x$.

 \begin{lemma}Let $P$ be irreducible and consider the initial distribution $\delta_{x}$ for some $x\in S$. Then $\mathbb{E}_x[\tau^+(y)]$ is finite for every $y\in S$.
\end{lemma}
\begin{proof}For all $s_j,s_k\in S$ there exists $s(j,k)>0$ such that  $P^s_{j,k}>0$. Let $r$ be the maximum of all such $s(j,k)$  and $\eps=\min\{P^{s(j,k)}_{j,k}\,|\, j,k=1,...,N\}$. 
The probability of the Markov chain going to some fixed state $y\in S$  between times $t$ and $t+r$ is at least $\eps$, or
\[ \Pro_x[X_k\not= y \;\text{for all}\; t< k \leq t+r | X_t=y'] \leq 1-\eps \]
a.s.\ for every $t\in\N_0$ and $y, y'\in S$. Now, if $\tau^+(y)>n$, then 
$X_k\not =y$ for all $0<k \leq n$. Hence, for $m\in \N$,
\begin{eqnarray*}
\Pro_x[\tau^+(y)>mr] &=& \Pro_x[X_k\not= y \;\text{for all}\; 0< k \leq mr] \\
&\leq& \Pro_x[X_k\not= y \;\text{for all}\; 0< k \leq (m-1)r]
\cdot (1-\eps)
 \leq \ldots \leq (1-\eps)^m.
\end{eqnarray*}
Now Exercise \ref{umdhg} implies 
\[ \mathbb{E}_x[\tau^+_y] = \sum_{n=0}^\infty \Pro_x[\tau^+_y>n] \leq \sum_{m=0}^\infty r\Pro_x[\tau^+_y>mr]\leq
 r \sum_{m=0}^\infty (1-\eps)^m,\]
which is a convergent sum because $\eps>0$.
\end{proof}

 \begin{theorem}\label{uniq_statuin}
If $P$ is irreducible, then it has a unique stationary distribution \\
$\pi=(\pi(s_1),...,\pi(s_N))^T\in \R^N$ given by 
\[ \pi(x) = \frac1{\mathbb{E}_x[\tau^+(x)]}. \]
\end{theorem}
\begin{proof}Let us start the Markov chain in some state $s\in S$, and let $n(s,y)$ be the random variable ``number of visits to $y\in S$ before returning to $s$'', where $n(s,s)=1$. Then 
\[\tilde{\pi}(s,y) := \mathbb{E}_s[n(s,y)] = \sum_{n=0}^\infty \Pro_s[X_n = y, \tau_s^+ > n],\] where $\tilde{\pi}(s,s)=1$. 
Clearly, $n(s,y)\leq \tau^+_s$ and thus $\tilde{\pi}(s,y)\leq \mathbb{E}_s[\tau^+(s)]$. The previous lemma implies that $\tilde{\pi}(s,y)<\infty$ for all $y\in S$.\\
Since $P$ is irreducible, the probability to visit $y$ at least once before returning to $s$ must be positive and thus $\tilde{\pi}(s,y)>0$.\\
Let $v(s)=(\tilde{\pi}(s,s_1),...,\tilde{\pi}(s,s_N))^T$. 
We now show that  $Pv(s)=v(s)$.\\
We have 
\[ \sum_{k=1}^N \tilde{\pi}(s,s_k) P_{j,k} = 
\sum_{k=1}^N \sum_{n=0}^\infty \Pro_s[X_n = s_k, \tau_s^+ > n] P_{j,k}. \]
The event $\tau_s^+ > n$ is only determined by $X_0, X_1,..., X_n$, and thus it is independent of $X_{n+1}=y$ when conditioned on $X_n=x$, i.e.\
\begin{eqnarray*} &&\Pro_s[X_n=s_k,X_{n+1}=s_j,\tau^+_s > n]=
 \Pro_s[X_n=s_k,\tau^+_s > n]\cdot \Pro_s[X_{n+1}=s_j|X_n=s_k]\\
&=& \Pro_s[X_n=s_k,\tau^+_s > n]\cdot P_{j,k}.
\end{eqnarray*}
Hence,
\begin{eqnarray*}\sum_{k=1}^N \tilde{\pi}(s,s_k) P_{j,k} &=& 
 \sum_{n=0}^\infty \sum_{k=1}^N\Pro_s[X_n = s_k, \tau_s^+ > n] P_{j,k} = \sum_{n=0}^\infty \sum_{k=1}^N \Pro_s[X_n=s_k,X_{n+1}=s_j,\tau^+_s > n]\\
&=& \sum_{n=0}^\infty \Pro_s[X_{n+1}=s_j,\tau^+_s > n] = 
\sum_{n=1}^\infty \Pro_s[X_{n}=s_j,\tau^+_s > n-1].
\end{eqnarray*}
So 
\begin{eqnarray*}
&&\sum_{k=1}^N \tilde{\pi}(s,s_k) P_{j,k} - \tilde{\pi}(s,s_j)  =
\sum_{n=1}^\infty \Pro_s[X_{n}=s_j,\tau^+_s > n-1] - \sum_{n=0}^\infty \Pro_s[X_n = s_j, \tau_s^+ > n]\\
&=&  \sum_{n=1}^\infty \Pro_s[X_{n}=s_j,\tau^+_s = n] - \Pro_s[X_0 = s_j, \tau_s^+ > 0].
\end{eqnarray*}
In case $s_j=s$, this expression becomes $1-1=0$, and otherwise it is $0-0=0$. Hence, we obtain a stationary distribution $(\pi(s_1),...,\pi(s_N))^T\in \R^N$ by 
$\pi(s_j) = v(s)/\sum_{j=1}^N v(s_j)$. We have 
$\sum_{j=1}^N v(s_j)= \mathbb{E}_s[ \sum_{j=1}^N n(s,s_j)] =\mathbb{E}_s[ \tau^+(s)]$. In particular, $\pi(s)=1/\mathbb{E}_s[ \tau^+(s)]$. We obtain a further stationary distribution by choosing another $s\in S$. However, due to Lemma \ref{atm_stat}, there is only one stationary distribution. Hence, we have 
\[ \pi(s)=\frac1{\mathbb{E}_s[ \tau^+(s)]} \quad \text{for all $s\in S$}. \]
\end{proof}

\begin{definition}For $j=1,...,N$ let $\mathcal{T}(s_j)=\{n\geq 1\,|\, P^n_{j,j}>0\}$. The \emph{period} of $s_j$ is defined to be  $\operatorname{gcd} \mathcal{T}(s_j)$, the greatest common divisor of $\mathcal{T}(s_j)$.
\end{definition}

\begin{lemma}If $P$ is irreducible, then $\operatorname{gcd} \mathcal{T}(s_j)=\operatorname{gcd} \mathcal{T}(s_k)$ for all $j,k=1,...,N$.
\end{lemma}
\begin{proof}Fix two states $s_j$ and $s_k$. As $P$ is irreducible, we find $m,n \in \N$ such that $P^m_{j,k}>0$ and $P^n_{k,j}>0$. Let $k=m+n$. Then $k\in \mathcal{T}(s_j)\cap \mathcal{T}(s_k)$ and $\mathcal{T}(s_j)+m\subset \mathcal{T}(s_k)$ and thus 
$\operatorname{gcd} \mathcal{T}(s_k)$ divides all elements of $\mathcal{T}(s_j)$. 
We conclude that $\operatorname{gcd} \mathcal{T}(s_k) \leq \operatorname{gcd} \mathcal{T}(s_j)$. 
In the same way, we also have $\operatorname{gcd} \mathcal{T}(s_j) \leq \operatorname{gcd} \mathcal{T}(s_k)$. Hence $\operatorname{gcd} \mathcal{T}(s_j) = \operatorname{gcd} \mathcal{T}(s_k)$.
\end{proof}

\begin{definition}If $\operatorname{gcd} \mathcal{T}(s_j)=1$ for all $j=1,...,N$, then $P$ is called \emph{aperiodic}. Otherwise, $P$ is called \emph{periodic}.
\end{definition}

\begin{example}
Consider a Markov chain on the state space $\{0,...,n-1\}$, $n\geq 2$, with transition probabilities $\Pro[X_{n+1}=m+1|X_n = m]=\Pro[X_{n+1}=m-1|X_n = m]=\frac1{2}$, where we identify $-1$ with $n-1$ and $n$ with $0$ (a random walk on $\Z_n$). Then the transition matrix $P$ is irreducible and we see that $ \mathcal{T}(s)=\{2,4,6,...\}$ whenever $n$ is even.  In this case $P$ is periodic with period $\operatorname{gcd} \mathcal{T}(s)=2$. If $n$ is odd, then $P$ is aperiodic.\hfill $\blacksquare$
\end{example}

If $P$ is aperiodic, then, for each $j=1,...,N$, we find some $m\in \N$ such that $\mathcal{T}(s_j)$ contains all natural numbers $\geq m$. This follows as $\mathcal{T}(s_j)$ is closed under addition and from a simple number theoretic argument, see \cite[Lemma 1.27]{LPW09}.

\begin{lemma}\label{ianda}If $P$ is irreducible and aperiodic, there exists $n\in\N$ such that $P^n_{j,k}>0$  for all $j,k=1,...,N$.
\end{lemma}
\begin{proof}For $k=1,...,N$ let $t(s_k)\in\N$ be such that $n\geq t(s_k)$ implies $n\in\mathcal{T}(s_k)$. As $P$ is irreducible, there exists $m(j,k)\in\N$ such that $P^m_{j,k}>0$ for every $j=1,...,N$. For $n\geq t(s_k)+m$ we have 
\[P^n_{j,k}= \sum_{p=1}^N P^{m}_{j,p}\cdot P^{n-m}_{p,k}\geq P^{m}_{j,k}\cdot P^{n-m}_{k,k}>0.\]
Thus $P^n_{j,k}>0$ for all $j=1,...,N$ and all $n\geq t'(s_k):=t(s_k)+ \max_{j=1,...,N}m(j,k)$, and finally $P^n_{j,k}>0$ for all $j,k=1,...,N$ and all $n\geq \max_{k=1,...,N}t'(s_k)$.
\end{proof}

Finally we can prove the following convergence result.

\begin{theorem}\label{limit_Markov}Assume that $P$ is irreducible and aperiodic with stationary distribution $\pi\in \R^N$. Then, for any initial distribution $v\in\R^N$, $X_n$ converges in distribution to $\pi$. In other words, 
\[ \lim_{n\to\infty}P^nv \to \pi.\]
\end{theorem}
\begin{proof} The previous lemma implies that there exists $m\in\N$ be such that $P^m_{j,k}>0$  for all $j,k=1,...,N$. Let $\Pi\in \R^{N\times N}$ such that each column is equal to $\pi=(\pi(s_1), ..., \pi(s_N))^T$. 
We need to show that $P^j-\Pi\to 0$ as $j\to\infty$ (with respect to some norm on $\R^{N\times N}$).\\

We find some $\delta \in (0,1)$ such that $P^n_{j,k} \geq \delta \pi(s_j)$ for all $j,k=1,...,N$. Let $\theta=1-\delta$ and let $Q\in\R^{N\times N}$ be defined by the equation $P^m = (1-\theta) \Pi + \theta Q$. By induction we prove that $P^{mn} = (1-\theta^{n}) \Pi + \theta^n Q^n$ for all $n\in\N$. Assume that this is true for some 
$n\in\N$. Then
$$ P^{m(n+1)}=P^m P^{mn}= P^m ((1-\theta^{n}) \Pi + \theta^{n} Q^{n}) =   (1-\theta^{n}) P^m\Pi + (1-\theta) \theta^{n} \Pi Q^{n} +  \theta^{n+1} Q^{n+1}.$$
As $P^m\Pi=\Pi$ and $\Pi Q^{n}=\Pi$ (as the sum of the elements in each column of $Q$ is equal to $1$), we obtain 
$$ P^{m(n+1)}=   (1-\theta^{n}) \Pi + (1-\theta) \theta^{n} \Pi  +  \theta^{n+1} Q^{n+1} = (1-\theta^{n+1}) \Pi  +  \theta^{n+1} Q^{n+1}.$$
Hence 
\[  P^{mn+j} -\Pi= \theta^{n} (P^jQ^n-\Pi). \]
Consider the maximum norm $\|\cdot\|_{max}$ on $\R^{N\times N}$. Then $\|P^jQ^n-\Pi\|_{max} \leq 
 \|P^jQ^n\|_{max} +\|\Pi\|_{max} = 1+1=2$. Thus
\[ \|P^{mn+j} -\Pi\|_{max} = \theta^{n}\|P^jQ^n-\Pi\|_{max} \leq 
2\theta^{n} \to 0.\]
\end{proof}

\section{Space-homogeneous Markov processes}

Now let $S=\Z$ or $S=\R$ (or, more generally, a metric space which is also an abelian group).

\begin{definition} A Markov process $(X_t)_{t\in T}$ on $S$ with transition kernels $k_{s,t}$ is called \emph{space-homogeneous} if 
$k_{s,t}(x,B)=k_{s,t}(0,B-x)$ for all $s,t\in T$ with $s\leq t$ and all $x\in S$, $B\in \mathcal{B}(S)$.
\end{definition}

\begin{definition} A stochastic process $(X_t)_{t\in T}$ has \emph{independent increments} if the $\sigma$-algebras 
$\sigma(X_t-X_s)$ and $\sigma(\{X_\tau\,|\, \tau\leq s\})$ are independent for all $s,t\in T$ with $s\leq t$.
\end{definition}

\begin{remark}The independence of  $\sigma(X_t-X_s)$ and $\sigma(\{X_\tau\,|\, \tau\leq s\})$ is equivalent to
the independence of $X_0, X_{t_1}-X_0, X_{t_1}-X_{t_0}, ..., X_{t_n}-X_{t_{n-1}}$ for any choice of $n\geq 1$ and $0\leq t_1 \leq... \leq t_n$ because
\begin{eqnarray*} \sigma(\{X_\tau\,|\, \tau\leq s\})&=&\sigma(\{(X_0, X_{t_1}, ..., X_{t_n})\,|\, 0\leq t_1 \leq ... t_n \leq s\}) \\
&=& 
\sigma(\{(X_0,X_{t_1}-X_0, ..., X_{t_n}-X_{t_{n-1}})\,|\, 0 \leq t_1\leq ... \leq t_n \leq s\}). 
\end{eqnarray*}
If $X_0=0$ a.s., then $(X_t)_{t\in T}$ has \emph{independent increments}  if and only if the $n$ increments $X_{t_1}-X_0, ..., X_{t_n}-X_{t_{n-1}}$ are independent.
\end{remark}

\begin{theorem}[Proposition 7.5 in \cite{kall}]\label{weneed}A stochastic process $(X_t)_{t\in T}$ is a space-homogeneous Markov process if and only if it has independent increments.
In this case, the transition kernels are given by 
\[ k_{s,t}(x,B) = \Pro[X_t-X_s \in B -x],  \quad x\in S, B\in \mathcal{B}(S),  s,t\in T, s\leq t. \]
\end{theorem}

The most important space-homogeneous Markov process, and the maybe most important stochastic process in all of probability theory, is the Brownian motion.

\begin{definition}[Brownian motion] A stochastic process $(B_t)_{t\geq 0}$ is called a \emph{Brownian motion} if
\begin{itemize}
	\item[(1)] $B_0=0$ a.s.,
	\item[(2)] the distribution of $B_t$ is $\mathcal{N}(0,t)$ for all $t> 0$,
	\item[(3)] the increments $B_{t_1}, B_{t_2}-B_{t_1}, ..., B_{t_n}-B_{t_{n-1}}$ are independent 
 for any choice of $n\geq 1$ and $0\leq t_1 \leq ... \leq t_n$.
	\item[(4)]  the sample paths $t\mapsto B_t$ are continuous a.s.
\end{itemize}
\end{definition}

A Brownian motion can be constructed in several ways. The first proof of its existence is due to N. Wiener, 1923. We follow L\'{e}vy's construction of the Brownian motion (from \cite[Theorem 1.3]{MP10}).

\begin{theorem}The Brownian motion exists.
\end{theorem}
\begin{proof}
We first we construct a Brownian motion on the interval $[0,1]$.\\
 For $n\in \N_0$ let \[D_n = \{k2^{-n}\,|\, 0\leq k\leq 2^n\} \quad \text{and}\quad  D = \cup_{n\in\N}D_n.\]
 Then $D$ is countable and there exists a probability space $(\Omega, \mathcal{F}, \Pro)$ and a collection $(Z_t)_{t\in D}$ of \textit{iid} random variables with distribution $\mathcal{N}(0,1)$. 
First we define random variables $(B_d)_{d\in D}$ by induction with respect to $D_n$. For $n=0$, we let $B_0=0$ and $B_1=Z_1$. 
Now let $n\geq 1$ and $d\in D_n\setminus D_{n-1}$. We let 
\[ B_d = \frac{B_{d-2^{-n}}+B_{d+2^{-n}}}{2} + \frac{Z_d}{2^{(n+1)/2}}. \]
It is easy to verify that the differences $B_d-B_{d-2^{-n}}$, $d\in D_n\setminus\{0\}$, are independent and have normal distribution 
$\mathcal{N}(0,2^{-n})$.\\
Next we define random functions $F_n:[0,1]\to\R$ for each $n\in \N_0$. For $n=0$, let $F_0(0)=0, F_0(1)=B_1=Z_1$, and define $F_0(t)$ on $(0,1)$ by linear interpolation.\\
For $n\geq 1$, let $F_n(t)=0$ for $t\in D_{n-1}$, $F_n(t)=2^{-(n+1)/2}Z_t$ for $t\in D_n\setminus D_{n-1}$ and $F_n(t)$ is defined by linear interpolation for $t \not\in D_n$.\\
Then we have \[B_d = \sum_{k=0}^n F_k(d) = \sum_{k=0}^\infty F_k(d), \quad d\in D_n. \]

Fix $c>1$. As each $Z_d$ is a $\mathcal{N}(0,1)$ random variable, we have $\Pro[|Z_d|\geq c\sqrt{n}]\leq \exp(-c^2n/2)$ for $n\in\N_0$. We conclude
\[ \sum_{n=0}^\infty \Pro[\exists d\in D_n: |Z_d|\geq c\sqrt{n}] \leq \sum_{n=0}^\infty (2^n+1)\exp(-c^2 n/2) < \infty \] 
if $c>\sqrt{2\log 2}$. Fix such a $c$. Then the Borel-Cantelli lemma (\cite[Theorem 4.3]{bill}) implies that there exists a random $N\in \N$ which is a.s.\ finite such that $|Z_d|<c\sqrt{n}$ for all $d\in D_n$ and $n>N$. In particular, \[ \sup_{t\in[0,1]} |F_n(t)| < c\sqrt{n}2^{-(n+1)/2} \]
for all $n>N$ and we see that, almost surely, 
\[B_t:=\lim_{n\to \infty}\sum_{k=0}^n F_k(t),\]
exists with respect to uniform convergence on $[0,1]$. Furthermore, $t\mapsto B_t$ is continuous a.s. The distribution of each $B_d$, $d\in D$, is $\mathcal{N}(0,d)$ and thus, by approximation, $B_t$ is $\mathcal{N}(0,t)$ distributed for all $t\in[0,1]$.
Similarly, one can show that $B$ has independent increments.\\

Finally, we construct a Brownian motion on $[0,\infty)$ by gluing together independent copies of $B|_{[0,1]}$.
\end{proof}

 \begin{figure}[h]
 \begin{center}
 \includegraphics[width=0.5\textwidth]{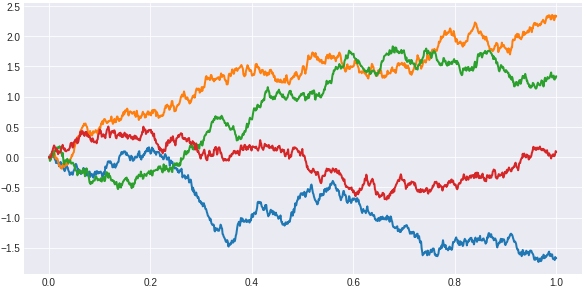}
 \caption{Four sample paths of a Brownian motion.}
 \end{center}
 \end{figure}

\begin{remark}Endow $C([0,1],\R)$ with the topology induced by the maximum norm. The \emph{Wiener measure} is the probability measure on $(C([0,1],\R), \mathcal{B}(C([0,1],\R)))$ induced by a Brownian motion and it can be seen as a normal distribution for functions.\\ Let $X_1,X_2,...$ be a sequence of \textit{iid} random variables with mean $0$ and finite, positive variance $\sigma^2$.
Define the random function $f_n\in C([0,1],\R)$ by 
\[\text{$f_{n}(k/n) = (X_1+...+X_k)/(\sigma \sqrt{n})$ for $k=0,...,n$}\]
 and by linear interpolation between the points $k/n$. \emph{Donsker's Theorem} states that, as $n\to\infty$, the distribution of $f_n$ converges weakly to the Wiener measure, see \cite[Theorem 8.2]{Bil99}.
\end{remark}

\begin{remark}The Brownian motion has many interesting properties, for which we refer to the book \cite{MP10}. For example, for any $t_0\geq0$, $B_t$ is not differentiable at $t_0$ a.s. More precisely, the modulus of continuity of a Brownian motion is given 
by \[ \limsup_{h\downarrow 0} \sup_{0\leq t\leq 1-h}\frac{|B_{t+h}-B_t|}{\sqrt{2h\log(1/h)}}=1, \]
see \cite[Theorem 1.14]{MP10}. 
\end{remark}

\section{Additive processes}

Space-homogeneous Markov processes with a continuity property lead to additive processes and L\'evy processes.

 \begin{definition}\label{addi_def}
A stochastic process $(X_t)_{t\geq 0}$ is called an \emph{additive process} if the following three conditions are satisfied.
\begin{enumerate} 
 \item[(1)] $X_0=0$ almost surely.
 \item[(2)] The increments $X_{t_0}, X_{t_1}-X_{t_0}, ..., X_{t_n}-X_{t_{n-1}}$ are independent 
 for any choice of $n\geq 1$ and $0\leq t_0 \leq t_1 \leq... \leq t_n$.
 \item[(3)] Continuity in probability: for any $\eps>0$ and $s\geq 0$, $\mathbb{P}[|X_{s+t}-X_s|>\eps]\to 0$ as $t\to 0$.
\end{enumerate}
Such a process is called a \emph{L\'evy process} if, in addition, 
 \begin{enumerate}
 \item[(4)] the distribution of $X_{t+s}-X_s$ does not depend on $s$.
\end{enumerate}
 \end{definition}
 
Let us look at some examples of L\'evy processes.

\begin{example}The deterministic process $X_t = ta$, $a\in\R$, is clearly a L\'evy process.\hfill $\blacksquare$
\end{example}

\begin{example}A Brownian motion is clearly a L\'evy process as sample path continuity implies continuity in probability. \\
L\'{e}vy processes are not only space-homogeneous, but also time-homogeneous. Thus the transition kernels satisfy have $k_{s,t}(x,B) = k_{0,t-s}(0,B-x)$. In case of a Brownian motion we have $k_{0,t}(0,B)= \frac1{\sqrt{2\pi t}} \int_B  e^{-y^2/(2t)} {\rm d}y$ for $t>0$ and thus
\[ k_{s,t}(x,B) = \frac1{\sqrt{2\pi (t-s)}} \int_B  e^{-(x-y)^2/(2(t-s))} {\rm d}y, \quad 0\leq s < t.\]\hfill $\blacksquare$
\end{example}

\begin{example}Fix $\lambda>0$ and let $Y_1, Y_2, ...$ be \textit{iid} random variables such that the distribution of 
$Y_k$ is the exponential distribution with parameter $\lambda$. The random variable $T_n=Y_1+...+Y_n$ has a $\Gamma(n,\lambda)$ distribution given by the density 
$\frac{\lambda^n e^{-\lambda x}x^{n-1}}{(n-1)!}$, 
$x\geq 0$.\\
Define the stochastic process $(X_t)_{t\geq 0}$ by 
$X_t = \sum_{n=1}^\infty 1_{[0,t]}(T_n)$. $(X_t)$ is called a \emph{Poisson process} and it can be seen as a counting process for events happening at arrival times $T_n$. If $T_n$ is the time of the $n$th occurrence of an event, then $X_t$ counts the number of such events in the time interval $[0,t]$. With $T_0:=0$, we can also write
\[ X_t = \sup\{n\geq 0\,|\, T_n \leq t\}. \] 

The Poisson process is a L\'{e}vy process with
\begin{equation*}
\Pro[X_t - X_s = n] = \frac{(\lambda(t-s))^n e^{-\lambda(t-s)}}{n!},\quad n\in\N_0,
\end{equation*}
i.e.\ the increments are stationary with a Poisson distribution. 
 \begin{figure}[H]
 \begin{center}
 \includegraphics[width=0.8\textwidth]{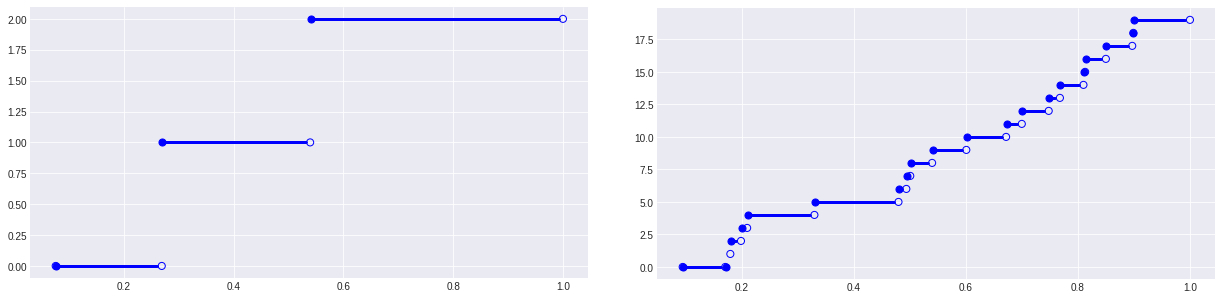}
 \caption{Sample paths of Poisson processes for $t\in[0,1]$  with $\lambda=2$ 
(left) and  $\lambda=20$ (right).}
 \end{center}
 \end{figure}
We can modify the Poisson process to obtain further examples of L\'{e}vy processes. Let $D_1,D_2,...$ be \textit{iid} random variables which are  independent of $(X_t)_{t\geq 0}$. The L\'{e}vy process
\[ Z_t := \sum_{n=1}^{X_t} D_n \]
is called a  \emph{compound Poisson process}. \hfill $\blacksquare$
\end{example}

The L\'{e}vy-It\^{o} decomposition theorem states that 
every L\'{e}vy process can be written as a sum of four elementary processes: a deterministic function $a\cdot t$, a Gaussian process, a compound Poisson process, and a pure jump process. Additive processes can be thought of processes that look locally like L\'{e}vy processes and can be decomposed into the non-stationary variations of these processes, see \cite[Section 19]{Sat99} or \cite[Chapter 13]{bill}.\\

We note the remarkable consequence that a L\'{e}vy process with continuous sample paths must be a Brownian motion plus a drift term.

\begin{theorem}\label{cont_levy}If $X_t$ is a L\'evy process with continuous sample paths, then $X_t = at + c \cdot B_t$ for some $a\in\R$, $c\geq0$, and a Brownian motion $B_t$.
\end{theorem} 

Which distributions can occur in L\'{e}vy processes? This question turns out to have a  nice answer. 

\begin{definition}A probability measure $\mu\in \mathcal{P}(\R)$ is called $*$-infinitely divisible if, for each $n\in\N$, there exists 
$\mu_n\in  \mathcal{P}(\R)$ with $\mu = (\mu_n)^{* n}$.
\end{definition}

The following result characterizes all distributions appearing in additive processes, see 
\cite[Theorems 1.1-1.3]{BNMR01} and \cite[\S 24, Theorem 2]{GK54} for (d).

\begin{theorem}\label{measures_incre_class}
 Let $\mu$ be a probability measure on $\R$. The following statements are equivalent:
 \begin{enumerate}
 \item[(a)] There exists an additive process $(X_t)_{t\geq0}$ 
 such that $\mu$ is the distribution of $X_1$.
 \item[(b)] There exists a L\'evy process $(X_t)_{t\geq0}$ 
 such that $\mu$ is the distribution of $X_1$.
 \item[(c)] $\mu$ is infinitely divisible.
\item[(d)] There exist random variables $(X_{j,n})_{n\in\N,1\leq j \leq k_n}$ with 
$k_n\to\infty$ as $n\to\infty$, $X_{n,1},...,X_{n,k_n}$ are independent, $\lim_{n\to\infty}\max_{j=1,...,k_n}\Pro[|X_{j,n}|\geq \eps]=0$ for every $\eps>0$, such that $\mu$ is the limit of the distribution of $X_{n,1}+...+X_{n,k_n}$.
 \item[(e)] (L\'evy-Khintchine representation) 
 There exist $a\in \R$, $\sigma \geq 0$, and a non-negative measure $\nu$ with $\nu(\{0\})=0$ and 
 $\int_\R (1\wedge t^2) \nu(dt)<\infty$ such that
 \begin{equation} \label{khin}
\varphi_\mu(t) =\exp\left(iat - \frac1{2}\sigma^2t^2 + \int_{\R} \left(e^{it s}-1- its 1_{\{|t|<1\}}
\right)\nu(ds)\right), \qquad t\in \R.
\end{equation}
The data $(a,\sigma, \nu)$ is also called the \emph{L\'{e}vy triple} of the infinitely divisible distribution $\mu$.
\end{enumerate}
\end{theorem}

\section{Further reading}

\begin{itemize}
\item More on the mathematics of Markov processes can be found in the books \cite{bill} and \cite{kall}. For more on the theory of Markov chains, see \cite{LPW09}.
\item 
Markov chains are of paramount importance for stochastic modeling and the interested reader will find an 
abundance of applications of Markov chains, e.g.\ modeling processes in physics, chemistry, biology (\cite{Par08}, \cite{Tam98}), stochastic algorithms (\cite{BGJM11}), models in queuing theory, models for reinforcement learning (\cite{Put05}, \cite{SB18}), \ldots
\item A lovely book on Brownian motion is \cite{MP10}. The Brownian motion is a basic process for building up more complicated processes, e.g.\ via stochastic differential equations, see \cite{Oks03}.
\end{itemize}

\section{Exercises}
\begin{exercise}Consider two time-homogeneous Markov chains $(X_n)_{n\in\N_0}$ on the state space $S=\{s_1,s_2,s_3\}$ with initial distribution $\delta_{s_1}$ and transition matrices
\[ P_1=\begin{pmatrix}
	1/4& 1/4&1/2\\
	1/4& 1/2&1/4\\
	1/2& 1/4& 1/4
\end{pmatrix}, \quad P_2=\begin{pmatrix}
	0& 1/2&0\\
	1& 1/2&0\\
	0&0&1
\end{pmatrix}. \]
Recall that $\tau^+(s) = \min \{n>0 \,|\, X_n = s\}$, $s\in S$.
Compute $\mathbb{E}[\tau^+(s_1)]$ for both cases.
\end{exercise}

\begin{exercise}\label{balanced} Consider a transition matrix $P=(P_{j,k})_{1\leq j,k\leq N}$ of a time-homogeneous Markov chain on a finite state space and assume that the distribution $v=(v_1,...,v_N)^T\in \R^N$ satisfies $P_{j,k}v_k  =P_{k,j}v_j$ for all $j,k$. Show that $v$ is a stationary distribution.
\end{exercise}

\begin{exercise}\label{cov_brown} Let $B_t$ be a Brownian motion. Show that 
$cov(B_t, B_s)=\min(s,t)$ for all $s,t\geq 0$. 
\end{exercise}

\begin{exercise}[Ornstein--Uhlenbeck process] Let $B_t$ be a Brownian motion and consider the process $X_t = e^{-t}B_{e^{2t}}$. Show that the process $(X_t)_{t\geq 0}$ is a Markov process where all $X_t$, $t\geq0$, have the same distribution.
Is $X_t$ space-homogeneous?
\end{exercise}

\begin{exercise}${}$
\begin{itemize}
\item[(a)] Show that $\delta_c$, $c\in\R$, is infinitely divisible.
	\item[(b)] Show that if $\mu$ is infinitely divisible with compact support, then $\mu=\delta_c$, $c\in\R$.
	\item[(c)] Show that a centered Cauchy distribution with scale $\gamma>0$ is infinitely divisible.

\end{itemize}
\end{exercise}
\newpage

\chapter{Quantum probability theory}

\section{The algebraization of probability theory}

The path to quantum probability theory starts with the fundamental observation that working with a probability space $(\Omega, \mathcal{F}, \Pro)$  is equivalent to working with the set of all its bounded complex-valued random variables $\mathcal{A}=\{X:\Omega\to \C\,|\, \text{$X$ measurable and bounded}\}$ together with their expectations, i.e.\ the function $\varphi:\mathcal{A}\to\C$, $\varphi(X)=\mathbb{E}[X]$:
\[ (\Omega, \mathcal{F}, \Pro) \quad \Longleftrightarrow \quad (\mathcal{A}, \varphi). \]
Using the complex numbers instead of $\R$ is not necessary, but it makes life easier. Recall that the characteristic function of a real-valued random variable $X$ is the expectation of the complex-valued random variable $e^{itX}$. \\
The same is true for the boundedness. Of course, we are typically also interested in unbounded random variables $X$, but their expectation might not exist. Instead, $X$ can be replaced by 
the set of all bounded random variables of the form $f(X)$, where $f:\C\to\C$ is bounded and continuous.\\

Notions in $(\Omega, \mathcal{F}, \Pro)$ can be translated into 
notions in $(\mathcal{A}, \varphi)$:

\begin{itemize}
\item An event $A\in \mathcal{F}$ can be represented by $\textbf{1}_A \in \mathcal{A}$ and a $\sigma$-subalgebra $\mathcal{G}\subset \mathcal{F}$ induces a subalgebra $\mathcal{B}\subset \mathcal{A}$ (consisting of all $\mathcal{G}$-measurable elements in $\mathcal{A}$).\\[-6mm]
\item The intersection and union of events corresponds to multiplication and addition:\\[-4mm]
 \[\textbf{1}_A \cdot \textbf{1}_B = \textbf{1}_{A\cap B}, \quad  
\textbf{1}_A + \textbf{1}_B - \textbf{1}_A \cdot \textbf{1}_B = \textbf{1}_{A\cup B}.\]\\[-11mm]
\item Real-valued $X_1,...,X_n\in \mathcal{A}$ are independent if and only if (see Theorem \ref{independence_classical})
\begin{equation}\label{ind_?}
 \mathbb{E}[X_1^{k_1}\cdots X_n^{k_n}] =  \mathbb{E}[X_1^{k_1}]\cdots  \mathbb{E}[X_n^{k_n}] \quad \text{for all $k_1,...,k_n\in \N_0$.}
\end{equation}\\[-15mm]
\end{itemize}

In $\mathcal{A}$, we can add and multiply elements and we have a scalar multiplication, i.e.\ we can form $X+Y, X\cdot Y, \lambda \cdot X$, where $\lambda\in\C$ and $X,Y\in \mathcal{A}$. These operations obey certain rules, in particular commutativity of the product: $X\cdot Y = Y\cdot X$. \\

The quantum probabilist is now confronted with the following task: which algebraic properties does $(\mathcal{A},\varphi)$ have and how can they be formalized and generalized? 
This leads to more abstract probability spaces. In particular, the product will not be necessarily commutative anymore, which is why quantum probability theory is also called ``noncommutative probability theory''.\\ 

Let us consider the space $\C^{n\times n}$ of all $n\times n$-matrices as a possible version of noncommutative complex-valued random variables. Let $A,B\in \C^{n\times n}$. When should we call $A$ and $B$ ``independent''? In different contexts, with might come up with very different notions of independence. We give some possible examples.  

\begin{itemize}
	\item Idea 1: In quantum mechanics, we could call $A$ and $B$ independent if $A$ and $B$ commute, i.e.\ $AB=BA$. 
Thus the commutator $[A,B]=AB-BA$ measures how far $A,B$ are away from being independent. If $A$ is the position operator and $B$ the momentum operator of a particle, then $A$ and $B$ satisfy the ``canonical commutation relation'' $[A,B] =  i \hbar I$, where $I$ is the identity. In this case, $A$ and $B$ are stuck together in some strange sort of dependence, which indeed has strong implications on $A$ and $B$ (e.g.\  Heisenberg's uncertainty relation and that, in fact, both $A$ and $B$ have to be unbounded operators).\\
Note that $AB=BA$ if and only if there is a common eigenbasis for $A$ and $B$, i.e.\ we find an invertible $P$ such that 
both $PAP^{-1}$ and $PBP^{-1}$ are diagonal matrices. Hence, this notion of independence is strongly related to the eigenvectors of $A$ and $B$. 
\item Idea 2: In contrast to eigenvectors, we might concentrate  on the eigenvalues of $A$ and $B$ only, which might have a certain  meaning in an application we are interested in. We might call $A$ and $B$ independent if the eigenvalues of $A$ and $B$ are disjoint, or if the convex hulls of the eigenvalues are disjoint.
\item Idea 3:  Let us define an expectation on $\C^{n\times n}$ by $\varphi(X)=\frac1{n}\tr(X)$, which is the arithmetic average of all eigenvalues of $X$. Consider two diagonal matrices
\begin{equation}\label{two_matrices}A=\begin{pmatrix}\alpha_1 & 0\\ 0 & \alpha_2\end{pmatrix},  \quad B=\begin{pmatrix}\beta_1 & 0\\ 0 & \beta_2\end{pmatrix}, \quad \alpha_1,\alpha_2,\beta_1,\beta_2\in\R.\end{equation}
Are $A$ and $B$ be (sub)independent in the sense of \eqref{ind_?}? 
We have 
\[\frac1{2}\tr(AB)=\frac1{2}(\alpha_1\beta_1+\alpha_2\beta_2), \quad \text{but} \quad 
\frac1{2}\tr(A)\cdot \frac1{2}\tr(B) = \frac1{4}(\alpha_1+\alpha_2)(\beta_1+\beta_2).\]
 Both numbers are equal only in special cases, e.\ g.\ if $\alpha_1=\alpha_2$.\\
 As $\varphi(A^n)=\frac1{2}\alpha_1^n + \frac1{2}\alpha_2^n$ and $\varphi(B^n)=\frac1{2}\beta_1^n + \frac1{2}\beta_2^n$, we might say that the distributions of $A$ and $B$ are $\mu_A = \frac{1}{2}\delta_{\alpha_1}+\frac{1}{2}\delta_{\alpha_2}$, $\mu_B = \frac{1}{2}\delta_{\beta_1}+\frac{1}{2}\delta_{\beta_2}$. 
The convolution of the distributions is given by \begin{equation}\label{teblkkk}\mu_A * \mu_B=\frac1{4}\delta_{\alpha_1+\beta_1}+\frac1{4}\delta_{\alpha_1+\beta_2}+\frac1{4}\delta_{\alpha_2+\beta_1}+\frac1{4}\delta_{\alpha_2+\beta_2}.\end{equation} In general, these are four different eigenvalues, and $A+B$ obviously has at most only $2$ eigenvalues. However, we can force $A$ and $B$ to be subindependent by embedding them into the higher dimensional space $\C^{4\times 4}$. Consider the following two tensor products:
\[  A\otimes I=
\begin{pmatrix}
\alpha_1 & 0&0&0\\ 0 &\alpha_1&0&0 \\
0&0&\alpha_2 & 0\\ 0&0&0& \alpha_2\end{pmatrix},
  \quad 
	  I\otimes B=\begin{pmatrix}
	\beta_1 & 0&0&0\\ 0 &\beta_2&0&0\\
0&0&\beta_1 & 0\\ 0&0&0& \beta_2
	\end{pmatrix}.\]
	We have $\mu_{A\otimes I}=\mu_A, \mu_{I\otimes B}=\mu_B$ and now the convolution $*$ appears: 
	$\mu_{A\otimes I+I\otimes B} = \mu_A * \mu_B$.
\end{itemize}

We will see that in quantum probability theory, independence is defined in the spirit of the last example.

\section{Quantum probability spaces}

We are now heading for the definition of a more abstract probability space with possibly noncommutative random variables.
We start with a complex Hilbert space $H$. Recall that $H$ is a vector space over the field $\C$ together with an inner product $\left<\cdot,\cdot \right>$ (we use inner products which are linear in the second argument: $\left<v,\lambda w\right>=\lambda \left<v, w\right>=\left<\overline{\lambda}v, w\right>$) which is complete with respect to the norm $\|x\| = \sqrt{\left<x,x\right>}$, i.e.\ every Cauchy sequence in $H$ converges to an element in $H$. We let $B(H)$ be the set of all linear and bounded mappings $A:H\to H$ and endow this space with the operator norm defined by 
\[ \|A\| = \sup_{v\in H, \|v\|=1} \|Av\|  = \sup_{v\in H, \|v\|=1}  |\left<v,Av\right>|. \]

The identity operator will be denoted by $I$.
We recall some further facts about operators, see e.g.\ \cite{Con94} for an introduction to functional analysis. 

\begin{itemize}
	\item If $A\in B(H)$, we can define an adjoint $A^*\in B(H)$, which is defined by the property
\[ \left<A^*v,w \right> = \left<v,Aw \right>\]
for all $v,w\in H$. If $A=A^*$, then $A$ is called \emph{self-adjoint}.
\item  We write $A\geq 0$ (and $A>0$) if $\left<v, Av\right> \geq 0$ ($\left<v, Av\right> > 0$) for all $v\in H\setminus \{0\}$. 
Clearly, $AA^*\geq 0$ for all $A\in B(H)$. 
\end{itemize}

\begin{example}If $H=\C^n$ with the usual inner product $\left<w,v\right> = v_1\overline{w_1}+...+v_n\overline{w_n}$, then 
$B(H)=\C^{n\times n}$ is the  space of all complex $n\times n$-matrices. 
For $A\in\C^{n\times n}$, the adjoint $A^*$ is given by the conjugate transpose $A^* = \overline{A}^T$.
The operator norm can also be expressed by 
\[ \|A\|=\sqrt{\lambda_{max}(AA^*)}, \quad A\in \C^{n\times n}, \] 
where $\lambda_{max}$ is the largest eigenvalue of the positive semi-definite 
matrix $AA^*$. \hfill $\blacksquare$
\end{example}

\begin{definition}\label{hjklpppp}Let $H$ be a complex Hilbert space. We define a \emph{quantum probability space} as a pair $(B(H), \varphi)$, where 
 $\varphi$ is an \emph{expectation} on $B(H)$, defined as a linear functional $\varphi:B(H)\to\C$ with 
\[\varphi(A)\geq 0 \quad \text{whenever $A\geq 0$}\quad \text{and} \quad \varphi(I)=1.\]
The self-adjoint elements of $B(H)$ will be called (real) \emph{random variables}. 
\end{definition}

If $X\in B(H)$ is self-adjoint, then there exists a unique Borel probability measure $\mu_X\in \mathcal{P}_c(\R)$ such that 
\begin{equation}\label{fdssamm}\varphi(p(X)) = \int_\R p(x) \mu_X(dx)\end{equation}
for all polynomials $p:\R\to \R$.\footnote{If $p:\R\to\R$, $p(x)=a_0+a_1x+...+a_nx^n$, then $p(X):=a_0 I + a_1 X +...+a_nX^n$.} This follows from the Riesz-Markov theorem. (Alternatively, one can use the spectral theorem.)\\
For $p(x)=x^n$ we see that $\varphi(X^n)$ is equal to the $n$-th moment of the measure $\mu$:
\[\varphi(X^n) = \int_\R x^n \mu_X(dx).\]
Thus $\varphi(X)$ can be seen as the first moment of $X$, $\varphi((X-\varphi(X))^2) = \varphi(X^2)-\varphi(X)^2$ as the variance of $X$, etc.

\begin{definition}The distribution of a random variable $X$ is defined as the unique Borel probability measure $\mu_X$ satisfying \eqref{fdssamm}.
\end{definition}

\begin{example}If $H$ is a Hilbert space and $v\in H$ with $\|v\|=1$, then \[\varphi(X)=\left<v,Xv\right>\] defines an expectation. 
In quantum mechanics, such a vector $v$ is also called  \emph{state} and due to this example, expectations are also called \emph{states} in quantum probability theory.\hfill $\blacksquare$
\end{example}

\begin{example}\label{999}In the case of $H=\C^n$ we have $B(H)=\C^{n\times n}$, the set of all $n\times n$-matrices. 
We obtain an quantum probability space by defining the expectation as
\[ \varphi(A) = \frac1{n} \tr(A). \]
Indeed, $\varphi(I)=1$ and as $\tr(A)$ is the sum of all eigenvalues of $A$, we have $\tr(A)\geq 0$ whenever $A\geq 0$.
Denote by $\lambda_1,...,\lambda_n\in\R$ the eigenvalues of a self-adjoint element $X\in B(H)$.  Then 
$\varphi(p(X)) = \int_\R p(x) \mu_X(dx)$ for any polynomial $p:\R\to\R$, where \[\mu_X  = \frac1{n}\delta_{\lambda_1} + ... + \frac1{n}\delta_{\lambda_n}.\]
	Other examples of expectations can be obtained as follows. Let $\rho\in \C^{n\times n}$ be self-adjoint, non-negative, and $\tr(\rho)=1$, a so called \emph{density matrix}. We now obtain an expectation via
\[\varphi(A) = \tr(A \rho).\]
If $\rho=\frac1{n}I$, then $\tr(A \rho)=\frac1{n} \tr(A)$.
For $v\in \C^n$ with $\|v\|=1$ and $\rho= v\overline{v}^T$, we have $\tr(A \rho)= \tr(A v\overline{v}^T)=\tr(\overline{v}^TA v)=\left<v,Av\right>$.  In fact, all expectations on $\C^{n\times n}$ can be written as $\tr(A \rho)$ for a density matrix $\rho$, see Exercise \ref{density}.\hfill $\blacksquare$
\end{example}

\begin{example}Let $P=(\Omega, \mathcal{F}, \Pro)$ be a classical probability space. Then $H=L^2(\Omega, \mathcal{F}, \Pro)$ is a Hilbert space, see Remark \ref{Hilbert_space}. Let $\textbf{1}\in H$ be the function constant 1 and consider the quantum probability space $(B(H), \varphi)$ with \[\varphi(A) = \left<\textbf{1},A(\textbf{1})\right> = \int_\Omega (A(\textbf{1}))(\omega) {\rm d}\Pro(\omega).\]
Then every bounded classical random variable $X:\Omega \to \R$ can be identified with a random variable $A_X\in B(H)$, namely as the multiplication operator $A_X(v) = X\cdot v$, $v\in H$. The distribution of $A_X$ is equal to the distribution $\mu$ of $X$ as 
\[ \varphi(p(A_X)) = \int_\Omega (p(X))(\omega) {\rm d}\Pro(\omega)= \int_\R p(x) \mu(dx) = \mathbb{E}[p(X)], \]
for any polynomial $p:\R\to\R$.\hfill $\blacksquare$
\end{example}

\begin{example}\label{444} We can combine the two previous examples to obtain random matrices as quantum random variables. 
Let $(\Omega, \mathcal{F}, \Pro)$ be a classical probability space and let $N\in \N$. We take the Hilbert space $H =L^2(\Omega, \mathcal{F}, \Pro) \otimes \C^{N}$.  Let $E_{j,k}\in \C^{N\times N}$ be the matrix with $(j,k)$-entry $1$ and $0$ entries otherwise. We construct an expectation on $B(L^2(\Omega, \mathcal{F}, \Pro) \otimes \C^{N\times N})$ by defining $\varphi$ for $A\otimes E_{j,k}$ as
$\varphi(A\otimes E_{j,k}) = \frac1{N}\int_\Omega (A(\textbf{1}))(\omega) {\rm d}\Pro(\omega)$ if $j=k$ and $0$ otherwise.\\

A bounded random matrix can be seen as a bounded random variable $(X_{j,k})_{1\leq j,k\leq N}:\Omega \to \C^{N\times N}$, which can be identified with the operator $X = \sum_{j,k}  A_{X_{j,k}}\otimes E_{j,k}$ and thus 
\[ \varphi(X) = \sum_{j,k} \varphi(A_{X_{j,k}}\otimes E_{j,k}) = \frac1{N} \sum_{j} \mathbb{E}[X_{j,j}] = 
\frac1{N} \mathbb{E}[\tr(X)]. \]\hfill $\blacksquare$
\end{example}

\section{Independence}

Just as in \eqref{teblkkk}, classical independence can be identified algebraically with the tensor product of quantum probability spaces. This leads to the question whether other products of quantum probability spaces might yield reasonable notions of independence. In the 1980's D. Voiculescu discovered the free independence. His works stimulated a systematic study of quantum probability spaces.\\
 One can single out certain properties of the tensor product (in particular a certain universality property), and define axioms for products that represent an abstract notion of independence. It turns out that there are five notions of independence satisfying these axioms: 
\emph{classical} or \emph{tensor independence}, \emph{Boolean independence}, \emph{free independence},  \emph{monotone} and \emph{anti-monotone independence}, see \cite{Spe97}, \cite{BGS99}, \cite{MR2016316}, and the section ``The Five Universal Independences'' in \cite{barndorff-nielsen+al}. 

\begin{definition}\label{111}Let $(B(H), \varphi)$ be a quantum probability space and let $X_1,...,X_n\in B(H)$ be random variables. 
We now define five different notions of independences for these random variables.. For $j=1,...,n$, let
\[\mathcal{A}_j=\{p(X_j)\,|\, p:\R\to\R\; \text{is a polynomial}\}.\] 
\begin{itemize}
\item[(1)] Tensor independence: $X_1,...,X_n$ are tensor independent if 
\[ \varphi(Y_1\cdots Y_m) = \prod_{j\in \{1,...,n\}} \varphi(\prod_{\substack{k\in\{1,...,m\}\\ Y_k\in \mathcal{A}_j}} Y_k) \]
for all $m\in\N$ and $Y_k\in \mathcal{A}_{j_k}$.
\item[(2)] Free independence:  $X_1,...,X_n$ are freely independent if 
\[\varphi(Y_{1}\cdots Y_{m})=0 \quad \text{whenever} \quad \varphi(Y_1)=...=\varphi(Y_m)=0,\]
for all $m\in\N$ and $Y_k\in \mathcal{A}_{j_k}$, where $j_k\not= j_{k+1}$ for all $k=1,...,m-1$.
\item[(3)] Boolean independence:  $X_1,...,X_n$ are Boolean independent if 
\[\varphi(Y_1\cdots Y_m)=\varphi(Y_1)\cdots \varphi(Y_m)\]
for all $m\in\N$ and $Y_k\in \mathcal{A}_{j_k}$, where $j_k\not= j_{k+1}$ for all $k=1,...,m-1$.
\item[(4)] Monotone independence: The tuple $(X_1,...,X_n)$ is monotonically independent if 
\[\varphi(Y_1 \cdots Y_k \cdots Y_m) = \varphi(Y_k)\varphi(Y_1\cdots Y_{k-1}Y_{k+1}\cdots Y_m) \]
for all $m\in\N$ and $Y_k\in \mathcal{A}_{j_k}$, whenever $j_{k-1} <j_k$ and $j_k >j_{k+1}$ \\(and if $k=1$ or $k=n$, the first or second equality resp. can be ignored).
\item[(5)] Anti-monotone independence: $(X_1,...,X_n)$ is anti-monotonically independent if \\
   $(X_n,...,X_1)$ is monotonically independent.
\end{itemize}

\end{definition}

Note that the tensor independence corresponds to the classical independence as in \eqref{ind_?}. 
Monotone and anti-monotone independence also depend on the order of the random variables. So $(X,Y)$ might be monotonically independent while 
$(Y,X)$ is not.

\begin{example}Let $X$ and $Y$ be two random variables. In all cases of independence of $X$ and $Y$, we have \[\varphi(XY)=\varphi(YX)=\varphi(X)\varphi(Y).\]
 Thus $X$ and $Y$ are uncorrelated. In the case of free independence, note that $\varphi((X-\varphi(X))(Y-\varphi(Y))=0$ and  $\varphi((Y-\varphi(Y))(X-\varphi(X))=0$. \hfill $\blacksquare$
\end{example}

\begin{example}Let $X$ and $Y$ be two random variables. If $(X,Y)$ is monotonically independent we have 
\[\varphi(XYX)=\varphi(Y)\varphi(X^2) \quad \text{but} \quad \varphi(YXY)=\varphi(X)\varphi(Y)^2,\] while the case of Boolean independence  gives
\[\varphi(XYX)=\varphi(Y)\varphi(X)^2, \quad \varphi(YXY)=\varphi(X)\varphi(Y)^2.\]
\hfill $\blacksquare$
\end{example}

\begin{remark}${}$\label{def-mon}In \cite{M01}, Muraki defines monotone independence by the following two stronger conditions:
\begin{itemize}
\item[(i)]
For all $i,j,k\in \{1,...,n\}$ with $i<j>k$ and any $X\in\mathcal{A}_i$, $Y\in\mathcal{A}_j$, $Z\in\mathcal{A}_k$ we have
\[
XYZ = \varphi(Y) XZ.
\]
\item[(ii)] \label{monotone_indep1}
For any $r,s\in\mathbb{N}\cup\{0\}$, $i_1,\ldots,i_r,j,k_1\ldots,k_s\in \{1,...,n\}$ with $i_1>\cdots>i_r>j < k_s<\cdots<k_1$ and
any $X_1\in\mathcal{A}_{i_1},\ldots,X_r\in\mathcal{A}_{i_r}$, $Y\in\mathcal{A}_j$, $Z_1\in\mathcal{A}_{k_1},\ldots,Z_s\in\mathcal{A}_{k_s}$, we have
\[
\varphi(X_1\cdots X_r Y Z_s\cdots Z_1)=\varphi(X_1)\cdots \varphi(X_r)\varphi(Y) \varphi(Z_s)\cdots \varphi(Z_1).
\]
\end{itemize}
The conditions (i) and (ii) imply monotone independence as defined in Definition \ref{111}. \cite[Remark 3.2 (c)]{franz07b} gives a condition under which these definitions are equivalent if $\varphi=\left<v,\cdot v\right>$.
\end{remark}

Similar to the tensor products $A\otimes I$ and $I\otimes B$ of the matrices \eqref{two_matrices}, we can construct models for the five independences.

Let $H_1$ and $H_2$ be Hilbert spaces with expectations $\varphi_1:B(H_1)\to\C$, $\varphi_1(A)=\left<v_1,Av_1\right>$, and $\varphi_2:B(H_2)\to\C$, $\varphi_2(A)=\left<v_2,Av_2\right>$ for some unit vectors $v_1\in H_1$ and $v_2\in H_2$. We can now produce a new probability space. Consider the tensor product $H=H_1\otimes H_2$. The inner product on $H$ is defined via $\left<(a\otimes b), (c\otimes d)\right> =\left<a, c\right> \cdot \left<b,d\right>$, and by linear and continuous extension. 
Let $v = v_1\otimes v_2$ and $\varphi:B(H)\to\C$,  $\varphi(A)=\left<v,Av\right>$.

\begin{theorem}\label{models}
Let $X\in B(H_1)$ and $Y\in B(H_2)$. Denote by $I_j$ the identity on $H_j$ and by $P_j$ the projection in $H_j$ onto the space spanned by $v_j$, $j=1,2$.  Then 
\begin{itemize}
	\item $X\otimes I_2$ and $I_1\otimes Y$ are tensor independent.
	\item $X\otimes P_2$ and $P_1\otimes Y$ are Boolean independent.
	\item $(X\otimes P_2, I_1\otimes Y)$ is monotonically independent.
		\item $(X\otimes I_2, P_1\otimes Y)$ is anti-monotonically independent.
\end{itemize}
\end{theorem}
\begin{proof}We only consider the Boolean case. 
Let $k\in \N_0$. Then $(X\otimes P_2)^{k} = X^k\otimes P_2$ and $(P_1\otimes Y)^{k} = P_1\otimes Y^k$. Instead of arbitrary polynomials $p(X), p(Y)$, it suffices to consider powers of $X$ and $Y$. Thus, for $k_1,...,k_m,l_1,...,l_m\in \N_0$, we obtain
\[ (X^{k_1}\otimes P_2) \cdot (P_1\otimes Y^{l_1})\cdots (X^{k_m}\otimes P_2)\cdot (P_1\otimes Y^{l_m}) = 
 (X^{k_1} P_1\cdots X^{k_m}P_1 \otimes  P_2 Y^{l_1}\cdots P_2 Y^{l_m}).
 \]
Hence,
\begin{eqnarray*} &&\left<(v_1\otimes v_2), (X^{k_1} P_1\cdots X^{k_m}P_1 \otimes  P_2 Y^{l_1}\cdots P_2 Y^{l_m})(v_1\otimes v_2)\right> \\
&=& 
 \left<v_1 , X^{k_1} P_1\cdots X^{k_m}P_1 v_1\right> \cdot 
 \left<v_2,  P_2 Y^{l_1}\cdots P_2 Y^{l_m} v_2\right> \\
&=&
\left<v_1 , X^{k_1}v_1\right> \cdots \left<v_1 , X^{k_m}v_1\right> \cdot 
 \left<v_2, Y^{l_1}v_2\right> \cdots \left<v_2,Y^{l_m} v_2\right>.
\end{eqnarray*}

\end{proof}

\begin{remark}Free independence is modeled in a different way. Let $K_j$ be the orthocomplement of $v_j$ in $H_j$ and define the free product

\[ H_1 * H_2 = \C v \oplus \bigoplus_{n\geq 1} \bigoplus_{j_1\not= j_2 \not= ... \not= j_n} K_{j_1}\otimes \cdots \otimes K_{j_n}, \]
where $\C v$ is a one-dimensional Hilbert space with $\|v\|=1$. Let $X_j\in B(H_j)$. We identify $X_j$ with an operator $\hat{X}_j\in B( H_1 * H_2)$ via $\hat{X}_j(v) = X_j(v_j)$ and 
\[  \hat{X}_j(k_j \otimes k_{j_1}\otimes ... \otimes k_{j_n})=
	X_j(k_j) \otimes k_{j_1}\otimes ... \otimes k_{j_n},\quad 
	 \hat{X}_j(k_{j_1}\otimes ... \otimes k_{j_n})=
	X_j(v_j) \otimes k_{j_1}\otimes ... \otimes k_{j_n},
 \]
where $k_j\in K_j$ and $k_{j_1} \in K_{3-j}$.
Now $\hat{X}_1$ and $\hat{X}_2$ are freely independent in $(B(H_1 * H_2), X\mapsto \left<v, X v \right>)$, see \cite{AN06}.
\end{remark}

\begin{example}Consider the two matrices $A$ and $B$ from \eqref{two_matrices}. We choose $v_1=v_2=(1,0)^T$. Then 
$\mu_{A} = \delta_{\alpha_1}$ and $\mu_{B} = \delta_{\beta_1}$. Let $v=v_1\otimes v_2 = (1,0,0,0)^T$.
 Then the following two matrices are monotonically independent in $(\C^{4\times 4}, \left<v,\cdot v\right>)$:
\[  A\otimes P_2=
\begin{pmatrix}
\alpha_1 & 0&0&0\\ 0 &0&0&0 \\
0&0&\alpha_2 & 0\\ 0&0&0& 0\end{pmatrix},
  \quad 
	  I_1\otimes B=\begin{pmatrix}
	\beta_1 & 0&0&0\\ 0 &\beta_2&0&0\\
0&0&\beta_1 & 0\\ 0&0&0& \beta_2
	\end{pmatrix}.\]
	 Boolean independence is realized by
\[  A\otimes P_2=
\begin{pmatrix}
\alpha_1 & 0&0&0\\ 0 &0&0&0 \\
0&0&\alpha_2 & 0\\ 0&0&0& 0\end{pmatrix},
  \quad 
	  P_1\otimes B=\begin{pmatrix}
	\beta_1 & 0&0&0\\ 0 &\beta_2&0&0\\
0&0&0 & 0\\ 0&0&0& 0
	\end{pmatrix}.\]\hfill $\blacksquare$
\end{example}

\begin{remark}Let us consider the monotone case in Theorem \ref{models} from a quantum mechanical perspective. Assume we make a quantum measurement in $H_2$. Then the state $v_2$ changes to some $w\in H_2$, $\|w\|=1$. Say $\left<w,v_2\right>=0$. (For example $w=(0,1)^T$ in the previous example.) \\
 Then $v=v_1\otimes v_2$ changes to $v'=v_1\otimes w$ ($v'=(0,1,0,0)^T$) and 
 \[\text{$\left<v', (A\otimes P_2)^n v'\right>=0$ for all $n\in\N$.}\] 
Conversely, let us change $v_1$ to some $w\in H_1$, $\|w\|=1$, with $\left<w,v_1\right>=0$. (For example $w=(0,1)^T$ in the previous example.) 
Then $v$ changes to $v'=w\otimes v_2$ ($v'=(0,0,1,0)^T$) and 
\[\text{$\left<v', (I_1\otimes B)^n v'\right>= \left<v_2, B^n v_2\right> $ for all $n\in\N$.}\]
We see that the measurement in $H_2$ ``annihilates'' the observable $A\otimes P_2$, while the measurement in $H_1$ leaves $I_1\otimes B$ unaffected. 
\end{remark}
\newpage
\section{Further reading}

A quantum probability space is often defined as a (unital) $C^*$-algebra $\mathcal{A}$ together with a linear functional $\varphi:\mathcal{A}\to\C$ which is non-negative and has norm $1$, see \cite[Def. 3.7]{NS10}. A $C^*$-algebra is a Banach algebra $\mathcal{A}$ over $\C$ together with a map $x\mapsto x^*$ on $\mathcal{A}$ such that 
\begin{itemize}
	\item[(1)] $(x^*)^* = x$ for all $x\in\mathcal{A}$,
	\item[(2)] $(x+y)^*=x^*+y^*$ and $(xy)^*=y^*x^*$ for all $x,y\in\mathcal{A}$,
	\item[(3)] $(\lambda x)^*=\overline{\lambda}x^*$ for all $\lambda\in \C$ and $x\in\mathcal{A}$,
	\item[(4)] $\|xx^*\|=\|x\|\|x^*\|$ for all $x\in\mathcal{A}$ ($C^*$-property).
\end{itemize}
[The term $C^*$-algebra was introduced by I. E. Segal in 1947 in \cite{seg47}, where ``C'' stands for 
``closed''; \cite[p. 75]{seg47}: ``Let $\mathcal{A}$  be  a  $C^*$-algebra,  by  which  we  mean  a  uniformly  closed,  self-adjoint  algebra  of  bounded  operators  on  a  Hilbert space.''] Our Definition \ref{hjklpppp} yields an example of such abstract probability spaces. However, the Gelfand-Naimark theorem (\cite[Section II.6.4]{Bla06}) states that in fact every $C^*$-algebra is isomorphic to a $C^*$-algebra consisting of bounded operators on a Hilbert space.\\

One could also make an even more abstract definition by dropping the Banach space structure, see \cite[Def. 1.1]{NS10}, or by considering  von Neumann algebras instead, see \cite{Voi97, Att, one, barndorff-nielsen+al}. \\
Furthermore, one can generalize the $\C$-valued expectation to an expectation $\varphi:A\to B$, where both $A$ and $B$ are $C^*$-algebras. The Cauchy transforms are now ``noncommutative holomorphic functions'' living on the (matricial) upper half-plane of $B$, see \cite{KVV14} for the theory of these functions.

\section{Exercises}

\begin{exercise}\label{csu}
Let $A, B\in B(H)$. Prove the Cauchy-Schwarz inequality
\[ |\varphi(A^*B)|^2  \leq \varphi(B^*B)\varphi(A^*A).\]
\end{exercise}

\begin{exercise}\label{use_con}Let $(B(H),\varphi)$ be a quantum probability space and let $A\in B(H)$. Show that 
\[ |\varphi(A)| \leq \|A\|. \]
\end{exercise}

\begin{exercise}Let $(B(H),\varphi)$ be a quantum probability space and denote by $I\in B(H)$ the identity. 
\begin{itemize}
	\item[(a)] For which self-adjoint $X\in B(H)$ is $(X,I)$ monotonically independent? 
	\item[(b)] For which self-adjoint $X\in B(H)$ is $(I,X)$ monotonically independent? 
\end{itemize}

\end{exercise}

\begin{exercise}Let $(B(H),\varphi)$ be a quantum probability space with $\varphi(AB)=\varphi(BA)$ for all $A,B\in B(H)$. 
Let $X,Y\in B(H)$ be random variables. Prove that if $(X,Y)$ is monotonically independent, then the distribution of $X$ or of $Y$ is a point measure.
\end{exercise}

\begin{exercise}\label{45555}Let $X,Y$ be random variables in a quantum probability space $(B(H),\varphi)$. 
Assume that $X$ and $Y$ are both classically and freely independent. Show that the distribution of $X$ or of $Y$ has to be a 
 point measure.
\end{exercise}

\begin{exercise}\label{density}Show that any expectation $\varphi$ on $B(\C^{N\times N})$ can be written as 
$\varphi(X) = \tr(X \rho)$ for a density matrix $\rho$.
\end{exercise}

\chapter{The complex toolbox}

In the case of the non-classical independences, the characteristic function is basically replaced by the Cauchy transform, which is a holomorphic function in the upper half-plane. Complex analysis offers us several powerful theorems that will help us to deal with these functions. 

\section{The Cauchy transform}

\begin{definition}
Let $\mu \in \mathcal{P}(\R)$. The Cauchy transform $G_\mu$ of $\mu$ is defined as the holomorphic function
\begin{equation}\label{eq:Cauchy}
G_\mu(z):= \int_{\R}\frac{1}{z-x} \mu(dx), \quad z\in \Ha:=\{w\in\C\,|\, \Im(w)>0\}.
\end{equation}
\end{definition}

We could regard $G_\mu$ also on the lower half-plane $\Ha^-:=\{w\in\C\,|\, \Im(w)<0\}$, but here we simply get the conjugation as $G_\mu(\overline{z})=\overline{G_\mu(z)}$. If $\mu$ has compact support, $G_\mu$ can in fact be extended analytically to $\hat{\C}\setminus \supp(\mu)$. The power series extension of $G_\mu$ at $z=\infty$ yields the moments of $\mu$ in this case:

\begin{equation}\label{Cauchy_moments}
\begin{aligned}
    G_\mu(z) &= \frac1{z} \int_\R  \frac1{1-\frac{x}{z}} \mu(dx)=
\frac1{z} \int_\R  \sum_{k=0}^\infty x^k/z^k \mu(dx) \\
    &= \sum_{k=0}^\infty \left(\int_\R x^k\mu(dx)\right)/z^{k+1}=\frac1{z}+\frac{m_1(\mu)}{z^2}+\frac{m_2(\mu)}{z^3}+...
\end{aligned}
\end{equation}

We will regard $G_\mu$ mostly as a function on the upper half-plane, but we keep the analytic extension \eqref{Cauchy_moments} for compactly supported measures in mind.

\begin{theorem}[Stieltjes-Perron inversion formula]\label{sp} Let $\mu \in \mathcal{P}(\R)$.
\begin{align}
\frac{1}{2}\mu(\{\alpha\})+\frac{1}{2}\mu(\{b\})+\mu((a,b)) 
&= -\frac{1}{\pi} \lim_{y\downarrow0}\int_a^b \Im(G_\mu(x+i y))\, dx, \quad a,b\in\R, a<b, \label{Stieltjes2} \\
\mu(\{a\}) &= \lim_{y\downarrow0} iy G_\mu(a+iy), \quad a \in \R. \label{eq:atom2}
\end{align}
In particular, if $G_\mu=G_\nu$ for $\mu,\nu\in\mathcal{P}(\R)$, then $\mu=\nu$.
\end{theorem}
\begin{proof}
\begin{eqnarray*}
&&\int_{a}^b \Im(G_\mu(x+iy)) dx = 
-\int_{a}^b \int_\R \frac{y}{(x-u)^2+y^2} \mu(du)dx = 
- y \int_\R \int_{a}^b \frac{1}{(x-u)^2+y^2} dx\mu(du) \\
&=& - y \int_\R \frac1{y}\left(\arctan\left(\frac{b-u}{y}\right)-\arctan\left(\frac{a-u}{y}\right)\right) \mu(du)\\
 &=& 
- \int_\R \arctan\left(\frac{b-u}{y}\right)-\arctan\left(\frac{a-u}{y}\right)\mu(du).
\end{eqnarray*}
Denote the integrand of the last integral by $s(u)$. Then 
\[\lim_{y\downarrow 0} s(u)= 
\begin{cases}
\frac{\pi}{2} & \text{if $u=a$ or $u=b$,}\\
0 & \text{if $u<a$ or $u>b$,}\\
\pi & \text{if $u\in(a,b)$.}\\
\end{cases}\]
By the dominated convergence theorem,
\begin{eqnarray*}
&&-\frac1{\pi}\lim_{y\to0}\int_{a}^b \Im(G_\mu(x+iy)) dx =
 \frac1{\pi}\lim_{y\to0} \int_\R   s(u) \mu(du) \\
&=& 
\frac1{\pi}\int_\R  \lim_{y\to0}s(u) \mu(du) = \frac1{2}\mu(\{a\}) +  \frac1{2}\mu(\{b\})  + \mu((a,b)).
\end{eqnarray*}
Furthermore,  
\[ iy G_\mu(a+iy) = 
 \int_\R \frac{iy}{a-u+iy} \mu(du) = 
 \int_\R \frac{1}{(a-u)/(iy)+1} \mu(du)
\]
and similarly we derive $\lim_{y\downarrow0} iy G_\mu(a+iy) = \mu(\{a\})$.
\end{proof}

Mappings from the upper half-plane into itself have a useful integral representation. It can be obtained from the Herglotz representation formula for holomorphic functions in the unit disc with non-negative real part, or from the Poisson integral representation of harmonic functions. A proof can be found in \cite[Theorem 1]{cau32}.

\begin{theorem}[Nevanlinna representation formula]\label{thm_nevanlinna}
Every holomorphic mapping $f$ from $\Ha$ into $\Ha\cup \R$ can be written as 
\begin{equation}\label{nevanlinna}
f(z)=a + bz + \int_{\mathbb{R}}\frac{1+z x}{x-z} \gamma(dx) ,\qquad z\in \mathbb{H},
\end{equation}
with  $a\in \R, b\geq 0$, and a non-negative  measure $\gamma$ on $\R$. The numbers $a$ and $b$ can be calculated via 
$a=\Re f(i)$, $b = \lim_{y\to\infty}f(iy)/(iy)$. Conversely, every such triple $(a, b, \gamma)$ produces a holomorphic mapping from $\Ha$ into $\Ha\cup \R$.
\end{theorem}

\begin{remark}\label{nontgh}The formula $b = \lim_{y\to\infty}f(iy)/(iy)$ can easily be verified. In fact, 
\[ b = \lim_{n\to \infty}\frac{f(z_n)}{z_n} \]
for any sequence $(z_n)\subset \Ha$ with $\Im(z_n)\to\infty$ and $\Im(z_n) > c|\Re(z_n)|$ for some $c>0$ (non-tangential approach to $\infty$).
\end{remark}

We now obtain a very simple characterization of Cauchy transforms.

\begin{theorem}\label{cor_nevanlinna1}Let $f:\Ha\to \Ha^-$ be holomorphic.
\begin{itemize}
\item[(a)] $f=G_\mu$ for some $\mu\in\mathcal{P}(\R)$ if and only if 
  \[\lim_{y\to \infty}(iy)f(iy)=1.\]
\item[(b)] $f=G_\mu$ for some $\mu\in\mathcal{P}_c(\R)$ if and only if 
 $f$ extends analytically to $\infty$ with the expansion 
\[f(z)=\frac{1}{z}+\frac{c_2}{z^2}+\frac{c_3}{z^3}+... \quad \text{as $y\to\infty$.}\]
\end{itemize}
\end{theorem}
\begin{proof}${}$
\begin{itemize}
\item[(a)] 
Let $f=G_\mu$. Then
\[ iyf(iy)=\int_\R \frac{iy}{iy-x} \mu(dx) = \int_\R \frac{y^2}{y^2+x^2}\mu(dx) -i \int_\R \frac{xy}{y^2+x^2} \mu(dx).  \]
The dominated convergence theorem implies   $\lim_{y\to \infty}(iy)f(iy)=1$.\\
Now assume that
  $\lim_{y\to \infty}(iy)f(iy)=1$. We have $-f(z)=a + bz + \int_{\mathbb{R}}\frac{1+z x}{x-z} \gamma(dx)$
for a non-negative  measure $\gamma$ on $\R$ and some $a\in \R, b\geq 0$. Next we write 
\[ \frac{1+z x}{x-z}  = \left(\frac1{x-z}-\frac{x}{1+x^2} \right)(1+x^2).\] 
Put $\mu(dx)=(1+x^2)\gamma(dx)$. Then $\mu$ is a finite, non-negative measure and we have 
\begin{equation}\label{sunny1} f(z) = -a - bz - \int_\R\frac{1}{x-z} \mu(dx) + \int_\R  \frac{x}{1+x^2}  \mu(dx). \end{equation}
Now we calculate $iyf(iy)$:
 \begin{equation}\label{sunny0}
\begin{aligned}&& iyf(iy) = -iya + by^2 - \int_\R\frac{iy}{x-iy} \mu(dx) + \int_\R  \frac{iyx}{1+x^2}  \mu(dx)
\\
&=&by^2 + \int_\R\frac{y^2}{x^2+y^2} \mu(dx) + i\left(-ya +\int_\R\frac{yx}{x^2+y^2} \mu(dx)+\int_\R  \frac{yx}{1+x^2} \mu(dx)\right).
\end{aligned}
\end{equation}
We know that $iyf(iy)$ is bounded as $y\to\infty$. Hence $\Im(iyf(iy))/y\to 0$, which gives us 
\begin{equation}\label{sunny2}-a +\int_\R\frac{x}{x^2+y^2} \mu(dx)+\int_\R  \frac{x}{1+x^2} \mu(dx)\to 
-a + \int_\R  \frac{x}{1+x^2} \mu(dx)=0,\quad y\to\infty.\end{equation}
Furthermore, $\Re(iyf(iy))/y^2\to 0$, which yields
\begin{equation}\label{sunny3} b + \int_\R\frac{1}{x^2+y^2} \mu(dx) \to b =0,\quad y\to\infty.\end{equation}
Equation \eqref{sunny1} together with \eqref{sunny2} and \eqref{sunny3} give
\[ f(z) = \int_\R\frac{1}{z-x} \mu(dx).\]
It remains to show that $\mu$ is a probability measure, i.e.\ $\mu(\R)=1$:\\
We have 
\[iyf(iy) = \int_\R\frac{y^2}{x^2+y^2} \mu(dx) - i\int_\R\frac{yx}{x^2+y^2} \mu(dx) =
 \int_\R\frac{1}{1+(x/y)^2} \mu(dx) - i\int_\R\frac{yx}{x^2+y^2} \mu(dx).\]
Consider the integrand of the imaginary part. We have $\frac{|yx|}{x^2+y^2}\leq \frac{\sqrt{x^2+y^2} |x|}{x^2+y^2}= \frac{|x|}{\sqrt{x^2+y^2}}$.
The dominated convergence theorem implies $\Im (iyf(iy))\to 0$ as $y\to\infty$. Because $iyf(iy)\to 1$, we must have $\Re(iyf(iy))\to 1$. Hence 
\[\int_\R\frac{1}{1+(x/y)^2} \mu(dx)\to \int_\R 1 \mu(dx) = \mu(\R) = 1.\]

\item[(b)]  If $f=G_\mu$, then we obtain the Laurent expansion by \eqref{Cauchy_moments}. 
Conversely, if $f$ extends analytically to $\infty$ with Laurent expansion $f(z)=1/z+...$, then $\lim_{y\to \infty}(iy)f(iy)=1$ and 
(a) implies that $f=G_\mu$ for some $\mu\in\mathcal{P}(\R)$. As $\overline{f(z)}=f(\overline{z})$ and $f$ extends analytically to $\hat{\C}\setminus[-M,M]$ for some $M>0$, we see that $f(x)\in \R$ for all $x\in \R\setminus [-M,M]$. The Stieltjes inversion formula implies that $\mu\in \mathcal{P}_c(\R)$.
\end{itemize}
\end{proof}

\begin{remark}\label{sunny00}
Let $f:\Ha\to\Ha\cup \R$ with Nevanlinna triple $(a,b,\gamma)$. 
Equation \eqref{sunny0} shows that 
\begin{equation}\label{Nevanlinna_b}
b = \lim_{y\to \infty}\frac{f(iy)}{iy}.
\end{equation}
As $\frac{1+ix}{x-i}=i$ for all $x\in\R$, we have 
$f(i)=a+ib+i\gamma(\R)$ and thus 
\[a = \Re(f(i)).\]
Finally, we can recover $\gamma$ from \eqref{sunny1} via the Stieltjes inversion formula. Let  
$\mu(dx)=(1+x^2)\gamma(dx)$. Then
\begin{align}
\frac{1}{2}\mu(\{a\})+\frac{1}{2}\mu(\{b\})+\mu((a,b)) &= \frac{1}{\pi} \lim_{y\downarrow0}\int_a^b \Im(F(x+i y))\, {\rm d}x, \label{Stieltjes3} \\
\mu(\{a\})&=-\lim_{y\downarrow0} iy F(a+iy). \label{eq:atom3}
\end{align}

\end{remark}

\begin{corollary}\label{cor_nevanlinna2}Let $f:\Ha\to \Ha\cup \R$ be holomorphic with Nevanlinna triple $(a,b,\gamma)$. Then 
\[\text{$\Im(f(z))\geq \Im(z)$ for all $z\in\Ha$ if and only if $b\geq 1$.}\]
Furthermore, if $b\geq 1$ and $\Im(f(z_0))=\Im(z_0)$ for some $z_0\in \Ha$, then 
$f(z)=z+a$ for some $a\in\R$.
\end{corollary}
\begin{proof}If $b\geq 1$, then the function $f(z)-z$ is of the form \eqref{nevanlinna} with Nevanlinna triple $(a,b-1,\gamma)$. Thus $f(z)-z$ maps $\Ha$ into $\Ha\cup \R$. Conversely, if $f(z)=z+g(z)$ for a holomorphic function $g:\Ha\to\Ha\cup \R$ with Nevanlinna triple $(a',b',\gamma')$, then $b=1+b'$ by \eqref{Nevanlinna_b}, thus $b\geq 1$.\\
Now assume that $b\geq 1$ and $\Im(f(z_0))=\Im(z_0)$ for some $z_0\in \Ha$. Then 
$g(z)=f(z)-z$ is a holomorphic mapping into $\Ha\cup \R$ with $g(z_0)=a\in\R$. The open mapping theorem 
implies that $g$ is constant.
\end{proof}

\begin{example}The translations $z\mapsto z+a$ are automorphisms of $\Ha$. Denote by $\operatorname{Aut}(\Ha)$ the set of all holomorphic automorphisms of $\Ha$. 
All these mappings are M\"obius transforms and we have the nice characterization (see, e.g.\ \cite[Theorem 13.17]{BN10})
\[ \operatorname{Aut}(\Ha)  = \left\{z\mapsto \frac{az+b}{cz+d}\,|\, a,b,c,d\in\R, ad-bc>0\right\}.\]
\hfill $\blacksquare$
\end{example}

\section{Convergence}

Let $f:\Ha\to\Ha\cup \R$ be holomorphic with Nevanlinna triple $(a,b,\gamma)$. We can think of $\R$ being embedded in the 
circle $\hat{\R}=\R\cup \{\infty\}$ in the Riemann sphere and $\gamma$ extends to a Borel measure on $\hat{\R}$. As  $\lim_{x\to\infty}\frac{1+xz}{x-z}=z$ for every $z\in \Ha$, we can write 
\[ f(z)=a + bz + \int_\R \frac{1+xz}{x-z}\gamma(dx)=a + \int_{\hat{\R}} \frac{1+xz}{x-z}\hat{\gamma}(dx),\]  
where $\hat{\gamma} = \gamma + b\delta_\infty$ is a finite non-negative Borel measure on $\hat{\R}$. Let us call $(a,\hat{\gamma})$ the Nevanlinna pair of $f$.

\begin{theorem}\label{lemmaconvergence0}
Let $f,f_1,f_2,...$ be holomorphic mappings from $\Ha$ into $\Ha\cup \R$ with Nevanlinna pairs $(a,\hat{\gamma}), (a_1,\hat{\gamma}_1),...$ Then the following statements are equivalent:
\begin{enumerate}[\rm(1)]
\item $f_n\to f$ converges locally uniformly in $\Ha$.
\item $a_n\to a$ and $\hat{\gamma}_n$ converges weakly to $\hat{\gamma}$.
\end{enumerate}
\end{theorem}
\begin{proof}
Assume that $a_n\to a$ and $\hat{\gamma}_n \to \hat{\gamma}$. The definition of weak convergence implies that 
$f_n(z)\to f(z)$ for all $z\in\Ha$. As the set of all holomorphic mappings from $\Ha$ into $\Ha\cup \R$ forms a normal family, pointwise convergence already implies locally uniform convergence by Vitali's theorem, see \cite[p.9]{Dur83}.\\

Now assume that $f_n\to f$ locally uniformly. Then, by Remark \ref{sunny00}, $f_n(i)=a_n+i\hat{\gamma}_n(\hat{\R})\to a+i\hat{\gamma}(\hat{\R})=f(i)$, and we see that $a_n\to a$ and $\hat{\gamma}_n(\hat{\R})\to \hat{\gamma}(\hat{\R})$. 
Now we use Helly's selection theorem, see \cite[Theorem 25.9]{bill}, which implies that there is a subsequence $(\hat{\gamma}_{n_k})_k$ which converges weakly to 
some finite, non-negative measure $\delta$ on $\hat{\R}$. As $f_{n_k}\to f$, we conclude that $\delta = \hat{\gamma}$ by the Stieltjes inversion formula. Thus every convergent subsequence of $(\hat{\gamma}_n)_n$ has the same limit $\hat{\gamma}$ and we conclude that 
$\hat{\gamma}_n\to \hat{\gamma}$ as $n\to\infty$.
\end{proof}

The same proof applies to Cauchy transforms.

\begin{lemma}\label{lemmaconvergence}
Let $\mu$ and $\mu_1,\mu_2...$ be probability measures on $\R$.
Then the following statements are equivalent:
\begin{enumerate}[\rm(1)]
\item $G_{\mu_n}$ converges to $G_\mu$ locally uniformly in $\Ha$.
\item $\mu_n$ converges weakly to $\mu$. 
\end{enumerate}
\end{lemma}

\section{Discrete semigroups}
 
Consider a (compositional) semigroup $(F_n)_{n\in \N}$ of holomorphic self-mappings of $\Ha$, i.e.\ $F_n = F^{n}=F\circ...\circ F$ with $F=F_1$. Such semigroups can be classified by looking at the behavior of $F^n$ as $n\to \infty$. This classification is usually stated for the unit disc $\D$, but we can simply pass from $\Ha$ to $\D$ via the Cayley transform $C:\Ha\to\D,$ $C(z)=\frac{z-i}{z+i}$. Note that 
$C\circ F_n \circ C^{-1} = (C\circ F \circ C^{-1})^n$ is a semigroup on $\D$.\\

To state the result, we need two further notions.\\

For a function $f:\D\to \C$ and $p\in\partial\D$, the existence of the non-tangential limit $\angle \lim_{z\to p}f(z)=c$ means that 
$\lim_{n\to \infty}f(z_n)=c$ for every sequence $(z_n)\subset \D$ which converges to $p$ within a sector, i.e.\ the angle of 
the vector $p-z_n$ must be in $(\arg p-\eps, \arg p + \eps)$ for some $\eps\in(0,\pi/2)$ and all $n$.

A horodisc in $\D$ at $p=1$ is the image of a set of the form $\{z\in\Ha\,|\, \Im(z)> c\}$, $c>0$, under the Cayley transform. 
By rotating, we define horodiscs in $\D$ at every other $p\in \partial \D$.

 \begin{figure}[h]
 \begin{center}
 \includegraphics[width=0.35\textwidth]{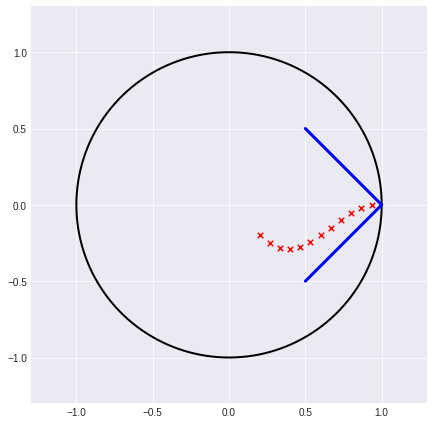}
 \includegraphics[width=0.35\textwidth]{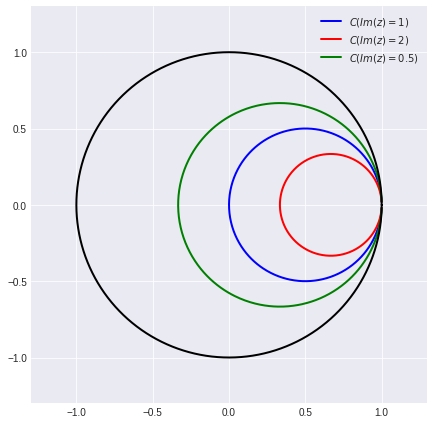}
 \caption{Left: non-tangential approach to $1$. Right: horodiscs at $1$.}
 \end{center}
 \end{figure}

\begin{theorem}\label{grand}
Let $F:\D\to \D$ be a holomorphic self-map which is not an elliptic automorphism, i.e.\ 
$F$ is not conjugated to a rotation $z\mapsto e^{i\alpha}z$, $\alpha \in \R$.
\begin{itemize}
\item[(a)] If $F$ has a fixed point $p\in \D$, then $F^n\to p$ locally uniformly as $n\to\infty$ and $|F'(p)|<1$.
\item[(b)] If $F$ has no fixed points in $\D$, then there is a point $p\in\partial\D$ such that $F^n\to p$ as $n\to\infty$. 
Furthermore, $\angle \lim_{z\to p}F(z)=p$ and $\angle \lim_{z\to p} F'(z)\in (0,1]$.
\end{itemize}
The point $p$ in (a) or (b) is called the \emph{Denjoy-Wolff point} of $F$.
\begin{itemize}
\item[(c)] Conversely, if there exists $p\in \partial \D$ with $\angle \lim_{z\to p}F(z)=p$ and $\angle \lim_{z\to p} F'(z)\in (0,1]$, then $p$ is the Denjoy-Wolff point of $F$.
\item[(d)] If $F$ has no fixed points in $\D$, then $p\in \partial \D$ is the Denjoy-Wolff point of $F$ if and only if 
$F$ maps every horodisc at $p$ into itself.
\end{itemize}
\end{theorem}
\begin{proof}The statements (a)-(c) can be found in \cite[The Grand Iteration Theorem]{shapiro}. Statement (d) follows from Wolff's theorem, see \cite[Section 5.3]{shapiro}.
\end{proof}

\begin{corollary}\label{deefdg}Let $f:\Ha\to\Ha\cup \R$ be holomorphic with Nevanlinna triple $(a,b,\gamma)$ and assume that $f$ is not the identity.
Then $\infty$ is the Denjoy-Wolff point of $f$, i.e.\  $f^n \to \infty$ as $n\to\infty$, if and only if $b\geq 1$.
\end{corollary}
\begin{proof}
If $\infty$ is the Denjoy-Wolff point of $f$, then Theorem \ref{grand} (d) implies that $\Im(f(z))\geq \Im(z)$ for all $z\in \Ha$, 
which implies $b\geq 1$ by Corollary \ref{cor_nevanlinna2}. \\
Conversely, assume that $b\geq 1$. Then Corollary \ref{cor_nevanlinna2} implies $\Im(f(z))\geq \Im(z)$ for all $z\in \Ha$. 
If $f$ has a fixed point $z_0$, then $f(z)=z$ for all $z\in\Ha$ by Corollary \ref{cor_nevanlinna2}. But as $f$ is not the identity, $f$ has no fixed points and Theorem \ref{grand} (d) implies that $\infty$ is the Denjoy-Wolff point of $f$.
\end{proof}

\begin{remark}\label{parabolic}In the extreme case $b=1$, the Denjoy-Wolff point $\infty$ of $f$ is also called \emph{parabolic}.
\end{remark}

\section{Continuous semigroups}\label{continuous_semigroup}

\begin{definition}Let $D\subset \C$ be a domain. A \emph{continuous semigroup} on $D$ is a family 
 $(F_t)_{t\geq0}$ of holomorphic self-mappings $F_t:D\to D$ such that 
\begin{itemize}
	\item[(1)] $F_0(z)=z$ for all $z\in D$,
	\item[(2)] $t\mapsto F_t$ is continuous with respect to locally uniform convergence,
	\item[(3)] $F_{t+s}=F_t\circ F_s$ for all $s,t \geq 0$.
\end{itemize}
\end{definition}

We will only need the case where $D\subsetneq \C$ is simply connected. Note that if $F_t$ is a continuous semigroup on $D$ and $C:E\to D$ is a biholomorphic mapping, then  $\hat{F}_t=C^{-1}\circ F_t \circ C$ is a continuous semigroup on $E$.
So, due to the Riemann mapping theorem, we may assume that $D=\D$.  \\ (In fact, all other domains, even all other Riemann surfaces, lead to rather boring continuous semigroups, see \cite[Section 1.4.3]{Aba89}.)\\

Remarkably, the continuity of $t\mapsto F_t$, assumption (2), makes a semigroup automatically differentiable.

\begin{theorem}\label{dsemi}Let  $(F_t)_{t\geq0}$ be a continuous semigroup on $\D$.
Then the locally uniform
limit \[\lim_{t\searrow 0}\frac{F_t(z)-z}{t}=:G(z)\]
exists and the function $G:\D\to\C$, called the \emph{infinitesimal generator} of the semigroup is holomorphic. 
Furthermore, $t\mapsto F_t(z)$ is the unique solution to the initial value problem
\begin{equation}\label{semigropu_init} \frac{\partial}{\partial t}F_t(z) = G(F_t(z)), \quad F_0(z)=z. \end{equation}

\end{theorem}
\begin{proof}
Let $K\subset \D$ be a compact subset. Choose some $\alpha > 0$. Then the compact set $\cup_{0\leq t\leq \alpha} F_t(K)$ has its convex hull $H$ (which is again a compact set) contained in $\D$. \\

Let $z\in K$ and consider the integral $\int_{z}^{F_t(z)} \frac{d}{dw} (F_t(w)-w) dw$, where we integrate along the line segment from $z$ to $F_t(z)$. This line segment is contained in $H$ and we find $\eta \in (0,\min(1/2,\alpha)]$ such that the integrand has modulus $\leq 1/10$ on $H$ for all $t\in[0,\eta]$. Thus 
\[\left|\int_{z}^{F_t(z)} \frac{d}{dw} (F_t(w)-w) dw\right| \leq \frac1{10}|F_t(z)-z|\]
for all $0\leq t\leq \eta$ and $z\in K$.
Furthermore,  
\[\int_{z}^{F_t(z)} \frac{d}{dw} (F_t(w)-w) dw = F_t(F_t(z)) - F_t(z) - (F_t(z)-z) = F_{2t}(z) - 2F_t(z) + z\]
 and we obtain
\[\frac1{10}|F_t(z)-z| \geq |F_{2t}(z) - 2F_t(z) + z| = 
|F_{2t}(z) - z -2F_t(z) + 2z|\geq 
 2|F_t(z)- z| - |F_{2t}(z) - z|\]
 and thus
\[ |F_t(z)-z| \leq \frac{10}{19}|F_{2t}(z) - z| \leq 2^{-2/3}|F_{2t}(z) - z|\] for all $0\leq t\leq \eta$ and $z\in K$. ($\left(\frac{10}{19}\right)^{3/2}=0.38... < 1/2$.)\\

Let $k\in\N$ be such that $2^k\eta \geq 1$ and put 
\[M = 2^{2k/3}\sup\{|F_t(z)-z| \,|\, z\in K, t\in[2^{-k}, 1]\}.\]
For $t\in[2^{-k},1]$, we have $Mt^{2/3} \geq M2^{-2k/3}\geq |F_t(z)-z|$. Now let $t\in(0,2^{-k}]$ and let $m\in\N$ be the smallest natural number with $2^m t > \eta \geq 2^{-k}$. Now we iterate the inequality above to get 
\begin{eqnarray*}|F_t(z)-z| &\leq& 2^{-2/3}|F_{2t}(z) - z| \leq 2^{-2\cdot 2/3}|F_{2^2t}(z) - z| \\
&\leq& ... \leq 2^{-2\cdot m/3}|F_{2^m t}(z) - z| 
\leq  2^{-2\cdot m/3}2^{-2\cdot k/3}M < t^{2/3}M.
\end{eqnarray*} 
(Here we used that $2^m t\leq 1$. As we assumed that $\eta \leq 1/2$, we have indeed $2^mt \leq \eta \cdot 2\leq 1$.)\\
All in all,  
\[ |F_t(z)-z| \leq M t^{2/3}\]
for all $0\leq t\leq 1$ and $z\in K$. (Note that for $t=0$, we have $F_0(z)=z$.)\\

Now we repeat the same argument for a compact set $K_1\subset \D$ which contains $H$, and we obtain a constant $M_1>0$ such that 
\[ |F_t(z)-z| \leq M_1 t^{2/3}\]
for all $0\leq t\leq 1$ and $z\in K_1$. The Cauchy inequalities show that there exists a constant $M_2>0$ such that 
\[ |F'_t(z)-1| \leq M_2 t^{2/3}\]
for all $0\leq t\leq 1$ and $z\in K_1$.\\

By using this estimate, our previous argument shows that 
\[ |F_{2t}(z)-2F_t(z)+z| \leq M_2 t^{2/3} |F_t(z)-z| \leq MM_2 t^{4/3} \]
for all $t\in[0,\alpha]$ and $z\in K$. Thus
\[ \left| \frac{F_{2t}(z)-z}{2t} - \frac{F_{t}(z)-z}{t} \right| \leq MM_2 \frac{t^{1/3}}{2}\]
for all $t\in[0,\alpha]$ and $z\in K$. We see that $\sum_{n=0}^\infty \frac{F_{2^{-n-1}}(z)-z}{2^{-n-1}} - \frac{F_{2^{-n}}(z)-z}{2^{-n}}$ 
is converging uniformly on $K$. From being a telescoping series, we see that 
\[ \lim_{n\to\infty} \frac{F_{2^{-n}}(z)-z}{2^{-n}} \]
exists uniformly on $K$. Hence this limit exists locally uniformly on $\D$ and defines a holomorphic function $G:\D\to\C$. \\
Fix $z_0\in \D$ and $t_0>0$. Then $\{F_t(z_0)\,|\, t\in[0,t_0]\}$ is a compact subset of $\D$. As $n\to\infty$, the function 
$2^n(F_{t+2^{-n}}(z_0)-F_t(z_0))=2^n(F_{2^{-n}}(F_{t}(z_0))-F_t(z_0))$ converges uniformly to $G(F_t(z_0))$ for each $t\in[0,t_0]$.\\
Let $t\in(0, t_0)$. Then, for $n$ large enough, we define $G_n:[0,t]\to\C$, $G_n(0)=0$, 
\[G_n(s)= 2^n(F_{(k+1)/2^{n}}(z_0)-F_{k/2^{n}}(z_0)),\; s\in (k/2^{n}, (k+1)/2^{n}],\; k=0,1,..., \lfloor t2^n \rfloor.\]
Then $G_n$ converges uniformly to $G(F_s(z_0))$ and $\int_0^{t} G_n(s) ds \to \int_0^t G(F_s(z_0)) ds$. We have 
\begin{eqnarray*}&&\int_0^{\lfloor t2^n+1 \rfloor/2^{n}} G_n(s) ds = \sum_{k=0}^{\lfloor t2^n \rfloor} F_{(k+1)/2^{n}}(z_0)-F_{k/2^{n}}(z_0)\\
&=& F_{\lfloor t2^n+1 \rfloor/2^{n} }(z_0)-F_{0}(z_0) \to  F_{t}(z_0)-F_{0}(z_0)=F_{t}(z_0)-z_0.
\end{eqnarray*}
This implies 
\[ F_t(z) = z + \int_0^t G(F_s(z)) ds \]
for all $z\in \D$ and $t\geq 0$. 
\end{proof}

Equation \eqref{semigropu_init} has a simple but interesting consequence. If we look at the initial value problem \eqref{semigropu_init} for two different initial values $z_0,w_0\in \D$, $z_0\not=w_0$, then 
$F_t(z_0)\not= F_t(w_0)$ for all $t\geq 0$. Otherwise, if $F_T(z_0)=F_T(w_0)$ for some $T>0$, we could solve the differential equation at time $T$ and go backward in time: 
\begin{equation*} \frac{\partial}{\partial t}v(t) = -G(v(t)), \quad v_0=F_T(z_0). \end{equation*}
This initial value problem would have two solutions for $t\in[0,T]$, namely $v(t)=F_{T-t}(z)$ and $v(t)=F_{T-t}(w)$, which would contradict the Picard-Lindel\"of uniqueness theorem. So $z\mapsto F_t(z)$ is injective. An injective holomorphic function is also called \emph{univalent}.

\begin{corollary}\label{fdsaKOERTZU}All elements $F_t$ of a continuous semigroup are univalent functions.
\end{corollary}

Infinitesimal generators can be represented by the so called Berkson-Porta formula.

\begin{theorem}\label{grand_iteration}Let  $(F_t)_{t\geq0}$ be a continuous semigroup on $\D$.
\begin{itemize}
\item[(a)]  Every generator $G$ on $\D$ has the following form (Berkson-Porta formula)
 \begin{equation}\label{Berkson-Porta}
 G(z) = (\tau-z)(1-\overline{\tau}z)p(z),
 \end{equation}
  where $\tau\in \overline{\D}$ and $p:\D\to\C$ is holomorphic with $\Re(p)\geq0$. \\
If the corresponding semigroup $(F_t)_{t\geq0}$ does not consist of elliptic automorphisms, then $\tau$ is the Denjoy-Wolff point of the semigroup, i.e.\ $F_t\to \tau$ as $t\to\infty$.
\item[(b)] Assume that $F_1$ is not an elliptic automorphism.\\ Then $F_1$ has a fixed point $p\in\D$ if and only if $G(p)=0$ and 
$F'_t(p)=e^{tG'(p)}$ for all $t\geq0$.\\
 Furthermore, $p\in\partial \D$ is the Denjoy-Wolff point of $F_1$ if and only if $\angle \lim_{z\to p} G(z)=0$ and $\angle \lim_{z\to p} G'(z)\in (-\infty,0]$. In this case, $\angle \lim_{z\to p} F_t'(z)=\angle \lim_{z\to p} e^{tG'(z)}$ for all $t\geq0$.
\end{itemize}
\end{theorem}
\begin{proof}For (a) we refer to  \cite{MR0480965} and (b) can be found in \cite[Theorem 1]{CDMP06}.
\end{proof}

\begin{remark}If $F_1$ is not an elliptic automorphism, then $\tau$ and thus also $p$ are uniquely determined. This is also true if $F_1$ is an elliptic automorphism which is not the identity, because then $\tau$ is the unique fixed point of $F_1$ in $\D$. Only if $F_t(z)=z$ for all $z$ and $t$, we have $G(z)=0$ for all $z$ and $\tau$ is not unique.
\end{remark}

\begin{example}$F_t(z)=e^{it}z$ is a continuous semigroup, where $F_1$ is an elliptic automorphism. We obtain the generator $G(z)=iz$, which corresponds to $\tau = 0$ and $p(z)\equiv-i$. The continuous semigroup $F_t(z)=e^{-t}z$ corresponds to $G(z)=-z$, i.e.\ $p(z)\equiv 1$ and $\tau=0$, which is the Denjoy-Wolff point of the semigroup.  \hfill $\blacksquare$
\end{example}

We see that any continuous semigroup can be described by its generator, which is a function of the form \eqref{Berkson-Porta} and should be seen as a vector field on $\D$. Now problems concerning semigroups can be translated to problems concerning generators.

\newpage

\section{Exercises}

\begin{exercise}Let \[f(z)=\frac{Az+B}{Cz+D}\] be an automorphism of $\Ha$ with $A,B,C,D\in \R$, $AD-BC>0$. Determine the Nevanlinna triple $(a,b,\gamma)$ of $f$.
\end{exercise}

\begin{exercise}\label{boundary_auto}
Consider three points $p_1,p_2,p_3\in\partial \D$ and three points $q_1,q_2,q_3\in \partial\D$, both in counter-clockwise order. Show that there exists a unique $f\in\operatorname{Aut}(\D)$ with $f(p_j)=q_j$ for all $j=1,2,3.$
\end{exercise}

\begin{exercise}\label{semi_exercie}The semicircle distribution $W(\mu,\sigma)$, also called Wigner's law, is given by the density 
\[\frac{1}{2\pi\sigma^2} \sqrt{4\sigma^2 -(x-\mu)^2}, \quad  x\in [\mu-2\sigma,\mu+2\sigma].\] 
It has mean $\mu$ and variance $\sigma^2$. 
\begin{itemize}
	\item[(a)] Let $G(z)=\frac{z-\sqrt{z^2-4}}{2}$, where the branch of the square root is chosen such that $\sqrt{\cdot}$ maps $\C^2\setminus [-4,\infty)$ into the upper half-plane. Prove that 
	$G=G_\mu$ for some $\mu\in \mathcal{P}(\R)$.
	\item[(b)] Use the Stieltjes inversion formula to show that 
	$\mu = W(0,1)$.
\end{itemize}
\end{exercise}

\begin{exercise}[Sokhotski-Plemelj formula]\label{equal_Hilbert}
For a probability measure $\mu$ on $\R$, consider the following limits (which may or may not exist):

\begin{equation*}\hat{\mathcal{H}}_\mu(x) := \lim_{\eps\downarrow 0} \hat{\mathcal{H}}_{\eps,\mu}(x), \quad
\hat{\mathcal{H}}_{\eps,\mu}(x):=\frac {1}{\pi }\Re\; G_\mu(x+i\eps),
\end{equation*}
and
\begin{equation*}\mathcal{H}_\mu(x) := 
\lim_{\eps\downarrow 0}\mathcal{H}_{\eps,\mu}(x), \quad 
\mathcal{H}_{\eps,\mu}(x) := \frac {1}{\pi }\int_{|x-t|> \eps}\frac {1}{x-t}\,\mu(dt) \quad  \text{(Hilbert transform)}.
\end{equation*}

Let $\mu$ be an absolutely continuous probability measure with compact support and continuous density $f(x)dx$. 
Let $x\in\R$. Show:
$\mathcal{H}_\mu(x)$ exists if and only if
 $\hat{\mathcal{H}}_\mu(x)$ exists. \\
If these limits exist, then $\mathcal{H}_\mu(x)=\hat{\mathcal{H}}_\mu(x)$.
\end{exercise}

\begin{exercise}\label{muhaha}${}$
\begin{itemize}
	\item[(a)]  If $F_t$ is a continuous semigroup on a simply connected domain $D\subsetneq \C$ with generator $G$ and $C:E\to D$ is a biholomorphic mapping, then 
 $\hat{F}_t=C^{-1}\circ F_t \circ C$ is a continuous semigroup on $E$.\\
 Use the Cayley transform $C:\Ha\to\D$, $C(z)=(z-i)/(z+i)$, to obtain a formula for generators on the upper half-plane $\Ha$.
\item[(b)] Prove the following statement: Every holomorphic mapping from $\Ha$ into $\Ha\cup \R$ is an infinitesimal generator on $\Ha$.
The set of these generators minus the generator $G(z)\equiv 0$ is exactly the set of all generators of continuous semigroups 
whose Denjoy-Wolff point is $\infty$.
\end{itemize}
\end{exercise}

\begin{exercise}\label{muhaha2}
Let $F:\Ha\to\Ha$ be holomorphic with Nevanlinna triple $(a,1,\gamma)$. Show that $F(iy), y>0,$ belongs to the set $\{x+iy\in\Ha\,|\, |x|<y\}$ for all $y$ large enough.
\end{exercise}

\newpage

\chapter{Convolutions and limit theorems}

In this chapter we will look at convolutions for the four non-classical independences and we will prove the corresponding central limit theorems, which can be summarized as follows:\\
\renewcommand{\arraystretch}{1.5}
\begin{center}
\begin{tabular}{|l|l|l|}\hline
Independence& Central limit law & Transform\\ \hline\hline
Tensor&Gaussian normal distribution $ \frac1{\sqrt{2\pi}} e^{-x^2/2} dx$& characteristic function \\\hline
Boolean& $\frac1{2}\delta_1 + \frac1{2}\delta_{-1}$& $B$-transform\\\hline
Free& semicircle distribution $\frac{1}{2\pi} \sqrt{4 -x^2} \textbf{1}_{[-2,2]} dx$& $R$-transform\\\hline
Monotone& arcsine distribution $\frac{1}{\pi\sqrt{2 -x^2}} \textbf{1}_{(-\sqrt{2},\sqrt{2})}  dx$& $F$-transform\\\hline
Anti-monotone& arcsine distribution $\frac{1}{\pi\sqrt{2 -x^2}} \textbf{1}_{(-\sqrt{2},\sqrt{2})}  dx$& $F$-transform\\ \hline
\end{tabular}\vspace{3mm}
\end{center}

The case of monotone independence will be handled with all details and proofs, while the other cases will be treated more relaxingly.\\

We will also see how the infinitely divisible distributions for the non-classical cases look like. Interestingly, these classes are all characterized by the set of holomorphic functions $f:\Ha\to\Ha\cup\R$ with 
\[f(z) = a + \int_\R \frac{1+xz}{x-z} \gamma(dx),\]
where $a\in\R$ and $\gamma$ is a non-negative finite measure.

\section{$F$-transform and monotone convolution}

In monotone probability theory, the role of the characteristic function is played by the $F$-transform, which is simply the multiplicative inverse of the Cauchy transform.

\begin{definition}For $\mu\in \mathcal{P}(\R)$, the $F$-transform $F_\mu$ is defined as the holomorphic function 
\[F_\mu:\Ha\to\Ha, \quad F_\mu(z)=\frac1{G_\mu(z)} = \left(\int_\R \frac{1}{z-x} \mu(dx)\right)^{-1}.\] 
\end{definition}

\begin{example}For $\mu=\frac1{2}\delta_a+\frac1{2}\delta_b$, $a,b\in\R$, we obtain the rational function
 $F_{\mu}(z) = \frac{(z-a)(z-b)}{z-(a+b)/2}$. If $a=b$, then $F_{\delta_a}(z)=z-a$ is a simple translation of the upper half-plane. \hfill $\blacksquare$
\end{example}

The arcsine distribution will play a special role for monotone independence. 

\begin{example}\label{ex_arcsine}The arcsine distribution $A(\mu,\sigma^2)$ with mean $\mu$ and variance $\sigma^2$ 
is given by the density 
\[\frac{1}{\pi\sqrt{2\sigma^2-(x-\mu)^2}}, \quad  x\in \left(\mu-\sqrt{2}\sigma, \mu+2\sqrt{2}\sigma\right).\] 
 \begin{figure}[H]
 \begin{center}
 \includegraphics[width=0.6\textwidth]{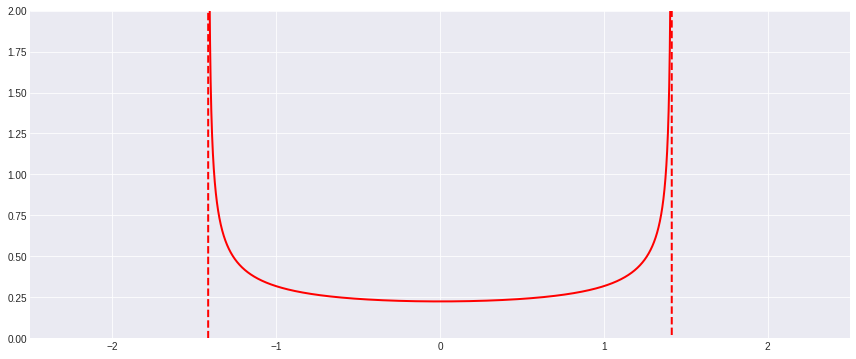}
 \caption{Density of the arcsine distribution $A(0,1)$.}
 \end{center}
 \end{figure}

Similarly to the semicircle distribution, see Exercise \ref{semi_exercie}, we can show that $G_{A(\sigma^2)}(z) = \frac1{\sqrt{(z-\mu)^2-2\sigma^2}}$ and thus \[F_{A(\mu, \sigma^2)}(z) =\sqrt{(z-\mu)^2-2\sigma^2}.\]
This function is a conformal mapping from $\Ha$ onto $\Ha$ minus the vertical line segment from $0$ to $\sqrt{2}\sigma i$.
 \begin{figure}[H]
 \begin{center}
 \includegraphics[width=0.9\textwidth]{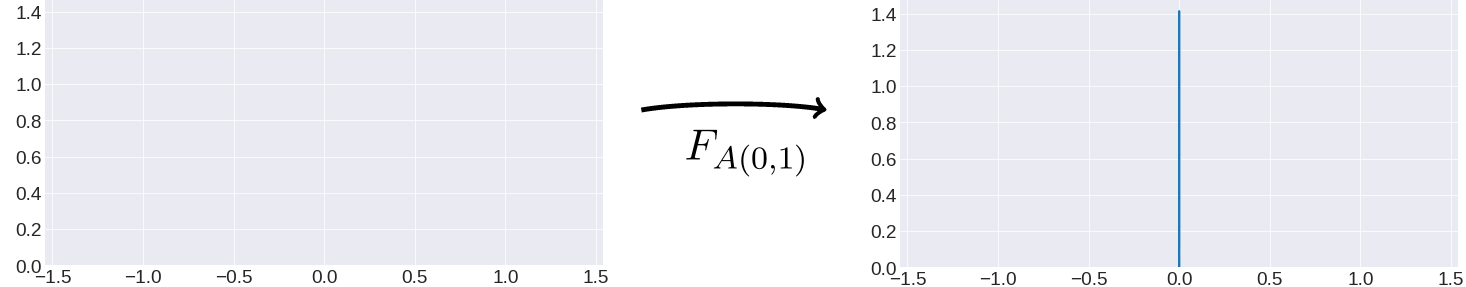}
 \caption{$F$-transform of the arcsine distribution $A(0,1)$.}
 \end{center}
 \end{figure}
\vspace{-10mm}\hfill $\blacksquare$
\end{example}

\begin{example}
Let $\mu=W(0,1)$ be the standard semicircle distribution. Then Exercise \ref{semi_exercie} shows that $F_\mu(z)=\frac{2}{z-\sqrt{z^2-4}}$. 
As $F_\mu$ extends continuously to $\Ha\cup \R$, $F_\mu(\Ha)$ is the unbounded complement of the curve $\{F_\mu(x)\,|\, x\in(-2,2)\}$ in $\Ha$. A simple calculation (see Exercise \ref{semi_exercie}) shows that $\Re(G_\mu(x))=\frac{x}{2}$ and $\Im(G_\mu(x))= \frac{-\sqrt{4-x^2}}{2}$ for $x\in[-2,2]$. Thus, for $x\in(-2,2)$, we have 
\[F_\mu(x) = \frac1{\frac{x}{2}+i\frac{-\sqrt{4-x^2}}{2}}=
\frac{2}{x-i\sqrt{4-x^2}}=
\frac{2x+2i\sqrt{4-x^2}}{4}=\frac{x}{2}+i\sqrt{1-\left(\frac{x}{2}\right)^2},\]
and we see that the curve is the semicircle $(\partial \D)\cap \Ha$.

 \begin{figure}[H]
 \begin{center}
 \includegraphics[width=0.9\textwidth]{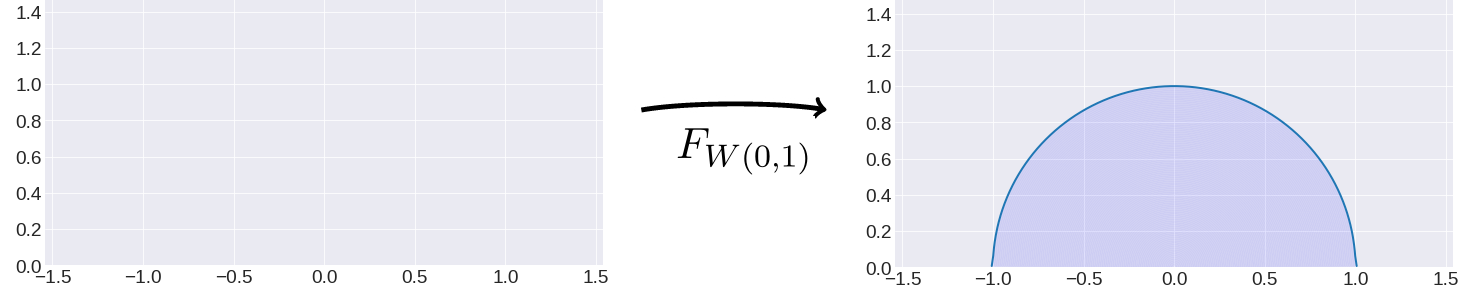}
 \caption{$F$-transform of the semicircle distribution $W(0,1)$.}
 \end{center}
 \end{figure}
\vspace{-8mm}\hfill $\blacksquare$

\end{example}

$F$-transforms of probability measures can be characterized as follows.
\begin{theorem}\label{Julia} Let $F\colon \Ha \to \Ha\cup\R$ be holomorphic. Then the following are equivalent.
\begin{enumerate}
\item[(a)] There exists a probability measure $\mu$ on $\R$ such that $F=F_\mu$.
\item[(b)] $F$ has the Pick-Nevanlinna representation
\begin{equation}\label{EV_eq:2}
F(z) =  z + b + \int_{\R}\frac{1+xz}{x-z}\rho({\rm d}x),
\end{equation}
where $b \in \R$ and $\rho$ is a finite, non-negative measure on $\R$.
\item[(c)]\label{Julia2} $\lim_{y\to \infty}\frac{F(i y)}{i y}=1$. 
\item[(d)] $F$ is the identity or the Denjoy-Wolff point of $F$ is $\infty$ and parabolic.
\end{enumerate}
If these equivalent conditions hold, then $\Im(F(z)) \geq \Im(z)$ for all $z\in\Ha$.\\
Furthermore, $\mu$ has compact support, mean $0$ and variance $\sigma^2$ if and only if there exists a
finite non-negative measure $\tau$ on $\mathbb{R}$ with compact support and $\tau(\R)=\sigma^2$
such that
\begin{equation}\label{EV_finite_var}
F_\mu(z) = z + \int_\mathbb{R} \frac1{x-z}\tau({\rm d}x).
\end{equation}
\end{theorem}
\begin{proof}
The equivalence between (a)  and (c) follows from Theorem \ref{cor_nevanlinna1}, the equivalence between (b) and (c) follows from Theorem \ref{thm_nevanlinna}, and the equivalence between (c) and (d) follows from Corollary \ref{deefdg} and Remark \ref{parabolic}.\\

Now let $F=F_\mu$. Then Corollary \ref{cor_nevanlinna2} implies that $\Im(F(z)) \geq \Im(z)$ for all $z\in\Ha$.\\

Now assume that $\mu$ has  compact support, mean $0$ and variance $\sigma^2$. The inversion formula \eqref{Stieltjes3} shows that also 
$\tau$ has compact support. As $G_\mu(z)=\frac1{z}+\frac{\sigma^2}{z^3} + \ldots$, see \eqref{Cauchy_moments}, a simple computation shows that 
$F_\mu(z)=z-\frac{\sigma^2}{z}+ \ldots$ at $\infty$. Now $H(z)=z-F_\mu(z)$ maps $\Ha$ into $\Ha^-\cup \R$. 
 Either $H$ is constant $c\in\R$ or $H:\Ha\to \Ha^-$. If $H(z)\equiv c$, then $H(z)=\frac{\sigma^2}{z}+ \ldots$, which implies $\sigma^2=0$ and thus $\mu=\delta_0$ and  \eqref{EV_finite_var} holds for $\tau=0$. Otherwise, $\sigma^2>0$ and $H(z)/\sigma^2$ maps $\Ha$ into $\Ha^-$ with 
$H(z)/\sigma^2=\frac{1}{z}+ \ldots$\\
Theorem \ref{cor_nevanlinna1} shows that $H(z)/\sigma^2$ is the Cauchy transform of some probability measure $\mu'$. With $\tau = \mu'\cdot \sigma^2$, we obtain \eqref{EV_finite_var}. These considerations can also be reversed: if $F_\mu$ has the form \eqref{EV_finite_var}, then $\mu$ has compact support, mean $0$ and variance $\sigma^2$.
\end{proof}

\begin{remark}\label{first_moment}
With Theorem \ref{cor_nevanlinna1} we can also say that $\mu\in\mathcal{P}(\R)$ has compact support with mean $0$ and variance $\sigma^2$ if and only if $F_\mu$
extends analytically to $\infty$ with the expansion 
\[F_\mu(z)=z-\frac{\sigma^2}{z}+... \quad \text{at $\infty$.}\]
  \end{remark}

If $X$ is a random variable in a quantum probability space with distribution $\mu\in \mathcal{P}_c(\R)$, then 
$\varphi(p(X)) = \int_\R p(x) \mu(dx)$ for all polynomials $p:\R\to \C$. For $z\in\Ha$ we find a sequence $p_n:\R\to\C$ of polynomials such that $p_n(x)\to \frac1{z-x}$ uniformly on the support of $\mu$. Then also $p_n(X)\to (z-X)^{-1}$ with respect to the operator norm and Exercise \ref{use_con} implies the continuity of $\varphi$, thus $\varphi\left((z-X)^{-1}\right) = \int_\R \frac{1}{z-x}\mu(dx)=G_\mu(z)$.

\begin{theorem}[Monotone convolution]\label{monotone_convolution} Let $(X,Y)$ be monotonically independent with distributions $\mu,\nu\in \mathcal{P}_c(\R)$. Let $\alpha$ be the distribution of $X+Y$. Then 
\[\text{$F_{\alpha}(z)=(F_\mu\circ F_\nu)(z)$ \quad for all $z$ in $\Ha$.}\]
\end{theorem}
\begin{proof}For $z\in \Ha$ we have
\begin{eqnarray*}
(z-(X+Y))^{-1} = \sum_{n=0}^\infty (X+Y)^n/z^{n+1} = 
\sum_{n=0}^\infty \sum_{k=0}^{n} 
\sum_{\substack{j_0+j_1+...+j_k=n-k \\ 0\leq j_1,...,j_k}}Y^{j_0}XY^{j_1}X\dots XY^{j_k}/z^{n+1}.
\end{eqnarray*}

Now we take the expectation and use the monotone independence of $(X,Y)$:

\begin{align}
G_\alpha(z)=\varphi((z-(X+Y))^{-1}) &=& 
\sum_{n=0}^\infty \sum_{k=0}^{n} \nonumber
\sum_{\substack{j_0+j_1+...+j_k=n-k \\ 0\leq j_1,...,j_k}}\varphi(Y^{j_0}XY^{j_1}X\cdots XY^{j_k})/z^{n+1}\\ 
&=& \sum_{n=0}^\infty \sum_{k=0}^{n} 
\sum_{\substack{j_0+j_1+...+j_k=n-k \\ 0\leq j_1,...,j_k}}\varphi(X^k)\varphi(Y^{j_0})\cdots \varphi(Y^{j_k})/z^{n+1}.\label{n00}
\end{align}

Furthermore, we have
\begin{eqnarray*}
G_{\nu}(z)(1-G_{\nu}(z)X)^{-1} = \sum_{k=0}^\infty G_{\nu}(z)^{k+1}X^k = \sum_{k=0}^\infty \left(\sum_{n=0}^\infty \varphi(Y^n)/z^{n+1}\right)^{k+1}X^k.
\end{eqnarray*}
Taking the expectation yields $G_\mu(1/G_\nu(z))$, so

\begin{equation}\label{n11}
G_\mu(1/G_{\nu}(z)) = \varphi(G_{\nu}(z)(1-G_{\nu}(z)X)^{-1}) =
\sum_{k=0}^\infty \left(\sum_{n=0}^\infty \varphi(Y^n)/z^{n+1}\right)^{k+1}\varphi(X^k). 
\end{equation}

By comparing the coefficients in \eqref{n00} and \eqref{n11}, we see that both sums are identical. Note that all the expansions only hold when $|z|$ is big enough. So $G_\alpha(z)=G_\mu(1/G_{\nu}(z))$  in a neighborhood of $\infty$. The identity theorem for holomorphic functions implies that $G_\alpha(z)=G_\mu(1/G_{\nu}(z))$ for all $z\in\Ha$.
\end{proof}

Let $\mu,\nu\in \mathcal{P}(\R)$. Then $F_\mu\circ F_\nu$ is again an $F$-transform, which can be seen as follows. Due to Exercise \ref{muhaha2}, the curve $y\mapsto F_{\nu}(iy)$ belongs to $\{x+iy\in\Ha\,|\, |x|<y\}$ for all $y$ large enough, and $\Im(F_{\nu}(iy))\to\infty$. With Remark \ref{nontgh} we obtain

\[ \lim_{y\to\infty}\frac{(F_\mu\circ F_\nu)(iy)}{iy} = \lim_{y\to\infty}\frac{F_\mu(F_\nu(iy))}{F_\nu(iy)} \frac{F_\nu(iy)}{iy}= 1\cdot 1=1,    \]
and Theorem \ref{Julia} implies that $F_\mu\circ F_\nu$ is an $F$-transform. We can thus make the following definition.

\begin{definition}[Monotone convolution] For $\mu,\nu\in\mathcal{P}(\R)$ we define the \emph{monotone convolution} $\mu \rhd \nu\in \mathcal{P}(\R)$ via \[F_{\mu\rhd\nu} = F_\mu \circ F_\nu.\]
A probability measure $\mu\in \mathcal{P}(\R)$ is called \emph{monotonically infinitely divisible} if, for every $n\in\N$, there exists $\mu_n\in \mathcal{P}(\R)$ such that 
$\mu = \mu_n \rhd  \cdots \rhd  \mu_n$ ($n$-fold convolution).
\end{definition}

\begin{lemma}For $\lambda>0$, denote by $\mu_\lambda$ the distribution $\mu(\cdot/\lambda)$. Then $F_{\mu_{\lambda}}(z)= \lambda F_{\mu_1}(z/\lambda)$ for all $z\in\Ha$.
\end{lemma}
\begin{proof}We have $G_{\mu_\lambda}(z)=G_{\mu_1}(z/\lambda)/\lambda$ and thus 
 $F_{\mu_{\lambda}}(z)= \lambda F_{\mu_1}(z/\lambda)$.
\end{proof}

\begin{theorem}[Monotone central limit theorem] Let $X_1,X_2,...$ be a sequence of monotonically independent identically distributed random variables with mean $\mu$ and variance $\sigma^2>0$. Let $\mu_n$ be the distribution of $S_n=(X_1+...+X_n-n\mu)/(\sigma\sqrt{n})$. Then $\mu_n$ converges weakly to the arcsine distribution $A(0,1)$ given by the density
\begin{equation*}
\frac{1}{\pi\sqrt{2 -x^2}}, \quad x\in (-\sqrt{2},\sqrt{2}).
\end{equation*}
\end{theorem}
\begin{proof}
Let $Y=(X_1-\mu)/\sigma$ and denote its distribution by $\nu$. By the previous lemma and Theorem \ref{monotone_convolution} we have  
\[F_{\mu_n}(z)=\frac1{\sqrt{n}} F_{\nu^{\rhd n}}(\sqrt{n}z) = (\frac1{\sqrt{n}} F_{\nu}\left(\sqrt{n}z)\right)^{\circ n}.\]
We need to show that $F_{\mu_n}(z)\to \sqrt{z^2-2}$ by Example \ref{ex_arcsine} and Lemma \ref{lemmaconvergence}. This can also be stated as
\[ \psi_1\circ F_{\mu_n} \circ  \psi_2 \to z-2\]
with $\psi_1(z)=z^2$ and $\psi_1(z)=\sqrt{z}$. 
Let  $F_n(z):=\frac1{\sqrt{n}} F_{\nu}(\sqrt{n}z)$ and write  $(\psi_1\circ F_{n} \circ  \psi_2)(z)=z+R_n(\psi_2(z))$, $z\in\C\setminus [0,\infty)$. We will show that \begin{equation}\label{Rsss}
 \lim_{n\to\infty} \sum_{j=0}^{n-1}R_n(F^{\circ j}_n(iy)) = -2 \end{equation}
for all $10<y<11$. Vitali's theorem implies locally uniform convergence of the sum in $\Ha$ and hence 
\[\psi_1\circ F_{\mu_n} \circ  \psi_2=
\psi_1\circ F_{n}^{\circ n} \circ  \psi_2=(\psi_1\circ F_{n} \circ  \psi_2)^{\circ n} =(z+R_n(\psi_2(z)))^{\circ n}=z+\sum_{j=0}^{n-1} R_n(F^{\circ j}_n(\sqrt{z}))\to z-2.\]

By \eqref{EV_finite_var}, we can write $F_\nu(z)=z+\int_\mathbb{R} \frac1{x-z}\tau({\rm d}x)$ for a non-negative measure $\tau$ with compact support and $\tau(\R)=1$ (the variance of $\nu$). We have 
\[F_n(z)=z+\frac1{\sqrt{n}}\int_\mathbb{R} \frac1{x-z\sqrt{n}}\tau({\rm d}x),\] and for $z\in\Ha$ we obtain
\begin{eqnarray*}
|F_n(z)-z|=\left|\frac1{\sqrt{n}}\int_\mathbb{R} \frac1{x-z\sqrt{n}}\tau({\rm d}x) \right| \leq 
\frac1{\sqrt{n}}\int_\R\left|  \frac1{x-\sqrt{n}z}\tau({\rm d}x) \right| \leq \frac{\tau(\R)}{\sqrt{n}} \frac1{\Im(z\sqrt{n})} = 
\frac{\sigma}{n\Im(z)}. 
\end{eqnarray*}

We have
\[R_n(z)=F_n(z)^2-z^2= \frac2{\sqrt{n}}\int_\mathbb{R} \frac{z}{x-\sqrt{n}z}\tau({\rm d}x)+\left(\frac1{\sqrt{n}}\int_\mathbb{R} \frac{1}{x-\sqrt{n}z}\tau({\rm d}x)\right)^2.\]

Let $y>10$. Then $\Im(F^{\circ j}_n(iy))>10$ for all $j=0,...,n-1$ and 
\begin{equation}\label{Rsss1} \sum_{j=0}^{n-1} \left(\frac1{\sqrt{n}}\int_\mathbb{R} \frac{1}{x-\sqrt{n}F^{\circ j}_n(iy)}\tau({\rm d}x)\right)^2 \leq  n \cdot \frac{1}{100 n^2}\to 0\end{equation}
as $n\to\infty$.

Let $y>10$ and put $z_j = F^{\circ j}_n(iy)$, $j=0,...,n-1$. Then also $10<\Im(z_j)$ for all $j=0,...,n-1$. We have
\begin{eqnarray*}&&\left|2+\sum_{j=0}^{n-1} \frac2{\sqrt{n}}\int_\mathbb{R} \frac{z_j}{x-\sqrt{n}z_j}\tau({\rm d}x)\right|=
\left|\sum_{j=0}^{n-1} \frac2{\sqrt{n}}\int_\mathbb{R} \frac1{\sqrt{n}}+ \frac{z_j}{x-\sqrt{n}z_j}\tau({\rm d}x)\right|\\
&=& \left|
\sum_{j=0}^{n-1} \frac2{\sqrt{n}}\int_\mathbb{R} \frac{x/\sqrt{n}}{x-\sqrt{n}z_j}\tau({\rm d}x)\right| \leq 
\sum_{j=0}^{n-1} \frac2{\sqrt{n}}\int_\mathbb{R} \frac{|x|/\sqrt{n}}{|x-\sqrt{n}z_j|}\tau({\rm d}x) \\
&\leq& \sum_{j=0}^{n-1} \frac2{\sqrt{n}}\int_\mathbb{R} \frac{|x|}{10n}\tau({\rm d}x)= \frac2{\sqrt{n}}\int_\mathbb{R} \frac{|x|}{10}\tau({\rm d}x) \to 0
\end{eqnarray*}
as $n\to \infty$. This shows that 
\[ \sum_{j=0}^{n-1} \frac2{\sqrt{n}}\int_\mathbb{R} \frac{z_j}{x-\sqrt{n}z_j}\tau({\rm d}x) \to -2\]
and, together with \eqref{Rsss1}, we conclude \eqref{Rsss}.
\end{proof}

\begin{theorem}\label{inf_div_mon}Let $\mu\in \mathcal{P}(\R)$. The following statements are equivalent:
 \begin{enumerate}
\item[(a)] $\mu$ is monotonically infinitely divisible.
\item[(b)] There exists a continuous $\rhd$-semigroup $(\mu_t)_{t\geq0}\subset \mathcal{P}(\R)$ ($\mu_0=\delta_0$, $\mu_{s+t}=\mu_s\rhd \mu_t$, $t\mapsto \mu_t$ is continuous) such that $\mu_1=\mu$.
\item[(c)] There exists a continuous semigroup $(F_t)_{t\geq0}$ of holomorphic functions $F:\Ha\to\Ha$ such that $F_1 = F_\mu$ and the generator $G:\Ha\to\Ha\cup \R$ of $(F_t)$ has the form 
\[G(z) = a + \int_\R \frac{1+xz}{x-z} \gamma(dx),\]
where $a\in\R$ and $\gamma$ is a non-negative finite measure. (Recall that $\frac{d}{dt}F_t = G(F_t)$.)
\end{enumerate}
\end{theorem}

\begin{remark}The pair $(a,\gamma)$ can be seen as the monotone analogue of the L\'{e}vy triple, see Theorem \ref{measures_incre_class}. 
\end{remark}

\begin{proof}[Sketch of the proof of Theorem \ref{inf_div_mon}]
If $(\mu_t)_{t\geq0}$ is a $\rhd$-semigroup as in (b), then $F_{\mu_t}$ is a continuous semigroup on $\Ha$. Exercise \ref{muhaha} implies that its generator $G$ maps $\Ha$ into $\Ha\cup \R$. Theorem \ref{grand_iteration} (b) and Theorem \ref{Julia} imply that $G$ has the form as in (c). Furthermore, $\mu_1 = \mu = \mu_{1/n} \rhd  \cdots \rhd  \mu_{1/n}$ for all $n\in\N$. Hence, $\mu$ is monotonically infinitely divisible.\\
If $(F_t)_{t\geq0}$ is a continuous semigroup as in (c), then Theorem \ref{Julia} and Theorem \ref{grand_iteration} (b) imply that 
all $F_t$ are $F$-transforms of probability measures $(\mu)_t$ and they clearly satisfy the conditions in (b).  \\
The difficult part is to show that (a) implies (b). This has been proven in \cite[Proposition 5.4]{Mur00} for compactly supported measures and later on in \cite[Proposition 3.8]{Bel05} for the general case. One can show that for each $n\in\N$ the measure $\nu_n$ with $(\nu_n)^{\rhd n}=\mu$ is uniquely determined. This allows us to define $\mu_{m/n}$, $n\in\N$, $m\in\N_0$, by $\mu_{m/n}=(\nu_{n})^{\rhd m}$ and now $\mu_t$ can be extended from all rational $t\geq0$  to all real $t\geq0$.
\end{proof}

\begin{remark}\label{oepf}The generators in (c) are exactly those  holomorphic $f:\Ha\to \Ha\cup\R$ whose Nevanlinna triple has the form $(a,0,\gamma)$. 
With Theorem \ref{cor_nevanlinna1} we see that, in particular, every $G(z) = \int_\R\frac1{x-z} \nu(dx)$, where $\nu\in \mathcal{P}(\R)$, is a generator of the form as in (c). 
\end{remark}

\section{$B$-transform and Boolean convolution}

\begin{definition}
Let $\mu \in \mathcal{P}(\R)$. The $B$-transform $B_\mu$ of $\mu$ is defined as 
\[B_\mu(z)=z-\frac1{G_\mu(z)}=z-F_\mu(z).\]
\end{definition}

\begin{example}\label{bhuuu}Let $\mu=\delta_0$. Then $B_\mu(z)=0$. More important is the example $\mu=\frac1{2}\delta_1+\frac1{2}\delta_{-1}$ for which we get 
\[B_\mu(z)=z-\frac2{1/(z-1)+1/(z+1)}=z-\frac{z^2-1}{z}=\frac{1}{z}.\]\hfill $\blacksquare$
\end{example}

\begin{theorem}\label{B_transform}Every $B$-transform maps $\Ha$ holomorphically into $\Ha^-\cup \R$.\\
Conversely, if $B:\Ha\to \Ha^-\cup \R$ is holomorphic, then the following statements are equivalent:
\begin{itemize}
	\item[(a)] $B=B_\mu$ for a probability measure $\mu\in\mathcal{P}(\R)$.
	\item[(b)] There exist  $a \in \R$ and a finite, non-negative measure $\gamma$ such that 
	\begin{equation*}
B(z) =   a + \int_{\R}\frac{1+xz}{z-x}\gamma(dx).
\end{equation*}
\end{itemize}
\end{theorem}
\begin{proof}Consider the function $-B_\mu(z)=\frac1{G_\mu(z)}-z$. As the Nevanlinna triple $(a,b,\gamma)$ of $1/G_\mu$ satisfies 
$b=1$, Corollary \ref{cor_nevanlinna2} implies that $-B_\mu$ maps $\Ha$ into $\Ha\cup \R$.\\
Thus, if (a) holds, then $-B$ has a Nevanlinna triple $(a,0,\gamma)$, which shows (b). \\
If (b) holds, then Theorem \ref{Julia} implies that $z-B(z)$ is the $F$-transform for some $\mu\in\mathcal{P}(\R)$. Hence $B=B_\mu$.
\end{proof}

\begin{theorem}[Boolean convolution]\label{Boolean_convolution} Let $X$, $Y$ be Boolean independent with distributions $\mu, \nu\in \mathcal{P}_c(\R)$.  Let $\alpha$ be the distribution of $X+Y$. Then 
\[\text{$B_{\alpha}(z)=B_\mu(z)+B_\nu(z)$ for all $z\in \Ha$.}\]
\end{theorem}
\begin{proof}
See \cite{SW97}.
\end{proof}

Theorem \ref{B_transform} shows that $B_\mu + B_\nu$ is again a $B$-transform for any $\mu,\nu\in\mathcal{P}(\R)$.

\begin{definition}For $\mu,\nu\in\mathcal{P}(\R)$ we define the \emph{Boolean convolution} $\mu \uplus \nu\in \mathcal{P}(\R)$ via 
\[  B_{\mu \uplus \nu} = B_\mu + B_\nu  .\]
A probability measure $\mu\in \mathcal{P}(\R)$ is called \emph{Boolean infinitely divisible} if, for every $n\in\N$, there exists $\mu_n\in \mathcal{P}(\R)$ such that 
$\mu = \mu_n \uplus  \cdots \uplus  \mu_n$ ($n$-fold convolution).
\end{definition}

\begin{lemma}For $\lambda>0$, denote by $\mu_\lambda$ the distribution $\mu(\cdot/\lambda)$. Then $B_{\mu_{\lambda}}(z)=\lambda B_{\mu}(z/\lambda)$ for all $z\in\Ha$.
\end{lemma}
\begin{proof}We have $G_{\mu_\lambda}(z)=G_{\mu_1}(z/\lambda)/\lambda$ and thus 
$B_{\mu_\lambda}(z)=z-\lambda/G_{\mu_1}(z/\lambda)=\lambda B_{\mu_1}(z/\lambda)$.
\end{proof}

\begin{theorem}[Boolean central limit theorem] Let $X_1,X_2,...$ be a sequence of Boolean independent identically distributed random variables with mean $\mu$ and variance $\sigma^2>0$. Let $\mu_n$ be the distribution of $S_n=(X_1+...+X_n-n\mu)/(\sigma\sqrt{n})$. As $n\to\infty$,  $\mu_n$ converges weakly to the distribution
 \[\frac1{2}\delta_1 + \frac1{2}\delta_{-1}.\]
\end{theorem}
\begin{proof}
Let $Y=(X_1-\mu)/\sigma$ and denote its distribution by $\nu$. Let $B_\nu(z)=\sum_{k=1}^\infty \kappa_{k}/z^{k-1}$. Then $\kappa_1=0$ and $\kappa_2=1$. By the previous lemma and Theorem \ref{Boolean_convolution} we have  
\[B_{\mu_n}(z)=\frac1{\sqrt{n}}B_{\nu^{\uplus n}}(\sqrt{n}z)=\frac{n}{\sqrt{n}}B_{\nu}(z/\sqrt{n})=
1/z + \sum_{k=3}^\infty \kappa_{k}n^{1-k/2}/z^{k-1}\to \frac1{z}\]
as $n\to\infty$ locally uniformly in $\Ha$. By Lemma \ref{lemmaconvergence} and Example \ref{bhuuu} we have $\mu_n\to\frac1{2}\delta_1 + \frac1{2}\delta_{-1}$ as $n\to\infty$.
\end{proof}

\begin{theorem}\label{measures_incre_Boolean} Every probability measure $\mu\in \mathcal{P}(\R)$ is infinitely divisible with respect to Boolean convolution.
\end{theorem} 
\begin{proof}For $n\in\N$, let $B_n(z)=B_\mu(z)/n$. Then Theorem \ref{B_transform} implies that $B_n=B_{\mu_n}$ for 
some $\mu_n\in\mathcal{P}(\R)$. Thus $\mu = \mu_n^{\uplus n}$.
\end{proof}

\section{$R$-transform and free convolution}

Let $\mu\in \mathcal{P}_c(\R)$. The Cauchy transform $G_\mu$ is analytic in a neighborhood of $\infty$ and $G_\mu'(\infty)\not=0$. So we can define its compositional right inverse $G_\mu^{-1}$ in a neighborhood $U_\mu$ of $0$, i.e.\ $(G_\mu \circ G_\mu^{-1})(z)=z$ for all 
$z\in U_\mu$. For all $z\in U_\mu$, we define the $R$-transform $R_\mu$ as $R_\mu(z)=G_\mu^{-1}(z)-\frac1{z}$. \\

For a general $\mu\in \mathcal{P}(\R)$, $G_\mu$ can be inverted within a set of the form $\Gamma_{\alpha,\beta}=\{z\in \Ha\,|\, \Re(z)<\alpha\Im(z), |z|>\beta \}$ (or in the corresponding set in the lower half-plane), for some $\alpha, \beta>0$ and $R_\mu(z)=G_\mu^{-1}(z)-\frac1{z}$ is defined as a holomorphic function that maps some subdomain of the lower half-plane into the lower half-plane or into $\R$; see \cite[Section 5]{BV93}.

\begin{remark}The transform $R_\mu(1/z)$ is usually called Voiculescu transform and often denoted by $\varphi_\mu$. 
\end{remark}

\begin{example}\label{ex_semicircle}Recall the semicircle distribution $W(0,\sigma^2)$ given by the density 
\[\frac{1}{2\pi\sigma^2} \sqrt{4\sigma^2 -x^2}, \quad  x\in [-2\sigma,2\sigma].\] 
We have $G_{W(0,1)}(z)=\frac{z-\sqrt{z^2-4}}{2}$, where the branch of the square root is chosen such that $\sqrt{\cdot}$ maps $\C^2\setminus [-4,\infty)$ into the upper half-plane. 
Solving the equation $\frac{z-\sqrt{z^2-4}}{4}=w$ gives $z=w+\frac{1}{w}$ and thus \[R_{W(0,1)}(z)= z.\]\hfill $\blacksquare$
\end{example}

 \begin{figure}[H]
 \begin{center}
 \includegraphics[width=0.9\textwidth]{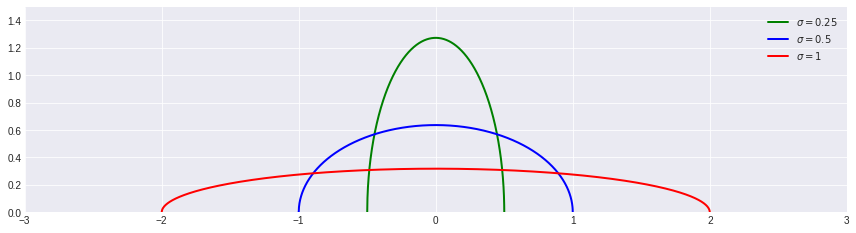}
 \caption{Densities of some centered semicircle distributions.}
 \end{center}
 \end{figure}

\begin{theorem}[Free convolution]\label{free_convolution}
Let $X, Y$ be freely independent with distributions  $\mu, \nu\in \mathcal{P}_c(\R)$. Let $\alpha$ be the distribution of $X+Y$. Then 
\[R_{\alpha}(z)=R_\mu(z)+R_\nu(z)\]
 for all $z$ in some neighborhood of $0$.
\end{theorem}
\begin{proof}See \cite{Voi86}.
\end{proof}

For general $\mu,\nu\in\mathcal{P}(\R)$, it is shown in \cite[Corollary 5.8]{BV93} that $R_\mu(z)+R_\nu(z)$ is again an $R$-transform of some probability measure.

\begin{definition}For $\mu,\nu\in\mathcal{P}(\R)$ we define the \emph{free convolution} $\mu \boxplus \nu\in \mathcal{P}(\R)$ via 
\[  R_{\mu \boxplus \nu} = R_\mu + R_\nu  .\]
A probability measure $\mu\in \mathcal{P}(\R)$ is called \emph{freely infinitely divisible} if, for every $n\in\N$, there exists $\mu_n\in \mathcal{P}(\R)$ such that 
$\mu = \mu_n \boxplus  \cdots \boxplus  \mu_n$ ($n$-fold convolution).
\end{definition}

\begin{lemma}Let $X$ be a quantum random variable with distribution $\mu$. For $\lambda>0$, denote by $\mu_\lambda$ the distribution of $\lambda X$. Then $R_{\mu_{\lambda}}(z)= \lambda R_{\mu_1}(\lambda z)$ for all $z$ in a neighborhood of $0$.
\end{lemma}
\begin{proof}We have $G_{\mu_\lambda}(z)=G_{\mu_1}(z/\lambda)/\lambda=:G(z)$ and from
$G_{\mu_1}(R_{\mu_1}(\lambda z)+1/(\lambda z)) = \lambda z$ we see that 
$G(\lambda R_{\mu_1}(\lambda z)+1/z)=z$. Hence $R_{\mu_{\lambda}}(z)= \lambda R_{\mu_1}(\lambda z)$.
\end{proof}

\begin{theorem}[Free central limit theorem] Let $X_1,X_2,...$ be a sequence of freely independent identically distributed random variables with mean $\mu$ and variance $\sigma^2>0$. Let $\mu_n$ be the distribution of $S_n=(X_1+...+X_n-n\mu)/(\sigma\sqrt{n})$. Then $\mu_n$ converges weakly to the semicircle distribution $W(0,1)$ as $n\to\infty$.
\end{theorem}
\begin{proof}
Let $Y=(X_1-\mu)/\sigma$ and denote its distribution by $\nu$. Let $R_\nu(z)=\sum_{k=0}^\infty \kappa_{k+1}z^k$. Then $\kappa_1=0$ and $\kappa_2=1$. By the previous lemma and Theorem \ref{free_convolution} we have  
\[R_{\mu_n}(z)=\frac1{\sqrt{n}}R_{\nu^{\boxplus n}}(z/\sqrt{n})=\frac{n}{\sqrt{n}}R_{\nu}(z/\sqrt{n})=z + \sum_{k=2}^\infty \kappa_{k+1}z^k/n^{(k-1)/2}\to z\]
as $n\to\infty$ locally uniformly in a neighborhood of $0$. By Lemma \ref{lemmaconvergence} and Example \ref{ex_semicircle} we have $\mu_n\to W(0,1)$ as $n\to\infty$.
\end{proof}

\begin{theorem}[Theorem 5.10 in \cite{BV93}] \label{thmBV93}
For a probability measure $\mu$ on $\R$, the following statements are equivalent.
\begin{enumerate}[\rm(1)]
\item $\mu$ is freely infinitely divisible.
\item $\mu=\mu_1$  for a $\boxplus$-semigroup $\{\mu_t\}_{t\geq0}$
 (i.e.\ $\mu_0=\delta_0$, $\mu_{t+s}=\mu_t \boxplus \mu_s$ for all $s,t\geq0$ and $t\mapsto\mu_t$ is continuous with respect to weak convergence).
\item $R_\mu(1/z)$ extends to a holomorphic function from $\Ha$ into $\Ha^- \cup \R$.
\item\label{FLK1} There exist $a \in \R$ and a finite, non-negative measure $\gamma$ on $\R$ such that
\begin{equation}\label{form0}
\R_\mu(1/z)=a +\int_{\mathbb{R}}\frac{1+z x}{z-x} \gamma({\rm d}x) ,\qquad z\in \mathbb{H}.
\end{equation}
\end{enumerate}
\end{theorem} 

\begin{example}The semicircle distribution $W(0,1)$ is infinitely divisible as $R_{W(0,1)}(1/z)=1/z$ extends analytically to the whole upper half-plane. Thus also $F^{-1}_{W(0,1)}(z)=z+\frac1{z}$ extends analytically to the upper half-plane. While $F_{W(0,1)}(\Ha)$ is the complement of a half-disc in $\Ha$, $F^{-1}_{W(0,1)}(\Ha)=\C\setminus ((-\infty, -2] \cup [2,\infty))$. 
 \begin{figure}[H]
 \begin{center}
 \includegraphics[width=0.9\textwidth]{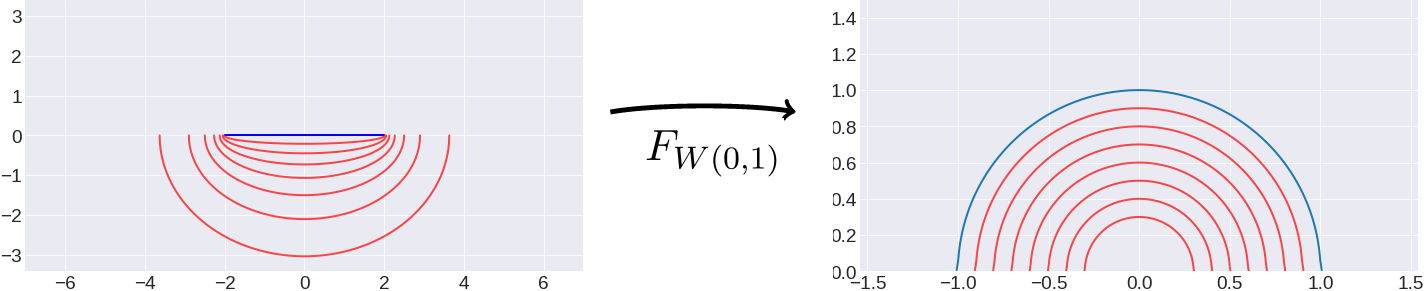}
 \caption{$F$-transform of the semicircle distribution $W(0,1)$.}
 \end{center}
 \end{figure}
\vspace{-8mm}\hfill $\blacksquare$
\end{example}

\begin{example}
The normal distribution $\mathcal{N}(0,1)$ has been shown to be freely infinitely divisible in \cite{BBLS11}.\hfill $\blacksquare$
\end{example}

\section{Further reading}

Many further examples of $\rhd$- and $\boxplus$-infinitely divisible distributions can be found in \cite{FHS}.\\

The five convolutions can also be approached from a combinatorial point of view. Suppose $\mu$ has compact support, then all moments $m_1,m_2,...$ of $\mu$ exist. Furthermore, we can expand the $R$-transform as $R_\mu(z)=\sum_{k=0}^\infty \kappa_{n+1}z^n$ and the numbers $\kappa_1, \kappa_2,...$ are called the \emph{free cumulants} of $\mu$. One can show that $\kappa_n$ can be written as a polynomial in $m_1,...,m_n$. For example,
\[\kappa_1=m_1, \quad \kappa_2=m_2-m_1^2, \quad \kappa_3=m_3-3m_2m_1+2m_1^3.\]
The free convolution of two probability measures $\mu,\nu\in\mathcal{P}_c(\R)$ can now be expressed as 
\[\text{$\kappa_{n}(\mu \boxplus \nu) = \kappa_{n}(\mu)+\kappa_{n}(\nu)$ for all $n\in\N$.}\]

In the classical case, cumulants of a bounded random variable $X$ are defined via the power series expansion of the function $\log \mathbb{E}[e^{tX}]$, see \cite[p.147]{bill}. We refer to \cite{NS10} for the free cumulants and to \cite{SW97} for the Boolean cumulants. The monotone case is studied in \cite{HS11}. 

\newpage

\section{Exercises}

\begin{exercise}Let $\star\in\{\boxplus, \uplus\}$. Prove that if $\mu, \nu$ are $\star$-infinitely divisible, then also 
$\mu\star\nu$ is $\star$-infinitely divisible.
\end{exercise}

\begin{exercise}Determine all $\mu\in \mathcal{P}(\R)$ having  a Boolean L\'{e}vy pair of the form $(0,b\cdot \delta_c)$ with $c\in \R$, $b\geq 0$, i.e.\ $B_\mu(z)=b\frac{1+cz}{z-c}$.
\end{exercise}

\begin{exercise}Determine $F_\mu$ for all $\rhd$-infinitely divisible $\mu$ such that the generator $G$ in Theorem \ref{inf_div_mon} has the form $G(z)=\frac{1}{x-z}$, $x\in\R$ (see Remark \ref{oepf}).
\end{exercise}

\begin{exercise}[Boolean Poisson limit theorem] 
For $n\in\N$, let $(X_{k,n})_{1\leq k\leq n}$ be Boolean \textit{iid} random variables with distribution $c_n\delta_{1} + (1-c_n)\delta_0$ and assume that $nc_n \to c>0$ as $n\to\infty$.  Determine the limit distribution of $X_{1,n}+...+X_{n,n}$.
\end{exercise}

\begin{exercise}[Free Poisson limit theorem] 
For $n\in\N$, let $(X_{k,n})_{1\leq k\leq n}$ be freely \textit{iid} random variables with distribution $c_n\delta_{1} + (1-c_n)\delta_0$ and assume that $nc_n \to c>0$ as $n\to\infty$.  Show that that limit distribution $\mu$ of $X_{1,n}+...+X_{n,n}$ exists and satisfies 
\[ G_\mu(z) = \frac{z+1-c - \sqrt{c^2-2c(z+1)+(z-1)^2}}{2z}. \]
(A \emph{Marchenko--Pastur} distribution.)
\end{exercise}

The monotone Poisson limit theorem can be found in \cite{Mur01}.

\newpage 

\chapter{Quantum stochastic processes}

Let $(B(H), \varphi)$ be a quantum probability space. We define a quantum stochastic process simply as a collection 
$(X_t)_{t\in T}$ of random variables in $B(H)$, where $T=\N_0$ or $T=[0,\infty)$.

\begin{example}
 Any sequence $X_0, X_1, X_2 , ... \in \C^{N\times N}$ of matrices can be regarded as a quantum process via Example \ref{999}, which may  or may not be a fruitful idea, depending on the concrete application. As matrices are really everywhere, there is at least a lot of potential of applying quantum probability theory in this way.  \hfill $\blacksquare$
\end{example}

\begin{example}
 Consider a multivariate time series $x_0, x_1, ..., x_n \in \R^{N}$, which might come from $N$ sensors for a certain machine.   Then we might construct a model for the time series by a random vector $X \in \R^N$. Fix some $m\in \N$ and let $R_{n}$ be sample covariance matrix calculated from the observations $x_n,...x_{n+m}$. According to our model, we obtain random matrices $R_0, R_1, ...$, which can be considered as a quantum process by Example \ref{444}. The matrices can be used for monitoring the functioning of the machine.
 If some pattern of the eigenvalues, which we derive from our model, suddenly breaks down, we might have registered some signal predicting an imminent failure of the machine. \hfill $\blacksquare$
\end{example}

\begin{example}
In quantum mechanics, the state of a particle is represented as a unit vector $\psi_0$ in some Hilbert space. 
	An observable quantity is represented by a (possibly unbounded) self-adjoint operator $A$. One obtains a dynamical description of this particle by choosing a Hamiltonian $H$, which is again a self-adjoint operator, and now the state $\psi_t$ depends on $t$ 
	(Schr\"odinger picture) and satisfies  the Schr\"odinger equation 
\[\displaystyle i\hbar {\frac {\partial }{\partial t}}\psi_t =H \psi_t.\]
In a mathematically equivalent way, one can leave the state $\psi_0$ constant and let the observable $A_t$ depend on $t$ (Heisenberg picture). In this way the dynamics can be described by the quantum process $(A_t)_{t\geq 0}$, which satisfies the Heisenberg equation 
\[\displaystyle {\frac {d}{dt}}A_t={\frac {i}{\hbar }}[H,A_t] = {\frac {i}{\hbar }}(HA_t-A_tH).\]
More on quantum stochastic processes in quantum mechanics can be found in \cite{one, barndorff-nielsen+al, Mey95, Att}. 
See also \cite{Kem03} for an overview article on quantum random walks.\hfill $\blacksquare$
\end{example}

\section{Additive processes}

\begin{definition}[Additive processes]
Let $(B(H), \varphi)$ be a quantum probability space and 
$(X_t)_{t\geq 0}\subset B(H)$ a family of random variables. 
 We call  $(X_t)_{t\geq 0}$ a (free, Boolean, monotone, anti-monotone) \emph{additive process} if the following conditions are satisfied:
\begin{itemize} 
 \item[(1)] $X_0 = 0$.
	\item[(2)]
	In case of Boolean or free independence: The random variables 
	 \[
X_{t_1},X_{t_2}-X_{t_1},\ldots,X_{t_n}-X_{t_{n-1}}
\]
are independent for all $n\in\mathbb{N}$ and all $t_1,\ldots,t_n\in\mathbb{R}$ with \ $0\leq t_1\leq t_2\leq\cdots\leq t_n$.\\
	In case of monotone or anti-monotone independence: The tuple
    \[
(X_{t_1},X_{t_2}-X_{t_1},\ldots,X_{t_n}-X_{t_{n-1}})
\]
is independent for all $n\in\mathbb{N}$ and all $t_1,\ldots,t_n\in\mathbb{R}$ with \ $0\leq t_1\leq t_2\le\cdots\leq t_n$.
\item[(3)] The mapping $(s,t)\mapsto \mu_{s,t}$ is continuous with respect to \ weak convergence, where $\mu_{s,t}$ denotes the distribution of the increment $X_t-X_s$, $0\leq s\leq t$.
\end{itemize}
Such a process is called a (free, Boolean, monotone, anti-monotone) \emph{L\'evy process} if, in addition, 
 \begin{enumerate}
 \item[(4)] the distribution of $X_{t+s}-X_s$ does not depend on $s$.
\end{enumerate}
\end{definition}

\subsection{Hemigroup distributions}

Consider an additive process $(X_t)_{t\geq 0}$ and let $\star \in \{\uplus, \boxplus,\rhd,\lhd \}$ be the convolution associated to the independence. By writing 
$X_{u}-X_s =  (X_t-X_s) + (X_u - X_t)$ for $0\leq s\leq t\leq u$, we see that
\[ \text{$\mu_{s,u} = \mu_{s,t}\star \mu_{t,u}$.} \]
 If the increments are also stationary, then we have $\mu_{s,t}=\mu_{0,t-s}=:\mu_{t-s}$ and 
\[ \mu_{s+t} = \mu_{t} \star \mu_s. \]

This motivates the following definition.

\begin{definition}Fix a convolution $\star \in \{*, \uplus, \boxplus,\rhd,\lhd \}$. A family $(\mu_{s,t})_{0\leq s\leq t}\subset \mathcal{P}(\R)$ is called a (continuous) \emph{$\star$-hemigroup} if the following conditions are satisfied:
\begin{itemize}
	\item[(1)] $\mu_{s,s}=\delta_0$ for all $s\geq 0$,
	\item[(2)] $(s,t)\mapsto \mu_{s,t}$ is continuous with respect to weak convergence,
	\item[(3)] $\mu_{s,u} = \mu_{s,t} \star \mu_{t,u}$ for all $0\leq s\leq t\leq u$.
\end{itemize} 
If, in addition, $\mu_{s,t}=\mu_{0,t-s}$ for all $0\leq s\leq t$, then $(\mu_{t})_{t\geq0}$ is called a 
(continuous) \emph{$\star$-semigroup}.
\end{definition}

\begin{lemma}\label{vvgg}Let $(X_t)_{t\geq 0}$ be an additive process and let $\mu_{s,t}$ be the distribution of $X_t-X_s$. Then 
 $(\mu_{s,t})_{0\leq s\leq t}\subset \mathcal{P}(\R)$  is a \emph{$\star$-hemigroup}, where 
$\star \in \{\uplus, \boxplus,\rhd,\lhd \}$ is the convolution associated to the notion of independence. 
If $(X_t)$ is a L\'evy process, then $(\mu_{t})_{t\geq0}$ is a \emph{$\star$-semigroup}.
\end{lemma}
 
We have already seen that the $\star$-semigroup distributions correspond to the $\star$-infinitely divisible distributions, which can be encoded by L\'{e}vy triples/pairs. (In the Boolean case, also every probability measure $\mu\in \mathcal{P}(\R)$ can be embedded into a semigroup, and thus also into a hemigroup, see the proof of Theorem \ref{measures_incre_Boolean}.) \\

The following natural question arises:\\
\[\text{Which probability measures arise in $\star$-hemigroups?}\]

\vspace{2mm}

\begin{theorem}\label{measures_incre_free}
Let $\star \in \{*,\uplus, \boxplus\}$. If $(\mu_{s,t})_{t\geq0}$ is a $\star$-hemigroup, then each $\mu_{s,t}$ is $\star$-infinitely divisible.
\end{theorem}
\begin{proof}In case $\star = \uplus$, every $\mu\in \mathcal{P}(\R)$ is infinitely divisible, see Theorem \ref{measures_incre_Boolean}.\\

Now fix $0\leq s \leq t$ and let $\mu_{j,n}=\mu_{s+(j-1)\cdot (t-s)/n, s+j\cdot (t-s)/n}$ for $n\in\N$, $j=1,...,n$. Then 
$\mu_{1,n} \star \ldots \star \mu_{n,n} = \mu_{s,t}$.
 As $(r,\tau)\mapsto \mu_{r,\tau}$ is continuous, and as $S=\{(r,\tau)\,|\, s \leq r \leq \tau \leq t\}$ is a compact set, 
$(r,\tau)\mapsto \mu_{r,\tau}$ is uniformly continuous on $S$. Hence, for $\eps>0$ and $\delta\in(0,1)$, we find $N\in \N$ such that 
$\mu_{j,n}((-\eps, \eps))>1-\delta$ for all $n\geq N$ and $j=1,...,n$.\\

In the case $\star=*$, Theorem \ref{measures_incre_class}, direction (d) $\rightarrow$ (c), implies that $\mu_{s,t}$ is $*$-infinitely divisible. The same conclusion is possible in the case $\star=\boxplus$, see \cite[Theorem 1]{BP00}. 
\end{proof}

The $\rhd$-hemigroups will be handled in Chapter \ref{sec_8} as they require more work and lead us deeper into complex analysis. We will see that there are more hemigroup distributions than semigroup distributions in this case. Furthermore, the evolution of the measure-valued process $t\mapsto \mu_{0,t}$ turns out to have a nice geometric interpretation, namely growth processes in $\Ha$ described by univalent functions. \\
Surprisingly, this geometric interpretation is also true for $\boxplus$-hemigroups, which we show in Chapter \ref{sec9}.

\subsection{Construction of monotone additive processes}

We now construct a monotone additive process from a given compactly supported hemigroup. 

\begin{definition}\label{def_equi_Markov}
A probability kernel $k$ on $\R$ is called \emph{$\rhd$-homogeneous} if it satisfies
\[
\delta_x\rhd k(y,\,\cdot\,) = k(x+y,\,\cdot\,)
\]
for all $x,y\in\mathbb{R}$. A classical Markov process $(M_t)_{t\geq0}$ on $\R$ is called a \emph{$\rhd$-homogeneous Markov process} if its transition kernels $(k_{s,t})_{0\leq s\leq t}$
 satisfy the following two conditions:
\begin{itemize}
\item[(a)] The mapping $(s,t)\mapsto k_{s,t}(x,\cdot)$ is continuous with respect to weak convergence for all $x\in\R$.
\item[(b)] The kernel $k_{s,t}$ is \emph{$\rhd$-homogeneous} for all $0\leq s\leq t$.
\end{itemize}
\end{definition}

\begin{theorem}\label{inf}${}$
 \begin{enumerate}[\rm(1)]
\item Let $(\mu_{s,t})_{0\leq s\leq t}$ be a $\rhd$-hemigroup. Then there exists a $\rhd$-homogeneous Markov process $(M_t)_{t\geq0}$ with transition kernels $k_{s,t}$ such that $k_{s,t}(x,\cdot) = \delta_x\rhd \mu_{s,t}(\cdot)$, i.e.
\begin{equation}\label{eq:additive_equivariance}
\int_\mathbb{R} \frac{1}{z-y}\,k_{s,t}(x,{\rm d}y) = \frac{1}{F_{\mu_{s,t}}(z)-x}.
\end{equation}
\item Let $(M_t)_{t\geq 0}$ be a $\rhd$-homogeneous Markov process such that $M_0=0$ and let $(\mathcal{F}_t)_{t\geq 0}$ be its natural filtration ($\mathcal{F}_t=\sigma((M_s)_{s\leq t})$). Assume that all distributions of $M_t$ have compact support. 
 Denote by $(\Omega,\mathcal F, \mathbb P)$ the underlying probability space and by $P_t$ the conditional expectation
\[
P_t = \mathbb{E}[\,\cdot\,|\mathcal{F}_t], \qquad t\geq 0.
\]
Define the operators $(X_t)_{t\geq 0}$ by
\begin{align}\label{def_SAIP}
X_t=P_t M_t. 
\end{align}
Then $(X_t)_{t\geq 0}$ is a monotone increment process on $(B(L^2(\Omega,\mathcal{F},\mathbb P)), \left<\mathbf{1}_\Omega,\cdot \mathbf{1}_\Omega\right>)$, where $M_t$ acts by multiplication on $L^2(\Omega,\mathcal{F},\mathbb P)$ and $\mathbf{1}_\Omega$ is the constant function with value $1$ on $\Omega$.
\end{enumerate}
\end{theorem}
\begin{proof}${}$
 \begin{enumerate}[\rm(1)]
\item We define $k_{s,t}(x,\cdot):=\delta_x \rhd \mu_{s,t}$.  We would like to apply Theorem \ref{existence_Markov}, which gives the existence of the Markov process. Thus we need to verify that $x\mapsto k_{s,t}(x,B)$ is a measurable function for all $B\in\mathcal{B}(\R)$, $0\leq s\leq t$, and that $(k_{s,t})$ satisfies the Chapman-Kolmogorov equation.\\

Concerning the measurability of $x\mapsto k_{s,t}(x,B)$, by the inversion formulas \eqref{Stieltjes2} and \eqref{eq:atom2}, we have
\begin{align*}
k_{s,t}(x,\{\alpha\}) &= \lim_{\epsilon\downarrow0} \frac{i\epsilon}{F_{s,t}(\alpha+i\epsilon)-x}, \\
\frac{1}{2}k_{s,t}(x,\{\alpha\})+ \frac{1}{2}k_{s,t}(x,\{\beta\})+k_{s,t}(x,(\alpha,\beta)) &= -\frac{1}{\pi}\lim_{\epsilon\downarrow0}\int_\alpha^\beta \frac{1}{F_{s,t}(y+i\epsilon)-x}\,{\rm d}y, 
\end{align*}
which implies that $x\mapsto k_{s,t}(x,B)$ is measurable for open and closed intervals $B$. By using the monotone convergence theorem, we obtain the measurability for all Borel sets.\\

 For $0 \leq s \leq t \leq u$, by using \eqref{eq:additive_equivariance}, we get
\begin{align*}
&\int_{\R} \frac{1}{z-w} \int_{\R} k_{s,t}(x,{\rm d}y) k_{t,u}(y,{\rm d}w)
= \int_{\R} \frac{1}{F_{t,u}(z)-y}\, k_{s,t}(x,{\rm d}y)\\ &= \frac{1}{F_{s,t}(F_{t,u}(z))-x}
=  \frac{1}{F_{s,u}(z)-x}
=  \int_{\R} \frac{1}{z-w}\, k_{s,u}(x,{\rm d}w).
\end{align*}
 Now we obtain the Chapman-Kolmogorov relation via the Stieltjes-Perron inversion.\\  

The $\rhd$-homogeneity of the transition kernels follows from the calculation
$$
F_{\delta_x \rhd k_{s,t}(y,\cdot)}(z) = F_{k_{s,t}(y,\cdot)}(z)-x =  F_{s,t}(z)-(x+y) = F_{k_{s,t}(x+y,\cdot)}(z).
$$
Finally, Lemma \ref{lemmaconvergence} implies the weak continuity of $(s,t)\mapsto k_{s,t}(x,\cdot)$ for all $x\in\R$.

\item The Markov property and \eqref{eq:additive_equivariance} imply
\begin{equation*}
\mathbb E\left[\frac{1}{z-M_t} \Bigg| \mathcal F_s\right] =\frac{1}{F_{s,t}(z)-M_s}~~a.s.,
\end{equation*}
which can also be stated as 
\begin{equation}\label{eq:additive_markovianity}
P_s \frac{1}{z-M_t}P_s = \frac{1}{F_{s,t}(z)-M_s}P_s. 
\end{equation}
As $M_0=0$, we have $X_0=0$.\\

Let $\mu_{s,t}:=k_{s,t}(0,\cdot)$. Next we show that
\begin{equation}\label{eq:additive_sandwich}
P_s \left(z-(X_t-X_s)\right)^{-1} P_s = G_{\mu_s,t}(z) P_s
\end{equation}
for $0 \leq s \leq t$ and $z\in \Ha$. By applying \eqref{eq:additive_sandwich} to the constant function $\mathbf1_\Omega$ and taking the expectation, we obtain that the distribution of $X_t-X_s$ with respect to the expectation $\langle \mathbf{1}_\Omega, \cdot \mathbf{1}_\Omega \rangle$ is equal to $\mu_{s,t}$.\\

 Using properties of the conditional expectation and property \eqref{eq:additive_markovianity}, for $\psi \in L^2(\Omega,\mathcal{F},\mathbb P)$ we obtain
\begin{align*}
&P_s \big(z-(P_tM_t-P_sM_s)\big)^{-1}P_s\psi \\
&= \frac{1}{z}P_s\psi + P_s\frac{M_t}{z(z-M_t)}P_s\psi - P_s\frac{ M_s\big(F_{\mu_{s,t}}(z)-M_s\big)}{F_{\mu_{s,t}}(z)(z-M_t)}P_s\frac{1}{z-M_t}P_s\psi\\
&=  \frac{1}{z}P_s\psi + \frac{1}{z}\left( P_s\psi \right) \cdot \left( P_s\frac{M_t}{z-M_t} \mathbf{1}_\Omega\right) - \left(\frac{M_s(F_{\mu_{s,t}}(z)-M_s)}{F_{\mu_{s,t}}(z)} P_s\psi \right) \cdot \left(P_s\frac{1}{z-M_t} \mathbf{1}_\Omega \right)^2 \\
&=  \frac{1}{z}P_s\psi
+\frac{1}{z}\left( P_s\psi \right) \cdot \left( - \mathbf{1}_\Omega+ \frac{z}{F_{\mu_{s,t}}(z)-M_s} \mathbf{1}_\Omega\right) \\&\qquad- \left(\frac{M_s(F_{\mu_{s,t}}(z)-M_s)}{F_{\mu_{s,t}}(z)}P_s\psi \right) \cdot \left(\frac{1}{F_{\mu_{s,t}}(z)-M_s} \mathbf{1}_\Omega \right)^2\\
&=  
\frac{1}{F_{\mu_{s,t}}(z)-M_s} P_s\psi - \frac{M_s}{F_{\mu_{s,t}}(z)(F_{\mu_{s,t}}(z)-M_s)}P_s\psi =  \frac1{F_{\mu_{s,t}}(z)}P_s \psi=  G_{\mu_s,t}(z)P_s \psi.
\end{align*}

The mapping $(s,t)\mapsto k_{s,t}(0,\cdot)$ is continuous by assumption. 
It remains to show that $(X_t)$ has monotonically independent increments.\\

For $0\leq s\leq t\leq u$ and a polynomial $p:\R\to \R $  we have

\begin{equation}\label{lem-mon-relations}
\begin{aligned}
p(X_t-X_s)P_u = P_u p(X_t-X_s) = p(X_t-X_s),\\
P_s p(X_u-X_t) P_s = \left\langle \mathbf{1}_\Omega, p(X_u-X_t)\mathbf{1}_\Omega \right\rangle P_s.
\end{aligned}
\end{equation}

Step 1. 
Let $t,s,t',s',t'',s''\in\mathbb{R}$ such that $0\leq s'\leq t'\leq s\leq t$ and $0\leq s''\leq t''\leq s\leq t$. Let $f,g,h:\R\to \R$ be polynomials and set $X=f(X_{t'}-X_{s'})$, $Y=g(X_t-X_s)$, $Z=h(X_{t''}-X_{s''})$. From \eqref{lem-mon-relations} we get
\[
XYZ = XP_{t'}YP_{t''}Z = \left\langle \mathbf{1}_\Omega, Y\mathbf{1}_\Omega \right\rangle_{L^2(\Omega)}XZ.
\]
This shows condition (i) from Remark \ref{def-mon}.\\

Step 2.
Let $t_1,\ldots,t_p, s_1,\ldots, s_p, t'_1,\ldots,t'_q, s'_1,\ldots, s'_q,t,s\geq0$ be such that
\[
t_1\geq s_1\geq t_2\geq  \cdots \geq t_p\geq s_p\geq t\geq s \leq t \leq t_q' \leq \cdots \leq t_2' \leq s_1' \leq t_1',
\]
let $f_1,\ldots, f_p,g, h_1,\ldots,h_q:\R\to\R$ be polynomials and set
\begin{gather*}
W_1=f_1(X_{t_1}-X_{s_1}), \ldots,\, W_p=f_p(X_{t_p}-X_{s_p}),\, Y=g(X_t-X_s), \\
Z_1=h_1(X_{t'_1}-X_{s'_1}),\ldots,\, Z_q=h_q(X_{t'_q}-X_{s'_q}).
\end{gather*}
Then we obtain from \eqref{lem-mon-relations} that
\begin{align*}
&\langle \mathbf{1}_\Omega, W_1\cdots W_pYZ_q\cdots Z_1 \mathbf{1}_\Omega \rangle\\
&= \langle \mathbf{1}_\Omega, P_{t_2}W_1 P_{t_2} W_2 W_3 \cdots W_p Y Z_q  \cdots Z_2 P_{t'_2} Z_1 P_{t'_2} \mathbf{1}_\Omega \rangle \\
&= \langle \mathbf{1}_\Omega, W_1\mathbf{1}_\Omega\rangle
\langle \mathbf{1}_\Omega, W_2 W_3 \cdots W_p Y Z_q  \cdots Z_2 \mathbf{1}_\Omega \rangle \langle \mathbf{1}_\Omega, Z_1 \mathbf{1}_\Omega \rangle \\
&= \cdots \\
&=  \langle \mathbf{1}_\Omega, W_1\mathbf{1}_\Omega\rangle  \cdots \langle\mathbf{1}_\Omega , W_p\mathbf{1}_\Omega\rangle \langle\mathbf{1}_\Omega ,Y\mathbf{1}_\Omega\rangle \langle\mathbf{1}_\Omega Z_q\mathbf{1}_\Omega\rangle \cdots \langle \mathbf{1}_\Omega, Z_1 \mathbf{1}_\Omega \rangle,
\end{align*}
and we have shown condition (ii) from Remark \ref{def-mon}. Hence, $X_t$ has monotonically independent increments.

\end{enumerate}
\end{proof}

\begin{remark}The same construction holds for unbounded additive processes from hemigroups with unbounded support, see \cite{FHS}.\\
 In \cite{MR1462227}, Muraki constructed a monotone Brownian motion, i.e.\ an additive process such that $X_t$ is $A(0,t)$-distributed, on a monotone Fock space, see also \cite{MR1483010, MR1455615}. Monotone L\'evy processes consisting of bounded self-adjoint operators have been constructed in \cite[Theorem 4.1]{franz+muraki04}. Monotone additive processes (bounded, but with operator-valued expectation) have also been constructed in \cite{Jek17}. 
\end{remark}

The classical case of the construction of (unbounded) additive processes is handled by Theorem \ref{existence_Markov}. For the free and Boolean case, we refer to \cite[p.112]{barndorff-nielsen+al}. 

\begin{remark}
A free additive process with distributions $\mu_t = W(0,t)$ is called a free Brownian motion. It can be constructed as follows. 
For a Hilbert space $H$ and $\Omega\in H$, $\|\Omega\|=1$, let $\mathcal{F}(H)$ be the full Fock space defined by 
\[ \mathcal{F}(H) = \C \Omega \oplus \bigoplus_{n=1}^\infty H^{\otimes n}, \]
which becomes a Hilbert space via the product 
\[\text{$\left<f_1\otimes ... \otimes f_n, g_1\otimes ... \otimes g_n \right> = \left<f_1,g_1\right> \cdots \left<f_n,g_n\right>$,}\]
\[\text{$\left<f_1\otimes ... \otimes f_n, g_1\otimes ... \otimes g_m \right> =0$ if $n\not=m$,}\]
\[\text{$\left<f_1\otimes ... \otimes f_n, v \right> = 0$, \quad 
$\left<v,v\right>=1$.}\]
For $f\in H$, let $a(f)$ (left annihilation operator) and $a^*(f)$ (left creation operator) be the elements from $B(\mathcal{F}(H))$ defined via 
\[ a(f)(f_1\otimes \ldots \otimes f_{n+1}) = \left<f,f_1\right>f_2\otimes \ldots \otimes f_{n+1},\quad 
a(f)(f_1) = \left<f,f_1\right>\Omega, \quad 
a(f)(\Omega) = 0, \]
\[  a^*(f)(f_1\otimes \ldots \otimes f_n) = f\otimes f_1\otimes f_2\otimes \ldots \otimes f_n, \quad  a^*(f)(\Omega) = f. \]
Now take $H=L^2([0,\infty),\C)$. For $t\geq0$, $\textbf{1}_{[0,t]}\in H$. Then 
\[ t\mapsto a(\textbf{1}_{[0,t]}) + a^*(\textbf{1}_{[0,t]})\]
is a free Brownian motion within $(B(\mathcal{F}(H)), X\mapsto \left<\Omega, X\Omega\right>)$, see \cite{Spe90}.
\end{remark}

\section{Quantum Markov chains}

In order to define quantum Markov processes, we would need a  notion of conditional independence (sometimes also called ``independence with amalgamation'') or conditional expectation  for quantum probability spaces. Due to Theorem \ref{condi}, we could define a conditional expectation as a linear mapping $\varphi(\cdot,|\mathcal{B}):\mathcal{A}\to \mathcal{B}$ between two quantum probability spaces  $\mathcal{A},\mathcal{B}$, where $\mathcal{B}$ is embedded in $\mathcal{A}$, which is a completely positive contraction and $\mathcal{B}$-linear. For this purpose, we would have to generalize our notion of quantum probability spaces and we would face some further technicalities. For the general theory, we refer instead to the literature; see \cite[II.6.10]{Bla06}, \cite{Ske04}, the Section ``Towards Markov Processes'' in \cite{barndorff-nielsen+al}. In this section, we will see how the most common quantum version of finite time-homogeneous Markov chains look like.\\

Instead of defining a Markov chain as a collection $(X_n)_{n\in\N_0}$ of random variables on a fixed quantum probability space $(B(H),\varphi)$, it is now more convenient to work with the Schr\"odinger picture, i.e.\ we have one fixed operator $X_0$ and change the expectation $\varphi_n$ with $n$.\\

Recall that every expectation $\varphi$ on $B(\C^{N\times N})$ can be written as 
$\varphi(X) = \tr(X \rho)$ for a density matrix $\rho$, see Exercise \ref{density}.

\begin{example}
Consider a state space $S=\{s_1,...,s_N\}\subset \R$ and a time-homogeneous Markov chain $(M_n)_{n\in\N_0}$ on $S$ with transition matrix 
 $P=(p_{j,k})_{1\leq j,k \leq N}$ and initial distribution $\nu$. Put $v=(v_1,...,v_N)^T=(\nu(\{s_1\}), ..., \nu(\{s_N\}))^T$. (Recall that $\mathbb{P}[M_{n+1} = s_{j}| M_n=s_k]=p_{j,k}$ a.s.)\\

This process can be modeled by a quantum stochastic process as follows. The diagonal matrix $X_0=\operatorname{diag}(s_1,...,s_N)$ encodes the state space and the distribution of $X_0$ with respect to the expectation $\varphi(X)= \left<v, Xv \right>$ is equal to $\nu$. We can write this expectation also as $\varphi(X) = \tr(X \rho_0)$ with $\rho_0 = \operatorname{diag}(v_1,...,v_N)$.

Now let $e_k$ be the $k$-th unit vector in $\C^N$ and let 
$Q_{j,k}=\sqrt{p_{j,k}} e_je_k^T$. Then  $Q_{j,k} \rho_0 Q_{j,k}^T = v_k p_{j,k} e_je_j^T$.\\

Define the matrix $\rho_1:=\sum_{1\leq j,k\leq N} Q_{j,k} \rho_0 Q_{j,k}^T$, which is a diagonal matrix of the form 
\[\rho_1=\operatorname{diag}\left(\sum_{k=1}^N v_k p_{1,k}, ..., \sum_{k=1}^N v_k p_{N,k}\right).\] 
Thus $\rho_1$ is a density matrix which corresponds to the probabilities given by the vector $Pv$ in the classical case. So the probabilities of $P^n v$ correspond to the density matrix $\rho_n$ defined recursively by $\rho_{n+1} = \sum_{1\leq j,k\leq N} Q_{j,k} \rho_n Q_{j,k}^T$. 
We conclude that the distribution of $M_n$ is equal to the distribution of $X_0$ with respect to $X\mapsto \tr(X \rho_n)$.\hfill $\blacksquare$
\end{example}

The key in the previous example is the transform $\mathcal{T}:X \mapsto \sum_{j,k}Q_{j,k}XQ_{j,k}^T$, as we have \[
\rho_{n} = \mathcal{T}^{n}(\rho_0).\]
This leads to the following definition. 

\begin{definition}
A linear mapping $\mathcal{T}:\C^{N\times N} \to \C^{N\times N}$  is called a \emph{quantum channel} if 
\[\mathcal{T}(X) = \sum_{j=1}^M E_j X E_j^*\] for matrices $E_1,...,E_M\in \C^{N\times N}$ with $\sum_{j=1}^M E_j^* E_j = I$.
\end{definition}

If $\rho\in\C^{N\times N}$ is a density matrix, then $\mathcal{T}(\rho)$ is again a density matrix as 
\[\tr(\mathcal{T}(\rho))=\tr(\sum_{j} E_j\rho E_j^* )=\sum_{j}\tr( E_j\rho E_j^* )=\sum_{j}\tr( E_j^*E_j\rho  )=\tr((\sum_{j}E_j^* E_j)\rho )=\tr(\rho )=1,\]
and as each $E_j\rho E_j^*$ is self-adjoint and non-negative, also the sum is self-adjoint and non-negative.\\

\begin{remark}One can define quantum channels (on separable  Hilbert spaces) as completely positive, Hermiticity preserving, trace-preserving linear mappings, and then prove that quantum channels can be represented as a (possibly infinite) sum $\sum_{j} E_j X E_j^*$ for bounded operators  $(E_j)_j$ with $\sum_{j} E_j^* E_j = I$ (Kraus representation); see \cite[Lecture 6]{Att}, \cite{Cho75}.
\end{remark}

\begin{definition}
A self-adjoint $X_0\in \C^{N\times N}$ together with a family $(\rho_n)_{n\in\N_0}\subset \C^{N\times N}$ of density matrices is called a \emph{(time-homogeneous) quantum Markov chain} if there exists a quantum channel $\mathcal{T}$ such that
\[ \rho_{n+1} = \mathcal{T}(\rho_n) \quad \text{for all $n\in\N_0$.} \]
\end{definition}

While $X_0$ encodes the ``state space'', the initial expectation $\rho_0$ represents the initial distribution of the state space. 
By comparing this setting with the classical case, we should say that a quantum Markov chain has a limit state if there exists a density matrix $\rho_\infty$ such that $\rho_\infty = \lim_{n\to \infty} \rho_n = (\mathcal{T})^n(\rho_0)$. \\

For a density operator $\rho$, let $\supp(\rho)$ be the subspace of $\C^n$ spanned by its eigenvectors for non-zero eigenvalues.

\begin{definition}A quantum channel $\mathcal{T}$ is called \emph{aperiodic} if for any density matrix $\rho\in\C^{N\times N}$, 
\[\operatorname{gcd}\{m\in \N\,|\, \supp(\rho)\subseteq \supp((\mathcal{T})^m(\rho))\}=1.\]
$\mathcal{T}$ is called \emph{irreducible} if, for any density matrix $\rho$, $\operatorname{span}\left(\cup_{m=0}^\infty \supp((\mathcal{T})^m(\rho))\right)=\C^N$.
\end{definition}

One can now prove the following analogue of Lemma \ref{ianda}:

\begin{theorem}[Corollary 1 in \cite{GFY18}] If $\mathcal{T}$ is irreducible and aperiodic, then, for any density matrix $\rho$, there exists $M>0$ such that $\supp((\mathcal{T})^m(\rho))=\C^N$ for all $m\geq M$.  
\end{theorem}

...and of the classical convergence result, Theorem \ref{limit_Markov}:

\begin{theorem}[Theorem 3 in \cite{GFY18}] If $\mathcal{T}$ is irreducible and aperiodic, then there exists a density matrix $\rho_\infty$ 
such that \[\rho_\infty = \lim_{n\to \infty} (\mathcal{T})^n(\rho)\]
  for any initial density matrix $\rho$. 
\end{theorem}

\chapter{Univalent functions}\label{sec_8}

In this chapter we will see that $\mu \in \mathcal{P}(\R)$ can be embedded into a $\rhd$-hemigroup if and only if 
$F_\mu$ is univalent.

\begin{theorem}${}$\label{main_sec_8}
\begin{itemize}
	\item[(a)] Let $(\mu_{s,t})_{0\leq s\leq t}$ be a $\rhd$-hemigroup. Then $F_{\mu_{s,t}}$ is univalent for all $0\leq s\leq t$.
	\item[(b)] Let $\mu \in  \mathcal{P}(\R)$ such that  $F_\mu$ is univalent. Then there exists a $\rhd$-hemigroup $(\mu_{s,t})_{0\leq s\leq t}$ such that $\mu_{0,1}=\mu$.
\end{itemize}
\end{theorem}

$\rhd$-hemigroups are in one-to-one correspondence with certain Loewner chains, which are a useful method to study conformal mappings. 
We will first consider only very special Loewner chains, which allow us to prove Theorem \ref{main_sec_8} (a) in Section \ref{sec_univalent_F}. 
We then look at more general Loewner chains in Section \ref{general_Loewner_chains}, where we basically refer to the literature for proofs.   
A result from \cite{bracci+al2015} allows us to conclude Theorem \ref{main_sec_8} (b).

\section{Loewner chains}

Let $D\subsetneq \C$ be a simply connected domain. The following definition generalizes the notion of continuous semigroups on $D$, see Section \ref{continuous_semigroup}. 

\begin{definition}\label{EV_def:evolution_family}${}$
\begin{enumerate}[\rm (1)]
\item  Let $(f_{s,t})_{0\leq s \leq t}$ be a family of non-constant holomorphic self-mappings $f_{s,t}:D\to D$ satisfying
 \begin{enumerate}[\quad(a)]
\item $f_{s,s}(z)=z$ for all $z\in D$ and $s\geq 0$,
\item $f_{s,u} = f_{s,t} \circ f_{t,u}$ for all $0\leq s \leq t \leq u$,
\item $(s,t)\mapsto f_{s,t}$ is continuous with respect to locally uniform convergence.
\end{enumerate}
The family $(f_t)_{t\geq 0}:=(f_{0,t})_{t\geq 0}$ is called a \index{Loewner chain}\emph{(decreasing) Loewner chain on $D$}. We will call the mappings $f_{s,t}$ the \index{transition mappings}\emph{transition mappings} of the Loewner chain.

\item We call a Loewner chain $(f_{t})_{t\geq 0}$ an \index{Loewner chain!additive}\index{F-transform@$F$-transform}\emph{additive Loewner chain} if $D=\Ha$ and \\
$\lim_{y \to \infty} f_{s,t}(iy)/ (iy)=1$, or equivalently
$$f_{s,t} = F_{\mu_{s,t}}$$
 for all $0\leq s\leq t,$ where each $\mu_{s,t}$ is a probability measure on $\R$. 
\end{enumerate}

\end{definition}

Due to property (b), the domains $f_t(D)$ are decreasing, i.e.\ $f_t(D)\subseteq f_s(D)$ for all $s\leq t$. Usually, the literature focuses on increasing Loewner chains, where $f_s(D) \subseteq f_t(D)$ whenever $0\leq s\leq t$. Clearly, 
if $(f_t)_{t\in [0,T]}$ is an increasing Loewner chain, then 
$(f_{T-t})_{t\in [0,T]}$ is (a part of) a decreasing Loewner chain.\\

 \begin{figure}[H]
 \begin{center}
 \includegraphics[width=0.9\textwidth]{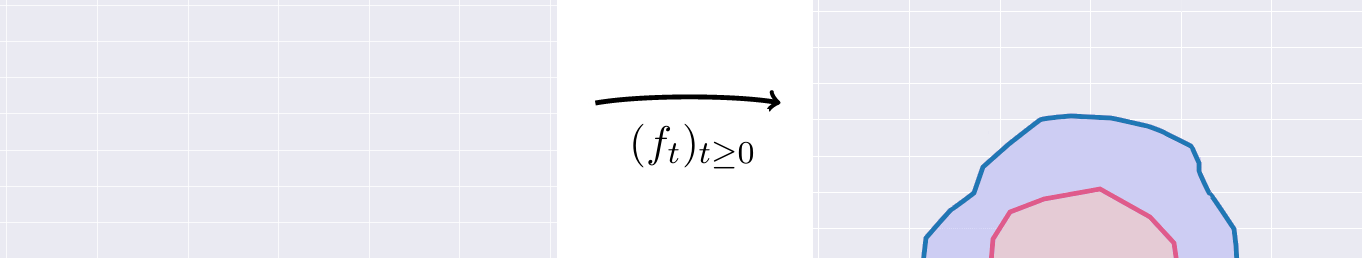}
 \caption{The complements $D\setminus f_t(D)$ of a Loewner chain $(f_t)_{t\geq 0}$ are growing subsets of $D$.}
 \end{center}
 \end{figure}

In 1923, C. Loewner introduced a differential equation for univalent functions that form an increasing Loewner chain 
 to attack the so called Bieberbach conjecture (\cite{Loewner:1923}), which we briefly explain in Appendix \ref{sec_Bie}. Afterwards, his ideas have been extended to more general settings, with recent applications in stochastic geometry (Schramm-Loewner evolution, see Section \ref{schramm}). We refer to \cite{AbateBracci:2010} for an historical overview of Loewner theory.

\begin{remark}\label{EV_remark_ev}
In Loewner theory, a family $(\phi_{s,t})_{0\leq s\leq t}$ of holomorphic mappings $\phi_{s,t}:D\to D$ is called an \emph{evolution family on $D$}
if
 \begin{itemize}
\item[(a)] $\phi_{s,s}(z)=z$ for all $z\in D$ and all $s \geq 0$,
\item[(b)] $\phi_{s,u} = \phi_{t,u} \circ \phi_{s,t}$ whenever $0\leq s\leq t\leq u$,
\item[(c)] $(s,t)\mapsto \phi_{s,t}$ is continuous with respect to locally uniform convergence.
\end{itemize}
If (b) is replaced by $\phi_{s,u} = \phi_{s,t} \circ \phi_{t,u}$, then the family is usually called a
\emph{reverse evolution family}. Thus the transition mappings of a decreasing Loewner chain form a reverse evolution family. 
\end{remark}

We have chosen the name \emph{additive} for the special Loewner chains in part (2) of Definition \ref{EV_def:evolution_family} as they are in one-to-one correspondence with $\rhd$-hemigroups.

\begin{lemma}\label{hem_loe}Let $(\mu_{s,t})_{0\leq s\leq t}$ be a $\rhd$-hemigroup. Then $F_{\mu_{0,t}}$ is an additive Loewner chain with transition mappings 
$F_{\mu_{s,t}}$. Conversely, if $(f_{t})_{t\geq 0}$ is an additive Loewner chain with transition mappings $f_{s,t}$, then $f_{s,t}=F_{\mu_{s,t}}$ for a 
$\rhd$-hemigroup $(\mu_{s,t})_{0\leq s\leq t}$. 
\end{lemma}
\begin{proof}
By Lemma \ref{lemmaconvergence}, condition (c) is equivalent to the continuity of $(s,t)\mapsto \mu_{s,t}$ with respect to weak convergence. 
\end{proof}

\section{Radial Loewner chains}\label{sec_radialssddf}

\begin{definition}${}$\begin{itemize}
	\item[(1)] A Loewner chain  $(f_t)_{t\geq0}$ on $\D$ with $f_{s,t}(0)=0$ for all $0\leq s\leq t$ will be called \emph{radial}.
	\item[(2)] A  \emph{normalized radial Loewner chain} is a radial Loewner chain with $f_{s,t}'(0)=e^{s-t}$ for all $0\leq s\leq t$.
\end{itemize}
\end{definition}

Radial Loewner chains have a probabilistic interpretation just like additive Loewner chains. We can define $\rhd$-hemigroups also for distributions on the unit circle $\partial\D$, which arise by replacing self-adjoint random variables by unitary random variables. See Exercises \ref{exxxa} and \ref{exxxb}.\\

For the rest of this section, we now fix a normalized radial Loewner chain $(f_t)_{t\geq0}$. 
We will see that $t\mapsto f_t$ is differentiable almost everywhere, and that each $f_{s,t}:\D\to\D$ is a univalent function. We basically follow \cite[Section 3.4]{Dur83} and \cite[Section 6.1]{P75} with slight modifications.\\

In order to obtain this result, we need the class $\mathcal{P}$ defined as
\[\mathcal{P} = \{p:\D\to\C \,|\, p\; 
 \text{holomorphic},\; \Re(p)> 0,\; p(0)=1 \}. \]
By the Herglotz representation formula, see Exercise \ref{Herglotz}, every $p\in \mathcal{P}$ can be written as 
\[p(z) = \int_{\mathbb{\partial \D}}\frac{u+z}{u-z} \rho(du) \]
for a probability measure $\rho$ on $\partial \D$.

\begin{lemma}\label{Herglotz_est}Let $p\in\mathcal{P}$. Then 
\[ |p(z)| \leq \frac{1+|z|}{1-|z|}\quad \text{and}\quad 
  |p'(z)| \leq \frac{2}{(1-|z|)^2} \quad \text{for all $z\in \D$.}\]
\end{lemma}
\begin{proof} The Herglotz representation formula yields
\[|p(z)| \leq
 \int_{\mathbb{\partial \D}}\left|\frac{1+z/u}{1-z/u}\right| \rho(du) \leq 
\int_{\mathbb{\partial \D}}\frac{1+|z/u|}{1-|z/u|} \rho(du) =
\int_{\mathbb{\partial \D}}\frac{1+|z|}{1-|z|} \rho(du)=\frac{1+|z|}{1-|z|}.\]
Furthermore,
\[|p'(z)| = \left|\int_{\mathbb{\partial \D}}\frac{2u}{(u-z)^2} \rho(du)\right|\leq \int_{\mathbb{\partial \D}}\left|\frac{2u}{(u-z)^2}\right| \rho(du)\leq \int_{\mathbb{\partial \D}}\frac{2}{(1-|z|)^2} \rho(du)=\frac{2}{(1-|z|)^2}.\]
\end{proof}

The set of all holomorphic $f:\D\to \D$ forms a normal family and hence also the set of its derivatives $f'$ is a normal family. In the following we let $m(r)$, $r\in(0,1)$, be such that 
\[\text{$|f'(z)|\leq m(|z|)$ for all $z\in\D$ and all holomorphic $f:\D\to \D$ with $f(0)=0$.}\]
One can derive an explicit bound $m$. Dieudonne\'{e} has shown that one can choose  $m(r)=1$ for $r\in[0,\sqrt{2}-1]$ and 
$m(r)=\frac{(1+r^2)^2}{4r(1-r^2)}$ for $r\in[\sqrt{2}-1,1)$, see \cite[p.352]{Die31}, \cite{Bea97}.

\begin{lemma}Let $s\geq 0$. Then
\begin{equation}\label{estiii}
|f_{s,u}(t)-f_{s,t}(z)|  \leq 2|z|\frac{1+|z|}{1-|z|} m(|z|) |e^t-e^{u}|
\end{equation}
for all $t,u\geq s$ and all $z\in \D$.
\end{lemma}
\begin{proof}
We have $f_t = f_s \circ f_{s,t}$ and $|f_{s,t}(z)|\leq |z|$ due to the Schwarz lemma. 
Hence $f_{s,t}(z)/z$ is an analytic function in $\D$ with modulus $\leq 1$. Define 
\begin{equation}\label{p_est}
p(z,s,t) = \frac{1+e^{s-t}}{1-e^{s-t}} \frac{1-f_{s,t}(z)/z}{1+f_{s,t}(z)/z}. 
\end{equation}
Then $p(\cdot,s,t)$ belongs to the class $\mathcal{P}$.
Since $|p(z,s,t)|\leq \frac{1+|z|}{1-|z|}$ by Lemma \ref{Herglotz_est}, \eqref{p_est} gives us
\[ |z-f_{s,t}(z)| \leq  2|z|(1-e^{s-t})\frac{1+|z|}{1-|z|}.\]
 Let $t\leq u$. Then
\begin{eqnarray*}
|f_{s,t}(t)-f_{s,u}(z)|&=&\left |\int_{f_{t,u}(z)}^z f'_{s,t}(w) dw\right| \leq 
|z-f_{t,u}(z)| m(|z|) \\
&\leq& 2|z|\frac{e^u-e^{t}}{e^u}\frac{1+|z|}{1-|z|} m(|z|)\leq 2|z|\frac{1+|z|}{1-|z|} m(|z|)(e^u-e^{t}).
\end{eqnarray*}
If $u\leq t$, then
\begin{eqnarray*}
|f_{s,u}(t)-f_{s,t}(z)| &=& \left|\int_{f_{u,t}(z)}^z f'_{s,u}(w) dw\right| \leq 
|z-f_{u,t}(z)| m(|z|) \\
&\leq& 2|z|\frac{e^t-e^{u}}{e^t}\frac{1+|z|}{1-|z|} m(|z|)\leq 2|z|\frac{1+|z|}{1-|z|} m(|z|) (e^t-e^{u}).
\end{eqnarray*}
\end{proof}

Recall that a function $f:[a,b]\to\C$, where $a,b\in\R$ with $a<b$, is called \emph{absolutely continuous} if for every $\eps>0$ 
there is $\delta>0$ such that $\sum_{k=1}^n|f(y_k)-f(x_k)|<\eps$ whenever $[x_1,y_1],...,[x_n,y_n]\subset [a,b]$ are disjoint subintervals with $\sum_{k=1}^n|y_k-x_k|<\delta$.\\
$f:[a,b]\to \C$ is absolutely continuous if and only if there exists a Lebesgue integrable function $g$ such that $g=f'$ almost everywhere and 
$f(x)=f(a)+\int_a^x g(y){\rm d}y$ for all $x\in[a,b]$.\\

We now obtain Loewner's partial differential equation. 

\begin{theorem}The function $t\mapsto f_{t}(z)$ is locally absolutely continuous for every $z\in \D$ and there exists a function $p:\D\times[0,\infty)\to\C$ such that $p(\cdot,t)\in \mathcal{P}$ for almost all $t\geq 0$ and $t\mapsto p(z,t)$ is measurable for all $z\in\D$ such that 
\begin{equation}\label{LE_1} \frac{\partial}{\partial t}f_{t}(z) = 
-z\frac{\partial}{\partial z}f_{t}(z) p(z,t) \quad 
\text{for almost all $t\geq 0$ and all $z\in\D$.}
  \end{equation} 
\end{theorem}
\begin{proof}
\eqref{estiii} implies that  $t\mapsto f_{t}(z)$ is locally absolutely continuous. It follows that $\frac{\partial}{\partial t}f_{t}(z)$ exists for almost all $t\geq s$. Since the union of countably many sets of Lebesgue measure $0$  is again a set of Lebesgue measure $0$, the derivative exists for all $z=\frac1{k}$, $k=2,3,...$, and all $t\not\in E$ for some set $E\subset [0,\infty)$ of Lebesgue measure $0$. The bound \eqref{estiii} and Vitali's theorem imply that $\frac{\partial}{\partial t}f_{t}(z)$ exists for all $z\in\D$ and all $t\not\in E$.
Let $t\not\in E$ and let $u> t$. Then we can write
\[ \frac{f_{t}(z)-f_{u}(z)}{t-u} = 
\frac{f_{t}(z)-f_{u}(z)}{z-f_{t,u}(z)} \frac{z-f_{t,u}(z)}{t-u}=
\frac{f_{t}(z)-f_{t}(f_{t,u}(z))}{z-f_{t,u}(z)} \frac{1-e^{t-u}}{t-u}\frac{z+f_{t,u}(z)}{1+e^{t-u}} p(z,t,u). \] 
\eqref{estiii} implies that $f_{t,u}(z)\to z$ locally uniformly as $u\to t$. Hence
\[ \frac{\partial}{\partial t}f_{t}(z) = 
-z\frac{\partial}{\partial z}f_{t}(z) p(z,t),
  \] 
	for some $p(\cdot,t)\in \mathcal{P}$, as $\mathcal{P}$ is closed with respect to locally uniform convergence.	
	As $\frac{\partial}{\partial t}f_{t}(z)$ and $\frac{\partial}{\partial z}f_{t}$ are both measurable in $t$, also $p(z,t)$ is measurable in $t$.
	\end{proof}
	
\begin{remark}
Fix $s\geq 0$, let $t\in E \cap [s,\infty)$ and let $u> t$. Then 
\begin{eqnarray*} \frac{f_{s,t}(z)-f_{s,u}(z)}{t-u} &=& 
\frac{f_{s,t}(z)-f_{s,u}(z)}{z-f_{t,u}(z)} \frac{z-f_{t,u}(z)}{t-u} \\
&=& \frac{f_{s,t}(z)-f_{s,t}(f_{t,u}(z))}{z-f_{t,u}(z)} \frac{1-e^{t-u}}{t-u}\frac{z+f_{t,u}(z)}{1+e^{t-u}} p(z,t,u).\end{eqnarray*}
We know that the right side converges as $u\to t$. Thus $\frac{\partial}{\partial t}f_{s,t}(z)$ exists with
\[ \frac{\partial}{\partial t}f_{s,t}(z) = -z\frac{\partial}{\partial z}f_{s,t}(z) p(z,t) \quad 
\text{for almost all $t\in E\cap [s,\infty)$ and all $z\in\D$.}  \]
\end{remark}

In order to study Loewner's partial differential equation, it is helpful to consider a related ordinary differential equation.

\begin{theorem}
Let $p:\D\times[0,\infty)\to\C$ be such that $p(\cdot,t)\in \mathcal{P}$ for almost all $t\geq 0$ and $t\mapsto p(z,t)$ is measurable for all $z\in\D$. Then there exists a unique solution of holomorphic functions $g_{s,t}:\D\to\D$, locally absolutely continuous with respect to $t$, to the initial value problem
\begin{equation}\label{LE_0} \frac{\partial}{\partial t}g_{s,t} = -g_{s,t} p(g_{s,t},t), \quad \text{for a.e.\ $t\geq s$, $g_{s,s}(z)=z$}. \end{equation}
Furthermore, each $g_{s,t}:\D\to \D$ is univalent and 
$g_{s,u} = g_{t,u} \circ g_{s,t}$ for all $0\leq s\leq t\leq u$.
\end{theorem}
\begin{proof}Let $s\geq 0$ and $r\in(0,1)$. 
Let $g_{s,t}^0(z)\equiv 0$ and for $n\in \N$ and $|z|\leq r$ let 
\[g_{s,t}^n(z):= z \exp\left(-\int_s^t p(g_{s,\tau}^{n-1}(z),\tau)d\tau\right).\]
As $\Re p(\cdot,t)>0$ in $\D$ for almost all $t$, we see by induction that $|g_{s,t}^n(z)\leq r|$, which justifies the definition. Since $|e^{-a}-e^{-b}|\leq |a-b|$ for complex numbers 
$a,b$ with $\Re(a),\Re(b)\geq0$, it follows together with the second estimate from Lemma \ref{Herglotz_est} that
\[ |g_{s,t}^{n+1}(z)-g_{s,t}^n(z)|\leq \int_s^t |p(g_{s,\tau}^{n}(z),\tau)- p(g_{s,\tau}^{n-1}(z),\tau)| d\tau \leq \frac{2}{(1-r)^2} \int_s^t |g_{s,t}^{n+1}(z)-g_{s,t}^n(z)|d\tau. \]
We can now show by induction that 
\[ |g_{s,t}^{n+1}(z)-g_{s,t}^n(z)|\leq \frac{2^n(t-s)^n}{(1-r)^{2n}n!} \]
for all $n=0,1,...$ and all $|z|\leq r$. Consequently, 
$\lim_{n\to\infty}g_{s,t}^n(z)=:g_{s,t}(z)$ exists uniformly in $|z|\leq r$, $s\leq t\leq T$ for every $T>s$. Thus $g_{s,t}(z)$ is defined for all $t\geq s$ and all $z\in\D$ and $z\mapsto g_{s,t}(z)$ is holomorphic with $g_{s,t}(\D)\subset \D$. Furthermore, 
\[g_{s,t}(z) = z \exp\left(-\int_s^t p(g_{s,\tau}(z),\tau)d\tau\right),\]
which shows that $g_{s,t}$ solves \eqref{LE_0} and that $t\mapsto g_{s,t}(z)$ is locally absolutely continuous.\\

Now let $h_{s,t}$ be another solution to \eqref{LE_0}. Fix $z\in\D$. Then we obtain the estimate 
\[|g_{s,t}(z) - h_{s,t}(z)| \leq \frac{2|z|}{(1-|z|)^2} \int_s^t |g_{s,\tau}(z) - h_{s,\tau }(z)|d\tau,
\]
which implies 
\[|g_{s,t}(z) - h_{s,t}(z)| \leq  \frac{2^{n+1}(t-s)^n}{(1-|z|)^{2n}n!}.\]
As $n\to\infty$, we obtain $g_{s,t}(z)= h_{s,t}(z)$.\\

Let $0\leq s\leq \tau$. Consider the function $t\mapsto g_{\tau, t} \circ g_{s,\tau}$ for $t\geq \tau$. It coincides with $g_{s,\tau}$ for $t=\tau$ and a simple calculation shows that it satisfies the differential equation \eqref{LE_0}. Hence, $g_{s,t} = g_{\tau, t} \circ g_{s,\tau}$ for all $0\leq s\leq \tau \leq t$.\\

Finally, we show that $g_{s,t}$ is univalent. Suppose that 
$g_{s,T}(z)=g_{s,T}(w)$ for some $z,w\in \D$ with $z\not=w$ and some $T>s$. Let $s\leq t\leq T$. Then 
\begin{eqnarray*}
&&\frac{d}{dt}(g_{s,t}(z)-g_{s,t}(w))=g_{s,t}(z)p(g_{s,t}(z),t) - g_{s,t}(w)p(g_{s,t}(w),t)\\
&=&
g_{s,t}(z) (p(g_{s,t}(z),t)-p(g_{s,t}(w),t) )+ p(g_{s,t}(w),t)(g_{s,t}(z)-g_{s,t}(w)).
\end{eqnarray*}
Thus, with Lemma \ref{Herglotz_est},
\[|\frac{d}{dt}|g_{s,t}(z) - g_{s,t}(w)|| \leq 
|\frac{d}{dt}(g_{s,t}(z) - g_{s,t}(w))| \leq  
K |g_{s,t}(z) - g_{s,t}(w)|\]
for some constant $K$. In particular, 
\[ \frac{d}{dt}|g_{s,t}(z) - g_{s,t}(w)| \geq 
-K |g_{s,t}(z) - g_{s,t}(w)|, \]
which implies $\frac{d}{dt}(e^{Kt}|g_{s,t}(z) - g_{s,t}(w)| \geq 
0.$ Integration from $s$ to $T$ yields
\[ |g_{s,T}(z) - g_{s,T}(w)| - |g_{s,s}(z) - g_{s,s}(w)| = - |z - w| \geq 0. \]
Thus $z=w$, a contradiction. 
\end{proof}

\begin{corollary}\label{EV_continutiy_transition}
 Let $T>0$. Then $s\mapsto g_{s,T}(z)$ is differentiable for almost every $s\in [0,T]$ with 
\[ \frac{\partial}{\partial s} g_{s,T}(z) = z \frac{\partial}{\partial z}g_{s,T}(z) \cdot  p(z, T-s).\]
\end{corollary}
\begin{proof}
We have $g_{s,T} \circ g_{0,s} = g_{0,T}$ for all $s\in [0,T]$. Thus

\[\frac{g_{s+h,T}(g_{0,s}(z))-g_{s,T}(g_{0,s}(z))}{h} =  -\frac{g_{s+h,T}(g_{0,s+h}(z))-g_{s+h,T}(g_{0,s}(z))}{h} \]
for all $h$ small enough. 
 If $\tau \mapsto g_{0,\tau}$ is differentiable at $s$, then we see that the right side converges to 
$\frac{\partial}{\partial z}g_{s,T}(g_{0,s}(z)) \cdot (- g_{0,s}(z) p(g_{0,s}(z),s))$ as $h\to 0$. 
This shows that $\tau \mapsto g_{\tau,T}(g_{0,s}(z))$ is differentiable at $\tau=s$ with
\[\frac{\partial}{\partial s} g_{s,T}(g_{0,s}(z)) =  \frac{\partial}{\partial z}g_{s,T}(g_{0,s}(z)) \cdot  g_{0,s}(z) p(g_{0,s}(z), T-s),\]
and we conclude that
\[ \frac{\partial}{\partial s} g_{s,T}(w) = w \frac{\partial}{\partial w} g_{s,T}(w) \cdot p(w, T-s)\]
for all $w\in g_{0,s}(\D)$ and almost all $s\in [0,T]$. However, the right side (and thus its integral with respect to $s$) can be extended holomorphically to $\D$ and thus 
$\frac{\partial}{\partial s} g_{s,T}(w)$ satisfies the PDE for all $w\in \D$.
\end{proof}

\begin{corollary}\label{mult_univ}
The initial value problem \eqref{LE_1} has exactly one solution 
$(h_t)_{t\geq0}$ of holomorphic mappings $h_t:\D\to\D$, locally absolutely continuous in $t$. The family $(h_t)_{t\geq0}$ is a normalized radial Loewner chain and and each $h_t$ is univalent.
\end{corollary}
\begin{proof}
Let $g_{s,t}$ be the solution to \eqref{LE_0} and fix some $T>0$.
Now define $h_t = g_{T-t,T}$, $0\leq t\leq T$. 
We have $h_t = g_{T-t,T} = g_{T-s, T} \circ g_{T-t,T-s}=h_s \circ g_{T-t,T-s}$ whenever $s\leq t$. Hence $(h_t)_{0\leq t\leq T}$ is (a part of) a decreasing Loewner chain consisting of univalent functions. We have 
\begin{equation*}
 \frac{\partial}{\partial t}h_t(z)= - z\frac{\partial}{\partial z}g_{T-t,T}(z) p(z,t) =   -z\frac{\partial}{\partial z}h_t(z) \cdot p(z,t) , \quad h_0(z)\equiv z \in D,\quad  0\leq t \leq T. 
\end{equation*}
 By choosing another $\hat{T}>T$, we obtain a family $(\hat{g}_{s,t})_{0\leq s\leq t\leq\hat{T}}$ with $\hat{g}_{s,\hat{T}} 
=g_{s+T-\hat{T},\hat{T}+T-\hat{T}}=g_{s+T-\hat{T},T}$ for all $s\in [\hat{T}-T, \hat{T}]$. Hence 
$\hat{h}_t := \hat{g}_{\hat{T}-t,\hat{T}} = g_{T-t,T}=h_t$ for all $t \in [0, T]$. \\
As we can choose $T>0$ arbitrarily large, we conclude that there exits a decreasing Loewner chain $(h_t)_{t\geq 0}$ of univalent functions satisfying \eqref{LE_1}. From \eqref{LE_0} and the fact that $p(0)=1$ whenever $p\in \mathcal{P}$, we see that $(h_t)$ is a 
normalized radial Loewner chain.\\

Now let $(J_t)_{t\geq0}$ be another family of holomorphic mappings, locally absolutely continuous in $t$, satisfying \eqref{LE_1}. Let $T>0$ and define 
$g_{s,t}$ as above. We have 
\begin{eqnarray*}
&&\frac{\partial}{\partial t}[(J_t(g_{0,T-t}))(z)]=\\
 && J'_t(g_{0,T-t}(z))\cdot (-g_{0,T-t}(z)p(g_{0,T-t}(z),t)) - J'_t(g_{0,T-t}(z))\cdot (-g_{0,T-t}(z)p(g_{0,T-t}(z),t)) = 0
\end{eqnarray*} 
and thus $J_t(g_{0,T-t}(z))=J_0(g_{0,T}(z))=g_{0,T}(z)$ for all $z\in \D$ and $0\leq t\leq T$. This implies $J_t = g_{0,T} \circ g^{-1}_{0,T-t}=g_{T-t,T}$ on $g_{0,T-t}(\D)$. As $g_{0,T-t}(\D)$ is an open set, the identity theorem implies 
$J_t = h_t$ on $\D$.
\end{proof}

\begin{remark}The uniqueness of the solution to \eqref{LE_1} can also be seen as follows. Write 
$p(z,t)=\sum_{n=0}^\infty c_n(t)z^n$ and let $f_t(z)=\sum_{n=1}^\infty a_n(t)z^n$ be a solution of holomorphic mappings, locally absolutely continuous in $t$. Then 
	\[ \dot{a}_1(t)z + \dot{a}_2(t)z^2+... = 
-z (a_1(t) + 2a_2(t)z + 3a_3(t)z^2+...)(1+c_1(t)z+c_2(t)z^2+...).
  \] 
	By comparing coefficients of both sides, we obtain 
	\[\dot{a}_1(t) = -a_1(t),\quad  a_1(0)=1, \quad \text{and} \quad \dot{a}_n(t) = - \sum_{k=1}^n  ka_k(t)c_{n-k}(t), \quad a_n(0)=0, n\geq 2.\]
In order to show the uniqueness of these initial value problems, we cannot apply the Picard-Lindel\"of uniqueness theorem, as the equations only hold almost everywhere. However, a similar result also holds in this situation, see \cite[Theorem 5.3]{Hal80}. Here we can use Exercise \ref{Herglotz}, which shows that $|c_n(t)|\leq 2$ for almost all $t$ and all $n\geq 1$. Thus each $t\mapsto a_n(t)$ is uniquely determined  and there exists at most one solution of holomorphic mappings, locally absolutely continuous in $t$, to \eqref{LE_1}.
\end{remark}

\section{Univalent $F$-transforms}\label{sec_univalent_F}

\begin{theorem}\label{EV_univalence}${}$
 Let $(f_{t})$ be a \index{Loewner chain}Loewner chain.
        Then every transition mapping $f_{s,t}$, in particular every $f_t=f_{0,t}$, is a univalent function.
\end{theorem}
\begin{proof} 
Because of $f_t = f_s \circ f_{s,t}$, it is sufficient to prove that $f_t$ is univalent for all $t\geq0.$ 

\vspace{3mm}

Step 1: Assume that $(f_t)$ is a Loewner chain on $\D$ with $f_{s,t}(0)=0$ for all $0\leq s\leq t$. Let $a_{t}:=f'_{t}(0).$ Due to the Schwarz lemma, we have $|a_{t}|\leq 1$ and, as $f_t = f_s \circ f_{s,t}$,
$t\mapsto |a_t|$ is non-increasing. Furthermore, as $t\mapsto f_t$ is continuous, also $t\mapsto a_t$ is continuous and we conclude that  $a_{t}\not=0$ for all $t\in[0,\varepsilon]$ and some $\varepsilon>0$.\\ 
 
First, assume that $a_{t}\not=0$ for all $t\geq 0$. Then we have $0<|a_t|\leq 1$ for all $t\geq 0$ and there exists a uniquely determined continuous function $C:[0,\infty)\to \{z\in\C \,|\, \Re (z)\leq 0\}$ with $C(0)=0$ such that $a_{t}=e^{C(t)}$.
It is easy to see that $$g_{t}(z):= f_{t}(e^{-i\Im(C(t))}z)$$ is also a radial Loewner chain with
$$g'_{t}(0)=e^{\Re(C(t))}.$$ The function $t\mapsto \Re(C(t))$ is
non-increasing and continuous. Note that $\Re(C(t))=\Re(C(s))$, $s\leq t$, implies that $g_{t}=g_{s}$, for  $g_t = g_s \circ g_{s,t}$ with $g_{s,t}:\D\to\D$, $g_{s,t}(0)=0, g'_{s,t}(0)=1$, i.e.\ $g_{s,t}$ is the identity by the Schwarz lemma.
        
We can reparametrize $g_{t}$ to $h_{s}:=g_{\tau(s)}$ such that $h'_{s}(0)=e^{-s}$ for all $s \in [0,S)$ for some $0<S \leq \infty$, where $\tau(s)$ is defined by 
$$
\tau(s) = \inf \{ t \geq0: \Re[C(t)]=-s\}, 
$$  
which is a strictly increasing, possibly discontinuous function. The reparametrization\\ $(h_{s})_{s\in[0,S)}$ is (part of) a normalized radial  Loewner chain and Corollary \ref{mult_univ} implies that each $h_{s}$ is univalent, which implies that each $f_{t}$ is univalent.\\    

No assume that $a_{\tau}=0$ for some $\tau > 0$ and
$a_{t}\not=0$ for $t< \tau.$  The previous case implies that $f_{t}$ is univalent for all $t<\tau.$
Hence, $f_{\tau}=\lim_{t\uparrow \tau}f_{t}$ is the limit of univalent
self-mappings of $\D$ fixing $0$. It follows that $f_{\tau}(z)\equiv 0$. This is a contradiction as all elements of a Loewner chain are non-constant by definition.\\ 
				
Step 2:	 Now we consider the general case. We can use a conformal mapping $I: D \to \D$  to transfer the problem to the unit disk, i.e.\ we define
$F_{s,t}:= I \circ f_{s,t} \circ I^{-1},$ which gives transition mappings of a Loewner chain on $\D$.
Next we use an idea from Proposition 2.9 in \cite{contreres+al2010}. Fix some $T>0.$ We define
$$ a(t) := F_{t,T}(0),
\quad h_t(z):=\frac{z+a(t)}{1+\overline{a(t)}z},$$
for $t\geq 0,$ $z\in \D.$ Note that $h_t$ is an automorphism of $\D$ mapping
$0$ onto $a(t)$. 

Define $(G_{s,t})_{0\leq s\leq t\leq T}:= (h^{-1}_{s}\circ F_{s,t} \circ h_t)_{0\leq s\leq t\leq T}$ and $G_t=G_{0,t}.$
Then $(G_t)$ is (a part of) a Loewner chain on $\D$ with
\[ G_t(0)=  (h^{-1}_{0}\circ F_t \circ h_t)(0)= (h^{-1}_{0}\circ F_t)(F_{t,T}(0))=h^{-1}_{0}(F_T(0))=0.\]
Hence, $(G_t)$ is a radial Loewner chain and a) implies that every $G_t$, $t\in[0,T]$, is univalent.
As $T>0$ can be chosen arbitrarily large, we conclude that every $F_t$ is univalent.
\end{proof}

With Lemma \ref{hem_loe} we conclude the following a corollary, which is Theorem \ref{main_sec_8} (a).

\begin{corollary}
Let $\mu_{s,t}$ be a $\rhd$-hemigroup. Then $F_{\mu_{s,t}}$ is univalent for all $0\leq s\leq t$.
\end{corollary}

Which domains have the form $F_{\mu}(\Ha)$ for univalent $F_\mu$? 
Roughly speaking, these are all simply connected subdomains of $\Ha$ having $\infty$ as a boundary point. However, a precise characterization is not known to us. We will only consider the case $\mu\in \mathcal{P}_c(\R)$, which is quite easy to handle. A partial result for the general case can be found in \cite[Theorem 3.18]{FHS}.

\begin{theorem}\label{ftransform_images}
Let $\Omega \subseteq\Ha$ be a simply connected domain such that $\Ha\setminus \Omega$ is
a bounded set. Then there exists a unique probability measure $\mu$ on $\R$ with mean $0$
and compact support such that $F_\mu(\Ha)=\Omega.$
\end{theorem}
\begin{proof}
By the Riemann mapping theorem we find a conformal mapping $f:\Omega\to \Ha$. 
 Consider the complement $B=\Ha\setminus\Omega$. As $B$ is bounded, we find a disc $R\cdot\D$ 
  such that $B\subset R\cdot\D$. Then $f$ maps the curve $\partial(R\cdot\D\cap \Ha)$ onto a simple curve in $\Ha$ with two endpoints  $a,b\in\R$, $a<b$ due to \cite[Prop. 2.14]{MR1217706}. Thus $f$ maps the Jordan domain $\Ha\setminus R\cdot \overline{\D}$ (as a domain in $\hat{\C}$) onto another Jordan domain and Theorem \ref{fortsetzen} implies that $f$ extends to a homeomorphism between the closures of these domains. We can postpone $f$ with an automorphism of $\Ha$ to get that $f(\infty)=\infty$. As $f(x)\in \R$ whenever $x\in(-\infty,-R)\cup(R,\infty)$, we can extend $f$ to a holomorphic mapping on $\hat{\C}\setminus R\cdot \overline{\D}$ with power series expansion 
	$f(z) = az + b + \sum_{n=1}^\infty c_n/z^n$, $a\not=0$, at $\infty$, as $f$ has a simple pole at $\infty$. 
	As $\Im(f(iy))/y \to a$, we see that $a>0$. Thus we can postpone $f$ with the automorphism $(z-b)/a$ of $\Ha$ to obtain a conformal mapping $f:\Omega\to \Ha$ with $f^{-1}(z) = z - \frac{c}{z} + \ldots$ with $c\in \R$.
	Due to Corollary \ref{cor_nevanlinna2}, $f^{-1}(z)-z$ maps $\Ha$ into $\Ha \cup \R$ which implies that $c\geq 0$.

Theorem \ref{first_moment} implies that $f^{-1} = F_\mu$ for a probability measure $\mu\in \mathcal{P}_c(\R)$ with mean $0$ and variance $c$.\\

Any other conformal mapping $G:\Ha\to \Omega$ has the form $F_\mu\circ \alpha$, for some automorphism $\alpha$ of $\Ha$, which is an $F$-transform of some $\nu\in\mathcal{P}(\R)$ only if $\alpha(z)=z-d$, $d\in\R$. However, for $d\not=0$, $\nu$ does not have mean $0$, see Remark \ref{first_moment}.
\end{proof}

\begin{remark}\label{rm_capacity}
The number $c$ in the proof is equal to the variance of $\mu$.
This value is also called the \index{half-plane capacity}\emph{half-plane capacity} of the ``hull''
$\Ha\setminus \Omega$, see \cite[Section 3.4]{lawler05}. It has a more or less geometric interpretation, see
\cite{LLN09}. An explicit probabilistic formula is given in \cite[Proposition 3.41]{lawler05}.
\end{remark}

\section[General Loewner chains]{Loewner' differential equation for general Loewner chains}\label{general_Loewner_chains}

The considerations of Section \ref{sec_radialssddf} can be generalized to arbitrary Loewner chains, by endowing them with a stronger regularity property. In this section, we will mainly refer to the literature, in particular to \cite{MR2995431}, \cite{contreres+al2014}, and \cite{bracci+al2015}.

 \begin{definition}
Let $d \in [1,\infty]$ and let $D\subsetneq \C$ be a simply connected domain. A Loewner chain $(f_{t})_{t \geq 0}$ on $D$
is called a \index{Loewner chain!of order $d$}\emph{Loewner chain of order $d$} if it satisfies the condition
\begin{itemize}
\item[(c')] for any $z \in D$ and any $S>0$ there exists a non-negative function
$k_{z,S} \in L^d([0,S],\mathbb{R})$ such that
$$|f_{s,t}(z) - f_{s,u}(z)| \leq \int_t^u k_{z,S}(\xi)
\,{\rm d}\xi$$ for all
$0\leq s\leq t \leq u \leq S$.
\end{itemize}
\end{definition}

\begin{example}\label{EV_not_abs_continuous}
 Let $D=\Ha$ and $f_{s,t}(z) = z + C(t)-C(s),$ where $C:[0,\infty)\to \mathbb{R}$ is
 continuous but not absolutely continuous. Then $(f_{0,t})$ is an additive Loewner chain
 with $\mu_{t}=\delta_{C(0)-C(t)}$, and we have
$$|f_{s,t}(z)-f_{s,u}(z)|= |C(t)-C(u)|.$$ Hence $(f_{t})$ is not a Loewner chain of any order $d$.\hfill $\blacksquare$
\end{example}

Property (c') ensures that $t\mapsto f_{s,t}$ is absolutely continuous and thus differentiable almost everywhere. For a precise statement,
 we also need the following notion.

\begin{definition} A \index{Herglotz vector field}\emph{Herglotz vector field of order $d \in [1,\infty]$ on $D$} is a function
$M:D\times [0,\infty) \to \mathbb{C}$ with the following properties:
\begin{itemize}
	\item[(i)] The function $t\mapsto M(z,t)$ is measurable for every $z\in D.$
	\item[(ii)] The function $z\mapsto M(z,t)$ is holomorphic for every $t\in[0,\infty).$
	\item[(iii)] For any compact set $K \subset D$ and for all $S>0$ there exists a non-negative function
$k_{K,S} \in L^d([0,S],\mathbb{R})$ such that $|M(z,t)| \leq k_{K,S}(t)$ for all $z \in K$ and for almost
every $t \in [0,S]$.
	\item[(iv)] $M(\cdot, t)$ is an infinitesimal generator on $D$ for a.e.\ $t\geq0$.
\end{itemize}
\end{definition}

Now we have the following one-to-one correspondence.

\begin{theorem}\label{89}
A Loewner chain $(f_{t})_{t\geq0}$ of order $d$ satisfies the Loewner partial differential equation
\begin{equation}\label{EV_Loewner}
\frac{\partial}{\partial t} f_{t}(z) = \frac{\partial}{\partial z}f_{t}(z)\cdot M(z,t) \quad \text{for a.e.\ $t\geq 0$, $f_{0}(z)=z\in D,$}
\end{equation}
for a Herglotz vector field $M$ of order $d$. Conversely, the unique solution of holomorphic mappings $f_t:D\to D$, locally absolutely continuous in $t$, to \eqref{EV_Loewner}
for a given Herglotz vector field of order $d$ is always a Loewner chain of order $d$.

Moreover, each element $f_t:D\to D$ of a Loewner chain of order $d$ is a \index{univalent function}univalent function.
\end{theorem}
\begin{proof}
In \cite{contreres+al2014}, Loewner chains consist of univalent functions by definition. However, the proof
of \cite[Theorem 3.2]{contreres+al2014} does not use this property and proves equation \eqref{EV_Loewner}. This can also be seen by looking at the family
$(f_{T-t,T-s})_{0\leq s\leq t\leq T}$ for some fixed $T>0$. It can be verified that it forms an evolution family (see Remark \ref{EV_remark_ev}) and we obtain \eqref{EV_Loewner} from \cite[Theorem 1.1]{MR2995431}. 

Conversely, by \cite[Theorems 1.11]{contreres+al2014}, the unique solution to \eqref{EV_Loewner} yields a Loew\-ner chain of order $d$ consisting of univalent functions.
\end{proof}

\begin{remark}From the relation $f_t = f_s \circ f_{s,t}$ we obtain
\begin{equation*}
\frac{\partial}{\partial t} f_{s,t}(z) = \frac{\partial}{\partial z}f_{s,t}(z)\cdot M(z,t) \quad \text{for a.e.\ $t\geq s$, $f_{s,s}(z)=z\in D.$}
\end{equation*}
Furthermore, we can also differentiate $f_{s,t}$ with respect to $s$ and obtain
 \begin{equation}\label{EV_ord} \frac{\partial}{\partial s} f_{s,t}(z) = -M(f_{s,t}(z), s) \quad \text{for a.e.\ $s\leq t$, $f_{t,t}(z)=z\in D$}.
\end{equation}
Conversely, this equation has a unique solution, which gives the transition mappings of a decreasing Loewner chain of order $d$,
see again \cite[Theorems 1.11 and 3.2]{contreres+al2014}.
\end{remark}

Our special Loewner chains now satisfy the following relationship.

\begin{theorem}\label{is_addi}
 Let $(f_{t})$ be an additive Loewner chain of order $d$.
 Then $(f_t)$ satisfies \eqref{EV_Loewner} with a Herglotz vector field $M$ of order $d$ having the form
\begin{equation}\label{EV_eq:4}
M(z,t)=a_t + \int_\mathbb{R}\frac{1+xz}{x-z} \rho_t({\rm d}x),
\end{equation}
 where $a_t\in\mathbb{R}$ and $\rho_t$ is a finite non-negative Borel measure on $\mathbb{R}$, for a.e.\ $t\geq 0$.\\
 
Conversely, let $M$ be Herglotz vector field of order $d$ of the above form. Then the solution $f_{t}$ to \eqref{EV_Loewner}
is an additive Loewner chain of order $d$.
\end{theorem}
\begin{proof} ``$\Longrightarrow$'':
{}We have $$\Im(f_{s,u}(z))=
\Im(f_{s,t}(f_{t,u}(z))) \geq \Im(f_{t,u}(z))$$ for all $0\leq s\leq t\leq u$, see Theorem \ref{Julia}. So
$s\mapsto \Im(f_{s,t}(z))$ is non-increasing for every $z\in\Ha$. From \eqref{EV_ord} we see that
$\Im(M(z,t))\geq 0$ for almost every $t\geq 0$ and every $z\in\Ha$. Hence, $M(\cdot, t)$ has the form \eqref{EV_eq:4} for a.e.\ $t\geq 0.$ (See also \cite[Thm. 8.1]{MR2995431}.) 

Assume that $M'(\infty,t)>0$ for a set $I\subset [0,T]$ of positive Lebesgue measure. Then, by \cite[Thm. 1.1]{bracci+al2015}, we obtain that
$ f_{T}'(\infty) >1$, a contradiction. This proves that $M'(\infty,t)=0$ for a.e.\ $t\geq 0$, i.e.\ $M$ is an additive Herglotz vector field.

\vspace{3mm}

``$\Longleftarrow$'': We have to show that every
$f_{t}$ can be written as $f_{t}=F_{\mu_{t}}$ for a probability measure $\mu_{t}$.
 Consider the Nevanlinna representation of $f_{t}$:
$$f_{t}(z) = A_{t} + B_{t} z + \int_\mathbb{R} \frac{1+xz}{x-z} \sigma_{t}({\rm d}x).$$
 As $M(\cdot, t)$ has the form \eqref{EV_eq:4} for a.e.\ $t\geq0,$ $M(\cdot, t)$ has a ``boundary regular null point'' at $\infty$ with dilation $0$ for a.e.\ $t\geq 0$, see \cite[Def. 2.6]{bracci+al2015} which handles the unit disk case.
 
By \cite[Thm. 1.1]{bracci+al2015}, the ``spectral function'' of $f_{t}$ at $\infty$ is equal to 0, which translates in our setting to
$$B_{t} = 1$$ for every $t\geq 0$ (note that $B_{t}=f_{t}'(\infty)$, which corresponds to $f'_{t}(\sigma)$ in
 \cite{bracci+al2015}, must be a non-negative real number by \cite[Thm. 2.2 (vi)]{bracci+al2015}). 
Hence, \eqref{EV_eq:2} implies $f_{t}=F_{\mu_{t}}$ for a probability measure $\mu_{t}$ for every $t\geq 0.$
\end{proof}

\begin{theorem}\label{EV_normal_add}
 Let $(f_{t})$ be an \index{Loewner chain!additive}additive Loewner chain such that the first and second moments
 of all $\mu_{t}$ exist with
 \begin{equation*}
\text{$\int_\mathbb{R} x \mu_{t}({\rm d}x)=0$ and $\int_\mathbb{R} x^2 \mu_{t}({\rm d}x)=t$ for all $t\geq0.$}
 \end{equation*}
 Then $(f_{t})$ satisfies
 \eqref{EV_Loewner} for a Herglotz vector field $M$ of the form
$$
M(z,t)= \int_\mathbb{R}\frac1{u-z}\,\tau_t({\rm d}u),
$$
where $\tau_t$ is a probability measure for a.e.\ $t\geq0$.

Conversely, let $M$ be a  Herglotz vector field of the above form. Then the solution $(f_{t})$ to \eqref{EV_Loewner}
is an additive Loewner chain having the above normalization.
\end{theorem}
\begin{proof}
See \cite{MR1201130} or \cite[Prop. 3.6]{monotone_schl}. We note that the normalization implies that
\begin{equation*}\label{EV_regularity2}|f_{s,t}(z) - f_{s,u}(z)| \leq \frac{u-t}{\Im(z)},\end{equation*}
for all $0\leq s \leq t\leq u$ and $z\in \Ha,$ see \cite[p. 1214]{MR1201130}.
Hence, $(f_{t})$ is an additive Loewner chain of order $\infty$.
\end{proof}

The following theorem  implies part (b) of Theorem \ref{main_sec_8}.

\begin{theorem}\label{EV_embed_F}${}$
\begin{enumerate}[\rm(a)]
 \item Let $\mu$ be a probability measure on $\R$ such that \index{F-transform@$F$-transform}$F_\mu$ is univalent.
Then there exists an \index{Loewner chain!additive}additive Loewner chain $(f_{t})_{t\geq 0}$ such that $f_{1} = F_\mu.$
\item Let $\mu$ be a probability measure on $\R$  such that $F_\mu$ is univalent and
$$\int_\mathbb{R} x\, \mu({\rm d}x)=0,\quad  \int_\mathbb{R} x^2\, \mu({\rm d}x)=:T<\infty.$$
 Then there exists an additive Loewner chain $(f_{t})_{t\geq 0}$ having the normalization from Theorem 
 \ref{EV_normal_add} such that  $f_{T} = F_\mu.$
\end{enumerate}
\end{theorem}
\begin{proof}
(a) Theorem 1.2 in \cite{bracci+al2015}, with $\Lambda(t)\equiv 0,$ implies that
we can write $F_\mu = f_{0,1}$ where $\{f_{s,t}\}_{0\leq s\leq t\leq 1}$ is an evolution family in the sense of Remark \ref{EV_remark_ev} and
\begin{itemize}
	\item[(i)] $f_{s,t}$ has a boundary regular fixed point at $\infty$ for all $0\leq s \leq t.$
	\item[(ii)] $f_{s,t}'(\infty) = 1$ for all $0\leq s \leq t.$
\end{itemize}
(Note that \cite[Thm. 1.2]{bracci+al2015} only gives $|f_{s,t}'(\infty)| = 1$. 
However, $f_{s,t}'(\infty)$ must be non-negative as $\infty$ is a fixed point of $f_{s,t},$
see again \cite[Thm. 2.2 (vi)]{bracci+al2015}.) 
We conclude that every $f_{s,t}$ has the form \eqref{EV_eq:2}. 
Finally, the family $(f_t)_{t\geq 0}$ with $f_t = f_{1-t,1}$ for $t\in[0,1]$, $f_t = f_{0,1}$ for $t>1$, is an additive Loewner chain with $f_1 = f_{0,1} = F_\mu.$ 

\vspace{3mm}

(b) This statement follows in a similar way by using \cite[Theorem 5]{MR1201130}.
\end{proof}

\section{Slit mappings}\label{p_p}

For $T>0$, let $\gamma:[0,T]\to \Ha\cup \R$ be continuous with $\gamma(0)\in\R$ and $\gamma(0,T]\subset \Ha$. Such a curve is also called a \emph{slit}.\\  

By Theorem \ref{ftransform_images}, for any $t\in [0,T]$, there exists a unique $\mu_{t}\in \mathcal{P}_c(\R)$ with mean $0$ such that 
$F_{\mu_t}$ is univalent and $F_{\mu_t}(\Ha)=\Ha\setminus \gamma(0,t]$. We will assume that $\gamma[0,T]$ is parametrized such that the variance of $\mu_t$ is equal to $t$. Then we have the normalization
\[ F_{\mu_t}(z) = z -\frac{t}{z} + \ldots \]
at $\infty$. (This parametrization is also called the  \emph{hydrodynamic parametrization} of the slit.) By Theorem \ref{loc_con}, $F_{\mu_t}$ can be extended continuously to $\Ha\cup \R$,  and by Theorem \ref{pomm}, there exists a unique $U(t)\in\R$ such that $F_{\mu_t}(U(t)) = \gamma(t)$. Thus $F_{\mu_t}^{-1}$ can be extended continuously to $(\Ha\setminus \gamma(0,t]) \cup \{\gamma(t)\}$ and 
\[ U(t) =  F_{\mu_t}^{-1}(\gamma(t)).\]

\begin{theorem}\label{slit_Loewner_eq_thm}The family $(F_{\mu_t})_{t\in[0,T]}$ is (a part of) an additive Loewner chain. The function $t\mapsto U(t)$, $t\in[0,T]$, is continuous with $U(0)=\gamma(0)$ and $F_{\mu_t}$ satisfies
\begin{equation}\label{slit_Loewner_eq}
\frac{\partial}{\partial t}F_{\mu_t}(z) = \frac{\partial}{\partial z}F_{\mu_t}(z)  \frac{1}{U(t)-z}
\quad \text{for all $t\in[0,T]$ and $z\in\Ha$.} \end{equation}
\end{theorem}
\begin{proof}See \cite{dMG16}.
\end{proof}

$U$ is also called the \emph{driving function} of the slit $\gamma$ and we will call equation \eqref{slit_Loewner_eq} \emph{Loewner's slit equation}.

\begin{example}Let $\mu_t=A(0,t)$, the arcsine distribution with mean $0$ and variance $t$. Then  $F_{\mu_t}=\sqrt{z^2-2t}$ maps $\Ha$ onto the complement of the vertical line segment $[0,i\sqrt{2t}]$, see Example \ref{ex_arcsine}, and $F_{\mu_t}$ satisfies \eqref{slit_Loewner_eq} with $U(t)\equiv 0$. \hfill $\blacksquare$
\end{example}

\begin{remark} Conversely, every continuous function $U:[0,T]\to\R$ generates an additive Loewner chain $F_t$ via equation \eqref{slit_Loewner_eq}. Then $F_{\mu_t}(\Ha)=\Ha\setminus K_t$ for some growing subsets $K_t\subset \Ha$. However, these sets do not necessarily describe a growing slit. This was noted first by Kufarev in \cite{Kuf47}. An example for such a driving function is the function $U:[0,1]\to\R,$ $U(t)=c \sqrt{1-t}$ with $c\geq 4$, see \cite{LMR10}. One can even generate spacefilling curves in this way, see \cite{LR12}.\\
  The set of all continuous driving functions that correspond to slits  is 
not known explicitly.  However, there are several partial results into that direction. 
 Roughly speaking, if $U$ is smooth enough, e.g.\ continuously differentiable, then 
 $K_t$ describe a slit. We refer to \cite{ZZ18} and the references therein for such results.
\end{remark}

We can now prove that, unlike the cases of classical, Boolean and free independence, there are more $\rhd$-hemigroup distributions than $\rhd$-infinitely divisible distributions.

\begin{theorem}There exists a $\rhd$-hemigroup distribution $\mu$ which is not $\rhd$-infinitely divisible.
\end{theorem}
\begin{proof}
Choose $\mu$ with mean $0$ such that $F_\mu$ is univalent and $F_\mu(\Ha)=\Ha\setminus \gamma(0,1]$ for a simple curve 
$\gamma:[0,1]\to\overline{\Ha}$ with $\gamma(0)\in\R$ and $\gamma(0,1]\subset \Ha$. This is possible for any such curve due to Theorem \ref{ftransform_images}. Assume that $\gamma$ is not a vertical line segment. 
Then $\mu$ is a $\rhd$-hemigroup distribution due to Theorem \ref{main_sec_8} or Theorem \ref{slit_Loewner_eq_thm}.\\

Let $F_{\mu_t}$ be the corresponding additive Loewner chain with $F_{\mu_t}(z) = z -\frac{t}{z} + \ldots$ 
at $\infty$. Then $F_{\mu_t}$ satisfies \eqref{slit_Loewner_eq}. Any other normalized Loewner chain generating $F_\mu$ clearly corresponds to a time change of the Loewner chain $(F_{\mu_t})$.\\

Now assume that $\mu$ is monotonically infinitely divisible. Then, by Theorem \ref{inf_div_mon}, $F_\mu$ can be embedded into an additive Loewner chain $(F_{\nu_t})_{t\geq0}$ which is a semigroup. We may assume that the first moments of all $\nu_t$ are equal to $0$. (Otherwise, consider the Loewner chain $F_{\nu_t}-\int_\R x \nu_t(dx)$). Then we have $\int_\R x^2 \nu_t(dx)=ct$ for some $c>0$ and 
$F_{\nu_t}=z-\frac{c t}{z}+...$. A time change yields $c=1$ and then $F_{\nu_t}=F_{\mu_t}$ for all $t\geq0$. \\
Now $F_{\nu_t}$ satisfies \eqref{slit_Loewner_eq} and as $(F_{\nu_t})_{t\geq0}$ is a semigroup, the Herglotz vector field $\frac{1}{U(t)-z}$ does not depend on $t$, i.e.\ $U(t)\equiv u\in\R$. In other words, $\gamma[0,1]$ must be a vertical line segment connecting $u$ to some $u+iT$, $T>0$, which is a contradiction to our assumption. Hence, $\mu$ is not monotonically infinitely divisible.
\end{proof}

If the $F$-transform of a probability measure $\mu$ is a univalent slit mapping, then $\mu$ has some special properties.

\begin{theorem}\label{slit_measures} Let $\mu$ be a probability measure on $\R$ such that $F_\mu$ is univalent and maps $\Ha$ conformally onto $\Ha\setminus \gamma$, where $\gamma$ is a slit starting at $C\in\R\setminus\{0\}$. Then  $\mu$ has the following properties:
\begin{itemize}
 \item[(a)] $\supp \mu= \{x_0\}\cup[a,b]$, where $\mu$ 
 has a continuous density $d(x)$ on the compact interval $[a,b]$ and an atom at some $x_0\in \R\setminus [a,b].$  
 Furthermore, $d(a)=d(b)=0$ and $d(x)>0$ in $(a,b)$.
 \item[(b)] $\mathcal{H}_\mu$ is defined and continuous on $\R\setminus\{x_0\}$ with $\mathcal{H}_\mu(a)=\mathcal{H}_\mu(b)=\frac1{\pi C}$.
 \item[(c)] There exists a decreasing homeomorphism $h:[a,b]\to[a,b]$ with 
 \[d(h(x))=d(x) \qquad \text{and} \qquad \mathcal{H}_\mu(h(x))=\mathcal{H}_\mu(x)\]
 for all $x\in [a,b]$. 
 \end{itemize}
\end{theorem}
\begin{proof}
As the domain $\Ha\setminus \gamma$ has a locally connected boundary, the mapping $F_\mu$ can be extended continuously
to $\overline{\Ha};$ see Theorem \ref{loc_con}.\\
There exists an interval $[a,b]$ such that $F_\mu([a,b])=\gamma$ and there is a unique $u\in (a,b)$ such that
$F_\mu(u)$ is the tip of the slit. All points $[a,u]$ correspond to the left side, all points $[u,b]$ to the right side 
of $\gamma.$ (This orientation follows from the behavior of $F_\mu(x)$ as  $x\to\pm\infty.$) 
Hence, there
exists a unique homeomorphism $h:[a,b]\to[a,b]$ with $h(u)=u,$ $h[a,u]=[u,b]$ such that 
$F_\mu(h(x))=F_\mu(x)$ for 
all $x\in[a,b].$\\
Furthermore, $F_\mu$ has exactly one zero $x_0\in\R\setminus[a,b]$ on $\R$, as the slit does not start at $0$. As $C=F_\mu(a)=F_\mu(b)$, we have $x_0 < a$ if and only if $C>0$.\\ 

 It follows from the Stieltjes-Perron inversion formula that $\supp \mu=\{x_0\}\cup [a,b]$ and that $\mu$ is absolutely continuous on
 $[a,b]$ and its density 
 $d(x)$ satisfies
 \[d(x) =  \lim_{\eps\to 0} -\frac{1}{\pi}\Im(1/F_\mu(x+i\eps))=-\frac{1}{\pi}\Im(1/F_\mu(x)).\] Hence,
 $d(h(x)) = d(x)$ for all $x\in[a,b]$, $d(x)>0$ on $(a,b)$, and $d(a)=d(b)=0$.\\
Let $\lambda=\mu(\{x_0\})$. Then we have
\[ \frac1{\pi} \Re\left(1/F_\mu(x)\right)
=\hat{\mathcal{H}}_\mu(x)=\hat{\mathcal{H}}_d(x)+\frac{\lambda}{\pi(x-x_0)} = \mathcal{H}_d(x)+\frac{\lambda}{\pi(x-x_0)}=\mathcal{H}_\mu(x)\]
for every $x\in\R\setminus\{x_0\}$ due to Exercise \ref{equal_Hilbert}. Here, $\mathcal{H}_d$ and $\hat{\mathcal{H}}_{d}$ are defined by replacing $\mu(dt)$ by $d(t)dt$ in the integration, and formally, we apply Exercise \ref{equal_Hilbert} to the probability measure defined by the density $d(t)/(1-\lambda)$.\\
Thus $\mathcal{H}_\mu(x)$ is continuous on $\R\setminus\{x_0\}$,
 $\mathcal{H}_\mu(a)=\mathcal{H}_\mu(b)=\frac1{\pi C}$, and $\mathcal{H}_\mu(h(x))=\mathcal{H}_\mu(x)$ on $[a,b]$.

\end{proof}
\begin{remark}\label{unio}
The proof shows that $x_0<a$ if $C>0$ and $x_0>b$ if $C<0$. \\
Furthermore, we note that there is a unique $u\in (a,b)$ with $d(u)=u$. This number is equal to the preimage of
the tip of  $\gamma$ under the map $F_\mu$.\\

Assume that only the density $d$ on $[a,b]$ is known. Then $\lambda:=\mu(\{x_0\})$ can simply be determined by $\lambda=1-\int_a^b d(x)\, dx$. 
Furthermore,  
$\mathcal{H}_{\mu}(x)=\mathcal{H}_d(x)+\frac{\lambda}{\pi(x-x_0)}$. As $\frac1{\pi C}=\mathcal{H}_{\mu}(a)=\mathcal{H}_{\mu}(b)$, we see that $x_0$ satisfies the quadratic equation
$\frac{\lambda(b-a)}{\pi(\mathcal{H}_d(a)-\mathcal{H}_d(b))}=(x_0-a)(x_0-b)$.
\end{remark}

The case of a slit starting at $0$ is quite similar.

\begin{theorem}\label{slit_measures2}Let $\mu$ be a probability measure on $\R$ such that $F_\mu$ is univalent and maps $\Ha$ conformally onto $\Ha\setminus \gamma$, where $\gamma$ is a slit starting at $C=0$. Then  $\mu$ has the following properties:
\begin{itemize}
 \item[(a)] $\supp \mu= [a,b]$, where $\mu$ 
 has a continuous density $d(x)>0$ on $(a,b)$. 
 \item[(b)] $\mathcal{H}_\mu$ is defined and continuous on $\R\setminus\{a,b\}$ with 
  $\lim_{x\downarrow a}|\mathcal{H}_\mu(x)|=\lim_{x\uparrow b}|\mathcal{H}_\mu(x)|=\infty$ or 
 $\lim_{x\downarrow a}d(x)=\lim_{x\uparrow b}d(x)=\infty$.
 \item[(c)] There exists a decreasing homeomorphism $h:[a,b]\to[a,b]$ with 
 \[d(h(x))=d(x) \qquad \text{and} \qquad \mathcal{H}_\mu(h(x))=\mathcal{H}_\mu(x)\]
 for all $x\in (a,b)$.
 \end{itemize}
\end{theorem}

\begin{remark}\label{rm:1}
Note that $F_\mu(\Ha)=F_\mu(\Ha-d)=F_{\mu'}(\Ha)$ whenever $\mu'$ is $\mu$ translated by $d\in\R$. \\
Conversely, if we have two univalent $F$-transforms with 
$F_\mu(\Ha)=F_{\mu'}(\Ha)=\Ha\setminus \gamma$, then $\alpha=F_\mu\circ F_{\mu'}^{-1}$ is an automorphism 
of $\Ha$ with $\alpha(\infty)=\infty$ and $\alpha'(\infty)=1$, which implies $\alpha(z)=z+d$ for some 
$d\in\R$. Hence $\mu'$ is a translation of $\mu$.
\end{remark}

\begin{remark}The homeomorphism $h$ is also called the \emph{welding homeomorphism} of 
the slit $\gamma$.\\
A slit $\gamma$ is called \emph{quasislit} if $\gamma$ approaches $\R$ nontangentially and 
$\gamma$ is the image of a line segment under a quasiconformal mapping.
The theory of conformal welding implies: 
$\gamma$ is a quasislit if and only if $h$ is quasisymmetric;
see \cite[Lemma 6]{Lind:2005} and \cite[Lemma 2.2]{MarshallRohde:2005}.\\
In this case, the slit is uniquely determined by $h$ and its starting point $C$. An example of a slit which is not uniquely 
determined by $h$ and $C$ is a slit with positive area.
\end{remark}

\begin{remark}
 Take a simple curve $\gamma:[0,1)\to\overline{\Ha}$ such that $\gamma(0)=0$, $\gamma(0,1)\subset \Ha\setminus[i,2i]$, 
 and the limit points of $\gamma$ as $t\to 1$ form the interval $[i,2i]$, as depicted in the figure 
 below. Let $D=\Ha\setminus (\gamma(0,1) \cup [i,2i])$. Then $D$ is simply connected. Let $F_\mu:\Ha\to D$ 
 be univalent. Then the limit $\lim_{\eps\downarrow 0}F_\mu(x+i\eps)$ exists for every 
 $x\in \R$ due to \cite[Exercises 2.5, 5]{MR1217706} and the fact that the prime end $p$ that corresponds to 
 $[i,2i]$ is accessible, i.e.\ the point $2i$ can be reached by a Jordan curve in $D$.
 In this case, $\mu$ has quite similar properties as in Theorem \ref{slit_measures2}, but the density $d$ is not continuous. The midpoint $u$ corresponds to the preimage of $p$ under $F_\mu$.\\
 If we replace the vertical interval $[i,2i]$ by a horizontal interval like $[i,1+i]$, a similar construction 
 yields a measure $\mu$ satisfying all properties as in Theorem \ref{slit_measures2} except that 
 $\mathcal{H}_\mu$ is not continuous. 
\end{remark}

 \begin{figure}[H]
 \begin{center}
 \includegraphics[width=0.4\textwidth]{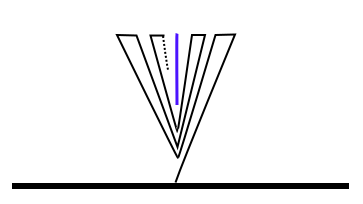}
 \caption{A curve $\gamma$ approaching a vertical line segment (blue).}
 \end{center}
 \end{figure}

\section{Further reading}

The radial Loewner chains correspond exactly to multiplicative $\rhd$-semigroups (see Exercises \ref{exxxa} and \ref{exxxb}) and describe the distributions of multiplicative quantum processes with monotonically independent increments, see \cite{FHS}. \\

In the spirit of geometric function theory, which couples geometric and analytic properties of holomorphic functions, one can translate several geometric properties of univalent $F$-transforms (and $\eta$-transforms) into properties of the probability measures, see \cite[Sections 6 and 7]{FHS}.

\newpage

\section{Exercises}

\begin{exercise}Show that the assumption in Definition \ref{EV_def:evolution_family} that all $f_{s,t}$ are non-constant can be dropped, i.e. if $f_{s,t}:D\to D$, $0\leq s\leq t$, is a family of holomorphic functions satisfying $(a)-(c)$ in Definition \ref{EV_def:evolution_family}, then  all $f_{s,t}$ must be non-constant. 
\end{exercise}

\begin{exercise}[Herglotz representation]\label{Herglotz}
Let $f:\D\to \C$ be holomorphic with $\Re(p(z))\geq 0$ for all 
$z\in \D$ and $f(0)=1$. Show that there exists a probability measure $\rho$ on $\partial \D$ such that 
\[f(z) = \int_{\mathbb{\partial \D}}\frac{u+w}{u-w} \rho(du). \]
Furthermore, consider the power series expansion $f(z)=1+\sum_{n=1}^\infty c_nz^n$ and show that
\[ |c_n|\leq 2 \quad  \text{for all $n\geq 1$.} \]
\end{exercise}

\begin{exercise}Assume that $\mu$ is a probability measure such that $F_\mu(\Ha)=\Ha\setminus \gamma$ 
 for a simple curve $\gamma$. Does $F_\mu$ have to be injective?
\end{exercise}

\begin{exercise}
Find an $F$-transform $F_\mu$ such that $\mu(A)=\mu(-A)$ for all $A\in \mathcal{B}(\R)$ and $\mathcal{H}_\mu(x)$ exists for all $x\in \R$ with 
$\mathcal{H}_\mu(-x)=-\mathcal{H}_\mu(x)$.
\end{exercise}

\begin{exercise}Show that if $\mu,\nu$ are $\rhd$-infinitely divisible, then $\mu\rhd \nu$ is not $\rhd$-infinitely divisible 
in general.
\end{exercise}

\begin{exercise}\label{exxxa}Let $H$ be a complex Hilbert space and let $\varphi:B(H)\to \C$ be an expectation. Let $U\in B(H)$ be a unitary, i.e.\ $UU^*=U^*U = I$. Then there exists a probability measure $\mu$ on $\partial \D$, the distribution of $U$, such that $\varphi(U^n) = \int_{\partial \D} x^n \mu({\rm d}x)$. We define the transforms
\begin{equation*}\label{def:psi-eta}
\psi_\mu(z)=\int_{\partial \D}\frac{xz}{1-xz}\, \mu({\rm d}x) = \varphi\left(\frac{Uz}{1-Uz}\right) \quad \text{and} \quad \eta_\mu(z)=\frac{\psi_\mu(z)}{1+\psi_\mu(z)}, \quad z \in \D.
\end{equation*}
\begin{itemize}
	\item[(a)] Show that $\eta_\mu$ maps $\D$ into $\D$ with $\eta_\mu(0)=0$ and that any holomorphic $f:\D\to \D$ with $f(0)=0$ has the form $f=\eta_\mu$ for some probability measure $\mu$ on $\partial \D$.  
	\item[(b)] Let $U,V\in B(H)$ be unitary with distributions $\mu$ and $\nu$ and let $\alpha$ be the distribution of $UV$. Assume that $(U-I,V)$ is monotonically independent. Show that $\eta_{\alpha} = \eta_\mu \circ \eta_\nu$. (We also write $\alpha=\mu\rhd \nu$.)
\end{itemize}
\end{exercise} 

\begin{exercise}\label{exxxb}Let $(\mu_{s,t})_{0\leq s\leq t}$ be a family of probability measures on $\D$ which form a $\rhd$-hemigroup on $\partial\D$, i.e.\ 
\begin{itemize}
	\item[(1)] $\mu_{s,s}=\delta_1$ for all $s\geq 0$,
	\item[(2)] $(s,t)\mapsto \mu_{s,t}$ is continuous with respect to weak convergence,
	\item[(3)] $\mu_{s,u} = \mu_{s,t} \rhd \mu_{t,u}$ for all $0\leq s\leq t\leq u$.
\end{itemize} 
Show: If  $(\mu_{s,t})_{0\leq s\leq t}$ is such a  $\rhd$-hemigroup, then $\eta_{\mu_{0,t}}$ is a radial Loewner chain with transition mappings 
$\eta_{\mu_{s,t}}$. Conversely, if $f_{t}$ is a radial Loewner chain with transition mappings $f_{s,t}$, then $f_{s,t}=\eta_{\mu_{s,t}}$ for a 
$\rhd$-hemigroup $(\mu_{s,t})_{0\leq s\leq t}$ on $\partial \D$.
\end{exercise}

\begin{exercise}\label{scaling_chordal}
Consider a slit $\gamma:[0,T]\to \Ha\cup \{\gamma(0)\}$ and let $F_{t}=F_{\mu_t}$ be the corresponding mappings satisfying \eqref{slit_Loewner_eq} with $U:[0,T]\to \R$.
Now consider the scaled slit $t\mapsto c\gamma(t)$ for some $c>0$. Find the hydrodynamic parametrization 
$\Gamma:[0,S]\to c \gamma[0,T]$ and find the driving function $V:[0,S]\to \R$ of $\Gamma$.
\end{exercise}

\begin{exercise}For probability measures $\mu, \nu$  on $\partial \D$, the Boolean convolution $\mu\uplus\nu$ is defined by 
\[\eta_{\mu\uplus\nu}(z)/z=\eta_\mu(z)/z \cdot \eta_\nu(z)/z,\] 
which is well-defined by Exercise \ref{exxxa} (a) and due to the Schwarz lemma.\\
Show that a family $(\mu_t)_{t\geq0}$ of probability measures on $\partial\D$ is a continuous $\uplus$-semigroup if and only if there exists a holomorphic function $p:\D\to \C$ with $\Re (p(z)) \geq0$ such that \[\eta_{\mu_t}(z)/z = \exp(-t p(z)).\] 
\end{exercise}

\newpage

\chapter[Free hemigroups and Loewner chains]{Free hemigroups, Loewner chains, and nonlinear resolvents}\label{sec9}

We have seen that $\rhd$-hemigroups correspond to additive Loewner chains and that $\rhd$-hemigroup distributions correspond to univalent $F$-transforms. In this chapter, we will see that the $F$-transforms of a $\boxplus$-hemigroup also form an additive Loewner chain, which shows that distributions of $\boxplus$-hemigroups  (=the freely infinitely divisible distributions), correspond to special univalent $F$-transforms.\\

The Loewner equation of an additive Loewner chain is given by
\begin{equation}\label{intro_1_eq}
\frac{\partial}{\partial t} f_{t}(z) = \frac{\partial}{\partial z}f_{t}(z)\cdot M(z,t) \quad \text{for a.e.\ $t\geq 0$, $f_{0}(z)=z\in \Ha,$}
\end{equation}
with 
\begin{equation*}
M(z,t)=a_t + \int_\mathbb{R}\frac{1+xz}{x-z} \rho_t({\rm d}x).
\end{equation*}

Now replace $M(t,z)$ by $M(t,f_t(z))$ to obtain the following modified equation:
\begin{equation}\label{intro_2_eq}
 \frac{\partial}{\partial t}f_t(z)=\frac{\partial}{\partial z}f_{t}(z)\cdot M(f_t(z),t), \quad \text{for a.e.\ $t\geq 0$, $f_{0}(z)=z\in \Ha.$}
\end{equation}

At first sight, it seems unnatural to look at this equation, for, if $G$ is a generator on $\Ha$ and $\varphi:\Ha\to \Ha$ a holomorphic self-mapping, then $G\circ \varphi$ need not be a generator. However, our generators have the form $G:\Ha\to \Ha\cup \R$, and then also $G\circ \varphi:\Ha\to \Ha\cup \R$ is a generator on $\Ha$ (see Exercise \ref{muhaha}).\\

In fact, the modified equation \eqref{intro_2_eq} is  much simpler than \eqref{intro_1_eq}. 
While the inverse functions $g_t=f_t^{-1}$ of the solution to \eqref{intro_1_eq} satisfy Loewner's ordinary differential equation 
\[\frac{\partial}{\partial t}g_t(z) = - G(t,g_t(z)),\] the inverse functions $g_t$ of a solution to \eqref{intro_2_eq} 
satisfy \[\frac{\partial}{\partial t}g_t(z) = - G(t,z),\quad \text{i.e.} \quad  g_t(z)=z-\int_0^t G(s,z)ds.\]
Note that this implies that $g_t=f_{t}^{-1}$ extends holomorphically to the whole domain $\Ha$. Thus the growing hulls $K_t$ defined by $f_t(\Ha)=\Ha\setminus K_t$ are always bounded by analytic curves in \eqref{intro_2_eq}, while they are arbitrary in \eqref{intro_1_eq}. \\

We will see that \eqref{intro_2_eq} characterizes Loewner chains that arise as $F$-transforms of $\boxplus$-hemigroups.

\section{Nonlinear resolvents}

\begin{definition}
Let $G$ be an infinitesimal generator of a continuous semigroup on a simply connected domain $D\subsetneq \C$. Let $t\geq 0$. If the equation
\begin{equation}\label{nl} w = z - t\cdot G(z) \end{equation}
has a unique solution for all $w\in D$, we call $z=J_t(w)=J_t(w,G)$ the \emph{nonlinear resolvent} (at ``time'' $t$) of $G$. 
\end{definition}

If $J_t$ exists, it is a holomorphic mapping from $D$ into itself. 

\begin{remark}\label{bounded_convex}
  If $D$ is bounded and convex, then nonlinear resolvents exist for all $t\geq 0$, see  \cite[Theorem 1.1]{RS97}. 
	Conversely, if $G:D\to\C$ is holomorphic and bounded such that all nonlinear resolvents in \eqref{nl} exist, then $G$ is an infinitesimal generator, see  \cite[Corollary 1.2]{RS97}.
\end{remark}

\begin{example}Let $D=\D$ and consider the semigroup $F_t(z)=e^{-t} z$. Then $G(z)=-z$ and the solution to the 
equation $w = z + t z$ is given by $J_t(w)=\frac{w}{1+t}$.\hfill $\blacksquare$
\end{example}

On unbounded domains, $J_t$ might exist for some, but in general not for all $t\geq0$.

\begin{example}\label{ex_H}Let $D=\Ha$ and consider the semigroup $F_t(z)=e^t z$. Then $G(z)=z$ and the equation $w = z - t z$ 
has no solution in $\Ha$ for $t\geq 1$, but $J_t(w)=\frac{w}{1-t}$ for $t\in[0,1)$.\hfill $\blacksquare$
\end{example}

\begin{lemma}\label{rm_1}If $D$ is bounded and convex, then $G\circ J_t$ is an infinitesimal generator on $D$ for all 
$t\geq 0$. 
\end{lemma}
\begin{proof}
This is clearly true for $t=0$. So let $t>0$. Then  $(G\circ J_t)(w)= (J_t(w)-w)/t$.
Now we use the fact that $f(z)-z$, $f:D\to D$ holomorphic, is always an infinitesimal generator on bounded convex domains, see \cite[Proposition 4.3]{RS96}, and the fact that $r\cdot G$ is an infinitesimal generator for every infinitesimal generator $G$ and $r>0$.
\end{proof}

\begin{theorem}\label{result1}${}$
\begin{itemize}
\item[(1)] Let $G$ be an infinitesimal generator on a bounded and convex domain $D\subset \C$ with resolvents $J_t:D\to D$. Then $(J_t)_{t\geq0}$ 
is a decreasing Loewner chain satisfying the Loewner partial differential equation
\begin{equation}\label{m}\frac{\partial}{\partial t}J_t(w)=J_t'(w)\cdot G(J_t(w))\quad \text{for all $t\geq 0$}, \quad J_0(w)\equiv w\in D.\end{equation}
The domains $J_t(D)$ contract to the zero set of $G$, i.e.\ $\bigcap_{t\geq 0} J_t(D) = G^{-1}(0)$.
\item[(2)] Let $D\subset \C$ be a (possibly unbounded) convex domain and let $G$ be an infinitesimal generator on $D$. 
Furthermore, let $(J_t:D\to D)_{t\geq0}$ be a family of holomorphic functions satisfying \eqref{m}. 
Then $J_t$ are the resolvents of $G$ on $D$.
\item[(3)] Let $D\subset \C$ be a bounded and convex domain and let $G:[0,\infty)\times D\to \C$ be such that $z\mapsto G(t,z)$ is holomorphic for a.e.\ $t\geq0$, $t\mapsto G(t,z)$ is locally integrable for all $z\in D$, and $H_t(z):=\int_0^t G(s,z) ds$ 
is an infinitesimal generator on $D$ for every $t\geq 0$. Let $J_t:D\to D$ be the nonlinear resolvent of $H_t$ at time $1$. Then $(J_t)_{t\geq 0}$ is the 
unique solution of holomorphic self-mappings of $D$, locally absolutely continuous in $t$, to  
\begin{equation}\label{m2}\frac{\partial}{\partial t}J_t(z)=J_t'(z)\cdot G(t,J_t(z)) \quad \text{for a.e.\ $t\geq 0$}, \quad J_0(z)\equiv z\in D.\end{equation}
\end{itemize}
\end{theorem}
\begin{proof}${}$
\begin{itemize}
\item[(1)] 
Let $J_t$ be the resolvents of $G$. From $w = J_t(w) - tG(J_t(w))$ we get by differentiation 
\[J_t'(w) - t G'(J_t(w)) \cdot J_t'(w)\equiv 1,\] i.e.\ $J_t'(w) = (1-t G'(J_t(w)))^{-1}$.
Differentiation with respect to $t$ yields \[0 = (1-t G'(J_t(w)))\frac{\partial}{\partial t}J_t(w) - G(J_t(w)),\] which gives us
\[\frac{\partial}{\partial t}J_t(w)=J_t'(w)\cdot G(J_t(w)).\]

Put $\varphi_t(z)=z-tG(z)$. As ($\varphi_t \circ J_t)(z)=z$, each $J_t$ is clearly a univalent function. Moreover, the continuity of 
$t\mapsto \varphi_t$ implies that also $t\mapsto J_t$ is continuous. \\
Now let $z\in D$. The statement $z\in J_t(D)$ is equivalent to $\varphi_t(z)\in D$. Hence, as $D$ is convex, we see that 
$z\in J_t(D)$ implies $z\in J_s(D)$ for all $s\in[0,t]$. We conclude that $J_t(D)\subset J_s(D)$ whenever $s\leq t$ and that $(J_t)_{t\geq0}$ is a decreasing Loewner chain. Due to Lemma \ref{rm_1}, we see that $(t,z)\mapsto G(J_t(z))$ is a Herglotz vector field and thus equation \eqref{m} is a Loewner partial differential equation of the form \eqref{EV_Loewner}.
[We can also argue as follows: $J_t$ solves \eqref{m} and $G\circ J_t$ is a Herglotz vector field. Hence, Theorem \ref{89} implies that $(J_t)_{t\geq 0}$ is a decreasing Loewner chain.]\\

Let $z\in D$. If $G(z)=0$, then $J_t(z)=z$ for all $t\geq 0$. If $G(z)\not = 0$, then there exists $T>0$ such that $z-t\cdot G(z)\not\in D$ for all 
$t\geq T$ as $D$ is bounded. Hence, $z\not\in J_t(D)$ for all $t\geq T$ and we conclude $\bigcap_{t\geq 0} J_t(D) = G^{-1}(0)$.\\

\item[(2)]
Now let $J_t:D\to D$ be a family of holomorphic functions satisfying \eqref{m}, where $D$ is a convex domain. (Note that now, we do not know yet whether $(t,z)\mapsto G(J_t(z))$ is a Herglotz vector field.)
Consider the differential equation 
\[\frac{\partial}{\partial t}\varphi_t(z) = -G(J_t(\varphi_t(z))), \quad \varphi_0(z)=z.\]

Fix $z\in D$. We can solve this equation at least for $t$ small enough.
A small computation shows that $\frac{d}{dt}[J_t(\varphi_t(z))]=0$, i.e.\ $J_t(\varphi_t(z))$ does not depend on $t$ and $J_t(\varphi_t(z))=J_0(\varphi_0(z))=z$. 
Hence $\frac{\partial}{\partial t}\varphi_t(z) = -G(z)$. We conclude that
$t\mapsto \varphi_t(z)$ simply describes a straight line: \[\varphi_t(z)=z-tG(z).\]

In particular, we can now define $\varphi_t(z)$ for all $z\in D$ and all $t\geq0$ by $\varphi_t(z)=z-tG(z)$.\\
Let $D_t=\{z\in D\,|\, \varphi_t(z)\in D\}$. The convexity of $D$ implies that $D_t\subset D_s$  whenever $s\leq t$. Thus, for all $z\in D_t$, we have 
$J_t(\varphi_t(z))=z$. Applying $\varphi_t$ gives $\varphi_t(J_t(w))=w$ for all $w\in \varphi_t(D_t)$. As the left side extends holomorphically to $D$, we see that $\varphi_t(J_t(w))=w$ for all $w\in D$. Applying $J_t$ gives $J_t(\varphi_t(z))=z$ for all $z\in J_t(D)$. 
We conclude that $J_t$ is the nonlinear resolvent of $G$.

\item[(3)]By definition,  $J_t$ is the inverse of $\varphi_t(z)=z\mapsto z- \int_0^t G(s,z) ds$. Clearly, $t \mapsto \varphi_t(z)$ is locally absolutely continuous for any $z\in D$.\\ 
Furthermore, as $D$ is bounded, the set $\{J_\tau\,|\, \tau\geq 0\}$ is a normal family and thus $\{J'_\tau\,|\, \tau\geq 0\}$ is locally uniformly bounded.\\

Let $s,t\geq 0$ and $z\in D$. Put $w=J_t(z)$. If $s$ is close enough to $t$, then $z'(s)=\varphi_s(w)\in D$. 
We have 
\[ J_t(z) - J_s(z)  = J_t(z) - J_s(z'(s)) +  J_s(z'(s)) - J_s(z)= J_s(z'(s)) - J_s(z).\]
From $z'(s)-z = \varphi_s(w) -\varphi_t(w)$, we see that $\tau \mapsto J_\tau(z)$ is locally absolutely continuous for any $z\in D$.\\ 

$J_t$ solves the equation $w = J_t(w) - \int_0^t G(s,J_t(w)) ds$.  We get by differentiation 
\[J_t'(w) - \int_0^t G'(s,J_t(w)) \cdot J_t'(w) ds = 1,\] i.e.\ $J_t'(w) = (1-\int_0^t G'(s,J_t(w))ds)^{-1}$.
Differentiation with respect to $t$ yields \[0 = (1-\int_0^t G'(s,J_t(w))ds)\frac{\partial}{\partial t}J_t(w) - G(t,J_t(w)),\] which gives us
\[\frac{\partial}{\partial t}J_t(w)=J_t'(w)\cdot G(t, J_t(w)).\]

Next we show uniqueness of the solution. Let $(f_t)_{t\geq 0}$ be another solution of holomorphic mappings $f_t:D\to D$, locally absolutely continuous in $t$, satisfying \eqref{m2}.\\
Choose some $z_0\in D$. Then we find some open disc $B\subset D$ with center $z_0$ and some $\eps>0$ such that 
$f_t$ is injective on $B$ for all $t\in[0,\eps)$. The inverse functions $g_t$ satisfy
\[ \frac{\partial}{\partial t}g_t(w) = -G(t,w)\]
for a.e.\ $t\in[0,\eps]$ and all $w\in \cap_{t\in[0,\eps]} f_t(B)$ (which is non-empty for $\eps$ small enough). This implies $g_t(w) = w-\int_0^t G(s,w) ds$, which shows $f_t=J_t$ on $B$ for all 
$t\in[0,\eps]$. The identity theorem implies $f_t=J_t$ on $D$ for all $t\in[0,\eps]$. \\In this way, we also see that the set of all 
$t\geq 0$ with $f_t=J_t$ is open (in $[0,\infty)$). At the same time, it is also closed, and thus equal to $[0,\infty)$.
\end{itemize}
\end{proof}

\section{Free hemigroups and Loewner chains}

\begin{theorem}\label{add_gen_nonaut} Let $G(t,z)$ be a Herglotz vector field on $\Ha$ such that, for a.e.\ $t\geq0$, 
$z\mapsto G(t,z)$ maps $\Ha$ into $\Ha\cup \R$ and $\lim_{y\to\infty} G(t,iy)/y=0$. 
Then there exists a unique solution $(f_t)_{t\geq0}$, locally absolutely continuous in $t$, to
\begin{equation}\label{op0} \frac{\partial}{\partial t}f_t(z)=f_t'(z)\cdot G(t,f_t(z))\quad \text{for a.e.\ $t\geq 0$}, \quad f_0(z)=z\in \Ha.
\end{equation}
The solution can also be written as $(f_t)_{t\geq0}=(F_{\mu_t})_{t\geq0}$ for a family of $\boxplus$-infinitely divisible probability measures on $\R$ and 
$(f_t)_{t\geq0}$ is a decreasing Loewner chain. \\
If $G(t,z)$ does not depend on $t$, then $G(t,z)=-R_{\mu_1}(1/z)$ and the family $(f_t)_{t\geq0}$ are the resolvents of $-R_{\mu_1}(1/z)$ in this case.
\end{theorem}
\begin{proof}Consider $H_t(z):=\int_0^t -G(s,z) ds$, $t\geq 0$. This function is also a function of the form \eqref{form0} 
and Theorem \ref{thmBV93} implies that we find $\boxplus$-infinitely divisible probability measures
 $(\mu_t)_{t\geq 0}$ such that $\varphi_{t}:=R_{\mu_t}(1/z) = H_t(z)$. 
Put $f_t=F_{\mu_t}$. Then $f_t(\varphi_{_t}(z)+z)=z$ for all $z\in\Ha$. Differentiation yields
 \[0=\frac{\partial}{\partial t}f_t(\varphi_{_t}(z)+z) + 
f'_t(\varphi_{t}(z)+z)\frac{\partial}{\partial t}\varphi_{t}(z)=\frac{\partial}{\partial t}f_t(\varphi_{t}(z)+z) - 
f'_t(\varphi_{t}(z)+z)G(t,z).\] 
Put $w=z+\varphi_{t}(z)=f_t^{-1}(z)$. Then 
\[\frac{\partial}{\partial t}f_t(w) = f'_t(w) G(t,f_t(w))\]
for all $w$ in the image domain $f_t^{-1}(\Ha)$, and thus, in particular, for all $w\in \Ha$. Clearly, $f_0$ is the identity as 
$\varphi_{\mu_0}=0$. Due to Exercise \ref{muhaha}, $G(t,f_t(z))$ is a Herglotz vector field and Theorem \ref{is_addi} implies that 
$(f_t)$ is a decreasing Loewner chain. 
Uniqueness of the solution is shown as in the proof of Theorem \ref{result1} (3).\\

If $G(t,z)$ does not depend on $t$, then $G(t,z)=-R_{\nu}(1/z)$ for some probability measure $\nu$ and 
$\varphi_{t} = H_t(z) = t R_{\nu}(1/z)$. This shows that $\nu=\mu_1$ and that $(\mu_t)_{t\geq 0}$ is a free semigroup due to 
Theorem \ref{thmBV93}. Theorem \ref{result1} (2) now implies 
that the functions $(J_t)_{t\geq0}$ are the resolvents of the generator $-R_{\mu_1}(1/z)$.
\end{proof}

\begin{corollary}
Let $(\mu_{s,t})_{0\leq s\leq t}$ be a $\boxplus$-hemigroup and let $F_{t}=F_{\mu_{0,t}}$.
\begin{itemize}
	\item[(1)] Then $(F_t)_{t\geq0}$ is an additive Loewner chain on $\Ha$.
	\item[(2)] If $(\mu_{t})_{t\geq 0}$ is a $\boxplus$-semigroup, then $(F_t)_{t\geq0}$ are the nonlinear resolvents of $-R_{\mu_1}(1/z)$.
\end{itemize}
\end{corollary} 
\begin{proof}
Each $\mu_{0,t}$ is $\boxplus$-infinitely divisible due to Theorem \ref{measures_incre_free} and Theorem \ref{add_gen_nonaut} implies that $F_t$ is univalent. Clearly, $t\mapsto F_t$ is continuous. It only remains to show that we can write $F_t = F_s \circ G$ for some holomorphic $G:\Ha\to \Ha$. 
 (Then $F_{s,t}:=F_s^{-1}\circ F_t$ defines the transition mappings of the Loewner chain.)\\

Fix $0\leq s\leq t$. Embed $\mu_{0,s}$ into a $\boxplus$-semigroup $(\alpha_t)_{t\geq0}$ with $\alpha_1=\mu_{0,s}$ and $\mu_{s,t}$ into a $\boxplus$-semigroup $(\beta_t)_{t\geq0}$ with $\beta_1=\mu_{s,t}$. Define $R_t = R_{\alpha_t}$ for $0\leq t\leq 1$ and 
$R_t = R_{\alpha_1}+R_{\beta_{t-1}}$ for $t>1$. Then we obtain probability measures $(\gamma_t)_{t\geq 0}$ defined via $R_{\gamma_t}=R_t$ and it is easy to verify that $F_{\gamma_t}$ satisfies \eqref{op0} for all $t\in [0,\infty)\setminus\{1\}$. Thus, $F_{\gamma_t}$ is a decreasing Loewner chain. We have $F_{\gamma_1}=F_{\mu_{0,s}}=F_s$ and 
$F_{\gamma_2}=F_{\mu_{0,s} \boxplus \mu_{s,t}}=F_{\mu_{0,t}} = F_t$.
Thus $F_t = F_s \circ G$ for some holomorphic $G:\Ha\to \Ha$.
\end{proof}

In particular, if $\mu$ is a $\boxplus$-hemigroup distribution, then $F_\mu$ is univalent and $\mu$ is also a $\rhd$-hemigroup distribution.

\begin{remark}
We see that for every $\boxplus$-hemigroup $(\mu_{s,t})_{0\leq s \leq t}$ there  exists a $\rhd$-hemigroup $(\nu_{s,t})_{0\leq s \leq t}$ with $\mu_{0,t}=\nu_{0,t}$ for all $t\geq 0$.\\
The special case $\mu_{s,t}=\mu_{0,t-s}$ corresponds to a free $\boxplus$-semigroup and to free L\'{e}vy processes.\\
If we require instead $\nu_{s,t}=\nu_{0,t-s}$, we obtain another class of processes, called L\'{e}vy processes of the second kind in \cite{Bia98}.
\end{remark}

 \part[Applications]{Applications} 

\chapter{Models with random matrices}

Random matrices $M\in \C^{N\times N}$ are widely used in statistical models. For example, every sample covariance matrix can be seen as a random matrix. Via the expectation
\[ M \mapsto \frac1{N} \mathbb{E}[\tr(M)]\]
we can view random matrices as quantum random variables (as in Example \ref{444}).

\section{A toy example: quantum autoregression}

Modeling a time series $x_0, x_2,..., x_N \in \R$ is all about finding a function $f$ that explains 
the hidden structure of the values $x_n$ such that 
\[ x_n = f(\text{some data available at time $n-1$}) + \text{``noise''}.\]
The data available at time $n-1$ clearly contains all $x_0,...,x_{n-1}$, but there might be further data, maybe from some other time series. 
We would model $x_n$ as a random variable $X_n$ such that
\[ X_n = f(\text{some data available at time $n-1$}) + \eps_n,\]
 and the noise  $(\eps_n)$ is usually modeled as \emph{white noise}: $\eps_0,\eps_1,...$ are \textit{iid} random variables with mean $0$ and variance $\sigma^2<\infty$.\footnote{This noise is called ``white'' because the autocorrelation $\rho(\eps_n,\eps_{n+m})$, which is constant $0$ for $m\in\N$, can be translated into a constant spectral density, which reminds us of white light. This analogy has produced further notions such as red, pink, brown, ... noise.} We can use this model to forecast the time series, which generally means to determine the distribution of some $X_{N+\Delta t}$ in the future. Usually, this is done by giving a point prediction and an prediction interval. Here we also need the distribution of the $\eps_n$. It can be estimated from the residuals of the observations, or it can be checked whether we can put an  additional assumption into our model (e.g.\ that $(\eps_n)$ is Gaussian white noise, i.e.\ $\eps_n$ is $\mathcal{N}(0,\sigma^2)$-distributed).\\

A simple example of this type is an \emph{autoregressive model} of order $p$, $p\geq 1$, denoted by $AR(p)$, where $(X_n)$ is given by
initial values $X_0,...,X_{p-1}$ and \[ X_{n} = c_1 X_{n-1} + \ldots + c_p X_{n-p} + \eps_n, \]
for all $n\geq p$, where $(\eps_n)$ is a white noise with variance $\sigma^2$ and $c_1, \ldots, c_p\in\R$ and $\sigma^2$ are the parameters of the model, see \cite[Chapter 3]{SS11}.  (The process $(X_n)$ is usually assumed to be stationary in the following sense: 
 $\mathbb{E}[X_n]=0$,  $\E[X_n^2]<\infty$ for all $n\in\N$ and the autocovariance does not vary with respect to time, i.e.\ $\E[X_{n}X_{n+m}]$ only depends on $m$. The stationarity now puts a further constraint on the parameters: the roots of the polynomial $1-\sum_{k=1}^p c_k z^{k}$ must lie outside the closed unit disc $\overline{\D}$.)\\

\begin{example}\label{778}
If $X_n$ is an $AR(1)$ process with $X_0=1$, $\sigma^2=1$, then $X_n$ has mean $c_1^n$ and variance 
$1+c_1^2+...+c_1^{2(n-1)}$. We have $c_1^n\to 0$ and $1+c_1^2+...+c_1^{2(n-1)}\to \frac{1}{1-c_1^2}$ if $|c_1|<1$.\\
If $c=1$, we obtain a random walk that approximates a Brownian motion.
  \begin{figure}[H]
 \begin{center}
 \includegraphics[width=0.9\textwidth]{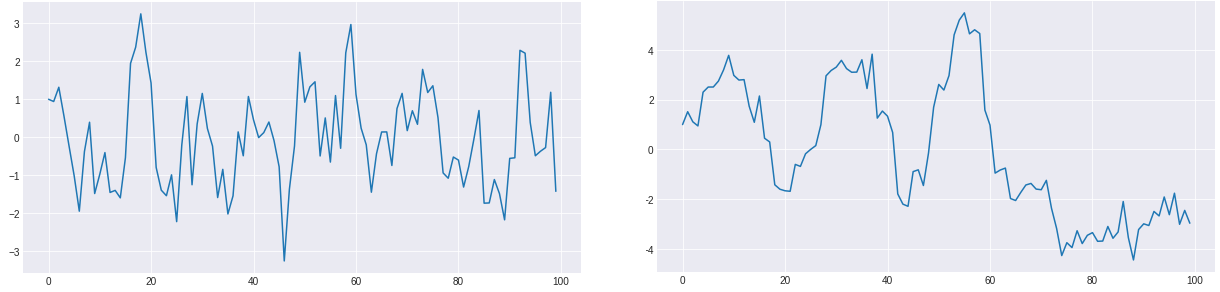}
 \caption{An $AR(1)$ process with $X_0=1$, $\sigma^2=1$, and $c_1=0.5$ (left), $c_1=0.95$ (right).}
 \end{center}
 \end{figure}\vspace{-5mm}\hfill $\blacksquare$
\end{example}

The parameters can be determined, e.g., by the Yule-Walker equations. For $m\in \Z$, let $\gamma_m = \E[X_nX_{n+m}]$. (We have $\gamma_{-m}=\gamma_m$.) Then

 \[ \displaystyle  {\begin{pmatrix}\gamma _{1}\\\gamma _{2}\\\gamma _{3}\\\vdots \\\gamma _{p}\\\end{pmatrix}}={\begin{pmatrix}\gamma _{0}&\gamma _{{-1}}&\gamma _{{-2}}&\dots \\\gamma _{1}&\gamma _{0}&\gamma _{{-1}}&\dots \\\gamma _{2}&\gamma _{{1}}&\gamma _{{0}}&\dots \\\vdots &\vdots &\vdots &\ddots \\\gamma _{{p-1}}&\gamma _{{p-2}}&\gamma _{{p-3}}&\dots \\\end{pmatrix}}{\begin{pmatrix}c_{{1}}\\c_{{2}}\\c_{{3}}\\\vdots \\c_{{p}}\\\end{pmatrix}}\]
In addition, \[ \displaystyle \gamma _{0}=\sum _{{k=1}}^{p}c_{k}\gamma _{{-k}}+\sigma^{2},\]
which can be solved for $\sigma^2$. (The order $p$ can be determined by looking at the partial autocorrelation function.)\\

We can now imitate the classical case to define quantum stationary processes and quantum white noise.

\begin{definition}Let $(B(H), \varphi)$ be a quantum probability space and fix an independence. A family
 $(X_n)_{n\in\N_0}$ of self-adjoint elements in $B(H)$ is called stationary if 
 $\varphi(X_n)=0$ for all $n\in\N_0$ and the autocovariance does not vary with respect to time, i.e.\ $\varphi(X_{n}X_{n+m}),$ $\varphi(X_{n+m}X_{n})$ only depend on $m$.\\
If $X_0, X_1,...$ are \textit{iid}, then $(X_n)_{n\in\N}$ is called a quantum white noise.
\end{definition}

For example, for matrices $X_0,X_1,... \in \C^{N\times N}$, we could consider a ``free quantum autoregressive model'' of order $p$, $p\geq 1$,
  given by the equation
\[ X_{n} = c_1 X_{n-1} + \ldots + c_p X_{n-p} + \eps_n, \]
for all $n\geq p$, where $(\eps)_n$ is a free quantum white noise with respect to $\varphi(X)=\frac1{N}\tr(X)$ with variance $\sigma^2$, which is also freely independent of $X_0,...,X_{p-1}$. The parameters of the model can be determined as in the classical case, namely by using the Yule-Walker equations which only depend on the autocovariance function of $(X_n)$.\\

Let $\mu_n$ be the distribution of $X_n$, $\mu_\eps$ the distribution of $\eps_n$, and denote by $\nu_n$ the distribution of 
$Y=c_1 X_{n-1} + \ldots + c_p X_{n-p}$. The free independence of $\eps_n$ and $Y$ gives us 
\[ \mu_{n} = \nu_n \boxplus \mu_\eps. \]

But is the free independence of matrices of any practical relevance? The answer is yes due to the existence of several strong theorems on the asymptotic behavior of random matrices as their size goes to $\infty$. 

\section{Asymptotically independent random matrices}

Let $N\in \N$ and consider classical independent random variables 
$\{X_{j,k}\}_{1\leq j\leq k\leq N} \cup$\\  $\{Y_{j,k}\}_{1\leq j<k\leq N}$ such that 
each $X_{j,k}$ and each $Y_{j,k}$ has a $\mathcal{N}(0,1/(2N))$ distribution. \\

Then $\mathbb{E}[X_{j,k}+iY_{j,k}]=0$ and $\mathbb{E}[|X_{j,k}+iY_{j,k}|^2]=\frac1{N}$. \\

Now consider the $N\times N$-random matrix $A_N=(A_{N,j,k})_{1\leq j,k\leq N}$ where 
\[\text{$A_{N,j,k} = X_{j,k} + iY_{j,k}$ for $j<k$, 
       $A_{N,j,k} = X_{k,j} - iY_{k,j}$ for $k<j$, and 
			$A_{N,j,j} = X_{j,j}$.}\] 
			Such a matrix is a self-adjoint Gaussian random matrix, also called a \emph{GUE} random matrix (Gaussian unitary ensemble). 
			For a random matrix $X\in \C^{N\times N}$, let $\varphi_N(X)=\frac1{N}\mathbb{E}[\tr(X)]$.
E.\ Wigner, who used random matrices for models in physics, found that the semicircle distribution describes the limit behavior of $A_N$.
	
\begin{theorem}	[\cite{Wig55, Wig58}]		
		\[\lim_{N\to\infty} \varphi_N(A_N^k) = \lim_{N\to\infty}\frac1{N} \mathbb{E}[\tr(A_N^k)] = \frac1{2\pi} \int_{-2}^2 x^k \sqrt{4-x^2} {\rm d}x.\]	
\end{theorem}

In other words, if $N$ is large, we will expect that the number of eigenvalues of $A_N$ in some interval $[a,b]$ is equal to 
\[\frac{N}{2\pi} \int_{a}^b \sqrt{4-x^2}\textbf{1}_{[-2,2]}(x) {\rm d}x.\]

  \begin{figure}[H]
 \begin{center}
 \includegraphics[width=0.9\textwidth]{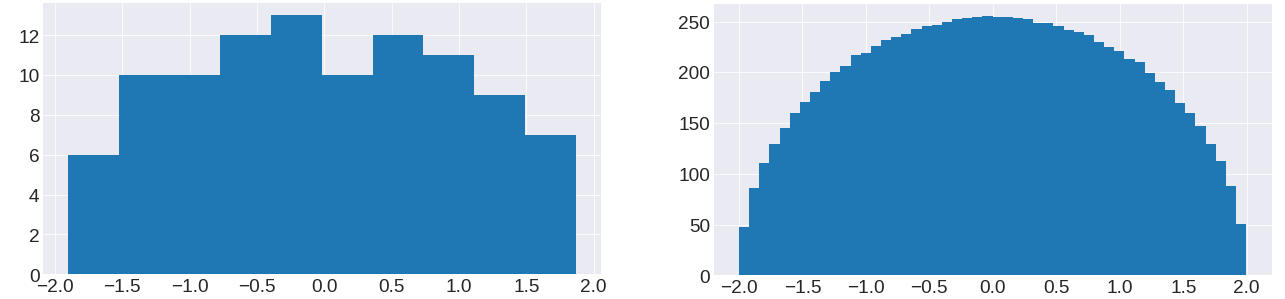}
 \caption{Histograms for the eigenvalues of simulated GUE random matrices for $N=100$ (left) and $N=10,000$ (right).}
 \end{center}
 \end{figure}

Voiculescu found that GUE random matrices are related to free independence.

\begin{theorem}[\cite{Voi91}] For every $N\in \N$,  let  $A_{N,1},...,A_{N,k}$ be independent GUE random matrices. 
Then $A_{N,1},...,A_{N,k}$ are asymptotically freely independent, i.e.\
\begin{equation}\label{injkoo}\varphi_N( (p_1(a_{N,i_1})-\varphi_N( p_1(a_{N,i_1})))\cdots  (p_m(a_{N,i_m})-\varphi_N( p_1(a_{N,i_m})))) \to 0 \end{equation}
as $N\to\infty$ for any polynomials $p_1,...,p_m$ and indices $i_1,...,i_m\in\{1,...,k\}$ with $i_j\not= i_{j+1}$.
\end{theorem}

Thus, for large $N$, the classically independent GUE matrices $A_{N,1},...,A_{N,k}$ can be treated as freely independent semicircle distributed random variables via the expectation $\frac1{N}\mathbb{E}[\tr(X)]$.\\

Let us take a look at a further result of this type. 

The set of all unitary $N\times N$-matrices $\mathcal{U}(N)=\{U\in \C^{N\times N}\,|\, UU^*=I\}$ forms a group and the normalized Haar measure is the unique Borel probability measure on $\mathcal{U}(N)$ which is invariant with respect to the group operation.
If $M\in \C^{N\times N}$ and $U$ is a Haar unitary $N\times N$-random matrix, then $UMU^*$ can be thought of as a random rotation of $M$. 

\begin{theorem}[See Section 4.3 in \cite{MS17}]\label{haar}Let $(A_N)_{N\in\N}$ and $(B_N)_{N\in\N}$ be sequences of (deterministic) $N\times N$-matrices such that $(A_N)_N$ and  $(B_N)_N$ converge in distribution with respect to $X\to \frac1{N}\tr(X)$ as $N\to \infty$. 
 Furthermore, let $(U_N)_{N\in\N}$ be a sequence of Haar unitary $N\times N$-random matrices. 
Then $A_N$ and $U_N B_N U^*_N$ are asymptotically freely independent in the sense of \ref{injkoo} as $N\to \infty$.
\end{theorem}

There are many further results on limit distributions of random matrices and on asymptotic free independence of random matrices, see the books \cite{Tao12}, \cite{MS17}. The other independences also appear in the asymptotic behavior of random matrices, see \cite{Len11, Len15}.\\

These results can be used in applications as follows: Assume we have $N\times N$-matrices $A_1,...,A_k$, and $N$ is large. If $A_1,...,A_k$ are (asymptotically) freely (or tensor, Boolean, ...) independent, and we can either calculate or model the eigenvalue distribution of these matrices, then we can calculate the eigenvalue distribution of sums and products of these matrices; without knowing the eigenvectors.

\begin{example}In \cite{ER06}, free probability theory is used to obtain a formula for the covariance matrix of a model for a random signal. Let $y\in \R^N$ be random vector modeled as 
\[y = Ax + w, \]
where $A \in \R^{N\times L}$, $x\in \R^L$ is a signal vector and $w\in R^N$ is a noise vector. If $x$ and $w$ are modeled as independent Gaussian vectors with identity covariance matrix, then the covariance matrix of $y$  is given by 
\[ R = \mathbb{E}[yy^T]  = AA^T + I.\]
Often,  the eigenvalues of $R$ are of interest in applications. However, $R$ is not known to us directly and we would rather estimate it by observing samples $y^1,...,y^n\in \R^N$ and taking, e.g., 
\[ \hat{R} = \frac1{N}\sum_{k=1}^N y_k y_k^T. \]
In our model $\hat{R}$ would converge to $R$ as $n\to\infty$ for fixed $N$. However, sometimes both $n,N\to\infty$, and then we cannot use the approximation $\hat{R}\approx R$ any longer. \\
Instead, under the assumption $N/n \to c>0$, one can calculate the eigenvalue distribution of $R$ by writing $R = \hat{R}^{1/2}(\frac1{N}GG^T) \hat{R}^{1/2}$, where $G$ is an $N\times n$-random matrix with independent entries, all $\mathcal{N}(0,1)$-distributed. \\
Under the assumption that $N/n \to c>0$, the distribution of $\frac1{N}GG^T$ converges to a Marchenko--Pastur distribution, and the matrix whose columns are the eigenvectors of $\frac1{N}GG^T$ is a Haar unitary $N\times N$-random matrix. (So $\frac1{N}GG^T\approx U_N B_N U^*_N$ as in Theorem \ref{haar}, where $B_N$ is a diagonal matrix whose eigenvalues approximate the limit distribution of $\frac1{N}GG^T$.)\\
Thus the eigenvalues of $R$ can now be calculated from the eigenvalues of $\hat{R}$ and the limit distribution of $\frac1{N}GG^T$. \hfill $\blacksquare$
\end{example}
 
\begin{example}
In \cite{PSG17}, the authors consider the product of (asymptotically freely independent) weight matrices in a neuronal net to obtain a more clever initialization of the weights before the training algorithm starts; see also \cite{LQ19}.  \hfill $\blacksquare$
\end{example}

Let us revisit our toy example. Define random matrices $(X_n)_{n\in\N_0}\subset \C^{N\times N}$, where $X_0,...,X_{p-1}$ are deterministic and
\[ X_{n} = c_1 X_{n-1} + \ldots + c_p X_{n-p} + \eps_n\] 
for $n\geq p$, where $(\eps_n)$ are classically independent GUE random matrices. If $N$ is large, we can consider $(\eps_n)$ as a free quantum noise and we obtain a model for a quantum $AR(p)$ process. 

\begin{example}
Consider an $AR(1)$ process in $\C^{N\times N}$ with 
\[X_0=I \quad \text{and} \quad X_n = 0.5 X_{n-1} + \eps_n,\] where  $(\eps_n)$ are independent GUE random matrices. 
We simulate this process for $N=500$ and $n=0,...,100=m$. With $\varphi(X)=\frac1{N}\tr(X)$, we obtain 
\[\varphi(X_{m-1})=-0.001...,\quad \varphi(\eps_m) = -0.001...,\quad \varphi(X_{m-1}\eps_mX_{m-1}\eps_m)=0.003...,\]
which is in accordance with free independence. For the variances we obtain
 $\varphi(\eps_m^2) = 0.997$ and  $\varphi(X_{m-1}^2) = 1.332...$ (Note that $\frac1{1-0.5^2}=\frac{4}{3}$, see Example \ref{778}.)\\
If we change our model to ``squared GUE noise''
\[X_0=I \quad \text{and} \quad X_n = 0.5 X_{n-1} + (\eps_n^2-\varphi(\eps_n^2)),\]
then we should still expect free independence. With $E_n=\eps_n^2-\varphi(\eps_n^2)$ we obtain 
\[\varphi(X_{m-1})=0.0000..., \quad \varphi(E_m)=0.0000...,\quad \varphi(X_{m-1}E_mX_{m-1}E_m) = 0.0001...\]
  \begin{figure}[H]
 \begin{center}
 \includegraphics[width=0.9\textwidth]{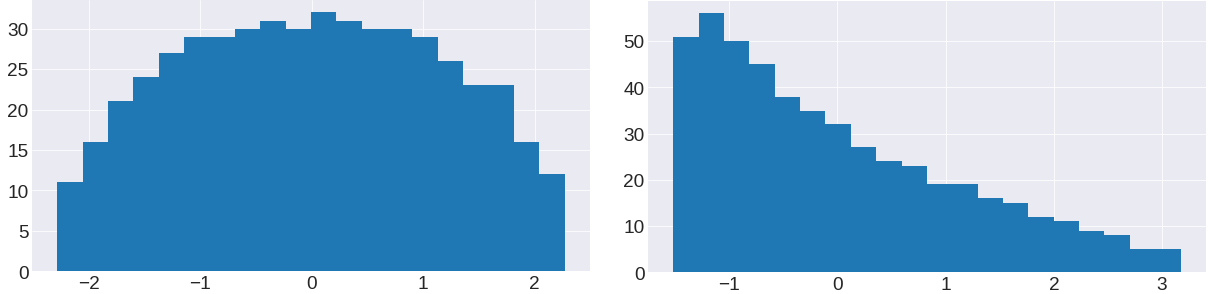}
 \caption{The eigenvalues of the last matrix $X_m$ for GUE noise (left, approx. a semicircle distribution) and squared GUE noise (right, approx. a Marchenko--Pastur distribution).}
 \end{center}
 \end{figure}
\vspace{-8mm}\hfill $\blacksquare$
\end{example}

\chapter{Reinforcement learning}

Reinforcement learning is a special case of machine learning where an agent interacts with an environment and collects rewards such that it can learn to improve its actions (to increase the expected reward). Markov processes yield a powerful model for this situation.

\section{Markov reward processes}

Let $\mathcal{S}=\{s_1,...,s_N\}$ be a finite set of states and consider a time-homogeneous Markov chain $(M_n)_{n\in\N_0}$ on $S$. The Markov chain is uniquely determined by the distribution $\mu_0$ of $M_0$ and the transition matrix $P=(p_{s,s'})$, where 
$p_{s,s'} = \mathbb{P}[M_{n+1}=s|M_n=s']$ a.s. \\

Let us look at a concrete example. We think of a basketball match and assume that team 1 consists of only two players $P_1$ and $P_2$ which attacks team 2. In order to model this offense, we define $\mathcal{S}=\{P_1, P_2, S, L\}$, where $P_j$, $j=1,2$, means that $P_j$ currently possesses the ball, $S$ means the offense ends by a score (we don't differentiate between scoring 2 or 3 points), and $L$ means that the ball is lost (because of a steal by a defensive player, an offensive foul, the ball is out-of-bounds and $P_1$ touched it last, etc.) The transition probabilities might be given by

\begin{center}
\begin{tabular}{llll}
$p_{P_1, P_1} = 0$,& $p_{P_1, P_2} = 0.1$,& $p_{P_1, S} = 0$,& $p_{P_1, L} = 0$,\\ 
$p_{P_2, P_1} = 0.6$, & $p_{P_2, P_2} = 0$,& $p_{P_2, S} = 0$,& $p_{P_2, L} = 0$,\\ 
$p_{S, P_1} = 0.2$,& $p_{S, P_2} = 0.8$,& $p_{S, S} = 1$,& $p_{S, L} = 0$,\\ 
$p_{L, P_1} = 0.2$,& $p_{L, P_2} = 0.1$,& $p_{L, S} = 0$,& $p_{L, L} = 1$. 
\end{tabular} 
\end{center}

This whole setting can be neatly summarized as a graph:

\begin{center}
\scalebox{1.2}{
 \begin{tikzpicture}[auto,node distance=8mm,>=latex,font=\small]

   \tikzset{vertex/.style = {shape=circle,draw,shade,top color=blue,fill=white,minimum size=1.5em}}
	  \tikzset{vertex1/.style = {shape=circle,draw,shade,top color=red,minimum size=1.5em}}
		\tikzset{vertex2/.style = {shape=circle,draw,shade,top color=green,minimum size=1.5em}}
\tikzset{edge/.style = {->,> = latex'}}

 \node[vertex] (a) at  (0,0) {$P_2$};
label={[red]right:B}, 
      label={[blue]135:A},
\node[vertex1] (b) at  (4,0) {$L_{}$};
\node[vertex2] (c) at  (4,4) {$S_{}$};
\node[vertex] (d) at  (0,4) {$P_1$};

 \begin{scope}[>=latex,
                        nodes={font=\footnotesize,circle,inner sep=2pt},
                        auto=right,
                        every edge/.style={draw=black,thick}]
           \draw [->, bend left] (a) edge node { $0.1$} (d);
					 \draw [->,bend right=20] (a) edge node { $0.8$} (c);
					 \draw [->] (a) edge node { $0.1$} (b);
				  	\draw [->, bend left] (d) edge node { $0.6$} (a);
				   \draw [->] (d) edge node { $0.2$} (c);
					 \draw [->,bend left=20] (d) edge node { $0.2$} (b);
					 \path[->] (b)   edge [loop right] node {$1$} ();
					 \path[->] (c)   edge [loop right] node {$1$} ();
						\end{scope}

\end{tikzpicture}
}
\end{center}

We see that player $P_1$ likely passes to $P_2$ and loses the ball with probability 0.2, while player $P_2$ scores with probability 0.8. 
Both $S$ an $L$ are terminal states, i.e.\ once we arrive in this state, we stay there, which can be seen as finishing the corresponding Markov chain.  If we start the offense with player $P_1$ having possession of the ball, we obtain a Markov chain modelling the offense. Possible instances of this chain might be:
$$ P_1 \rightarrow P_2 \rightarrow P_1 \rightarrow L \;(\rightarrow L \rightarrow L, ...),\; \text{or} $$
$$ P_1 \rightarrow S \;(\rightarrow S \rightarrow S, ...) $$

\begin{definition}A Markov reward process is a tuple 
$(\mathcal{S}, P, \gamma, R)$ consisting of a state space $\mathcal{S}=\{s_1,...,s_N\}$, a transition matrix $P=(p_{s,s'})_{s,s'\in\mathcal{S}}$  for $\mathcal{S}$, a discount factor $\gamma\in[0,1)$, and a reward function $R:\mathcal{S}\to \R$.
\end{definition}
 The reward function $R:\mathcal{S}\to \R$ assigns to each state 
a reward of having arrived in this state. 
 The rewards and the transition probabilities define the expected reward function $\mathcal{R}:\mathcal{S}\to \R$ given by 
$\mathcal{R}(s) = \sum_{s'\in \mathcal{S}} p_{s',s}R(s')$.\\

In our basketball example, we might define $R(S)=+1$, $R(L)=-1$, $R(P_1)=R(P_2) = 0$. Then a small calculation yields 
\[ \mathcal{R}(P_1) = 0, \quad \mathcal{R}(P_2) = 0.7,\quad \mathcal{R}(S)=+1,\quad \mathcal{R}(L)=-1. \]

Another way to represent $\mathcal{R}(s)$ is using a Markov chain $(M_k)_{k\in\N_0}$ under the condition $M_0=s$. We have 
$\mathcal{R}(s)=\mathbb{E}[R(M_1)]=\mathbb{E}[R(M_1)|M_0=s]$. More generally, we can now look further into the future and define 
\[\mathcal{R}(s,k):=\mathbb{E}[R(M_k)|M_0=s],\]
 which is the expected reward after $k$ time steps.\\

Finally, we can sum up all these expected rewards and use $\gamma$ to control how important future rewards are. We define the 
state value function $v:\mathcal{S}\to \R$ by

\[v(s) := \mathcal{R}(s,1) + \gamma \mathcal{R}(s,2) + ... = \sum_{k=0}^\infty \gamma^k \mathcal{R}(s,k+1). \] 

Note that our assumption $\gamma\in[0,1)$ ensures convergence of this infinite series.
However, it looks very complicated to actually compute $v(s)$ as this involves calculating infinitely many expectations.
This problem can be solved by splitting the sum after the first term:

\begin{eqnarray*}v(s)&=&\mathbb{E}[R(M_1)+\gamma R(M_2)+...|M_0=s]\\
&=&\mathbb{E}[R(M_1)+\gamma (R(M_2)+\gamma R(M_3)+...)|M_0=s]=\mathcal{R}(s,1) + \gamma\sum_{s'\in \mathcal{S}} p_{s',s} v(s').
\end{eqnarray*}

Now consider $v$ as a (row) vector $v=(v(s_1),...,v(s_N))$ and let $\mathcal{R}=(\mathcal{R}(s_1,1),...,\mathcal{R}(s_N,1))$. Then the above equation can be written as 
\begin{equation}\label{BE} v = \mathcal{R} + \gamma v P, \end{equation}
where $P$ is the transition matrix and $v P$ denotes the vector-matrix multiplication of $v$ and $P$. 
Equation \eqref{BE} is known in the field of dynamic programming as the \emph{Bellman equation}.\\

Denote by $I$ the identity matrix. Then $I-\gamma P$ is always invertible and the solution to \eqref{BE} is given by 
\[ v = \mathcal{R}(I-\gamma P)^{-1}. \]

We now come back to our basketball example and solve this equation for some fixed values of $\gamma$. If $\gamma=0$, then we simply have 
$v=R$, so 
\[\gamma = 0:\quad v(P_1) = 0, \quad v(P_2) = 0.7,\quad v(S)=1,\quad v(L)=-1. \]
Numerical solutions for two other values of $\gamma$ give the following values:
\begin{eqnarray*}&&\gamma = 0.5:\quad v(P_1) = 0.43,\quad v(P_2) = 1.42,\quad v(S)=2,\quad v(L)=-2;\\
&&\gamma = 0.9:\quad v(P_1) = 3.97,\quad v(P_2) = 7.36,\quad v(S)=10,\quad v(L)=-10. \end{eqnarray*}

\begin{remark}
The fact that we still collect rewards after having arrived in $S$ or $L$ can be avoided by adding an absorbing state $Ab$ with 
$p_{Ab, S}=1, p_{Ab, L}=1, p_{Ab,Ab}=1$ and $R(Ab)=0$.
\end{remark}

\section{Markov decision processes}

\begin{definition}A Markov decision process is a tuple 
$(\mathcal{S}, \mathcal{A}, (p_{s',(s,a)})_{s,s'\in\mathcal{S}, a\in \mathcal{A}}, \gamma, R)$ consisting of a finite state space $\mathcal{S}=\{s_1,...,s_N\}$, a finite set of actions $\mathcal{A}=\{a_1,...,a_m\}$, transition probabilities $(p_{s',(s,a)})$, a discount factor $\gamma\in[0,1)$, and a reward function $R:\mathcal{S}\to \R$.
\end{definition}
In every state $s\in \mathcal{S}$ we can choose an action $a\in \mathcal{A}$ and the transition probability $p_{s',(s,a)}$ is the probability of arriving in state $s'$ when action $a$ is applied in state $s$. We can thus define a conditional expected reward 
\[\mathcal{R}(s,a) = \sum_{s'\in\mathcal{S}} p_{s',(s,a)} R(s').\]

 In our example $\mathcal{S}=\{P_1,P_2,S, L\}$, $R(P_1)=R(P_2)=0, R(S)=+1, R(L)=-1$, we might have the actions 
\[\mathcal{A}=\{\text{pass}, \text{shoot}\}.\] 
If player $P_1$ wants to pass, he/she would like to change the state $P_1$ to $P_2$, which will often be successful, but sometimes go wrong and actually lead to $(P_1, \text{pass})\to L$, the ball is lost to the opponent. Very rarely, we might see that the player is trying to pass but actually scores, i.e.\ $(P_1, \text{pass})\to S$. We could have the following probabilities:

\begin{center}
\begin{tabular}{llll}
$p_{P_1, (P_1, \text{pass})} = 0$,& $p_{P_2, (P_1, \text{pass})} = 0.89$,& $p_{S, (P_1,\text{pass})} = 0.01$,& 
$p_{L, (P_1,\text{pass})} = 0.1$.
\end{tabular} 
\end{center}

While the Markov reward process of the previous section corresponds to modeling a match from the point of view of a spectator, the Markov decision process now put's us in the position of the coach of team 1.  We see all options a player has during a game (the action set $\mathcal{A}$) and the strengths and weaknesses of each player (the probabilities $p_{s', (s,a)}$).\\
Obviously, we are now facing the problem of giving the players a strategy, i.e.\ for each state we must find the best action. It is more convenient to allow still some randomness here, which leads to the notion of a policy.
 
\begin{definition}
A policy $\pi$ assigns to each state $s\in \mathcal{S}$ a probability distribution $\pi_s$ on the action set $\mathcal{A}$. 
We denote the probability $\pi_s(\{a\})$ by $\pi_{s,a}$.
\end{definition}

Once we have defined a policy $\pi$, the Markov decision process becomes a Markov reward process on the state space $\mathcal{S}\times \mathcal{A}$. (Translated to our basketball example, this means that once a strategy is defined, we can watch a basketball match consisting of a series of states and actions.) The probability of passing from $(s,a)$ to $(s',a')$ is given by 
\[p_{s',(s,a)}\cdot \pi_{s',a'}.\]

We now obtain an expected reward function $\mathcal{R}_\pi:\mathcal{S}\to\R$ and a state-value function $v_\pi:\mathcal{S}\to\R$ depending on $\pi$. The Markov decision problem is now to find the policy which maximizes $v_\pi$. The coach has to find the best strategy for the team.\\

Define the pointwise maximum of $v_\pi$ by $v^{*}$, i.e.
\[ v^{*}(s) = \max_{\pi} v_\pi(s). \]
This function is called the optimal value function. 

\begin{theorem}${}$
\begin{itemize}
\item[(a)]  The optimal value function satisfies the \emph{Bellman optimality equation}
\[v^*(s) = \max_{a\in \mathcal{A}} \left(\mathcal{R}(s,a) + \gamma \sum_{s'\in \mathcal{S}} p_{s', (s,a)} v^*(s')\right), \quad s\in \mathcal{S}. \]
\item[(b)] There exists an optimal policy $\pi^*$ such that 
$v_{\pi^*}(s)=v^*(s)$ for all $s\in \mathcal{S}$.
\item[(c)] There is always a deterministic policy $\pi^d$ which is optimal.
\end{itemize}
\end{theorem}
\begin{proof}
For $v,w:\mathcal{S}\to\R$ we write $v\leq w$ if $v(s)\leq w(s)$ for all $s\in \mathcal{S}$. For any function $v:\mathcal{S}\to\R$, let 
\[(B^*v)(s) = \max_{a\in \mathcal{A}} \left(\mathcal{R}(s,a) + \gamma \sum_{s'\in \mathcal{S}} p_{s', (s,a)} v(s')\right).\]

Next define the operator $Q$ which maps a function $v:\mathcal{S}\to\R$ to a deterministic policy $Q(v)=\pi$ with $\pi_{s,a}=1$ for 
\[ a = \underset{a\in \mathcal{A}}{\text{argmax}}  \left(\mathcal{R}(s,a) + \gamma \sum_{s'\in \mathcal{S}} p_{s', (s,a)} v(s')\right). \]
(As there might exist several $a_{j_1}, a_{j_2}, ...$ that maximize this expression, we choose the one with the smallest index. Then $Q$  is well-defined.)\\ 

Finally, for a policy $\pi$, let \[B_\pi(v)=\mathcal{R}_\pi + \gamma  v P_\pi.\]
The definitions immediately imply $B_{Q(v)}=B^*(v)$ for any $v:\mathcal{S}\to\R$.\\

Let $N$ be the cardinality of $\mathcal{S}$. Then both $B_\pi$, $\pi$ fixed, and $B^*$ are mappings from $\R^N$ into itself. 
We equip $\R^N$ with the maximum norm $\|\cdot\|_{max}$. Then both mappings are $\gamma$-contractions, i.e.\ 
\[\|B_\pi(v)-B_\pi(w)\|_{max}\leq \gamma \|v-w\|_{max}, \quad \text{and} \quad \|B^*(v)-B^*(w)\|_{max}\leq \gamma \|v-w\|_{max}, \] 
and Banach's fixed point theorem implies that each mapping has a unique fixed point, namely
\[B_\pi(v_\pi)=v_\pi \quad \text{(Bellman equation)}\, \quad \text{and}\quad  B^*(v_{opt})=v_{opt}\]
for some function $v_{opt}:\mathcal{S}\to\R$.\\

Now let $\pi_0$ be any policy. We define a sequence $(\pi_k)_{k\in\N}$ 
of policies by $\pi_k = Q(v_{\pi_{k-1}})$.  Then we have
\begin{equation}\label{nui0}v_{\pi_{k-1}} = B_{\pi_{k-1}}(v_{\pi_{k-1}}) \leq B^*(v_{\pi_{k-1}}) = B_{Q(v_{\pi_{k-1}})}(v_{\pi_{k-1}}) = 
 B_{\pi_k}(v_{\pi_{k-1}}),\end{equation}
and thus  $0\leq B_{\pi_k}(B_{\pi_k}(v_{\pi_{k-1}})-v_{\pi_{k-1}})=B^2_{\pi_k}(v_{\pi_{k-1}})-B_{\pi_k}(v_{\pi_{k-1}})$. Hence, 
 $ v_{\pi_{k-1}} \leq B^n_{\pi_k}(v_{\pi_{k-1}})$ for all $n\in\N$. The limit $n\to\infty$ yields the fixed point of $B_{\pi_k}$ on the right side, which is $v_{\pi_k}$, thus
\[ v_{\pi_{k-1}} \leq  v_{\pi_k},\]
i.e.\ the sequence of policies improves the value function. As $v_{\pi_k}(s)$ is monotonically increasing and bounded, 
$v_{\pi_k}(s)\to v_L(s)$ for a limit function $v_L$. Let $\pi_L=Q(v_L)$. Then we have equality in \eqref{nui0} and thus 
$v_{\pi_L} = B_{\pi_{L}}(v_{\pi_{L}}) = B^*(v_{\pi_{L}})$, which implies $v_{\pi_L} = v_{opt}$. \\

Now we prove (a): For any value function $v_\pi$ we have $v_\pi = B_\pi(v_\pi)\leq B^*(v_\pi)$. 
 
Fix $\eps>0$. Then we find $\pi_{0,\eps}$ such that $v^*(s)-v_{\pi_{0,\eps}}(s) \leq \eps$ for all $s\in \mathcal{S}$.
Consider the sequence $\pi_{k,\eps}$ of policies starting with $\pi_{0,\eps}$. 
This sequence is monotonically increasing and we conclude $0 \leq v^*(s)-v_L(s) \leq \eps$ for all $s\in \mathcal{S}$. 
But as $\eps$ can be chosen arbitrarily small, we have $v^* = v_L$. Hence, we have proven (a) and (b).\\

A deterministic policy that maximizes $v_\pi$ is given by $\pi=Q(v^*)$. 
\end{proof}

\begin{remark}
In applications, an optimal policy can be calculated by various numerical methods. For example, we can construct a neuronal net which takes a state $s$ as input and gives a probability distribution on $A$ as output. We initialize the neuronal net with random weights and thus obtain a policy $\pi_0$.\\
We can now interact with the environment (we play many basketball games) and collect returns for $(s,a)$ pairs, i.e.\ we start in some state $s_0$, choose an action $a_0$, we land in state $s_1$ with return $R(s_1)$, we choose an action $a_1$, etc. In this way, we obtain a collection 
of data: 
\[(s_0,a_0), (s_1,a_1, R(s_1)), (s_1,a_1, R(s_2)), \ldots\]
Here, we choose our action in state $s$ randomly with probabilities given by $\pi_0(s)$.\\

The critical step is now: how should we update the weights of the neuronal net (the training step) such that we finally come close to the maximum of the value function after several iterations of interacting with the environment and training? Usually, one defines a loss function for the ``true outcomes'' (labels) and applies the  gradient descent method.\\

Here, the Policy Gradient Theorem helps, see \cite[13.2]{SB18}.
It expresses the gradient of the value function $v(s)$ with respect to the weights such that ``gradient descent for $v(s)$'' is the same as:
``Pretend that the sampled actions $a_1,a_2,...$ are the labels of $s_1,s_2,...$ and use the cross-entropy loss function with weight of $(s_k,a_k)$ eqal to $\sum_{n\geq 0} \gamma^n R(s_{k+n})$.'' See also \cite{Bah19}.
\end{remark}

\section{Quantum Markov decision processes}

Machine learning can be coupled with  quantum mechanics to obtain quantum versions of mathematical models and algorithms. These are, in particular, relevant for quantum computing. We refer to \cite{DB18} for an overview on quantum machine learning.\\ 

Recall that a quantum channel $\mathcal{T}:\C^{N\times N}\to\C^{N\times N}$ is a linear mapping of the form $\mathcal{T}(X)=\sum_{j=1}^M E_j X E_j^*$ for matrices $E_1,...,E_n$ with $\sum_j E_j^* E_j=I$, and that $\mathcal{T}(\rho)$ is a density matrix whenever $\rho$ is a density matrix. \\

As in \cite{BBA14}, a \emph{quantum  observable  Markov decision  process} on $\C^{N\times N}$ can be defined as the following collection:
\begin{itemize}
	\item  a self-adjoint $S\in\C^{N\times N}$ (the state space),
	\item  a density matrix $\rho_0\in \C^{N\times N}$ (initial state),
	\item a set $\mathcal{A}_1, ..., \mathcal{A}_{n_a}$ of quantum channels on $\C^{N\times N}$ (the set of actions),
 \item a set $R_1, ..., R_{n_a}\subset \C^{N\times N}$ of self-adjoint matrices (reward operators),
\item and a discount factor $\gamma\in[0,1)$.
\end{itemize}

Now an agent chooses an action $\mathcal{A}_j$ and then makes a random observation $o_1\in\R$ described by the quantum random variable $S$ with respect to the state $X\mapsto \tr(X\mathcal{T}_j(\rho_0))$. The outcome $o_1$ is an eigenvalue of $S$. The expected reward 
``$\mathcal{R}(s,a)$'' from the classical case is now replaced by the expectation of the reward operator $R_j$, i.e.\ $\tr(R_j \mathcal{T}_j(\rho_0))$.\\

One can now define policies and state-value functions and let an agent interact with this quantum environment.

\chapter{A Markovian look at the Ising model}

The Ising model is one of the most important models in statistical physics. On the one hand, it is easy to define, and on the other hand, it is already complicated enough to be able to model the phase transition of a ferromagnetic material occurring at its Curie temperature. We will see how Markov processes enter the study of the Ising model at two completely different points.

\section{The two-dimensional Ising model}\label{Ising_chap}

Phase transitions are state changes of materials due to the variation of external conditions like the temperature, pressure, magnetic field, etc. The most common example is liquid water, which becomes solid at $0^\circ\text{C}$ (and standard pressure) and boils at $100^\circ\text{C}$.\\
Mathematically, phase transitions can be described by discontinuities of certain macroscopic quantities (or their derivatives). 
When water freezes, its volume increases discontinuously with respect to the temperature. Another example:\\
At a temperature $T<T_C=768^\circ\text{C}$, iron is ferromagnetic. When we place a piece of iron in a magnetic field and then draw it out, it will have 
a magnet field with the same direction, due to the parallelization of the magnetic moments of the iron atoms. Above the critical temperature $T_C$, also called Curie temperature, this property suddenly disappears and there is no magnet field left. Here, iron is paramagnetic.

\begin{remark}For $T<T_C$, the magnet attracts the piece of iron. This is also true for $T>T_C$, but the attraction is much weaker. 
There are also diamagnetic materials, which are repelled by the magnetic field. 
\end{remark}

In order to explain the phase transition ferromagnetic $\longrightarrow$ paramagnetic mathematically, Ernst Ising studied a statistical model proposed by 
 Wilhelm Lenz in his PhD thesis (1924). 

\begin{remark} Ising considered the model in one dimension and showed that there is no phase transition, i.e.\ there is no ferromagnetism in this model. 
He conjectured that the same should be true also in higher dimensions, which (fortunately) turned out to be wrong. 
First, however, physicists investigated other, more complicated models (e.g.\ the Heisenberg model). In 1936,  Rudolf Peierls showed that 
the Ising model has in fact a phase transition in two dimensions. 
Afterwards, the Ising model has been studied intensively. In fact, Ising had not noticed this for a long time (see \cite{Kobe} for his biography).\\
Because of its simplicity and the ability to model phase transitions, the Ising model is regarded as one of the most important models in statistical physics.
\end{remark}

Let $D\subset \C$ be a Jordan domain, representing a two-dimensional piece of iron.\\
First, we discretize the domain $D$. For $\delta>0$, we define the lattice $$\C_{\delta}=\{\delta k+i \delta l\,|\, k,l\in\Z\}.$$ 
Here, we consider $\C_\delta$ as an undirected graph, where $\C_\delta$ represents the set of all vertices, and $x,y\in \C_\delta$ are connected by an edge, written as $x\sim y,$ if and only if $|x-y| = \delta.$\\ 
A vertex in $\C_\delta$ has $4$ neighbors.
Furthermore, we let $D_\delta = D\cap \C_\delta$. (More precisely, we consider the subgraph of $\C_\delta$ induced by $D\cap \C_\delta$, i.e.\ $x,y\in D_\delta$ are connected
within this subgraph if and only if they are connected in $\C_\delta$.) This graph might not be connected anymore. So we define
$\Omega_\delta$ as the largest connected subgraph of $D_\delta$ (uniquely determined for $\delta$ small enough).\\

In the Ising model, every vertex $x\in\Omega_\delta$ represents an atom carrying a spin $\sigma_x \in \{-1,1\}$ and 
so we define a configuration as a function $\sigma:\Omega_\delta\to \{-1,1\}.$ Let $\Sigma_\delta$ be the set of all configurations for $\Omega_\delta$. A configuration $\sigma$ 
is now generated randomly as follows: \\
First, we define the energy $H_{\delta,B}(\sigma)$ for $\sigma\in \Sigma_\delta$ by 
\begin{equation}\label{ising}
H_{\delta,B}(\sigma) = -J\sum_{x\sim y} \sigma_x\sigma_y - B \sum_{x} \sigma_x, \quad J>0, B\geq 0,
\end{equation}
where we sum over all edges of $\Omega_\delta$ in the first sum. The value $B$ stands for an external magnetic field and $J$ describes the coupling of connected vertices.\\
We let $T>0$ be the absolute temperature and we define the partition function $Z_{\delta,T,B}$ by
$$Z_{\delta,T,B} = \sum_{\sigma \in \Sigma_\delta} e^{-\beta H_{\delta,B}(\sigma)},$$
where $\beta = \frac1{k_B T}$ and $k_B$ is the Boltzmann constant.\\ 
Now we define the probability for $\sigma \in \Sigma_\delta$ by 
$$\mathbb{P}_{\delta,T,B}(\{\sigma\}) = \frac1{Z_{\delta,T,B}}e^{-\beta H_{\delta,B}(\sigma)}.$$

The most likely configurations are those having the smallest energy $H_{\delta,B}(\sigma),$ i.e.\ all spins have the same direction (and are equal to $+1$ provided that $B>0$). 
This tendency to order is countered by the thermal energy, which is represented here only by the variable $T$.\\  
If we let $T\to\infty$, the distribution of the configurations converges to a uniform distribution, i.e.\ each configuration then has the same probability.\\

By looking at the behavior of $Z_{\delta,T,B}$ as $\delta\to 0,$ one can derive macroscopic quantities describing the phase transition. 
The magnetization per spin $M(\delta,T,B)$ is defined by 
$$M(\delta,T,B) = \mathbb{E}\left[\frac1{N}\sum_{x\in \Omega_\delta} \sigma_x\right] = 
\sum_{\sigma \in \Sigma_\delta}  \left(\frac1{N}\sum_{x\in \Omega_\delta} \sigma_x\right)\cdot  \mathbb{P}_{\delta,T,B}(\{\sigma\}), $$
where $N$ is the number of vertices in $\Omega_\delta$. If we look at this quantity only for $B>0$ and define $M(T,B)=\lim_{\delta \to 0} M(\delta,T,B),$ then
$$M_{0^+}(T):=\lim_{B\downarrow 0}M(T,B) = \begin{cases} \left(1-\sinh ^{-4}(2J/(k_B T))\right)^{1/8}  &T<T_C,\\ 0, & T \geq T_C,\end{cases}$$
where the critical temperature $T_C$ is given by
$$  T_C = \frac{2J}{k_B\log(1+\sqrt{2})} \approx \text{2.2692} \frac{J}{k_B};$$
see \cite[p.118]{Bax89}. The function $M_{0^+}(T)$ shows that there remains a rest magnetization for $T<T_C$ when we remove the piece of iron from the magnetic field $B$. It is continuous at $T=T_C$, but not differentiable. (In physics language: The derivative is not continuous.)

\begin{figure}[H]
\rule{0pt}{0pt}
\centering
\includegraphics[width=8cm]{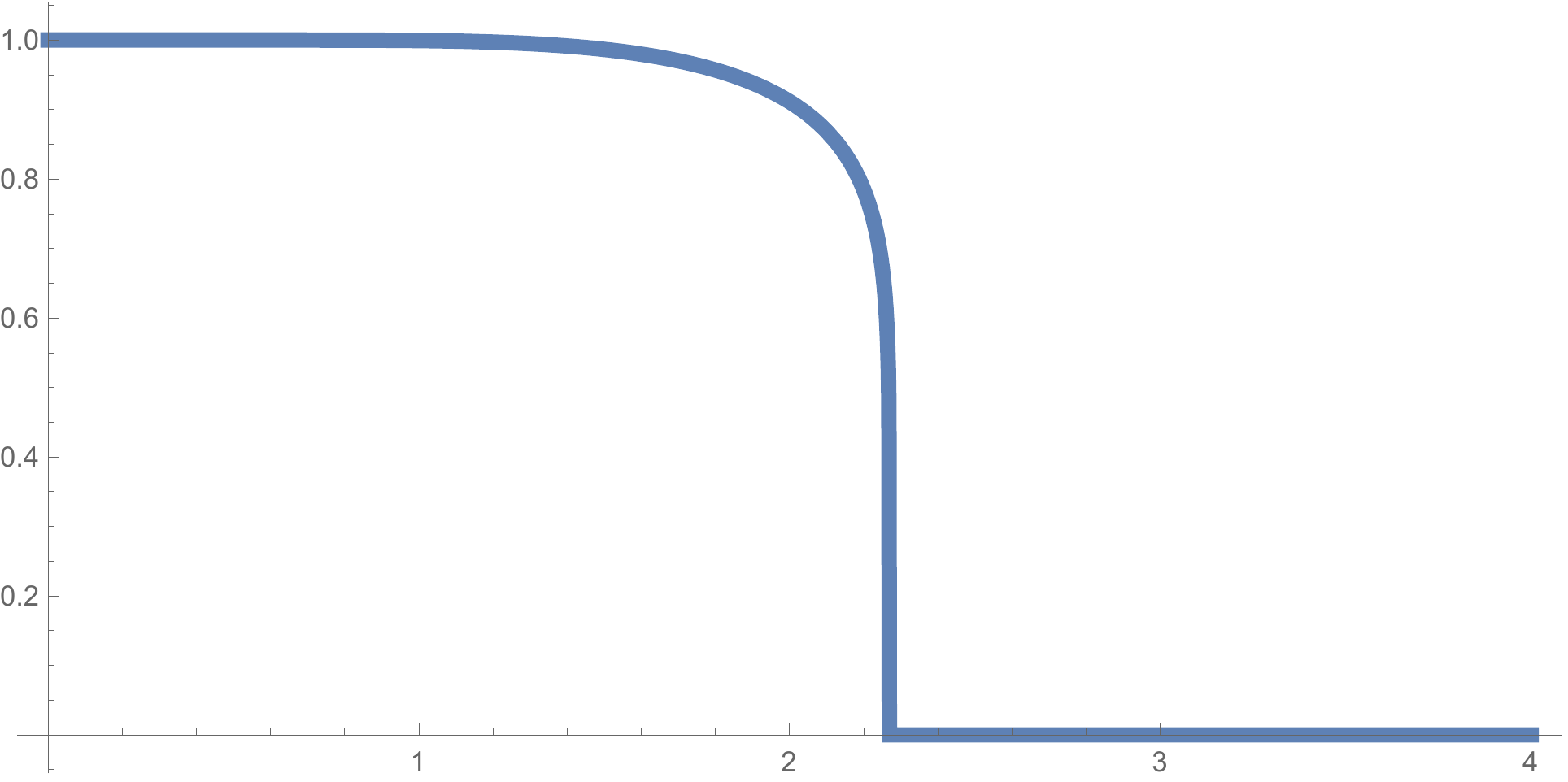}
\caption{$M_{0^+}(T)$ for $J=k_B=1,$ $T_C=\frac{2}{\ln(1+\sqrt{2})} \approx \text{2.2692}.$}
\end{figure}

The following figure shows three random configurations of the model \eqref{ising} for the values $J=1, B=0,$ and $\beta=$ 0.2 (paramagnetic), 
$\beta=\log(1+\sqrt{2})/2\approx$ 0.4407 (critical) and $\beta=100$ (ferromagnetic) on a square lattice with $200 \cdot 200$ vertices. 
The method used to simulate theses configurations is explained in the next section.

\begin{figure}[H]
\centering
\includegraphics[width=12cm]{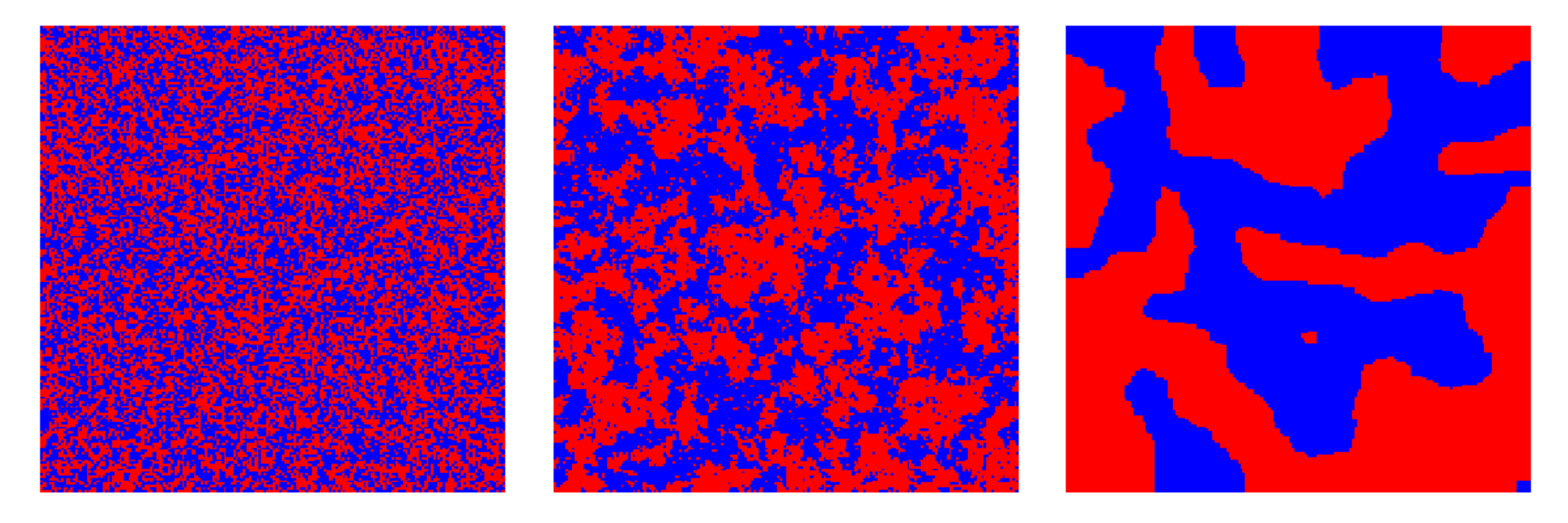}\label{betas}
\caption{Ising model for $\beta=$ 0.2 (left); $1/2\log(\sqrt{2}+1)$ (middle); 100 (right).}
\end{figure}

We can also impose some boundary restrictions on the configurations. The following configurations are based on the same values as before, but now the spins on the left and right half of the boundary are kept constant. 

\begin{figure}[H]
\centering
\includegraphics[width=12cm]{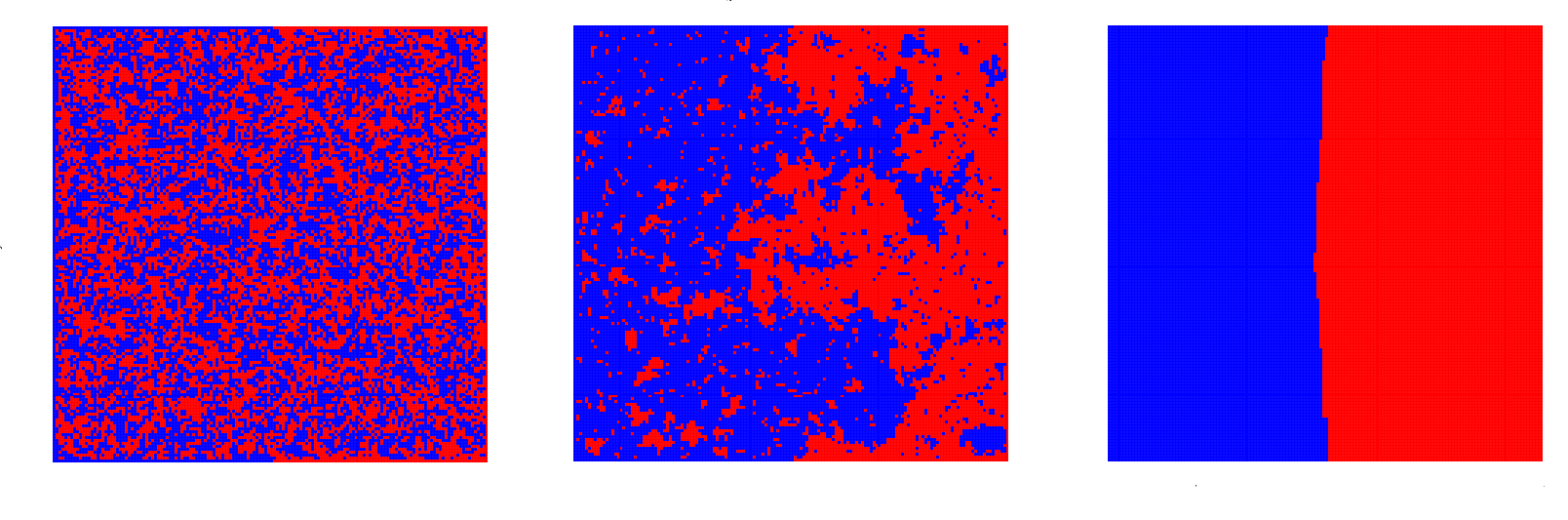}\label{betas2}
\caption{The simulations with boundary conditions.}
\end{figure}

In the case of a high temperature, the sum of all spins in some subregion is close to $0$, in contrast to the case of a low temperature. 
The critical Ising model, i.e.\ $B=0$ and $T=T_C$, shows interesting self-similar patterns. In fact, this model is conformally invariant.
Let $D,E\subset \C$ be Jordan domains and let $f:D\to E$ be conformal. If we let $\delta\to 0$ for the critical Ising model on the discretization of $D$ and then map the model to $E$ via $f$, we obtain the same as taking the limit $\delta\to 0$ for the critical Ising model on the discretization of $E$.

\begin{remark}We take a look at a concrete example.\\  
Let $B=0$, $J=1$, and $T=T_C$. Choose $a\in D$ and let $x_\delta, y_\delta \in \Omega_\delta$ be connected by a horizontal edge lying closest to $a$. Then  
$$\lim_{\delta\to 0} \mathbb{E}[\sigma_{x_\delta} \sigma_{y_\delta}] = \frac1{\sqrt{2}}.$$

Hence, the quantity $\mathbb{E}[\sigma_{x_\delta} \sigma_{y_\delta}-\frac1{\sqrt{2}}]$ converges to $0$. 
By scaling with $\delta^{-1},$ one basically obtains the hyperbolic metric $\lambda_D$ of $D$ as $\delta \to 0$, \cite[Theorem 1]{HoSm}: 
$$ \mathbb{E}[\sigma_{x_\delta} \sigma_{y_\delta}] = \frac1{\sqrt{2}}-\frac{\lambda_D(a)}{2\pi}\delta + \Landauo(\delta) \qquad \text{as} \quad \delta\to0.$$
\end{remark}

\section{Metropolis algorithm}\label{metro}

A distribution on a discrete sample space $\Omega=\{\omega_1,...,\omega_N\}$ is usually simulated as follows. 
Construct a partition of the interval $[0,1]$ into intervals $I_1,...,I_N$ where the length of $I_k$ is equal to $\Pro(\{\omega_k\})$. 
Now we simulate a random number $p$ uniformly distributed on $[0,1]$ and we choose the sample $\omega_k$ according to the interval $I_k$ satisfying $p\in I_k$.\\ 

In Figure 12.2, we have $200\cdot 200 = 40,000$ vertices and thus there are $2^{40,000}$ possible configurations, each with a positive probability. Thus the usual simulation method does not work anymore. In fact we cannot even calculate the partition function $Z_{\delta,T,B}$.
  The \emph{Metropolis algorithm}, introduced in \cite{MRRTT53} and extended by Hastings in \cite{Has70}, provides a different method to simulate the distribution on $\Sigma_\delta$.\\   

First, choose an initial configuration $\sigma_0\in\Sigma_\delta$. For $n\in\N$ we now generate $\sigma_{n}$ from $\sigma_{n-1}$ as follows.

\begin{itemize}
	\item Step 1: Let $x\in  \Omega_\delta$ be a vertex drawn randomly with a uniform distribution on the set of all vertices. Let $\sigma'$ be the configuration that we obtain from $\sigma_{n-1}$ by flipping the spin at $x$. Let $\Delta H=H_{\delta,B}(\sigma')-H_{\delta,B}(\sigma)$ be the energy difference between the two configurations. 
	\item Step 2: If $\Delta H\leq 0$, then $\sigma_{n} = \sigma'$. If $\Delta H> 0$, put $p=\exp(-\beta\Delta H)$ and let $\sigma_{n} = \sigma'$ with probability $p$ and $\sigma_n = \sigma_{n-1}$ with probability $1-p$.
\end{itemize}

This algorithm defines a stochastic process $(\sigma_n)_{n\in \N_0}$ on $\Sigma_\delta$ and it is clear that it has the Markov property. 
It is easy to see that it is irreducible and aperiodic. So we would like to apply Theorem \ref{limit_Markov} to show that the limit distribution on $\Sigma_\delta$ coincides with the distribution from the Ising model.
 To this end, we only need to show that the distribution $\mathbb{P}_{\delta,T,B}$ is stationary for the Markov process $(\sigma_n)$.\\

Let $P_{\sigma', \sigma} = \mathbb{P}[\sigma_{n+1} = \sigma'| \sigma_n = \sigma]$ be the transition probabilities for the process.
Furthermore, consider the modified Markov process on $\Sigma_\delta$, where we pick a vertex $x$ as in step 1, but then flip the spin at $x$ with probability $1$, and let $Q_{\sigma', \sigma}$ be the transition probabilities for this modified process.	We clearly have $Q_{\sigma', \sigma}=Q_{\sigma, \sigma'}$ for all $\sigma, \sigma' \in \Sigma_\delta$.\\

For $\sigma'\not=\sigma$,  we have
\[P_{\sigma', \sigma} = Q_{\sigma', \sigma} \cdot 
 \min\left(1, \exp \left(\beta H_{\delta,B}(\sigma)-\beta H_{\delta,B}(\sigma') \right)\right)\]

and consequently
\begin{eqnarray*}
&&
P_{\sigma', \sigma} \cdot \mathbb{P}_{\delta,T,B}(\{\sigma\}) = 
Q_{\sigma', \sigma} \cdot 
 \min\left(1, \exp(\beta H_{\delta,B}(\sigma)-\beta H_{\delta,B}(\sigma'))\right) \cdot \frac{\exp(-\beta H_{\delta,B}(\sigma))}{Z_{\delta,T,B}} \\
&=& Q_{\sigma, \sigma'} \cdot 
 \min\left(\exp(-\beta H_{\delta,B}(\sigma)), \exp(-\beta H_{\delta,B}(\sigma'))\right) \cdot \frac{1}{Z_{\delta,T,B}} \\
&=& Q_{\sigma, \sigma'} \cdot 
 \min\left(\exp(\beta H_{\delta,B}(\sigma')-\beta H_{\delta,B}(\sigma)), 1 \right) \cdot \frac{\exp(-\beta H_{\delta,B}(\sigma'))}{Z_{\delta,T,B}} \\
&=& P_{\sigma, \sigma'} \cdot \mathbb{P}_{\delta,T,B}(\{\sigma'\}).
\end{eqnarray*}

Thus $\mathbb{P}_{\delta,T,B}$ is a stationary distribution of the Markov process due to Exercise \ref{balanced}.\\

The Metropolis algorithm is a beautiful example of the power of Markov processes in applications. It can also be stated for more general settings and in turn it is a special case of \emph{Markov Chain Monte Carlo Methods}, see \cite{BGJM11}.

\section{Schramm-Loewner evolution}\label{schramm}

Consider the Ising model at the critical temperature, the critical Ising model, which is known to be conformally invariant. 
  The interface curves, i.e.\ the random curves that separate +1 from -1 spin clusters, seem to have a fractal-like shape. 

\begin{figure}[h]
\centering
\includegraphics[width=5cm]{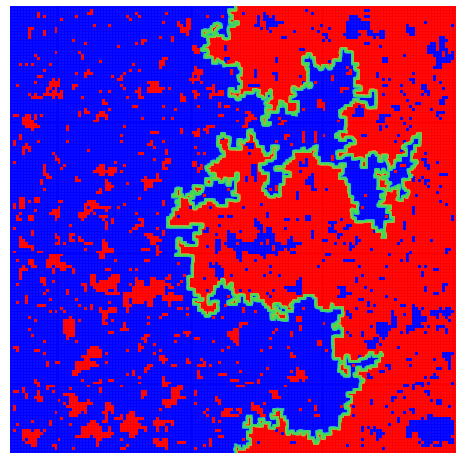}
\caption{An interface curve in the critical Ising model.}
\end{figure}

How can we investigate properties of these random curves? If we think of all the steps needed to define this curve, a calculation of its distribution seems to be beyond hope.\\

There are further conformally invariant models from statistical mechanics and stochastic geometry that generate random curves of a similar type.
Let us look at another example. Consider a random walk on the lattice $\{a+bi\,|\, a,b\in\Z\}$ which starts at $z=0$ and at each step it goes to one of the $4$ neighboring points, to each with probability $1/4$.  If we stop such a random walk after $N$ steps, we can define the corresponding \emph{loop erased random walk} (LERW) as the simple curve obtained by removing all loops in chronological order.

\begin{figure}[H]
\centering
\includegraphics[width=12cm]{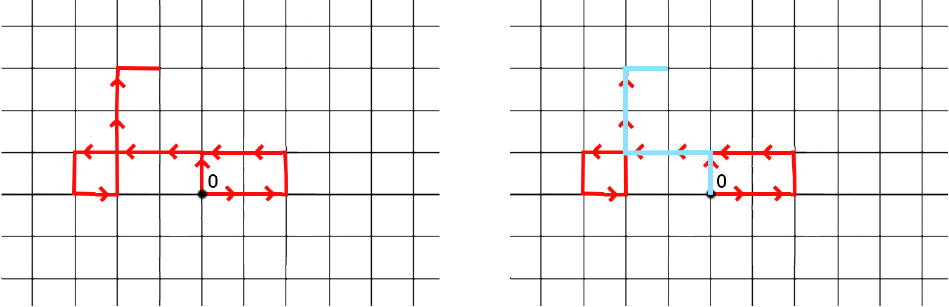}
\caption{Random walk (left) and the corresponding loop erased random walk (right).}
\end{figure}

For every $N\in\N$, let $S_N:[0,1]\to\C$ be the linear interpolation of a random walk with $N$ steps. Then $S_N/\sqrt{N}$ converges to a two-dimensional Brownian motion 
as $N\to \infty$ (Donsker's theorem).\footnote{If we equip $C([0,1],\C)$ with the topology induced by the $\sup$-norm, then 
 $S_N/\sqrt{N}$ converges in distribution with respect to this topology to a two-dimensional Brownian motion $B:[0,1]\to \C.$} 
Provided the limit exists, how can we describe the distribution of the limit of the loop erased random walk? In contrast to the random walk $(M_n)_n$, its loop erased subcurve $(M_{n_k})_k$ is far away from being Markov. If we know the first $N$ points, then $(M_{n_k})_{k>N}$ must be disjoint from $\{M_{n_1},...,M_{n_N}\}$.

\begin{figure}[h]
\centering
\includegraphics[width=12cm]{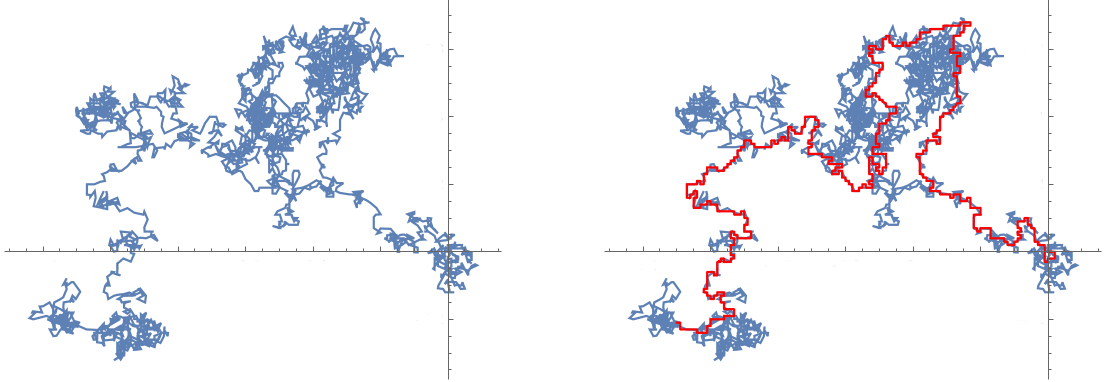}
\caption{A simulation of a random walk with $N=10,000$ steps and the corresponding loop erased random walk.}
\end{figure}

O. Schramm had a beautiful idea to solve such problems, published 2000 in \cite{MR1776084}. The highly non-Markovian curves in the complex plane can be encoded by the real-valued driving functions of Loewner's (slit-)differential equation, and it turns out that these functions become Markov processes:

\[ \text{(non-Markov) curves from the Ising model, LERW, etc.}\]
 \[  \underset{\text{Loewner equation}}{\Longrightarrow}\quad  \text{Markov process on $[0,\infty)$.} \]
For a simply connected domain $D\subsetneq \C$, we denote by $\partial_\infty D$ its boundary with respect to $\hat{\C}$.

Consider a triple $(D,x,y)$ where $D\subsetneq \C$ is a simply connected domain, and $x,y\in \partial_\infty D$ are two different points such that 
$\partial_\infty D$ is locally connected in neighborhoods of $x$ and $y$. If $D$ is a Jordan domain, we can choose any two $x,y$ from the locally connected boundary. It will soon become clear that we need this more general setting.\\

If $(D',x',y')$ is another triple of this kind, then the Riemann mapping theorem provides us a conformal mapping 
$f:D\to D'$ and we can choose $f$ such that it extends continuously to $x$ and $y$ with $f(x)=x'$ and $f(y)=y'$. For such a mapping, we simply say that $f:(D,x,y)\to (D',x',y')$ is conformal.

\begin{lemma}For two triples there exists a conformal mapping $f:(D,x,y)\to (D',x',y')$.
\end{lemma}
\begin{proof}Let $g:D\to \D$ be conformal. Due to Theorem \ref{loc_con} (applied locally), $g$ extends continuously to $x$ and $y$ with $g(x)=\alpha$, $g(x)=\beta$, and $\alpha, \beta$ are two points on $\partial \D$. In the same way we obtain a conformal mapping $g':D'\to\D$ with $g'(x')=\alpha'$, $g'(y')=\beta'$. Now there exists an automorphism $h:\D\to\D$ with $h(\alpha)=\alpha'$ and $h(\beta)=\beta'$, see Exercise \ref{boundary_auto}. Finally, $f = g'^{-1}\circ h\circ g$ satisfies the required conditions.
\end{proof}

Let $S(D,x,y)$ be the set of all $\gamma([0,1])$, where $\gamma:[0,1]\to D\cup \partial_\infty D$ is injective and continuous with $\gamma(0,1)\subset D$, $\gamma(0)=x$, and $\gamma(1)=y$. We equip $S(D,x,y)$ with some metric and the corresponding Borel $\sigma$-algebra.\\ 

We would like to find a probability measure $\mu_{D,x,y}$ on $S(D,x,y)$ having some special properties. The first property:
\begin{itemize}
	\item[$(i)$] Conformal invariance: if $f:(D,x,y)\to (D',x',y')$ is conformal, then the pullback $f^* \mu_{D',x',y'}$ of $\mu_{D',x',y'}$ with respect to $f$ is equal to $\mu_{D,x,y}$.
\end{itemize}

Due to $(i)$, it is sufficient to consider only the case $D=\Ha$, $x=0$, $y=\infty$. The automorphisms of $\Ha$ fixing $0$ and $\infty$ are the linear mappings $z\mapsto cz$, $c>0$. So $(i)$ implies further that $\gamma$ and $c\gamma$ have the same distribution. We see that $(i)$ corresponds to a scale-invariant random simple curve in $\Ha$ from $0$ to $\infty$.\\

The second property, the domain Markov property, is usually stated as follows. For every $t\in(0,1)$, the conditional distribution of $\gamma([t,1])$ given $\gamma([0,t])$ is a.s.\ equal to $\mu_{D\setminus \gamma([0,t]), \gamma(t),y}$. Note that $D\setminus \gamma([0,t])$ is not a Jordan domain, which explains why we consider more general domains. (Some conditions on the metric are needed such that we can induce Borel probability measures on the subcurves $\gamma([0,t])$, $\gamma([t,1])$.)

 \begin{figure}[H]
\rule{0pt}{0pt}
\centering
\includegraphics[width=7.5cm]{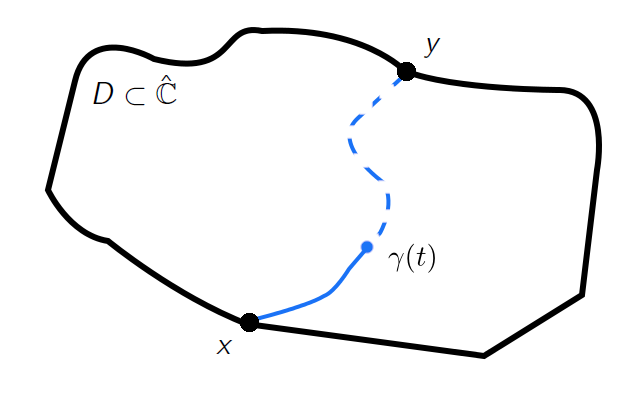}
\caption{The curve $\gamma([0,t])$ produces the triple $(D\setminus \gamma([0,t]), \gamma(t),y)$.}
\end{figure}

We aim for a slightly different definition. Consider again the case $D=\Ha$, $x=0$, $y=\infty$ and reparametrize $\gamma$ such that the conformal mapping $f_t:\Ha\to \Ha\setminus  \gamma([0,t])$ has hydrodynamic normalization, i.e.\ $f_t(z)=z-\frac{t}{z}+...$ at $\infty$. Due to Theorem \ref{slit_Loewner_eq_thm}, $f_t$ satisfies 

\begin{equation}\label{nhgggggg}
\frac{\partial}{\partial t} f_{t}(z) = \frac{\partial}{\partial z}f_{t}(z)\cdot \frac{1}{U(t)-z} \quad \text{for all $t\geq 0$, $f_{0}(z)=z\in \Ha$,}
\end{equation}
where $U:[0,\infty)\to\R$ is continuous with $U(0)=x=0$.

\begin{itemize}
	\item[$(ii)$] Assume that the metric on $S(D,x,y)$ is chosen such that, for each $t\geq 0$, the mapping $S(D,x,y)\ni \Gamma \mapsto U(t)$ 
	is continuous (e.g.\ the Hausdorff metric as in \cite{MR1776084}).
\end{itemize}

Under this assumption, the probability measure $\mu_{D,x,y}$ induces a stochastic process $(U(t))_{t\geq0}$ 
on $(\R, \mathcal{B}(\R))$.  The conformal mapping $f_t^{-1}-U(t)$ maps $\gamma[t,\infty]$ onto a curve $\hat{\gamma}[t,\infty]$ from $0$ to $\infty$ within $\Ha$. The driving function of this curve is simply given by $s\mapsto U(t+s)-U(t)$. We require:

\begin{itemize}
	\item[$(iii)$] Domain Markov property: For any $0\leq s,t$, $\sigma(U(t+s)-U(t))$ and $\sigma(\{U(\tau)\,|\, \tau\leq t\})$ are independent and 
	 the distribution of $U(t+s)-U(t)$ only depends $s$.
\end{itemize}

\begin{theorem}Under the assumptions (i)-(iii), there exists $\kappa\geq0$ such that $U(t) = \sqrt{\kappa/2}B_t$, where $B_t$ is a Brownian motion.
\end{theorem}
\begin{proof}The domain Markov property implies that $U$ has independent and stationary increments. As $U$ is continuous a.s., Theorem \ref{cont_levy} implies that $U$ has the form $U(t)=at+b B_t$ with $a\in\R, b\geq 0$, and a Brownian motion $B_t$. \\
We apply conformal invariance once more for the automorphisms $z\mapsto cz$ of $\Ha$. Due to Exercise \ref{scaling_chordal}, $U(t)$ has the same distribution as $cU(t/c^2)$ for every $c>0$. Hence we have $a=0$ and we can put $b=\sqrt{\kappa/2}$.
\end{proof}

Conversely, if we fix $\kappa\geq 0$ and solve \eqref{nhgggggg} for $U_t = \sqrt{\kappa/2}B_t$, do we obtain a probability measure $\mu_{D,x,y}$ satisfying $(i)$-$(iii)$?\\
At least we obtain a random Loewner chain $(f_t)_{t\geq 0}$, which corresponds to the growth of some random sets $K_t$ via $f_t(\Ha)=\Ha\setminus K_t$ in $\Ha$. This evolution is called \emph{Schramm-Loewner evolution SLE$(\kappa)$}. If $\kappa \in [0,4]$, then, for any $t>0$, $K_t$ is indeed the image of a simple curve almost surely. If $\kappa\in (4,8)$, then, a.s., 
 $K_t$ is the image of a curve that touches itself plus the compact components of the complement. For $\kappa\geq 8$, $K_t$ is the image of a space-filling curve a.s. We refer the reader to \cite{lawler05}.\\
It follows that all measures $\mu_{D,x,y}$ satisfying $(i)$-$(iii)$ are parametrized by the parameter $\kappa \in [0,4]$ (modulo different choices of the metric).

\begin{remark}SLE as described here is also called \emph{chordal SLE}. 
We can also write $U_t = \sqrt{\kappa/2}B_t = \hat{B}_{\kappa/2 \cdot t}$ for another standard Brownian motion $\hat{B}$.
The factor $\frac1{2}$ is due to a slightly different convention in the literature. The Schramm Loewner evolution is usually described via the Loewner equation   
\[\text{$\frac{\partial}{\partial t} f_{t}(z) = \frac{\partial}{\partial z}f_{t}(z)\cdot \frac{2}{U(t)-z}$,\quad i.e. \quad \ $f_t(z)=z-\frac{2t}{z}+...$ at $\infty$,}\] and here SLE$(\kappa)$ corresponds to $U(t)=B_{\kappa\cdot t}$ for a Brownian motion $B_t$. This different normalization is better suited when chordal SLE is compared to other versions of SLE, e.g.\ radial SLE in the unit disc, see \cite{SW05}.
\end{remark}

 \begin{figure}[H]
\rule{0pt}{0pt}
\centering
\includegraphics[width=0.7\textwidth]{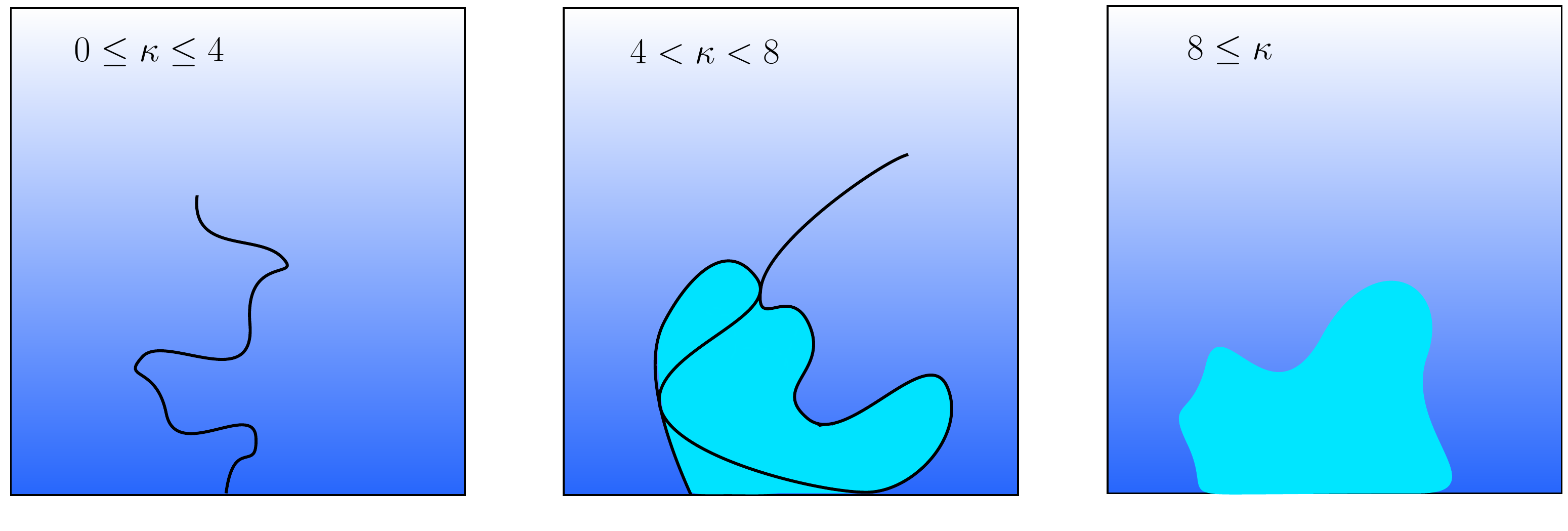}
\caption{A sketch for the three cases of SLE curves.}
\end{figure}

 \begin{figure}[H]
\rule{0pt}{0pt}
\centering
\includegraphics[width=0.9\textwidth]{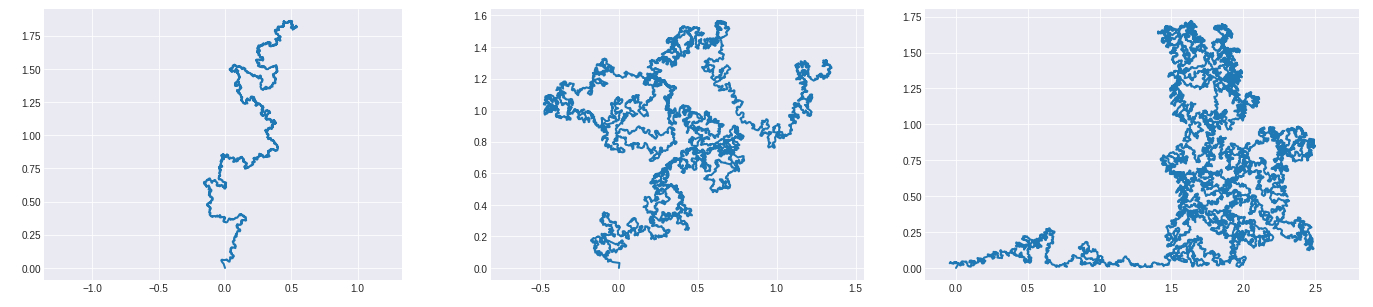}
\caption{Simulations for the cases $\kappa=2,6,9$.}
\end{figure}

The corresponding random growth processes $\{K_t\}_{t\geq0}$ have been shown to be the scaling limits of random curves from different models, depending on the value of $\kappa$:
\begin{itemize}
	\item $SLE(2)$: loop erased random walk
	\item $SLE(3)$: critical Ising model
	\item $SLE(4)$: harmonic explorer
	\item $SLE(4)$: contour lines of the discrete Gaussian free field
	\item $SLE(6)$: critical percolation
	\item $SLE(8)$: uniform spanning tree
\end{itemize}

\begin{remark}
Thinking of the Ising model, it is natural to look for an extension of SLE to a multivariate version, which describes several disjoint interface curves simultaneously. The construction of such a \emph{multiple SLE} requires more work and can also be approached from different points of view, see \cite{Karrila} and the references therein for the historical development and the recent progress.  \\

The works \cite{delMonacoSchleissinger:2016, dMHS18, HK18, HS19} consider the question whether there exists a limit Loewner chain as the number of the slits tends to $\infty$. 
Under certain assumptions (e.g.\ equal growth speed of each curve), the limit does exist and, funnily, it is described by the Loewner equation whose Herglotz vector field is given by the Voiculescu transforms of a free Brownian motion. 
\end{remark}

\chapter[Products of graphs]{Growing graph products as quantum random walks}

Undirected graphs can be represented by their adjacency matrices, which are symmetric and real-valued and thus self-adjoint quantum random variables. The five independences now arise as five products of graphs.\\

In this chapter we take a quantum probabilistic look at graph theory and we will see that certain independent increment processes can be regarded as growing graphs.

\section{Graphs as noncommutative random variables}

Let $V$ be a vertex set, finite or countable infinite, with a distinguished vertex $o\in V$.\\
Let $A:V\times V\to \{0,1\}$ be a symmetric matrix with $A_{xx}=0$ for all $x\in V$. \\

We can interpret $A$ as the adjacency matrix of an undirected, loop-free graph with vertex set $V$,
where $A_{xy}=1$ if and only if $x\sim y$, i.e.\ $x$ and $y$ are connected by an edge.

\begin{definition}
We define a \emph{(rooted)} \emph{graph} as such a triple $G=(V,A,o)$.\\
 For $x\in V$, the degree $\deg(x)$ of $x$ is defined as $\sum_{y\in V} A_{xy}$. The degree of the graph is defined as $\deg(G):=\deg(A):=\sup_{x\in V}\deg(x)$.
\end{definition}

If $\deg(A)<\infty$, then $A$ can be regarded as a bounded self-adjoint operator 
on the Hilbert space $l^2(V)$, see \cite[Theorem 3.2]{Moh82}.
The distinguished vertex $o\in V$ enables us to regard $A$ as a quantum 
random variable on the quantum probability space 
$(B(l^2(V)), \left<\delta_o, \cdot \delta_o\right>)$, where $\delta_o\in l^2(V)$ with $(\delta_o)(o)=1, (\delta_o)(x)=0$
for $x\not=o$. 

\begin{example}Let $V=\{0,1,2\}$ and connect $0$ with $1$ and $2$.\\[-1mm]
\begin{center}
\begin{tikzpicture}
\node[radius=3mm,label={[label]270: $0$}] (E1) at (0,0) [circle, fill=red] {};
\node[radius=3mm,label={[label]270: $1$}] (E2) at (-1,-1) [circle, fill=red] {};
\node[radius=3mm,label={[label]270: $2$}] (E3) at (1,-1) [circle, fill=red] {};
\coordinate (E8) at (3.5,0);
\draw[-] (E1) to (E2);
\draw[-] (E3) to (E1);
\end{tikzpicture}
\end{center}
 We choose $o=0$. The matrix $A$ has eigenvalues $\sqrt{2}$ with eigenvector $v_1=(\sqrt{2},1,1)/\sqrt{4}$, $-\sqrt{2}$ with eigenvector $v_2=(-\sqrt{2},1,1)/\sqrt{4}$, and $0$ with eigenvector $v_3=(0,-1,1)/\sqrt{2}$. Hence, the distribution $\mu$ of $A$ is given by 
\[\mu = \left<\delta_o,v_1\right>^2 \delta_{\sqrt{2}}+   \left<\delta_o,v_2\right>^2 \delta_{-\sqrt{2}}+ \left<\delta_o,v_3\right>^2 \delta_{0} = 
        \frac{1}{2} \delta_{\sqrt{2}}+  \frac{1}{2}  \delta_{-\sqrt{2}}.\]
	\hfill $\blacksquare$			
\end{example}

\newpage
\begin{example}\label{zet}
 Let $V=\Z$ with $A_{jk}=1$ if and only if $|j-k|=1$ and $0$ otherwise and choose $o=0$.\\[-1mm]
\begin{center}
\begin{tikzpicture}
\coordinate[label={[label distance=0mm]180: $\cdots$}] (E0) at (-3.5,0);
\node[radius=3mm,label={[label]270: $-3$}] (E1) at (-3,0) [circle, fill=red] {};
\node[radius=3mm,label={[label]270: $-2$}] (E2) at (-2,0) [circle, fill=red] {};
\node[radius=3mm,label={[label]270: $-1$}] (E3) at (-1,0) [circle, fill=red] {};
\node[radius=3mm,label={[label]270: $0$}] (E4) at (0,0) [circle, fill=red] {};
\node[radius=3mm,label={[label]270: $1$}] (E5) at (1,0) [circle, fill=red] {};
\node[radius=3mm,label={[label]270: $2$}] (E6) at (2,0) [circle, fill=red] {};
\node[radius=3mm,label={[label]270: $3$}] (E7) at (3,0) [circle, fill=red] {};
\coordinate[label={[label distance=0mm]0: $\cdots$}] (E8) at (3.5,0);
\draw[-] (E0) to (E1);
\draw[-] (E1) to (E2);
\draw[-] (E2) to (E3);
\draw[-] (E3) to (E4);
\draw[-] (E4) to (E5);
\draw[-] (E5) to (E6);
\draw[-] (E6) to (E7);
\draw[-] (E7) to (E8);
\end{tikzpicture}
\end{center}
 Then  the distribution of $A$ within the probability space $(B(l^2(\Z)), \left<\delta_o, \cdot \delta_o\right>)$ is given by the arcsine 
 distribution with mean 0 and variance 2, see \cite[Section 6.1]{acc}.\hfill $\blacksquare$
\end{example}

Only special distributions arise from such graph random variables.

\begin{theorem}\label{non_omentd}Let $\mu$ be the distribution of a rooted graph $(V,A,o)$, where $A$ is interpreted as a quantum random variable from $(B(l^2(V)), \left<\delta_o, \cdot \delta_o\right>)$. Then 
\[\int_\R x \mu(dx) = 0\quad \text{and} \quad \int_\R x^n \mu(dx) \in \N_0\quad \text{for all} \quad n\geq 2.\]
\end{theorem}
\begin{proof}
Clearly, $\int_\R x \mu(dx) = \left<\delta_o, A \delta_o\right>= A_{oo}= 0$ and for $n\geq 2$, 
$\int_\R x^n \mu(dx) = \left<\delta_o, A^n \delta_o\right>$ is a sum of $0$s and $1$s and thus belongs to $\N_0$.
\end{proof}

Theorem \ref{non_omentd} provokes the following inverse problem. 

\begin{question}Let $\mu\in \mathcal{P}_c(\R)$ with $\int_\R x \mu(dx) = 0$ and $\int_\R x^n \mu(dx) \in \N_0$ for all $n\geq 2$. 
Is there a graph $(V,A,o)$ with distribution $\mu$?
\end{question}

\begin{remark}
We note that one could generalize our setting by allowing loops and weighted edges. In this way, every $A\in B(l^2(V))$ with $A_{jk}=A_{kj}\in\R$ can be interpreted as a graph. 
\end{remark}

\section{Graph products and independence}

Let $G_1=(V_1, A^1, o_1), G_2=(V_2, A^2, o_2)$ be two graphs. We now construct new graphs with vertex set $V_3=V_1\times V_2$ and distinguished vertex $o_3=(o_1,o_2)$. For the cases other than the comb product, we refer to \cite{hora}.\\

\underline{\textbf{The comb product}}\\

The comb product $G_1 \rhd G_2 = (V_3,A^3,o_3)$ (with respect to $o_2$) 
is defined via 
\begin{equation}\label{comb_p} A^3_{(xx')(yy')} = A^1_{xx'}\delta_{yo_2}\delta_{y'o_2} + \delta_{xx'}A^2_{yy'}. \end{equation}
Here we use the symbol $\delta_{xy}=1$ if $x=y$, $\delta_{xy}=0$ if $x\not=y$.
It can be verified that $(x,y)\sim (x',y')$ if and only if 
\begin{itemize}
 \item $x\sim x'$, $x\not= x'$ and $y=y'=o_2$, or
 \item  $x= x'$, $y=y'=o_2$, and $x\sim x$ or $o_2\sim o_2$, or
 \item $x=x'$ and $y\sim y',$ $(y,y')\not=(o_2,o_2)$.
\end{itemize}

 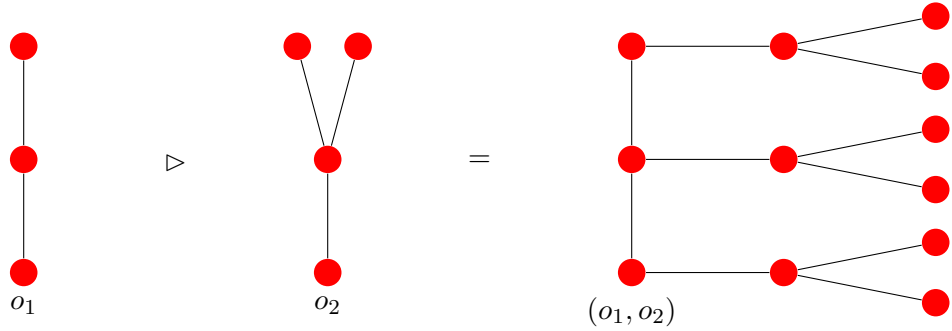
\begin{figure}[H]\vspace{-5mm}
\rule{0pt}{0pt}
\begin{center}
\begin{tikzpicture}
\node[radius=3mm][label={[label]270: $o_1$}] (A1) at (0,0) [circle, fill=red] {};
\node[radius=3mm] (A2) at (0,1.5) [circle, fill=red] {};
\node[radius=3mm] (A3) at (0,3) [circle, fill=red] {};
\coordinate[label={[label distance=-2mm]270: $\rhd$}] (A) at (2,1.5);
\node[radius=3mm][label={[label]270: $o_2$}] (B1) at (4,0) [circle, fill=red] {};
\node[radius=3mm]  (B2) at (4,1.5) [circle, fill=red] {};
\node[radius=3mm]  (B3) at (3.6,3) [circle, fill=red] {};
\node[radius=3mm]  (B4) at (4.4,3) [circle, fill=red] {};
\coordinate[label={[label distance=-2mm]270: $=$}] (AA) at (6,1.5);
\node[radius=3mm][label={[label]270: $(o_1,o_2)$}] (E1) at (8,0) [circle, fill=red] {};
\node[radius=3mm] (E2) at (8,1.5) [circle, fill=red] {};
\node[radius=3mm] (E3) at (8,3) [circle, fill=red] {};
\node[radius=3mm] (E4) at (10,0) [circle, fill=red] {};
\node[radius=3mm] (E5) at (10,1.5) [circle, fill=red] {};
\node[radius=3mm] (E6) at (10,3) [circle, fill=red] {};
\node[radius=3mm] (C1) at (12,0.4) [circle, fill=red] {};
\node[radius=3mm] (C2) at (12,-0.4) [circle, fill=red] {};
\node[radius=3mm] (C3) at (12,1.9) [circle, fill=red] {};
\node[radius=3mm] (C4) at (12,1.1) [circle, fill=red] {};
\node[radius=3mm] (C5) at (12,3.4) [circle, fill=red] {};
\node[radius=3mm] (C6) at (12,2.6) [circle, fill=red] {};
\draw[-] (A1) to (A2);
\draw[-] (A2) to (A3);
\draw[-] (B1) to (B2);
\draw[-] (B2) to (B3);
\draw[-] (B2) to (B4);
\draw[-] (E1) to (E2);
\draw[-] (E2) to (E3);
\draw[-] (E1) to (E4);
\draw[-] (E2) to (E5);
\draw[-] (E3) to (E6);
\draw[-] (E4) to (C1);
\draw[-] (E4) to (C2);
\draw[-] (E5) to (C3);
\draw[-] (E5) to (C4);
\draw[-] (E6) to (C5);
\draw[-] (E6) to (C6);
\end{tikzpicture}\vspace{-5mm}
\end{center}
\caption{The comb product of two graphs.}
\end{figure}

If $\deg(G_1), \deg(G_2)<\infty$, then the adjacency matrix $A^3$ of $G_1 \rhd G_2$ acts on $l^2(V_1\times V_2)\simeq l^2(V_1) \otimes l^2(V_2).$ Denote by $I^1$ the identity on $l^2(V_1)$ and by $P^2$ the projection from $l^2(V_2)$ onto the subspace spanned by 
$\delta_{o_2}$, i.e.\ $(P^2(\psi))(y) = \delta_{yo_2}\psi(o_2).$  Then one can verify that
 \begin{equation*} A_3 =  A^1 \otimes P^{2} + I^1 \otimes A^2.    
 \end{equation*}

More generally, we have the following decomposition. 

\begin{lemma}\label{dreidrei}
Let $G_1=(V_1,A^1,o_1),...,G_n=(V_n,A^n,o_n)$ be graphs. Denote by 
$I^k$ the identity on $l^2(V_k)$ and by $P^k$ the projection from $l^2(V_k)$ onto the subspace spanned by 
$\delta_{o_k}$, i.e.\ $(P^k(\psi))(y) = \delta_{yo_k}\psi(o_k).$  Denote by $B$ the adjacency matrix of the graph 
$G_1 \rhd G_2 \rhd ... \rhd G_n$. 
Then \begin{equation}\label{sumii}
 B= \sum_{j=1}^n I^1 \otimes ... \otimes I^{j-1} \otimes A^j \otimes P^{j+1} \otimes ... \otimes P^n.       
        \end{equation}
\end{lemma}

The lemma is a slightly more general version of \cite[Theorem 3.1]{acc}. Its proof follows 
 from definition \eqref{comb_p} and by induction. \\

The decomposition reminds us of of monotone independence, see Theorem \ref{models}.\\
So assume that $\sup\{\deg(v)\,|\, v\in V_j\}<\infty$ for all $j=1,...,n$. Then 
the adjacency matrix $B$ can be regarded as a quantum random variable in 
\[(B(l^2(V_1\times ... \times V_n)), \left<\delta_{o_1}\otimes ... \otimes \delta_{o_n}, \cdot (\delta_{o_1}\otimes ... \otimes \delta_{o_n})\right>).\] 
By \cite[Proposition 4.1]{acc}, the random variables 
$(I^1 \otimes ... \otimes I^{j-1} 
\otimes A^j \otimes P^{j+1} \otimes ... \otimes P^n)_{j\in(1,...,n)}$ are monotonically independent. 
Thus the distribution of $B$ is given by the monotone convolution of the distributions of 
the summands in \eqref{sumii}.  Furthermore, it is easy to see that 
the moments of 
$I^1 \otimes ... \otimes I^{j-1} \otimes A^j \otimes P^{j+1} \otimes ... \otimes P^n$ agree with the moments 
of $A^j$ within $(B(l^2(V_j)), \left<\delta_{o_j}, \cdot \delta_{o_j}\right>).$ Thus we obtain:

\begin{lemma}\label{viervier}Assume that $\sup\{\deg(v)\,|\, v\in V_j\}<\infty$ for all $j=1,...,n$. 
	Then the random variables 
$(I^1 \otimes ... \otimes I^{j-1} 
\otimes A^j \otimes P^{j+1} \otimes ... \otimes P^n)_{j\in(1,...,n)}$ are monotonically independent in the quantum probability space 
$(B(l^2(V_1\times ... \times V_n)), \left<\delta_{o_1}\otimes ... \otimes \delta_{o_n}, \cdot (\delta_{o_1}\otimes ... \otimes \delta_{o_n})\right>)$. Let $\mu_j$ be the distribution of $A_j$ within $(B(l^2(V_j)), \left<\delta_{o_j}, \cdot \delta_{o_j}\right>)$. Then $B$ has the
distribution 
\[ \mu_1 \rhd \mu_2 \rhd ... \rhd \mu_n.\]
\end{lemma}

\underline{\textbf{The direct product}}\\

The direct product $G_1 \times G_2 = (V_3,A^3,o_3)$ is defined by  $(x,y)\sim (x',y')$ if and only if 
$x=x'$ and $y\sim y'$ or $x\sim x'$ and $y=y'$.

 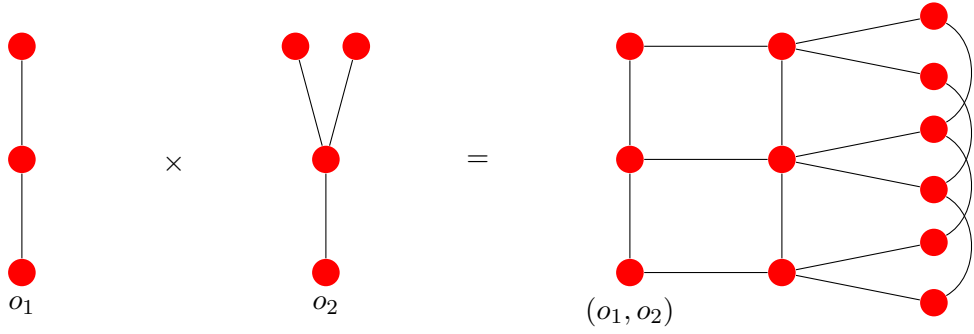
\begin{figure}[H]\vspace{-5mm}
\rule{0pt}{0pt}
\begin{center}
\begin{tikzpicture}
\node[radius=3mm][label={[label]270: $o_1$}] (A1) at (0,0) [circle, fill=red] {};
\node[radius=3mm] (A2) at (0,1.5) [circle, fill=red] {};
\node[radius=3mm] (A3) at (0,3) [circle, fill=red] {};
\coordinate[label={[label distance=-2mm]270: $\times$}] (A) at (2,1.5);
\node[radius=3mm][label={[label]270: $o_2$}] (B1) at (4,0) [circle, fill=red] {};
\node[radius=3mm]  (B2) at (4,1.5) [circle, fill=red] {};
\node[radius=3mm]  (B3) at (3.6,3) [circle, fill=red] {};
\node[radius=3mm]  (B4) at (4.4,3) [circle, fill=red] {};
\coordinate[label={[label distance=-2mm]270: $=$}] (AA) at (6,1.5);
\node[radius=3mm][label={[label]270: $(o_1,o_2)$}] (E1) at (8,0) [circle, fill=red] {};
\node[radius=3mm] (E2) at (8,1.5) [circle, fill=red] {};
\node[radius=3mm] (E3) at (8,3) [circle, fill=red] {};
\node[radius=3mm] (E4) at (10,0) [circle, fill=red] {};
\node[radius=3mm] (E5) at (10,1.5) [circle, fill=red] {};
\node[radius=3mm] (E6) at (10,3) [circle, fill=red] {};
\node[radius=3mm] (C1) at (12,0.4) [circle, fill=red] {};
\node[radius=3mm] (C2) at (12,-0.4) [circle, fill=red] {};
\node[radius=3mm] (C3) at (12,1.9) [circle, fill=red] {};
\node[radius=3mm] (C4) at (12,1.1) [circle, fill=red] {};
\node[radius=3mm] (C5) at (12,3.4) [circle, fill=red] {};
\node[radius=3mm] (C6) at (12,2.6) [circle, fill=red] {};
\draw[-] (A1) to (A2);
\draw[-] (A2) to (A3);
\draw[-] (B1) to (B2);
\draw[-] (B2) to (B3);
\draw[-] (B2) to (B4);
\draw[-] (E1) to (E2);
\draw[-] (E2) to (E3);
\draw[-] (E1) to (E4);
\draw[-] (E2) to (E5);
\draw[-] (E3) to (E6);
\draw[-] (E4) to (C1);
\draw[-] (E4) to (C2);
\draw[-] (E5) to (C3);
\draw[-] (E5) to (C4);
\draw[-] (E6) to (C5);
\draw[-] (E6) to (C6);
\draw[-] (E4) to (E5);
\draw[-] (E5) to (E6);
\draw [-,bend right=60] (C1) edge node {} (C3);
\draw [-,bend right=60] (C3) edge node {} (C5);
\draw [-,bend right=60] (C2) edge node {} (C4);
\draw [-,bend right=60] (C4) edge node {} (C6);
\end{tikzpicture}\vspace{-5mm}
\end{center}
\caption{The direct product of two graphs.}
\end{figure}

\begin{lemma}\label{dreidrei_direct}
Let $G_1=(V_1,A^1,o_1),...,G_n=(V_n,A^n,o_n)$ be graphs. Denote by 
$I^k$ the identity on $l^2(V_k)$ and by $B$ the adjacency matrix of the graph 
$G_1 \times G_2 \times ... \times G_n$. 
Then \begin{equation*}
 B= \sum_{j=1}^n I^1 \otimes ... \otimes I^{j-1} \otimes A^j \otimes I^{j+1} \otimes ... \otimes I^n.       
        \end{equation*}
\end{lemma}

\begin{lemma}\label{viervier_direct}Assume that $\sup\{\deg(v)\,|\, v\in V_j\}<\infty$ for all $j=1,...,n$. 
	Then the random variables 
$(I^1 \otimes ... \otimes I^{j-1} 
\otimes A^j \otimes I^{j+1} \otimes ... \otimes I^n)_{j\in(1,...,n)}$ are tensor independent in the quantum probability space 
$(B(l^2(V_1\times ... \times V_n)), \left<\delta_{o_1}\otimes ... \otimes \delta_{o_n}, \cdot (\delta_{o_1}\otimes ... \otimes \delta_{o_n})\right>)$. Let $\mu_j$ be the distribution of $A_j$ within $(B(l^2(V_j)), \left<\delta_{o_j}, \cdot \delta_{o_j}\right>)$. Then $B$ has the
distribution 
\[ \mu_1 * \mu_2 * ... * \mu_n.\]
\end{lemma}

\underline{\textbf{The star product}}\\

The star product $G_1 \star G_2 = (V_3,A^3,o_3)$ is defined by  $(x,y)\sim (x',y')$ if and only if 
$x=x'=o_1$ and $y\sim y'$ or $x\sim x'$ and $y=y'=o_2$.

 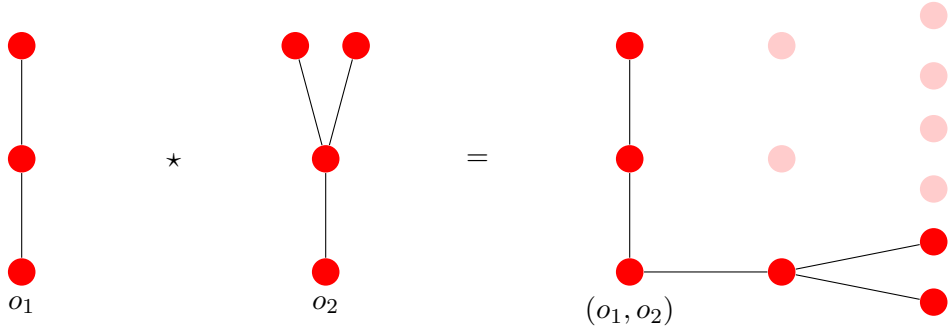
\begin{figure}[H]\vspace{-5mm}
\rule{0pt}{0pt}
\begin{center}
\begin{tikzpicture}
\node[radius=3mm][label={[label]270: $o_1$}] (A1) at (0,0) [circle, fill=red] {};
\node[radius=3mm] (A2) at (0,1.5) [circle, fill=red] {};
\node[radius=3mm] (A3) at (0,3) [circle, fill=red] {};
\coordinate[label={[label distance=-2mm]270: $\star$}] (A) at (2,1.5);
\node[radius=3mm][label={[label]270: $o_2$}] (B1) at (4,0) [circle, fill=red] {};
\node[radius=3mm]  (B2) at (4,1.5) [circle, fill=red] {};
\node[radius=3mm]  (B3) at (3.6,3) [circle, fill=red] {};
\node[radius=3mm]  (B4) at (4.4,3) [circle, fill=red] {};
\coordinate[label={[label distance=-2mm]270: $=$}] (AA) at (6,1.5);
\node[radius=3mm][label={[label]270: $(o_1,o_2)$}] (E1) at (8,0) [circle, fill=red] {};
\node[radius=3mm] (E2) at (8,1.5) [circle, fill=red] {};
\node[radius=3mm] (E3) at (8,3) [circle, fill=red] {};
\node[radius=3mm] (E4) at (10,0) [circle, fill=red] {};
\node[radius=3mm] (E5) at (10,1.5) [circle, fill=red!20] {};
\node[radius=3mm] (E6) at (10,3) [circle, fill=red!20] {};
\node[radius=3mm] (C1) at (12,0.4) [circle, fill=red] {};
\node[radius=3mm] (C2) at (12,-0.4) [circle, fill=red] {};
\node[radius=3mm] (C3) at (12,1.9) [circle, fill=red!20] {};
\node[radius=3mm] (C4) at (12,1.1) [circle, fill=red!20] {};
\node[radius=3mm] (C5) at (12,3.4) [circle, fill=red!20] {};
\node[radius=3mm] (C6) at (12,2.6) [circle, fill=red!20] {};
\draw[-] (A1) to (A2);
\draw[-] (A2) to (A3);
\draw[-] (B1) to (B2);
\draw[-] (B2) to (B3);
\draw[-] (B2) to (B4);
\draw[-] (E1) to (E2);
\draw[-] (E2) to (E3);
\draw[-] (E1) to (E4);
\draw[-] (E4) to (C1);
\draw[-] (E4) to (C2);
\end{tikzpicture}\vspace{-5mm}
\end{center}
\caption{The star product of two graphs.}
\end{figure}

\begin{lemma}\label{dreidrei_star}
Let $G_1=(V_1,A^1,o_1),...,G_n=(V_n,A^n,o_n)$ be graphs. Denote by  $P^k$ the projection from $l^2(V_k)$ onto the subspace spanned by 
$\delta_{o_k}$, i.e.\ $(P^k(\psi))(y) = \delta_{yo_k}\psi(o_k).$  Denote by $B$ the adjacency matrix of the graph 
$G_1 \star G_2 \star ... \star G_n$. 
Then \begin{equation*}
 B= \sum_{j=1}^n P^1 \otimes ... \otimes P^{j-1} \otimes A^j \otimes P^{j+1} \otimes ... \otimes P^n.       
        \end{equation*}
\end{lemma}

\begin{lemma}\label{viervier_star}Assume that $\sup\{\deg(v)\,|\, v\in V_j\}<\infty$ for all $j=1,...,n$. 
	Then the random variables 
$(P^1 \otimes ... \otimes P^{j-1} 
\otimes A^j \otimes p^{j+1} \otimes ... \otimes P^n)_{j\in(1,...,n)}$ are Boolean independent in the quantum probability space 
$(B(l^2(V_1\times ... \times V_n)), \left<\delta_{o_1}\otimes ... \otimes \delta_{o_n}, \cdot (\delta_{o_1}\otimes ... \otimes \delta_{o_n})\right>)$. Let $\mu_j$ be the distribution of $A_j$ within $(B(l^2(V_j)), \left<\delta_{o_j}, \cdot \delta_{o_j}\right>)$. Then $B$ has the
distribution 
\[ \mu_1 \uplus \mu_2 \uplus ... \uplus \mu_n.\]
\end{lemma}

\underline{\textbf{The free product}}\\

Free independence can be realized via the free product of graphs, which is more complicated than the other cases. We refer to \cite{acc2}.

\section{Approximation of additive processes}

Fix an independence and its convolution $\star\in\{*,\boxplus,\uplus,\rhd\}$ and let $\star_G$ be the corresponding graph  product. Let $T>0$ and let $(X_t)_{t\in [0,T]}$ be an additive process with distributions $(\mu_t)_{t\in [0,T]}$.\\

Next let $G_{n,k}$, $n\in\N$, $k=1,...,n$, be graphs and consider the discrete quantum process $(Y^n_k)_{k=1,...,n}$ given by the $n$ graphs
\[G_{n,1}, \quad G_{n,1} \star_G G_{n,2},\quad ..., \quad G_{n,1} \star_G G_{n,2} \star_G \cdots \star_G G_{n,n}.\]

Also $Y^n_k$ has independent increments. Can we approximate $X_t$ by $Y^n_k$, such that $G_{n,k}$ approximates the increment $X_{kT/n}-X_{(k-1)T/n}$? If $\mu_{n,k}$ is the distribution of $G_{n,1}  \star_G \cdots \star_G G_{n,k}$, then we would like that
 the scaled distribution $\mu_{n,\lfloor tn/T \rfloor}(c(n)\; \cdot)$ converges to $\mu_t$ as $n\to\infty$, where $c(n)>0$ is some scaling factor.\\ 

We already know a necessary condition. Consider the set $Gr(\star)$ of all $\star$-hemigroup distributions with compact support, mean $0$, and non-negative $n$-th moments for all $n\geq 2$. If $(X_t)_{t\in [0,T]}$ can be approximated in the way described, then necessarily $\mu_t\in Gr(\star)$ for all $t\in[0,T]$ due to Theorem \ref{non_omentd}. Conversely, we can ask the following question.

\begin{question}
Assume that $\mu_t\in Gr(\star)$ for all $t\in[0,T]$. Is it possible to approximate $(X_t)_{t\in [0,T]}$ by the graph products $(Y^n_k)_{k=1,...,n}$?
\end{question}

In the remaining sections we will look at the case of monotone independence and we see that an additional assumption on $(\mu)_{t\in [0,T]}$ guarantees that we can approximate $(X_t)_{t\in [0,T]}$ where all $G_{n,k}$ are special spidernets.

\section{Spidernets}

The main result of the work \cite{acc} (Theorem 5.1) can be interpreted as a discrete approximation of a monotone Brownian motion, a ``monotone quantum random walk'', via adjacency matrices of certain graphs. We now follow \cite{sch0}, which extends this idea.\\ 

First we construct special graphs whose distributions will be related to the Loewner equation.\\
We denote by $d(x,y)$ the length of the shortest walk 
within  a graph connecting $x$ and $y$. For $\eps\in\{-1,0,+1\}$, we define 
for any $x\in V$, 
\[\omega_{\eps}(x)=|\{y\in V\,|\, y\sim x, d(o,y)=d(o,x)+\eps\}|.\]

Let $a\in\N, b\in\N\setminus\{1\}$ and $c\in \N$ with $c \leq b-1$.
A \emph{spidernet with data $(a,b,c)$}, see \cite[Def. 4.25]{hora}), is a graph $(V,A,o)$ with root $o\in V$
such that 
\[ \omega_{+1}(o)=a,\quad \omega_{-1}(o)=\omega_0(o)=0, \quad \text{and} \quad
 \omega_{+1}(x)=c,\quad \omega_{-1}(x)=1,\quad \omega_{0}(x)=b-1-c\]
for all $x\in V\setminus\{o\}$ (and $A_{xy}\in\{0,1\}$ for all $x,y\in V$). 

  \begin{figure}[H]
\rule{0pt}{0pt}
\centering
\includegraphics[width=7cm]{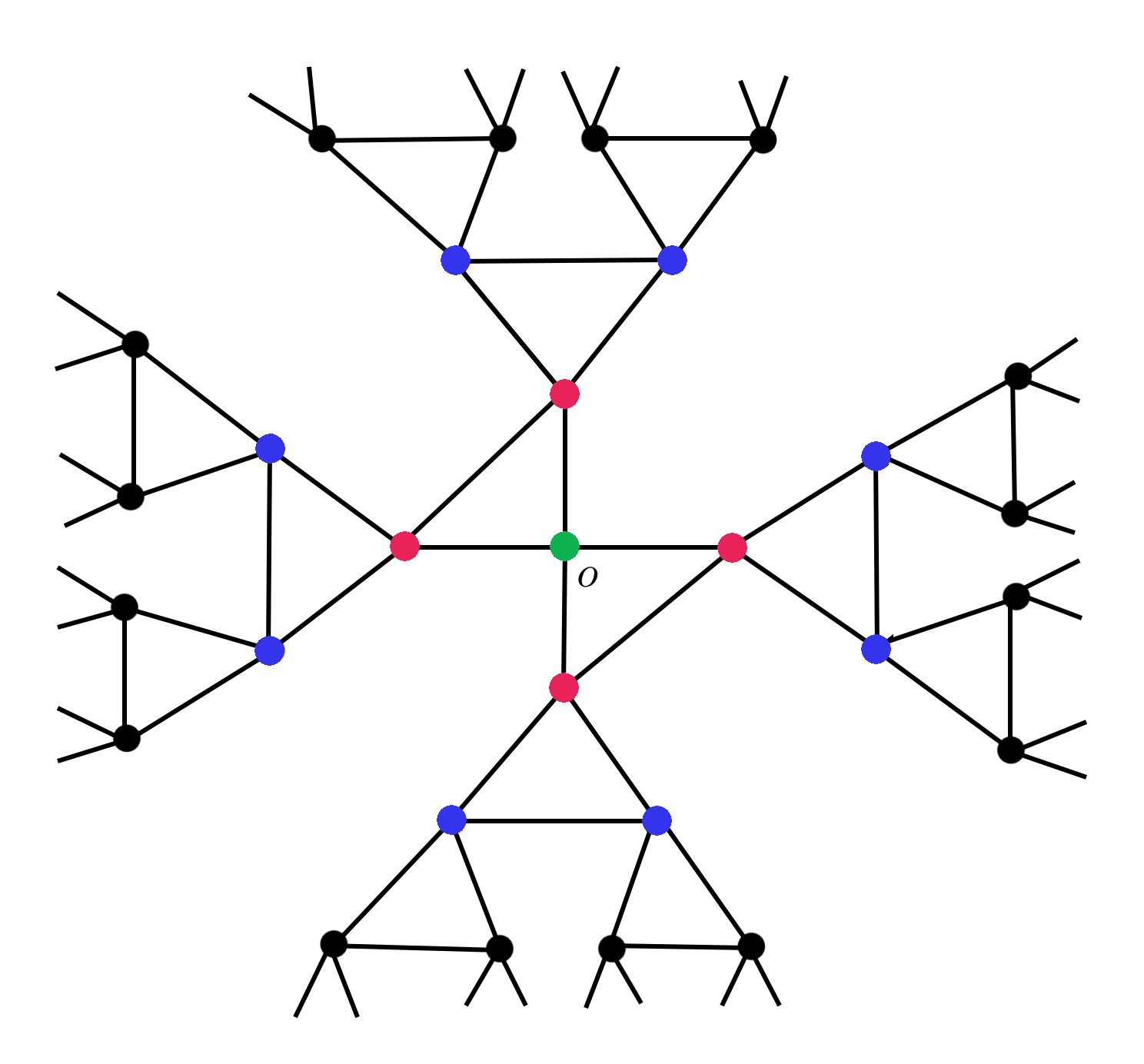}
\includegraphics[width=7cm]{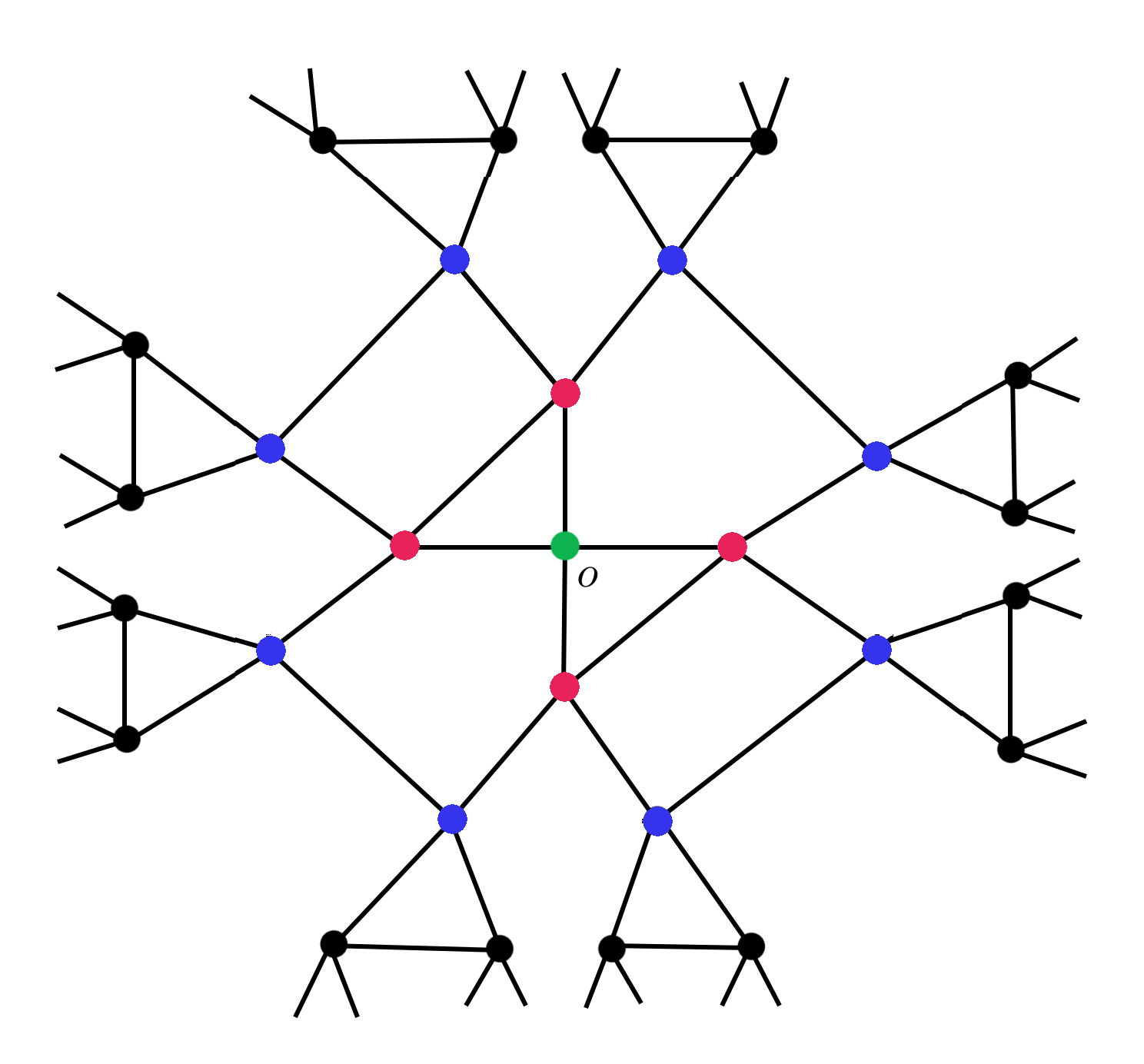}
\caption{Two spidernets with data $(4,4,2).$}
\end{figure}

\begin{example}\label{Wig_graph}The distribution of the spidernet with data $(1,2,1)$ is the Wigner law $\mu_{W,1}$, see \cite[Remark 8.3]{io}. 
This spidernet corresponds to the graph with vertex set $V=\N_0$, $o=0$, and $j\sim k$ if and only if $|j-k|=1$.\\[-5mm]
\begin{center}
\begin{tikzpicture}
\node[radius=3mm,label={[label]270: $0$}] (E4) at (0,0) [circle, fill=red] {};
\node[radius=3mm,label={[label]270: $1$}] (E5) at (1,0) [circle, fill=red] {};
\node[radius=3mm,label={[label]270: $2$}] (E6) at (2,0) [circle, fill=red] {};
\node[radius=3mm,label={[label]270: $3$}] (E7) at (3,0) [circle, fill=red] {};
\coordinate[label={[label distance=0mm]0: $\cdots$}] (E8) at (3.5,0);
\draw[-] (E4) to (E5);
\draw[-] (E5) to (E6);
\draw[-] (E6) to (E7);
\draw[-] (E7) to (E8);
\end{tikzpicture}\vspace{-8mm}
\end{center}\hfill $\blacksquare$
\end{example}

The spectrum of the adjacency matrix of a spidernet with respect to the quantum probability space $(B(l^2(V)), \left<\delta_o,\cdot \delta_o\right>)$ is the free Meixner law 
$m_{a,c,b-1-c}$, see \cite[Thm. 7.3]{io}. This distribution is described explicitly in \cite[Section B]{io}. 
For us it is sufficient to know its $F$-transform for the following special case.

\begin{lemma}\label{lemma000}
Let $n\in\N$ and $u\in\{0,...,2n-1\}$. Then there exists a spidernet $S_{n,u}$ with data $(2n, n+1+u, n)$.\\
The distribution $\mu$ of $S_{n,u}$ (the free Meixner law $m_{2n, n, u}$) has $0$ mean and variance $2n$, and its $F$-transform is given by 
\[F_\mu(z) = \sqrt{(z-u)^2-4n}+u.\] 
\end{lemma}
\begin{proof}From looking at the $2n$ vertices with $d(o,x)=1$, we get the necessary condition 
 $b-1-c = u \leq 2n-1$ for the existence of a spidernet with data $(2n, n+1+u, n)$. Conversely, one can verify by induction that for each $n\in\N$ and every $u\in\{0,...,2n-1\}$ there exists a spidernet with data $(2n, n+1+u, n)$. \\
The remaining statements follow from \cite[Thm. 7.3]{io} and the formula for the Cauchy transform of the free Meixner law in \cite[Equation (B.1)]{io}. 
\end{proof}

In the following, we denote by $S_{n,u}$  a fixed spidernet with data $(2n, n+1+u, n)$, $n\in\N$ and $u\in\{0,...,2n-1\}$.\\

Lemma \ref{lemma000} is a lucky coincidence for us. The spidernets $S_{n,u}$ have a distribution whose $F$-transform is univalent, and moreover, $\sqrt{(z-u)^2-4n}+u$ is a simple slit mapping for a vertical line segment in the upper half-plane from $u$ to $u+i2\sqrt{n}$.  
In other words, this mapping is the solution of the slit Loewner equation \eqref{slit_Loewner_eq} with $U(t)\equiv u$ at $t=2n$.  \\

Hence, approximating a non-negative driving function by piecewise constant non-negative driving functions is related to 
approximating the corresponding measures by distributions of spidernets. Together with Lemma \ref{viervier}, we expect that 
the monotone additive process associated to the measures driven by a non-negative driving function can be approximated by a sequence of growing graphs.
 
 \begin{figure}[ht]
\rule{0pt}{0pt}
\centering
\includegraphics[width=7.5cm]{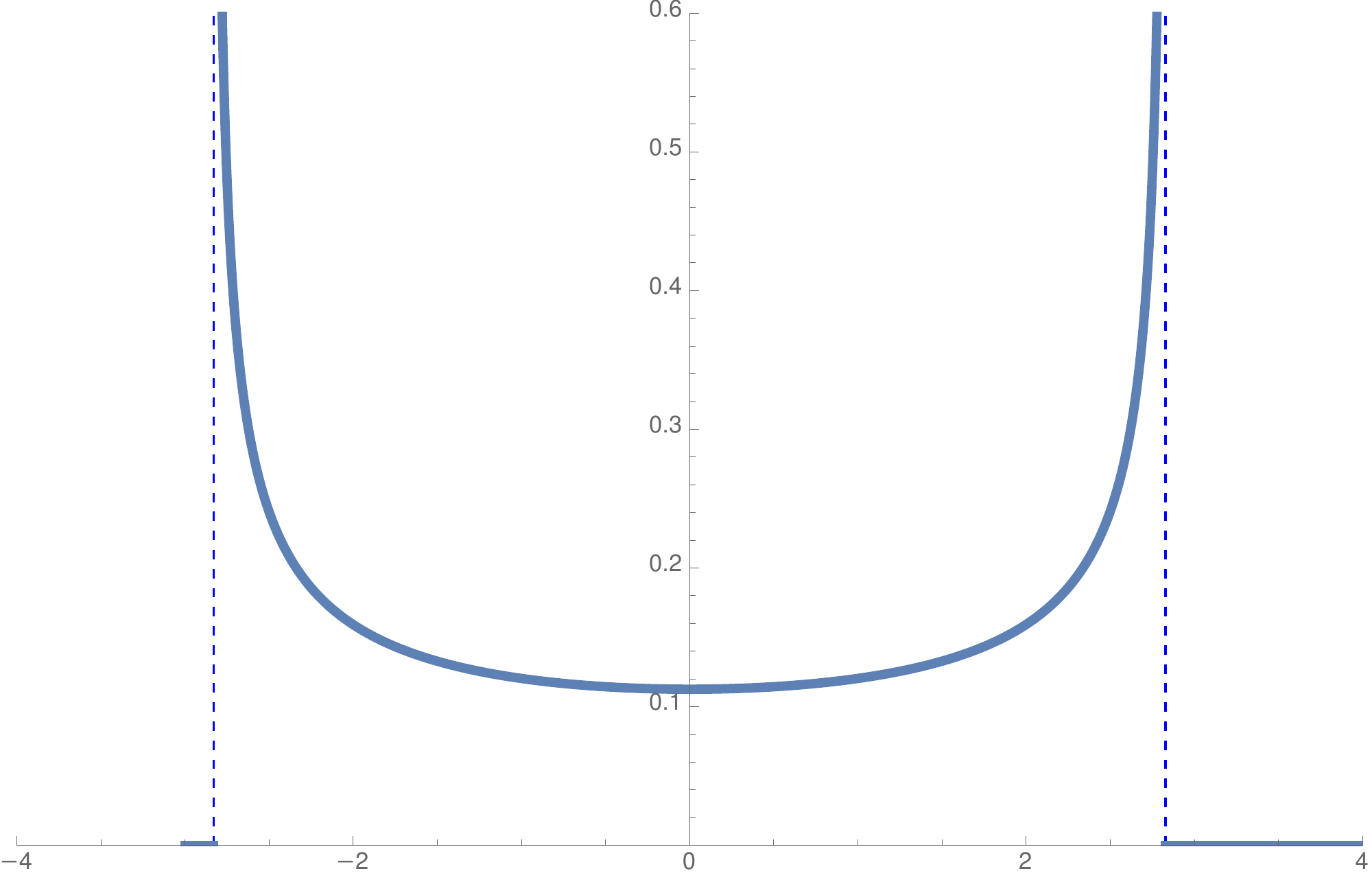}
\includegraphics[width=7.5cm]{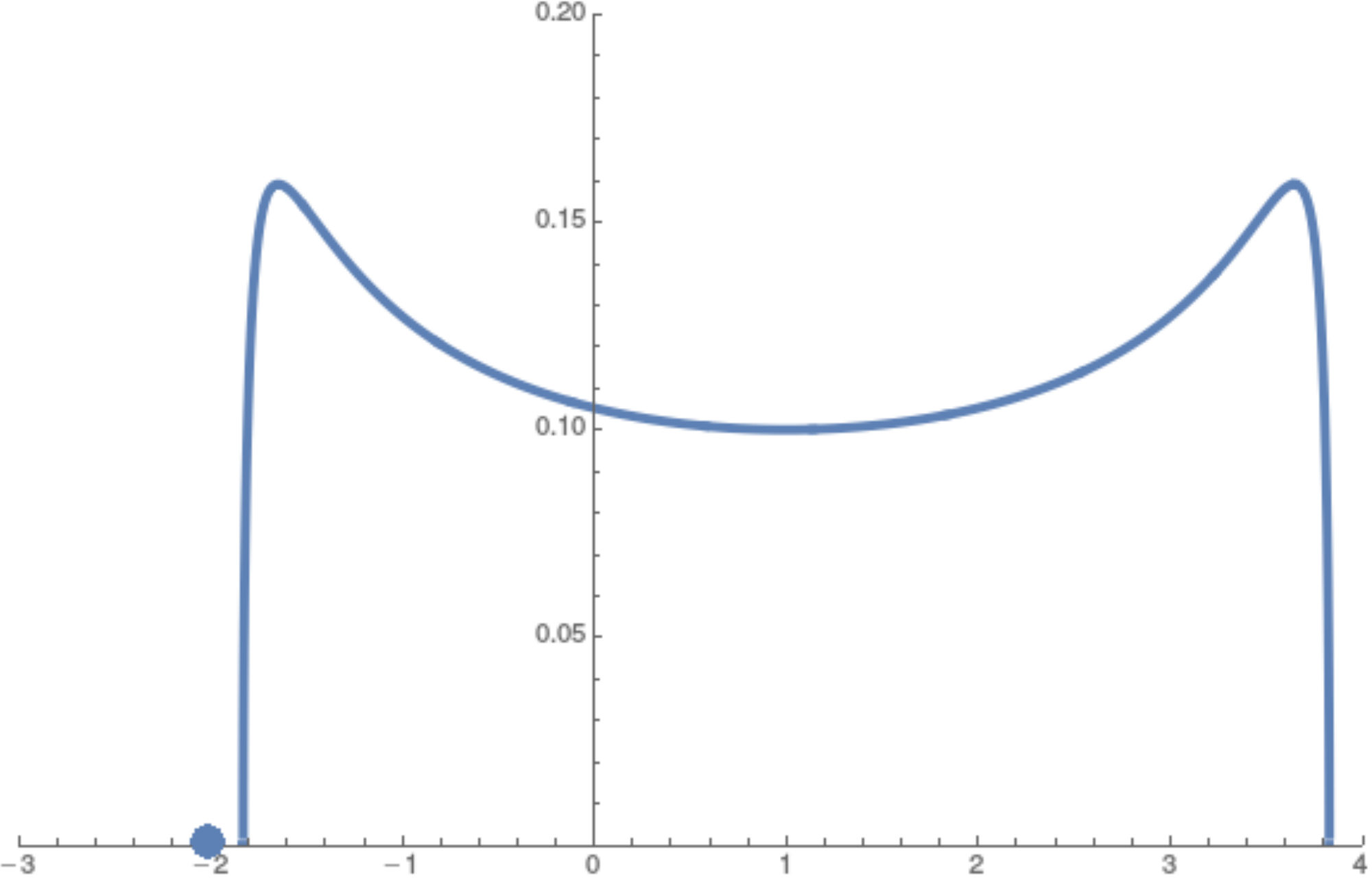}
\caption{Left: The free Meixner law $m_{4,2,0}$ is simply the arcsine distribution $A(0,4)$. Right: The density of $m_{4,2,1}$ in $[1-2\sqrt{2},1+2\sqrt{2}]$ and its atom at $-2$.}
\end{figure}

\section{Auxiliary approximation results}

\begin{definition}
Let $(\nu_t)_{t\geq0}$ be a family of probability 
measures on $\R$ such that $t\mapsto H(t,z):=\int_\R\frac{\nu_t(du)}{z-u}$ is measurable for every $z\in\Ha$, and  assume that there exists $M>0$ such that $\supp \nu_t \subset [-M,M]$ for all $t\geq0$. We denote the set of all 
such functions by $\mathcal{H}_M$.
\end{definition}

For  $H\in \mathcal{H}_M$ we consider the Loewner equation

\begin{equation}\label{slit33}
\frac{\partial}{\partial t} f_{t}(z) = -\frac{\partial}{\partial z}f_{t}(z)\cdot H(t,z)
\quad \text{for a.e.\ $t\geq 0$, $f_{0}(z)=z\in \Ha.$}
\end{equation}

\begin{theorem}\label{Houston} Let $(f_t)_{t\geq0}$ be the solution  of equation \eqref{slit33}. Then  each $f_t$ maps $\Ha$ conformally onto $\Ha\setminus K_t$ for a bounded set 
$K_t\subset {\Ha}$ and there exists a bound $C(t,M)>0$ such that $\sup_{z\in K_t} |z|< C(t,M)$.\\

There exists a unique family $(\mu_t)_{t\geq0}$ of probability measures, each with compact support and mean $0$, such that $f_t=F_{\mu_t}$. Furthermore, there exists a bound $D(t,M)>0$ such that $\supp \mu_t \subset [-D(t,M),D(t,M)]$.
\end{theorem}
\begin{proof}
The condition $\supp \nu_t \subset [-M,M]$ can be used to show that there is a bound $A(t,M)>0$ 
such that every $f_t$ extends conformally onto $I(t,M):=\R\setminus[-A(t,M),A(t,M)]$ with $f_t(I(t,M))\subset \R$.\\
This implies that there exists a bound $C(t,M)>0$ such that $\sup_{z\in K_t} |z|< C(t,M)$, see \cite[Inequality (3.14) on p.74]{lawler05}. \\

As in the proof of Theorem \ref{ftransform_images}, we see that the probability measures $\mu_t$ have compact support and mean $0$.
The existence of the uniform bound $D(t,M)$ follows from the bound $C(t,M)$ and the Stieltjes-Perron inversion formula.
\end{proof}

The following convergence result is standard in Loewner theory, see e.g.\ \cite[Lemma 4.12]{ghkk} for a slightly different setting.

\begin{lemma}\label{aprox_lemma}
Fix $T>0$. For every $n\in\N$, let $H_n(t,z)\in \mathcal{H}_M$. Assume that there exists 
$H(t,z)\in \mathcal{H}_M$ such that  \[\int_0^t H_n(s,z)ds \to \int_0^t H(s,z)ds\]
for every $t\in[0,T]$ locally uniformly in $\Ha$ as $n\to\infty$.\\
Let $f_{n,t}$ and $f_t$ be the solutions to \eqref{slit33} for the Herglotz 
vector fields $H_n(t,z)$ and $H(t,z)$ respectively. Then 
$f_{n,t}\to f_t$ for every $t\in[0,T]$ locally uniformly in $\Ha$.
\end{lemma}
\begin{proof}
Let $\nu_{n,t}$ and $\nu_t$ be the measures associated to $H_n(t,z)$ and $H(t,z)$ respectively. \\
The set 
$\{\int_\R \frac{\nu(du)}{z-u}\,|\, \text{$\nu$ is a prob. measure with $\supp \nu\subset [-M,M]$}\}$ is a normal family. 
Thus, if $G\in \mathcal{H}_M$ and $K\subset \Ha$ is a compact set, then there exists 
$L(K)>0$ such that $|G(t,z)-G(t,w)|\leq L(K)|z-w|$ for all $z,w\in K$ and all $t\in[0,T]$.\\ 
We now look at $g_{n,t}:=f_{n,t}^{-1}$, $g_t:=f_t^{-1}$. These functions satisfy 
\begin{equation*}
\frac{\partial g_{n,t}(z)}{\partial t} =  \int_\R\frac{\nu_{n,t}(du)}{g_t(z)-u},\quad  \frac{\partial g_{t}(z)}{\partial t} =  \int_\R\frac{\nu_t(du)}{g_t(z)-u}
\quad \text{for a.e.\ $t\geq0$, $g_{0}(z)=z\in \Ha,$}
\end{equation*}
 and we have
\begin{equation*}
g_{n,t}(z) =  z + \int_0^t H_n(s,g_{n,s}(z)) ds,\quad 
g_{t}(z) =  z+\int_0^t H(s,g_{s}(z)) ds.
\end{equation*}
Now let $K\subset \Ha$ be a compact set on which all $g_{n,t}$ and $g_t$ are defined. Due to 
Theorem \ref{Houston}, there exists a second compact set 
$K'\subset \Ha$, $K\subset K'$, such that $g_{n,t}(z), g_{t}(z)\in K'$ for all $z\in K$, $n\in\N$, and 
$t\in[0,T]$.\\
 We know that $\int_0^t H_n(s,z)ds$ converges uniformly on $K'$ to 
$\int_0^t H(s,z)ds$ for all $t\in[0,T]$. Now fix $t\in[0,T]$. For $z\in K$ we have 
\[ |g_{n,t}(z)-g_t(z)|\leq \left| \int_0^t H_n(s,g_{n,s}(z))-H_n(s,g_{s}(z)) ds \right| + 
\left| \int_0^t H_n(s,g_{s}(z))-H(s,g_{s}(z)) ds \right| \leq \]
\[ L(K') \int_0^t |g_{n,s}(z)-g_{s}(z)| ds  + 
\eps_n, \] for a sequence $(\eps_n)_n$ converging to $0$.
Gronwall's lemma implies that $g_{n,t}\to g_t$ uniformly on $K$. Hence also $f_{n,t}\to f_t$ locally uniformly in $\Ha$.
\end{proof}

We can now prove the following result, which will reduce our problem of constructing graphs for equation \eqref{slit33} to the slit equation \eqref{slit_Loewner_eq}.

\begin{lemma}\label{Whitney}Let $H(t,z)=\int_\R\frac{\nu_t(du)}{z-u}\in \mathcal{H}_M$ and let $(f_t)_{t\geq0}$ be the corresponding 
solution to \eqref{slit33}.
 Furthermore, assume that $\supp \nu_t \subset [0,M]$ for all $t\geq0$.\\
 Fix $T>0$. Then there exists 
a sequence $U_n:[0,T]\to [0,M]$ of continuous non-negative driving functions such that the corresponding solutions $(f_{n,t})_{t\geq0}$ to \eqref{slit_Loewner_eq} 
converge locally uniformly to $f_t$  for every $t\in[0,T]$ as $n\to\infty$.
\end{lemma}
\begin{proof}
Step 1: Assume that $H(t,z)=\frac1{z-U(t)}$ for a piecewise continuous and non-negative driving function $U$. 
Then we can clearly approximate $H(t,z)$ by a sequence 
$H_{n}(t,z)=\frac1{z-U_n(t)}$ with continuous non-negative driving functions $U_n:[0,T]\to [0,M]$
in the sense of Lemma \ref{aprox_lemma}.\\

Step 2: Next we consider the multi-slit equation, i.e.\ $H(t,z) = \sum_{k=1}^N\frac{\lambda_k(t)}{z-V_k(t)}$,
where $\lambda_1,...,\lambda_N:[0,T]\to[0,1]$ are continuous weight functions with
$\sum_{k=1}^N \lambda_k(t) = 1$ for all $t\in[0,T]$, and all driving functions $V_1,...,V_N:[0,T]\to[0,M]$ are continuous.\\

This Herglotz vector field can be approximated by a single-slit equation with a
piecewise continuous non-negative driving function.
We choose $m\in\N$ and divide the interval $[0,T]$ into $m$ intervals
$I_1:=[0,\frac{T}{m}], I_2:=(\frac{T}{m},\frac{T}{m}+\frac{1}{m}],...,I_m:=(T-\frac{1}{m},T]$. 
We define the driving function $U_m$ on $I_1$ as follows:
\begin{eqnarray*}
U_m(t)&=&V_1(t) \text{\quad on \quad} \left[0,T/m\cdot \lambda_1\left(T/m\right)\right],\nonumber\\
U_m(t)&=&V_2(t) \text{\quad on \quad} \left(T/m\cdot \lambda_1(T/m),T/m\cdot (\lambda_1(T/m)+\lambda_2(T/m))\right], ..., \nonumber\\
U_m(t)&=&V_N(t) \text{\quad on \quad} \left(T/m\cdot (\lambda_1(T/m)+...+\lambda_{N-1}(T/m)), T/m\right].
\end{eqnarray*}
We now repeat this construction for $I_2$,...,$I_m$. \\
Define $H_m(t,z)=\frac1{z-U_m(t)}$. Then $H_m(t,z)$ approximates 
$H(t,z)$ 
in the sense of Lemma \ref{aprox_lemma}. Together with step 1, we see that this multi-slit equation can be approximated 
by continuous non-negative driving functions.\\

Step 3:  Next we consider $H(t,z) = \sum_{k=1}^N\frac{\lambda_k(t)}{z-V_k(t)}$, where $\lambda_1,...,\lambda_N:[0,T]\to[0,1]$ are measurable weight functions with
$\sum_{k=1}^N \lambda_k(t) = 1$ for all $t\in[0,T]$, and all driving functions $V_1,...,V_N:[0,T]\to[0,M]$ are continuous.\\
For $m\in\N$, we let $H_m(t,z) = \sum_{k=1}^N\frac{\lambda_{k,m}(t)}{z-V_k(t)}$, 
where each $\lambda_{k,m}:[0,T]\to[0,1]$ is continuous, $\sum_{k=1}^N \lambda_{k,m}(t) = 1$ for all $t\in[0,T]$ and all $m\in\N$, 
and $\lambda_{k,m}\to \lambda_k$ in the $L^1$-norm as $m\to\infty$. Then $H_m(t,z)$ approximates $H(t,z)$ in the sense of Lemma \ref{aprox_lemma} as $m\to\infty$.\\

Step 4: Finally, assume that $H(t,z)=\int_\R\frac{\nu_t(du)}{z-u}\in\mathcal{H}_M$ is a general Herglotz vector field. 
 Divide $[0,M]$ into $m\in\N$ intervals: $I_{1,m}=[0,M/m], I_{2,m}=(M/m,2M/m],...,I_{m,m}=((m-1)M/m,M]$. 
For $k=1,...,m$, define $\lambda_{k,m}(t)=\nu_t(I_{k,m})$ and let $V_{k,m}(t)$ be the midpoint of $I_{k,m}$ for all $t\in[0,T]$. 
Each $\lambda_{k,m}$ is measurable, which follows from the Stieltjes-Perron inversion formula and the fact that $t\mapsto H(t,z)$ is measurable. The Herglotz vector field $H_m(t,z) = \sum_{k=1}^m\frac{\lambda_{k,m}(t)}{z-V_{k,m}(t)}$ approximates $H(t,z)$ in the sense of Lemma \ref{aprox_lemma} as $m\to\infty$.
\end{proof}

Finally, we also need the following quite useful scaling behavior.

\begin{lemma}\label{scale}Let $c,d>0$ and let $f_t=F_{\mu_t}$ be the solution to the slit equation \eqref{slit_Loewner_eq} 
with a piecewise continuous driving function $U(t)$.
 Consider the scaled measures $\nu_t(B)= \mu_{d\cdot t}(c\cdot B)$. Let $h_t=F_{\nu_t}$. Then $h_t$ solves 
\[\frac{\partial}{\partial t}h_{t}(z) = 
\frac{\partial}{\partial z}h_{t}(z)\cdot \frac{d/c^2}{h_{t}(z)-U(d\cdot t)/c}.\]
\end{lemma}
\begin{proof}
We have
\begin{eqnarray*}h_t(z) &=& \left(\int_\R \frac1{z-u} \mu_{d\cdot t}(c\cdot du)\right)^{-1} = \left(\int_\R 
\frac1{z-u/c} \mu_{d\cdot t}(du)\right)^{-1} \\
&=&
\left(\int_\R \frac{c}{cz-u} \mu_{d\cdot t}(du)\right)^{-1}=f_{dt}(cz)/c.\end{eqnarray*}
Then \eqref{slit_Loewner_eq} leads to
\[\frac{\partial}{\partial t}h_{t}(z) = \frac{d}{c} \frac{\partial}{\partial t}f_{dt}(cz) =
\frac{d}{c}\frac{\partial}{\partial z}f_{dt}(cz)\cdot \frac{1}{f_{dt}(cz)-U(d\cdot t)}=
\frac{\partial}{\partial z}h_{t}(z)\cdot \frac{d/c^2}{h_{t}(z)-U(d\cdot t)/c}.\]
\end{proof}

\section{Approximation via spidernets}

We now consider a driving function $U:[0,\infty)\to\R$ which is continuous and non-negative. \\
Let $(f_t)_{t\geq0}$ be the solution to \eqref{slit_Loewner_eq}  and denote by $(\mu_t)_{t\geq0}$ 
the probability measures with $F_{\mu_t}=f_t$. Furthermore, let $(X_t)_{t\geq 0}$ be a 
corresponding monotone additive process  given by Theorem \ref{inf}.\\

Fix some $T>0$. We would like to approximate $(X_t)_{t\in[0,T]}$ by a discrete quantum process, where each random variable is 
the adjacency matrix of a graph.
By means of the lemmas above, we can now proceed as follows.\\

Choose $n_0\in\N$  such that   
\begin{equation}\label{est0}0 \leq U(t) \leq \sqrt{\frac{T}{2}}\left(2\sqrt{n}-\frac1{\sqrt{n^3}}\right) \quad \text{on $[0,T]$}
\end{equation}
for all $n\geq n_0$.\\

Now assume that $n\geq n_0$. For $k=1,...,n$, we define
\[u_{n,k}= 
\lfloor \sqrt{2T}\sqrt{n}\cdot \frac{U(k/n\cdot T)}{\frac{T}{n}}\rfloor \in 
\{0, ..., 2n^2-1\}.\]
Here, $\lfloor x \rfloor$ denotes the largest $m\in \N_0$ with $m\leq x$. Note that 
\eqref{est0} implies that the spidernet $S_{n^2,u_{n,k}}$ exists for all $k=1,...,n$. 
We denote by $V_{n,k}$ the vertex set and by $o_{n,k}$ the root of $S_{n^2,u_{n,k}}$.

\begin{theorem}\label{theorem10}

For $k=1,...,n$, let $\mathcal{C}_{n,k}$ be the graph
\[ \mathcal{C}_{n,k} := S_{n^2, u_{n,1}} \rhd  S_{n^2, u_{n,2}}  \rhd ... \rhd 
S_{n^2, u_{n,k}}.\] 
 Then $(\mathcal{C}_{n,k})_{k=1,...,n}$ is a an approximation 
of the quantum process $(X_t)_{t\in[0,T]}$ in the following sense:
\begin{itemize}
\item[(a)]Let $A_{n,k}$ be the adjacency matrix of $\mathcal{C}_{n,k}$. Denote by $\mu_{n,k}$
the distribution of $A_{n,k}$ with respect to the 
quantum probability space 
\[(B(l^2(V_{n,1}\times ... \times V_{n,k})),\left<\delta_{o_{n,1}} \otimes ... \otimes \delta_{o_{n,k}}, \cdot (\delta_{o_{n,1}} \otimes ... \otimes \delta_{o_{n,k}})\right>).\] Then 
\[ \lim_{n\to\infty}
 \mu_{n,\lfloor tn/T \rfloor}(\sqrt{2n^3/T}\; \cdot ) = \mu_t(\cdot)
\]
 with respect to  weak convergence for all $t\in[0,T]$. The limit also holds true with respect to 
 the convergence of all moments.
 \item[(b)] 
 Consider 
the quantum probability space \[(B(l^2(V_{n,1}\times ... \times V_{n,n})),\left<\delta_{o_{n,1}} \otimes ... \otimes \delta_{o_{n,n}}\cdot (\delta_{o_{n,1}} \otimes ... \otimes \delta_{o_{n,n}})\right>).\]
Extend $A_{n,k}$ to $l^2(V_{n,1}\times ... \times V_{n,n})$ by $\mathcal{A}_{n,k}:=A_{n,k}\otimes P^{n,k+1}\otimes ... 
\otimes P^{n,n},$ where $P^{n,j}$ denotes the projection in $l^2(V_{n,j})$ onto 
$\delta_{o_{n,j}}$. Then the increments
$(\mathcal{A}_{n,1}, \mathcal{A}_{n,2}-\mathcal{A}_{n,1}, ..., \mathcal{A}_{n,n}-\mathcal{A}_{n,n-1})$ are monotonically independent.

\end{itemize}
\end{theorem}

\begin{remark}Note that the graph that corresponds to $\mathcal{A}_{n,k}$ is simply an embedding of $\mathcal{C}_{n,k}$ within 
a larger vertex set.  
\end{remark}

\begin{proof}Statement (b) follows directly from Lemmas \ref{dreidrei} and \ref{viervier}.\\

Let $U_n:[0,2n^3]\to\R$ be the function which is constant $u_{n,1}$ on $[0,2n^2]$, 
constant $u_{n,2}$ on $(2n^2,4n^2]$, etc.\\
 Let $f_{n,t}$ be the 
solution to \eqref{slit_Loewner_eq} with this driving function and define the measures $\alpha_{n,t}$ by 
$F_{\alpha_{n,t}}=f_{n,t}$. 
By Example \ref{ex_1} and Lemma \ref{lemma000} we have 
\[\alpha_{n,2n^2} = m_{2n^2,n^2,u_{n,1}}.\]
Starting the Loewner equation \eqref{slit_Loewner_eq} for $h_t$ at $t=2n^2$ with initial value $h_{2n^2}(z)=z$
and driving function $U_n(t)$ yields the mappings $(h_t)$ that satisfy 
$f_{n,t} = f_{n,2n^2} \circ h_{t}.$ Obviously, $h_{4n^2}=F_{m_{2n^2,n^2,u_{n,2}}}$ and thus 
$\alpha_{n,4n^2} =m_{2n^2,n^2,u_{n,1}}\rhd m_{2n^2,n^2,u_{n,2}}.$ By induction we obtain 
\begin{equation*} \alpha_{n,2kn^2} = \rhd_{j=1}^{k} m_{2n^2,n^2,u_{n,j}}.
\end{equation*}
On the other hand, Lemmas \ref{dreidrei}, \ref{viervier}, \ref{lemma000} imply
\begin{equation}\label{uu00}\mu_{n,k}=\rhd_{j=1}^{k} m_{2n^2,n^2,u_{n,j}}
\end{equation}
for all $k=1,...,n$.

The function $V_n:[0,T]\to\R, V_n(t):=\sqrt{\frac{T}{2n^3}} \cdot U_n(t/T\cdot 2n^3)$ is constant 
 on the intervals $(\frac{(k-1)T}{n}, \frac{kT}{n}],$ $k=1,...,n$. We have 

\begin{eqnarray*}
&& U(k/n\cdot T)- V_n(k/n\cdot T) = 
U(k/n\cdot T)- \sqrt{\frac{T}{2n^3}} \cdot U_n(k\cdot 2n^2)=\nonumber\\
&&U(k/n\cdot T)- \sqrt{\frac{T}{2n^3}} \cdot \lfloor \sqrt{2T}\sqrt{n}\cdot \frac{U(k/n\cdot T)}{\frac{T}{n}}\rfloor \leq 
\sqrt{\frac{T}{2n^3}}.
\end{eqnarray*}

Now let $t\in(\frac{(k-1)T}{n}, \frac{kT}{n})$ and denote by $\omega:[0,T]\to[0,\infty)$ a modulus of continuity of $U$ for $[0,T]$, i.e.\ 
$|U(x)-U(y)|\leq \omega(|x-y|)$ for all $x,y\in[0,T]$, and $\omega$ is increasing, vanishes at $0$, and is continuous at $0$. We have
\begin{eqnarray*}
&&|U(t)-V_n(t)|=|U(t)-V_n(kT/n)|\leq \nonumber \\
&&|U(t)-U(kT/n)| + |U(kT/n)-V_n(kT/n)|\leq \omega\left(\frac{T}{n}\right) + \sqrt{\frac{T}{2n^3}}.
\end{eqnarray*}
Finally, for $t=0$ we have $V_n(0)=V_n(T/n)$ and thus
\[|U(0)-V_n(0)|=|U(0)-U(T/n)| + |U(T/n)-V_n(T/n)|\leq \omega\left(\frac{T}{n}\right) + \sqrt{\frac{T}{2n^3}}.
\]
Hence,  we obtain
\begin{eqnarray}\label{conv0}
\sup_{t\in[0,T]}|U(t)-V_n(t)|\to 0 \quad \text{as $n\to\infty$}.
\end{eqnarray}

Let $(h_{n,t})_{t\in [0,T]}$ be the Loewner chain that corresponds to $V_n$. Define the measures 
$\nu_{n,t}$ by $h_{n,t}=F_{\nu_{n,t}}.$ 
Note that $V_n$ has the form $V_n = U_n(d \cdot t)/c$ with $d=c^2$. Hence, by Lemma \ref{scale} we have 
\[\nu_{n,t}(M) = \alpha_{n,t/T\cdot 2n^3}(\sqrt{2n^3/T}\cdot M)\]
for all $t\geq0$ and all Borel subsets $M\subset \R$. If $t$ has the form $t=kT/n, k=1,...,n,$ then
\eqref{uu00} gives
\begin{eqnarray*}\nu_{n,t}(M) &=& \mu_{n,k}(\sqrt{2n^3/T}\cdot M)  
= (\rhd_{j=1}^{k} m_{2n^2,n^2,u_{n,j}} )(\sqrt{2n^3/T}\cdot M) \nonumber\\ 
&=& 
 (\rhd_{j=1}^{tn/T} m_{2n^2,n^2,u_{n,j}})(\sqrt{2n^3/T}\cdot M).\end{eqnarray*}
For every $t\in[0,T]$ we have $h_{n,t} \to f_t$ locally uniformly because of
\eqref{conv0} and Lemma \ref{aprox_lemma}. 
By Lemma \ref{lemmaconvergence} we have $\nu_{n,t} \to \mu_t$ with respect to weak convergence, or 
\[ \mu_{n,\lfloor tn/T \rfloor}(\sqrt{2n^3/T}\cdot ) = (\rhd_{j=1}^{\lfloor tn/T \rfloor} m_{2n^2,n^2,u_{n,j}})(\sqrt{2n^3/T}\, \cdot )  \to \mu_t(\cdot).\] 
It remains to show that this limit also holds 
with respect to convergence of all moments.\\
As there is a uniform bound for the family $(V_n)_{n}$ on $[0,T]$, Theorem \ref{Houston} implies that there exists 
$D(t)>0$ such that $\supp \nu_{n,t}\subset[-D(t),D(t)]$ for all $n$ and all $t\in[0,T]$. Thus, weak convergence of $\nu_{n,t}$ 
is equivalent to convergence of all its moments. 
\end{proof}

Consider equation \eqref{slit33} with the additional condition that $\supp \nu_t \subset [0,M]$ for all $t\geq0$.
 Let $(f_t)_{t\geq0}$ be the solution to the corresponding Loewner equation and denote by $(\mu_t)_{t\geq0}$ the probability measures with $F_{\mu_t}=f_t$. Furthermore, let $(X_t)_{t\geq 0}$ be a 
corresponding additive process process given by Theorem \ref{inf}. The process $(X_t)_{t\in[0,T]}$ can be approximated by graphs in the following way.

\begin{theorem}\label{theorem11}
Choose $n_0\in\N$ such that $M\leq\sqrt{\frac{T}{2}}\left(2\sqrt{n}-\frac1{\sqrt{n^3}}\right)$ for all  $n\geq n_0$.
There exists a family $(\mathcal{C}_{n,k})_{n\geq n_0, k=1,...,n}$ of rooted graphs such that:
\begin{itemize}
 \item[(a)]  For each $n\geq n_0$,  $(\mathcal{C}_{n,k})_{k=1,...,n}$ can be considered as  graphs with common vertex set $V_n$ and common root $o_n$. 
Let $A_{n,k}$ be the adjacency matrix of $\mathcal{C}_{n,k}$. Then the increments
$(A_{n,1}, A_{n,2}-A_{n,1}, ..., A_{n,n}-A_{n,n-1})$ are monotonically independent with respect to the quantum probability space
 $(B(l^2(V_n)),\left<\delta_{o_n}, \cdot \delta_{o_n}\right>)$.
\item[(b)] Denote by $\mu_{n,k}$ the distribution of $A_{n,k}$. Then 
\[ \lim_{n\to\infty}
 \mu_{n,\lfloor tn/T \rfloor}(\sqrt{2n^3/T}\; \cdot ) = \mu_t(\cdot)
\]
 with respect to  weak convergence for all $t\in[0,T]$. The limit also holds true with respect to 
 the convergence of all moments.
\end{itemize}
\end{theorem}

\begin{proof}
Due to Lemma \ref{Whitney} there exists a sequence of continuous non-negative driving functions 
$U_m:[0,T]\to[0,M]$ such that the corresponding solution $f_{m,t}$ to 
\eqref{slit_Loewner_eq} converges locally uniformly to $f_t$ for all 
$t\geq0$ as $m\to\infty$. Write $f_{m,t}=F_{\mu_{m,t}}$. Then Lemma \ref{lemmaconvergence} implies that $\lim_{m\to\infty}\mu_{m,t}=\mu_t$.\\

Let $\mathcal{C}_{n,k;m}$ be the graphs from Theorem \ref{theorem10} for the driving function 
$U_m$ with distributions $\mu_{n,k;m}$. Note that $n\geq n_0$ and \eqref{est0} together with the bound $U_m(t)\leq M$ imply that $n$ is large enough 
to construct these graphs. Then 
\[  \lim_{n\to\infty} \mu_{n,\lfloor tn/T \rfloor;m}(\sqrt{2n^3/T}\; \cdot ) = \mu_{m,t}(\cdot). \]

A diagonalization argument (note that there is a metric for probability measures on $\R$
which is compatible with weak convergence, e.g.\ the L\'evy-Prokhorov distance) gives us a sequence $m(n)$ converging to $\infty$
  such that
\[  \lim_{n\to\infty} \mu_{n,\lfloor tn/T \rfloor;m(n)}(\sqrt{2n^3/T}\; \cdot ) = \mu_{t}(\cdot). \]

Hence, the graphs $\mathcal{C}_{n,k}:=\mathcal{C}_{n,k;m(n)}$ (where $\mathcal{C}_{n,k}$ is regarded as a 
subgraph of $\mathcal{C}_{n,n}$) satisfy all required conditions.

\end{proof}

\begin{remark}
Each measure $\mu_t$ has a univalent $F$-transform, the first moment of $\mu_t$ is $0$, and all higher odd moments are non-negative, i.e.\ $\mu_t \in Gr(\rhd)$. \\
Unfortunately, not all probability measures from $Gr(\rhd)$ arise via \eqref{slit33} with $\supp \nu_t \subset [0,M]$ for all $t\geq0$ and some $M>0$. Consider a measure $\mu$ with compact support, symmetric with respect to $0$, and having a univalent $F$-transform. Then all odd moments are $0$ and thus $\mu\in Gr(\rhd)$.  One can show that this measure can be generated in our setting only if $\mu$ is an arcsine distribution. 
All further distributions of this kind (e.g.\ the Wigner law $\mu=W(0,1)$) are not covered by \eqref{slit33}.
\end{remark}

\section{Further reading}

\begin{itemize}
\item A survey on spectra of infinite graphs is given in \cite{mw89}.
\item  The quantum probabilistic view on graphs is covered in the books \cite{hora} and \cite{Oba17}, which also treat several further aspects such as the quantum decomposition of graphs. We also refer to \cite{Len19} for further relations between independences and products of graphs.
\end{itemize}

\begin{appendices}

\chapter{Continuous extension of univalent functions}\label{A1}

In this section we regard the question under which conditions a conformal mapping in $\D$ can be extended continuously to the boundary. More on the boundary behavior of conformal mappings can be found in \cite{MR1217706}.\\

We denote by $B_r(z)$, $z\in \C$, $r>0$, the Euclidean disc with center $z$ and radius $r$. If $\gamma$ is a continuously  differentiable curve in $\C$, then $L(\gamma)$ denotes the Euclidean length of $\gamma.$\\

Let $D$ and $E$ be two conformally equivalent subdomains of $\C$. Then the boundaries $\partial D$ and $\partial E$ can be quite different from each other. For example, $\D$ is conformally equivalent to 
$$E=\Ha\setminus \left([0,i] \cup \bigcup_{n\in\N}[1/n,1/n+i] \cup \bigcup_{n\in\N}[-1/n,-1/n+i]\right).$$ 

 \begin{figure}[ht]
\rule{0pt}{0pt}
\centering
\includegraphics[width=7.5cm]{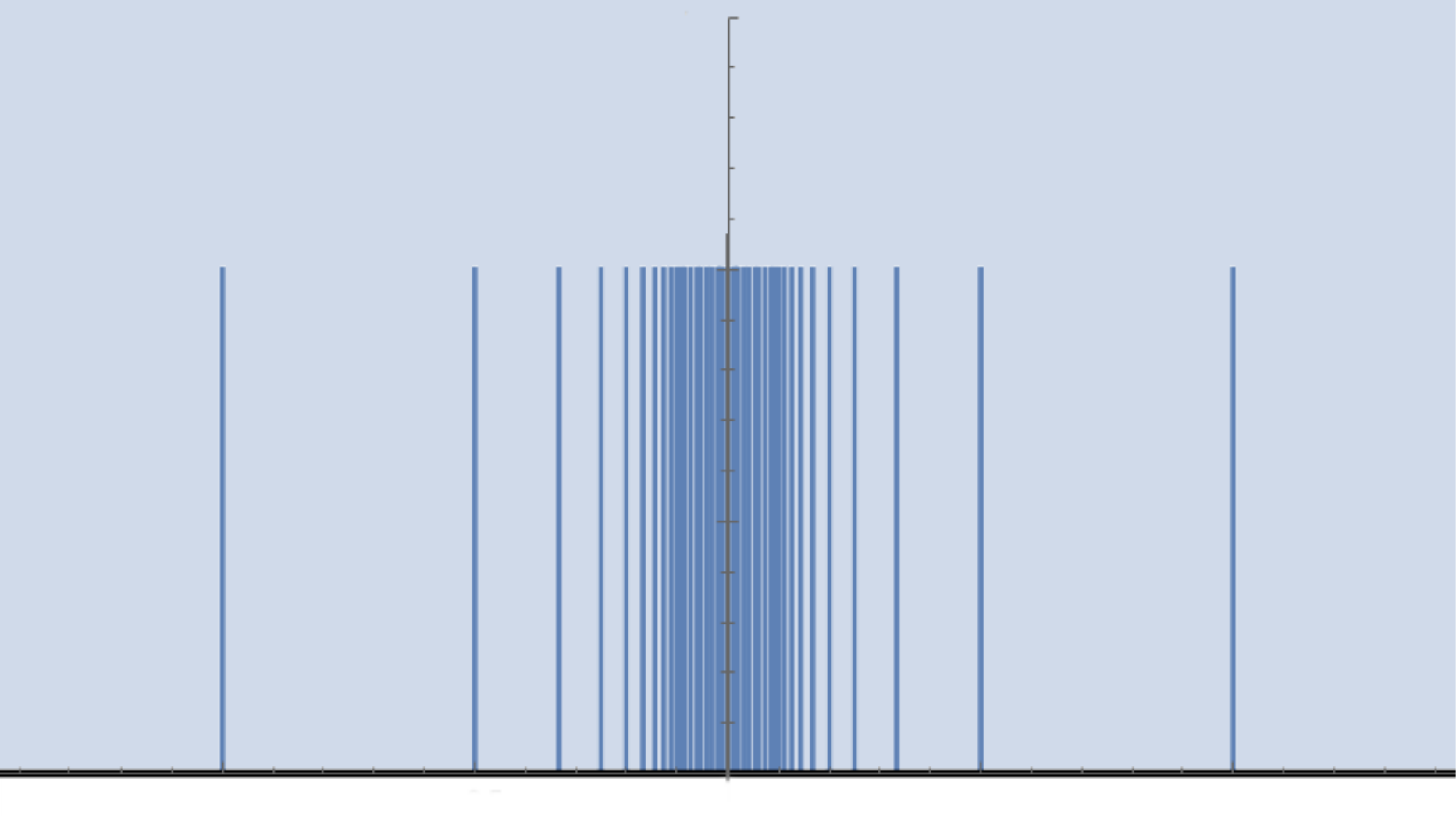}

\end{figure}

Let $f:\D\to E$ be conformal. One can show that there exists a point $q\in\partial \D$, such that every sequence $(z_n)_n\subset E$  with $z_n \to p \in [0,i]$ satisfies $f^{-1}(z_n)\to q.$ This implies that $f$ cannot be extended continuously to $q$.\\
The boundary behavior of a conformal mapping $f:\D\to E$ can be quite irregular, but it can be related to the topology of $\partial E$. \\

\underline{Recall:} A set $X\subset \C$ is called \emph{connected} if $X$ cannot be written as $X=A\cup B,$ where $A,B\not=\emptyset$ are disjoint and open in $X.$\\
A compact set $X\subset \C$ is called \emph{locally connected} if for every $x\in X$ and $r>0$ we can find a set $U\subset X$ which is connected and open in $X$ and satisfies $x\in U\subset B_r(x)$. \\
This is equivalent to: for every $\eps>0$ there exists $\delta >0$ such that for all points $a,b\in X$ with $|a-b|<\delta$ there exists a compact and connected set $B \subset X$ with 
$a,b\in B$ and $\diam B<\eps.$\footnote{For $B\subset \C,$ $\diam B := \sup_{z,w\in B}|z-w|.$}

\newpage 

\begin{theorem}\label{loc_con} Let $f:\D \to D$ be conformal and assume $D$ is bounded. Then the following statements are equivalent:
\begin{itemize}
	\item[a)] $f$ can be extended continuously to $\overline{\D}.$
	\item[b)] $\partial D$ is locally connected.
\end{itemize}
\end{theorem}

In the example above, $\partial E$ is not locally connected at $z=i$ (and the unboundedness of $E$ can be neglected as the same statements  also hold for the domain $B_2(0)\cap E$).

\begin{lemma}\label{pomm0} Let $U\subset \C$ be open and let $f:U\to B_R(0)$ be univalent. Furthermore, let $c\in \C$ and $C(r)=U\cap \partial B_r(c).$ Then, for $0<\rho < 1,$ we have:
$$\inf_{\rho < r <\sqrt{\rho}} L(f(C(r))) \leq \frac{2\pi R}{\sqrt{\log 1/ \rho}}. $$
In particular, there is a null sequence $(r_n)_n$ with
$$L(f(C(r_n)))\to 0 \quad \text{for $n\to\infty$}.$$
\end{lemma}
\begin{proof} We have
\begin{eqnarray*}
&&L(f(C(r)))^2 = \left(\int_{C(r)} |f'(z)||dz|\right)^2 \underset{\text{Cauchy-Schwarz inequality}}{\leq} \int_{C(r)} |dz| \int_{C(r)} |f'(z)|^2|dz| \\
&\leq&2\pi r
 \int_{c+re^{i\varphi}\in U} |f'(c+re^{i\varphi})|^2 r d\varphi. 
\end{eqnarray*}
Integration with respect to $r$ yields
\begin{eqnarray*}
&&\int_0^\infty L(f(C(r)))^2/r dr = 2\pi \int_0^\infty\int_{c+re^{i\varphi}\in U} |f'(c+re^{i\varphi})|^2 r d\varphi dr=2\pi F(f(U))\leq 2\pi^2 R^2. 
\end{eqnarray*}
Hence 
\begin{eqnarray*}
&&2\pi^2R^2 \geq \int_\rho^{\sqrt{\rho}} L(f(C(r)))^2/r dr \\&\geq&  
\inf_{\rho < r <\sqrt{\rho}} L(f(C(r)))^2 \cdot \int_\rho^{\sqrt{\rho}} 1/r dr 
= \inf_{\rho < r <\sqrt{\rho}} L(f(C(r)))^2 \cdot \frac1{2}\log(1/\rho). 
\end{eqnarray*}

\end{proof}

For the proof of Theorem \ref{loc_con} we need the following topological statements:
\begin{itemize}
	\item[(i)] Continuous extension theorem: If $f:\D\to \C$ is uniformly continuous (for every $\eps>0$ there is $\delta>0$ such that for every $z,w\in\D$ the inequality $|z-w|<\delta$ implies $|f(z)-f(w)|<\eps$), then $f$ has a continuous extension to $\overline{\D}$.
	\item[(ii)] Janiszewski's theorem: Let $A,B\subset \C$ be closed sets, such that $A\cap B$ is connected. If $a,b\in \C$ are neither disconnected by $A$ nor by $B$ (they lie in the same connected component of the complements $\C\setminus A$ and $\C\setminus B$ resp.), then they are not disconnected by $A\cup B$ either. 
\end{itemize}

\begin{proof}[Proof of Theorem \ref{loc_con}]
Assume a), then  $f(\partial\D)$ is a compact, continuous curve, and thus it is locally connected.\\
Next, assume b), i.e.\ $\partial D$ is locally connected. We assume that $f(0)=0$. So there exists $R_0<R$ with
$$B_{R_0}(0)\subset D \subset B_R(0).$$ 
 Let $0<\eps < R_0$ and $0<\delta < \eps$, such that, for $a,b\in \partial D$ with $|a-b|<\delta$, there exists a compact, connected set $B\subset \partial D$  with
 $a,b\in B$ and $\sup_{u,v\in B}|u-v|<\eps.$\\
Next, we choose  $0<\rho<1/4$ such that $2\pi R/(\sqrt{\log(1/\rho)})<\delta.$ Let $z,z'\in \D$ with $1/2<|z|,1/2<|z'|$ and $|z-z'|<\rho.$ Due to Lemma \ref{pomm0} with $U=\D$ and $c=z$, there exists a $\rho < r<1/2$ such that
\begin{equation}\label{help}\tag{$*$}L(f(C))<\delta < \eps, \quad C=\D\cap \partial B_r(z).\end{equation}
We show that this implies
$$|f(z)-f(z')|<2\eps,$$
i.e.\ $f$ is uniformly continuous and consequently has a continuous extension to $\overline{\D}$.\\
Assume that $|f(z)-f(z')|\geq 2\eps$ with $\overline{C}\cap \partial\D \not=\emptyset.$ The curve $\overline{f(C)}$ has two end points $a,b\in \partial D;$ otherwise, the Euclidean length of 
$f(C)$ would be $\infty$. Because of \eqref{help}, we have $|a-b|<\delta.$ Consequently, there exists a compact, connected set $B\subset\partial D$ with $a,b\in B$ and $\sup_{u,v\in B}|u-v|<\eps.$ Hence
$$B\cup f(C)\subset B_\eps(a), \quad 0\not\in B_\eps(a).$$
Because of $|f(z)-f(z')|\geq 2\eps$, the points $0$ and $f(z)$ (or $0$ and $f(z')$) are not disconnected by $B\cup f(C)$. As $0$ and $f(z)$ are not disconnected by $\partial D$ either, and because  $(B\cup f(C))\cap \partial D = B$ is connected,  Janiszewski's theorem implies that  $f(0)=0$ and $f(z)$ are not disconnected by the set $(B\cup f(C))\cup \partial D= f(C)\cup \partial D$. Then, $0$ and $z$ are not disconnected by  $C\cup \partial\D$, a contradiction.
\end{proof} 

Let $E\subset \C$ be connected. Then $p\in E$ is called a cut point of $E$ if  $E\setminus\{p\}$ is not connected.

\begin{theorem}\label{pomm} Let $f:\D \to D$ be conformal, $D$ bounded and $\partial D$ locally connected. Let $a\in \partial D$ and
$$A=f^{-1}(\{a\}), \quad m=|A|\leq \infty.$$
Then  $a$ is a cut point of $\partial D$ if and only if $m>1$. The components of $\partial D\setminus\{a\}$ have the form $f(I_k),$ where $k=1,...,m$ or $k\in\N$, where
$I_k\subset \partial \D$ are open circular arcs.
\end{theorem}

\begin{proof} Note that $f$ can be extended continuously to $\overline{\D}$ by Theorem \ref{loc_con}.  \\
If $m<\infty,$ then the set $\partial\D\setminus A$ consists of $m$ pairwise disjoint open circular arcs $I_k$.\\ 
Now let $m=\infty$. 
Then the set $\partial\D\setminus A$ is open in $\partial\D$ and thus it is a countable union of disjoint open circular arcs $I_k$, $k\in\N$.
So
$$\partial D\setminus\{a\} = f(\partial\D\setminus A) = \bigcup_{k=1}^m f(I_k)\quad \text{if $m<\infty$}, \qquad 
\partial D\setminus\{a\} = \bigcup_{k\in\N} f(I_k)\quad \text{if $m=\infty$}.$$
As $I_k$ is connected, also $f(I_k)$ is connected. \\
Let $m>1.$ Then $f(I_j),f(I_k),$ $j\not=k,$ are not connected within $\partial D\setminus\{a\}$:\\
Let $C$ be a curve in $\D,$ that connects the end points of $I_k$. Then $f(C)\cup \{a\}$ is a closed Jordan curve in $D\cup\{a\}$. 
Let $C'$ be a curve in $\D$ with end points in $I_j$ and $I_k$ such that $C$ and $C'$ intersect exactly once. Then $f(C)$ and $f(C')$ intersect exactly once. So  $f(I_j)$ and $f(I_k)$ lie in different components of 
$\C\setminus f(C),$ and consequently they are not connected in $\partial D\setminus\{a\}.$\\
If $m=1,$ then $\partial D\setminus\{a\}=f(\partial \D\setminus\{p\})$ for some $p\in\partial\D.$ Hence $\partial D\setminus\{a\}$ is connected and $a$ is not a cut point of $\partial D$. 
\end{proof}

\begin{theorem}\label{fortsetzen} Let $f:\D \to D$ be a conformal mapping onto a bounded domain $D.$ Then the following statements are equivalent:
\begin{itemize}
	\item[a)] $f$ can be extended to a homeomorphism on $\overline{\D}$.
	\item[b)] $\partial D$ is a Jordan curve.
\end{itemize}
\end{theorem}
\begin{proof} The direction $a) \Longrightarrow b)$ is obvious.\\
Let $\partial D$ be a Jordan curve. Then $f$ can be extended continuously to $\overline{\D}$ because of Theorem \ref{loc_con}. Each $a\in \partial D$ is not a cut point of $\partial D,$ so $|f^{-1}(\{a\})|=1$ due to Theorem \ref{pomm}, i.e.\ $f$ is injective on $\overline{\D}.$
\end{proof}

\begin{remark} In fact, every conformal mapping $f:\D\to D,$ $D$ bounded, can be extended to a homeomorphism on $\overline{\D}$. Here,  $f:\partial\D \to P(D)$ maps the boundary $\partial \D$ onto the set $P(D)$ of all \emph{prime ends} of $D$.\\
A prime end corresponds to a set $M\subset \partial D$ of points and an ``approach direction''. Example: If $f:\D\to \D\setminus [0,1]$ is conformal, then there are two points $a,b\in\partial \D$ with $f(a)=f(b)=1/2.$ One of these points corresponds to the prime end ``$1/2$ from above'' and the other one to ``$1/2$ from below''.     
\end{remark}

\chapter{The Bieberbach conjecture}\label{sec_Bie}

A holomorphic function is locally injective if and only if it's derivative is not $0$. Is it possible to describe the global injectivity of a holomorphic function analytically? We define the class $S$ of all injective holomorphic mappings on $\D$ with normalization of $f(0)$ and $f'(0):$
$$S=\{f:\D\to\C \,|\, \text{$f$ is injective and holomorphic, $f(0)=0, f'(0)=1$}\}.$$
The property ``injective and holomorphic'' is also called \emph{univalent} or \emph{\underline{s}chlicht}.
If $g:\D\to \C$ is univalent, then $(g(z)-g(0))/g'(0)\in S.$

\begin{example} The identity $z\mapsto z$ and the Koebe function $k(z)=\frac{z}{(1-z)^2}$ belong to $S$. The Koebe function maps $\D$ conformally onto $\C\setminus (-\infty,-\frac{1}{4}]$, which can be seen by writing \[k(z)=\frac1{4}\left(\left(\frac{1+z}{1-z}\right)^2-1\right).\] \hfill $\blacksquare$
\end{example}

\begin{example}\label{sqrt_tr} Let $f(z)=z+ \sum_{n=2}^\infty a_nz^n\in S.$ Then there are several transformations generating new functions in $S.$ For example the rotation:\\
For $\alpha\in \R$, we have $$e^{-i\alpha}f(e^{i\alpha}z) = z + \sum_{n=2}^\infty a_n e^{i(n-1)\alpha} z^n \in S.$$ 
The square root transform $\sqrt{f(z^2)}\in S:$\\
Here, we define $\sqrt{\cdot}$ by writing
$$\sqrt{f(z^2)}=g(z):=z \sqrt{1 + a_2z^2 + a_3z^4 + ...}$$ (The function $z\mapsto 1 + a_2z^2 + a_3z^4 + ...$ has no zeros in $\D$. Otherwise, $f$ would have at least two zeros.) Then  $$\sqrt{f(z^2)}=z + \frac{a_2}{2} z^3 + c_5z^5+...,$$ is an odd mapping. It is univalent, for $g(z_1)=g(z_2)$ implies $f(z_1^2)=f(z_2^2),$ thus $z_1=\pm z_2$. If $z_1=-z_2,$ then $g(z_2)=-g(z_2),$ thus $g(z_2)=0,$ i.e.\ $z_2=z_1=0.$\hfill $\blacksquare$
\end{example}

In order to understand $S$ better, we first look at the class $\Sigma$ of all univalent mappings $g$ defined in $\C\setminus\overline{\D}$ with $g(z)=z+\sum_{n=0}^\infty b_n z^{-n}.$ If $f\in S,$ then
\begin{equation}\label{SSigma} \frac1{f(1/z)} = \frac{z}{1+a_2 z^{-1} +...} = z - a_2 + (a_2^2-a_3)z^{-1}+... \in \Sigma.
\end{equation}   

\begin{theorem}[Area theorem] Let $g\in \Sigma$ and $E=\C\setminus g(\C\setminus \overline{\D}).$ Then
$$Area(E) = \pi\left(1-\sum_{n=1}^\infty n|b_n|^{2}\right).$$
\end{theorem}
\begin{proof}
For $r>1$ we let $E(r) = \C \setminus g(\{z\in \C\,|\, |z|\geq r\})$ and $C_r=\{z\in \C\,|\, |z|= r\}.$ Then $E(r)$ is bounded by the Jordan curve $g(C_r)$. Green's theorem and the observation $\overline{w}dw = (x-iy)(dx+idy)=xdx+ydy+ixdy-iydx$ imply \begin{eqnarray*}&&Area(E(r)) \underset{\text{Green's theorem}}{=} \frac1{2}\int_{g(C_r)} xdy-ydx = \frac1{2i}\int_{g(C_r)} ixdy-iydx\\
&=& \frac1{2i}\int_{g(C_r)} xdx +ydy+ ixdy-iydx=\frac1{2i}\int_{g(C_r)} \overline{w}dw \\
&=& \frac1{2i}\int_{C_r} \overline{g(z)}g'(z)dz=\frac1{2i}\int_0^{2\pi} ire^{it}\overline{g(re^{it})}g'(re^{it})dt\\
&=& \frac1{2}\int_0^{2\pi} (r e^{-it}+\sum_{n=0}^\infty \overline{b_n} r^{-n}e^{int})
(re^{it}-\sum_{n=1}^\infty n b_n r^{-n}e^{-int})dt =\pi( r^2 - \sum_{n=1}^\infty n |b_n|^2 r^{-2n}).
\end{eqnarray*}
As $Area(E(r))\geq0,$ we have $ \sum_{n=1}^N n |b_n|^2 r^{-2n}\leq r^2$. Hence $ \sum_{n=1}^N n |b_n|^2 \leq 1$ for every $N\in\N.$ Consequently, the sum $\sum_{n=1}^\infty n |b_n|^2$ converges and we can take the limit $r\to 1$ to obtain
$$Area(E) = \pi\left( 1 - \sum_{n=1}^\infty n |b_n|^2\right).$$
\end{proof}

\begin{corollary}\label{Palazzone} Let $g(z)=z+\sum_{n=0}^\infty b_n z^{-n}\in \Sigma.$ Then  $$\sum_{n=1}^\infty n|b_n|^2 \leq 1.$$
In particular, $|b_1|\leq 1,$ and equality holds if and only if $g(z)=z+b_0 +a/z$ for some $a\in\partial \D.$
\end{corollary}
\begin{proof} Follows directly from the area theorem, as $Area(E)\geq0.$
\end{proof}

\begin{corollary}\label{body} Let $f(z)=z+\sum_{n=2}^\infty a_nz^{n}\in S.$ Then 
$$|a_2^2-a_3|\leq 1.$$
\end{corollary}
\begin{proof} Follows directly from Corollary \ref{Palazzone} and \eqref{SSigma}.
\end{proof}

In 1916, L.\ Bieberbach proved the following inequality.

\begin{corollary}[Bieberbach's theorem] Let $f(z)=z+\sum_{n=2}^\infty a_nz^{n}\in S.$ Then
$$|a_2| \leq 2 $$
and equality holds if and only if $f(z)= \frac{z}{(1-e^{i\alpha}z)^2}=e^{-i\alpha}k(e^{i\alpha}z),$ $\alpha\in \R.$
\end{corollary}
\begin{proof}According to Example \ref{sqrt_tr}, the mapping $\sqrt{f(z^2)}$ is in $S,$ and with \eqref{SSigma}, the mapping $(f(1/z^2))^{-1/2}$ is in $\Sigma.$ As $$(f(1/z^2))^{-1/2}=z-\frac{a_2}{2}z^{-1}+...,$$ it follows from Corollary \ref{Palazzone} that $|a_2|\leq 2.$
Equality holds if and only if  $g(z) = z - e^{i\alpha}z^{-1},$  $\alpha\in \R,$ which leads to  $f(z)= \frac{z}{(1-e^{i\alpha}z)^2}$.
\end{proof}

Bieberbach then stated the conjecture that $z+\sum_{n=2}^\infty a_n z^n \in S$ should imply $$|a_n|\leq n \quad \text{for all $n\geq 2.$}$$ Note that the Koebe function $k(z)=\sum_{n=1}^\infty nz^n$ shows that there is no lower bound for $|a_n|$. \\

This conjecture had a strong impact on the development of complex analysis and until the complete proof there have been many partial results: 
\begin{itemize}
\item 1916: $|a_2|\leq 2$ (Bieberbach), 
\item 1917: $|a_n|\leq 1$ for all $f\in S$ whose image domain is convex (Loewner),
\item 1921: $|a_n|\leq n$ for all $f\in S$ whose image domain is starlike w.r.t. 0 (Nevanlinna),
\item 1923: $|a_3|\leq 3$ (Loewner),
\item 1925: $|a_n|< en$ (Littlewood),
\item 1955: $|a_4|\leq 4$ (Garabedian and Schiffer), 
\item 1965: $|a_n| < 1.243 \cdot n$ (Milin)
\item 1968: $|a_6|\leq 6$ (Pederson, Ozawa),
\item 1972:  $|a_n| < \sqrt{7/6} \cdot n = 1.0801... \cdot n$ (FitzGerald)
\item 1972: $|a_5|\leq 5$ (Pederson and Schiffer)
\item 1984: $|a_n|\leq n$ for all $n\geq 2$ and all $f\in S$ (de Branges);
\end{itemize}
 see \cite{Koepf} for the historical development.  Loewner proved his result in \cite{Loewner:1923} by introducing a new dynamical description of univalent functions (Loewner chains \& Loewner's differential equation). Note that the title of \cite{Loewner:1923} ends with an ``I'', expressing the hope that the new method might soon have solved the complete conjecture, but there has never been a successive ``II''. 
The Bieberbach conjecture has finally been proven in 1985 by Louis de Branges, \cite{dBr85}, and indeed, this final proof also uses Loewner's method. 

\end{appendices}

\newpage


\begin{thebibliography}{BCDMG15}

\bibitem[Aba89]{Aba89} M.\ Abate, \emph{Iteration theory of holomorphic maps on taut manifolds},  Mediterranean Press, 1989.

\bibitem[ABCD10]{AbateBracci:2010} M.\ Abate, F.\ Bracci, M.\ D.\ Contreras, S.\ D{\'{\i}}az-Madrigal, \emph{The evolution of {L}oewner's differential equations}, Eur. Math. Soc. Newsl. 78 (2010), 31--38.

\bibitem[AGO04]{acc} L. Accardi, A. Ben Ghorbal, N. Obata, 
\emph{Monotone independence, comb graphs and Bose-Einstein condensation},
Infin. Dimens. Anal. Quantum Probab. Relat. Top. 7 (2004), 419--435.

\bibitem[ALS07]{acc2} L.\ Accardi, R.\ Lenczewski, R.\ Salapata, \emph{Decompositions of the free product of graphs}, 
Infin. Dimens. Anal. Quantum Probab. Relat. Top. 10 (2007), 303--334.

\bibitem[ABKL05]{one} D.\ Applebaum, B.V.R.\ Bhat, J.\ Kustermans, J.M.\ Lindsay, \emph{Quantum independent increment processes I}, Lecture Notes in Mathematics 1865, Springer, Berlin, 2005.

\bibitem[AN06]{AN06} S.\ Attal, I.\ Nechita, \emph{Discrete Approximation of the Free Fock Space}, in: S\'{e}minaire de Probabilit\'{e}s XLIII, 
    C.\ Donati-Martin, A.\ Lejay, A.\ Rouault, Springer, 2006, 379--394.

\bibitem[Att]{Att} S.\ Attal, \emph{Lectures in Quantum Noise Theory}, \\
http://math.univ-lyon1.fr/homes-www/attal/chapters.html (access date: 01 Dec. 2017).

\bibitem[BS20]{BS20} M.\ T.\ Barlow, G.\ Slade, \emph{Random Graphs, Phase Transitions, and the Gaussian Free Field}, 
Springer, 2020.

\bibitem[BNMR01]{BNMR01} O.\ E.\ Barndorff-Nielsen, T.\ Mikosch, S.\ I.\ Resnick, \emph{L\'{e}vy Processes,
Theory and Applications}, Birkh\"uaser, 2001.

\bibitem[BN+al06]{barndorff-nielsen+al}
O.E.\ Barndorff-Nielsen, U.\ Franz, R.\ Gohm, B.\ K\"ummerer, S. Thorbj{\o}rnsen, \emph{Quantum independent increment processes II}, Lecture Notes in Mathematics 1866, Springer, 2006.

\bibitem[Bah19]{Bah19} R.\ Bahumi, \emph{Reinforcement Learning - A mathematical introduction to Policy Gradient},\\ {\small http://machinelearningmechanic.com/deep\_learning/reinforcement\_learning/}
{\small 2019/12/06/a\_mathematical\_introduction\_to\_policy\_gradient.html},\\ 2019 (access date: 22 May 2020).

\bibitem[BN10]{BN10} J.\ Bak, D.J.\ Newman, \emph{Complex Analysis}, Springer, 2010.

\bibitem[BBA14]{BBA14} J.\ Barry, D.T.\ Barry, S.\ Aaronson, \emph{Quantum POMDPs}, Phys. Rev. A 90 (2014).

\bibitem[Bax89]{Bax89} R.\ J.\ Baxter, \emph{Exactly Solved Models in Statistical Mechanics}, Academic Press, 1989.

\bibitem[Bea97]{Bea97}A.\ F.\ Beardon, \emph{The Schwarz-Pick Lemma for derivatives}, Proc. Amer. Math. Soc. 125 (1997), 3255--3256.

\bibitem[Bel05]{Bel05} S.\ T.\ Belinschi, \emph{Complex analysis methods in noncommutative probability}, Doctoral Dissertation, Indiana University, 2005. (arXiv:0602343) 

\bibitem[BBLS11]{BBLS11} S.\ T.\ Belinschi,  M.\ Bo$\dot{\text{z}}$ejko, F.\ Lehner, R.\ Speicher, \emph{The normal distribution is $\boxplus$-infinitely divisible}, Adv. Math. 226 (2011), 3677--3698.

\bibitem[BGS99]{BGS99} A.\ Ben Ghorbal, M.\ Sch \"urmann, \emph{On the algebraic foundations of non-commutative probability theory}, preprint, Nancy, 1999.

 \bibitem[BV93]{BV93} H.\ Bercovici, D.\ Voiculescu, \emph{Free convolution of measures with unbounded support}, Indiana Univ. Math. J. 42 (1993), 733--773.

\bibitem[BP00]{BP00} H.\ Bercovici, V.\ Pata, \emph{A free analogue of Hin\v{c}in's characterization of 
infinite divisibility}, Proc. Amer. Math. Soc. 128 (2000), 1011--1015.

\bibitem[BP78]{MR0480965} E.\ Berkson, H.\ Porta, \emph{Semigroups of analytic functions and composition operators}, Michigan Math. J. 25 (1978), 101--115.

\bibitem[Bia98]{Bia98} P.\ Biane, \emph{Processes with free increments}, Math. Z. 227 (1998), 143--174.

\bibitem[Bil95]{bill} P.\ Billingsley, \emph{Probability and Measure}, Wiley Series in Probability and Statistics, 3rd edition,  1995.


\bibitem[Bil99]{Bil99} P.\ Billingsley, \emph{Convergence of Probability Measures}, Wiley Series in Probability and Statistics, 1999.

\bibitem[BS00]{BS00} T.M.\ Bisgaard, Z.\ Sasv\'{a}ri, \emph{Characteristic Functions and Moment Sequences: Positive Definiteness in Probability},  Nova Science Publishers, 2000.


\bibitem[Bla06]{Bla06} B.\ Blackadar, \emph{Operator Algebras}, Springer, 2006.

\bibitem[BCDM12]{MR2995431} F.\ Bracci, M.D.\ Contreras, S.\ D{\'{\i}}az-Madrigal, \emph{Evolution families and the {L}oewner equation {I}: the unit disk},  J. Reine Angew. Math. 672 (2012), 1--37.

\bibitem[BCDMG15]{bracci+al2015} F.\ Bracci, M.D.\ Contreras, S.\ D\'iaz-Madrigal, P.\ Gumenyuk, 
\emph{Boundary regular fixed points in Loewner theory}, Annali di Matematica Pura ed Applicata 196 (2015), 221--142.

\bibitem[BGJM11]{BGJM11} S.\ Brooks, A.\ Gelman, G.\ Jones, X.-L.\ Meng, \emph{Handbook of Markov Chain Monte Carlo}, 
Chapman \& Hall/CRC, 2011.

\bibitem[Cau32]{cau32} W.\ Cauer, \emph{The Poisson integral for functions with positive real part}, Bull. Amer. Math. Soc. 38 (1932), 
713--717. 

\bibitem[Cho75]{Cho75} M.\ D.\ Choi, \emph{Completely positive linear maps on complex matrices}, Linear Algebra and Its Applications 10 (1975), 285--290.

\bibitem[CK12]{CK12} T. Christ, J. Kalpokas, \emph{Upper bounds for discrete moments of the derivatives of the Riemann zeta-function on the critical line}, Lith. Math. Journal 52 (2012),  233--248.


\bibitem[CD17]{CD17}  P.\ Constantinou, A.\ P.\ Dawid, 
\emph{Extended conditional independence and applications in causal inference},  Ann. Statist.   45 (2017), 2618--2653.


\bibitem[CDMP06]{CDMP06} M.\ D.\ Contreras, S. D\'iaz-Madrigal, C. Pommerenke, \emph{On  boundary  critical  points  for  semigroups of analytic functions}, Math. Scand. 98 (2006), 125--142.

\bibitem[CDMG10]{contreres+al2010} M.\ D.\ Contreras, S.\ D{\'{\i}}az-Madrigal, P.\ Gumenyuk, \emph{Loewner chains in the unit disk}, Rev. Mat. Iberoamericana 26 (2010), 975--1012.

\bibitem[CDMG14]{contreres+al2014} M.\ D.\ Contreras, S.\ D{\'{\i}}az-Madrigal, P.\ Gumenyuk, 
\emph{Local duality in {L}oewner equations}, J. Nonlinear Convex Anal. 15 (2014), 269--297.

\bibitem[Con84]{Con84} J.\ B.\ Conway, \emph{Functions of
One Complex Variable}, second edition, Springer, 1984.

\bibitem[Con94]{Con94} J.\ B.\ Conway, \emph{A Course in Functional Analysis}, Springer, 1994.


\bibitem[CDH12]{CDH12} O. Curtef, G. Dirr, Uwe Helmke,  \emph{Riemannian Optimization on Tensor Products of Grassmann Manifolds: Applications to Generalized Rayleigh-Quotients}, SIAM J. Matrix Anal. Appl., 33(1) (2012), 210--234.

\bibitem[dBr85]{dBr85} L.\ de Branges, \emph{A proof of the Bieberbach conjecture}, Acta Mathematica 154 (1985), 137--152.

\bibitem[dGL97]{MR1483010}M.\ De~Giosa, Y.G.\ Lu, \emph{The free creation and annihilation operators
  as the central limit of the quantum {B}ernoulli process}, Random Oper.  Stochastic Equations 5 (1997), 227--236.

\bibitem[dMG16]{dMG16} A.\ del Monaco, P.\ Gumenyuk, \emph{Chordal Loewner equation}, Complex Analysis and Dymamical Systems VI Part 
2: Complex Analysis, Quasiconformal Mappings, Complex Dymamics, Contemp. Math. 667 (2016), 63--77.

\bibitem[dMS16]{delMonacoSchleissinger:2016}
A.\ del Monaco, S.\ Schlei{\ss}inger, \emph{Multiple sle and the complex burgers
  equation}, Math. Nachr. 289 (2016), 2007--2018. 

\bibitem[dMHS18]{dMHS18} A.\ del Monaco, I.\ Hotta, S.\ Schlei{\ss}inger, \emph{Tightness results for infinite-slit limits of the chordal Loewner equation}, Comput. Methods Funct. Theory 18 (2018), 9--33. 

\bibitem[Die31]{Die31} J.\ Dieudonn\'{e}, \emph{Recherches sur quelques probl\`{e}mes relatifs aux polyn\^{o}mes et aux fonctions born\'{e}es d'une variable complexe}, Ann. Sci. Ecole Norm. Sup. 48 (1931), 247--358.

\bibitem[DB18]{DB18}  V.\ Dunjko, H.J. Briegel, \emph{Machine learning \& artificial intelligence in the quantum domain : a review of recent progress},	Reports on Progress in Physics 81 (2018).

\bibitem[Dur83]{Dur83} P.L.\ Duren, \emph{Univalent functions}, volume 259 of Grundlehren der
  Mathematischen Wissenschaften [Fundamental Principles of Mathematical
  Sciences], Springer, 1983.  
	
\bibitem[ER06]{ER06} A.\ Edelman, N.R.\ Rao, \emph{Free Probability, Sample Covariance Matrices, and Signal Processing}, 
2006 IEEE International Conference on Acoustics Speech and Signal Processing 
Proceedings. 
	
\bibitem[Esp15]{Esp15} C. Esposito, \emph{Formality Theory: From Poisson Structures to Deformation Quantization},  Sringer, 2015.
	
	
\bibitem[Fra03]{franz-unif} U.\ Franz, \emph{Unification of Boolean, monotone, anti-monotone, 
and tensor independence and L\'evy processes}, Math. Z. 243 (2003), 779--816.

\bibitem[Fra06]{franz-mult} U.\ Franz, \emph{Multiplicative monotone convolutions}, Banach Center Publ., 73, Polish Acad. Sci. Inst. Math. (2006), 153--166.

\bibitem[Fra09a]{franz07b} U.\ Franz, \emph{Monotone and Boolean convolutions for non-compactly supported
  probability measures}, Indiana Univ. Math. J. 58 (2009), 1151--1186.	
	
\bibitem[Fra09b]{franz_boolean} U.\ Franz, \emph{Boolean convolution of probability measures on the unit circle}, 
Analyse et probabilit\'{e}s, S\'{e}minaires et Congr\`{e}s 16 (2009), 83--93.	
	
\bibitem[FM05]{franz+muraki04} U.\ Franz, N.\ Muraki, \emph{Markov structure on monotone {L}\'evy processes}, 
in: Infinite dimensional harmonic analysis III, Herbert Heyer, Takeshi Hirai, Takeshi Kawazoe, Kimiaki Sait\^o (eds.), World Sci.\ Publ., Hackensack, NJ, 2005, 37--57.
	
\bibitem[FHS20]{FHS} U.\ Franz, T.\ Hasebe, S.\ Schlei{\ss}inger,  \emph{Monotone Increment Processes, Classical Markov Processes, and Loewner Chains},  Dissertationes Mathematicae 552 (2020). 

	\bibitem[GK54]{GK54} B.V.\ Gnedenko, A.N.\ Kolmogorov, \emph{Limit Distributions for Sums of Independent Random Variables}, 
Addison-Wesley Publishing Company, Inc., 1954.

\bibitem[GB92]{MR1201130} V.V.\ Goryainov, I.\ Ba, \emph{Semigroup of conformal mappings of the upper
  half-plane into itself with hydrodynamic normalization at infinity}, Ukrain. Mat. Zh. 44 (1992), 1320--1329.

\bibitem[GHKK14]{ghkk} I.\ Graham, H.\ Hamada, G.\ Kohr, M.\ Kohr, 
\emph{Extremal properties associated with univalent subordination chains in $\C^n$}, 
Math. Ann. 359 (2014), 61--99. 

\bibitem[GFY18]{GFY18} J.\ Guan, Y.\ Feng, M.\ Ying, \emph{Decomposition of quantum Markov chains and its applications},  Journal of Computer and System Sciences 95 (2018), 55--68.

\bibitem[Hal80]{Hal80} J.\ K.\ Hale,  \emph{Ordinary Differential Equations}, second edition, Robert E. Krieger Publishing Company, 1980.

\bibitem[Has11]{H11} T.\ Hasebe, \emph{Conditionally monotone independence I: 
Independence, additive convolutions and related convolutions}, Infin. Dimens. Anal. Quantum Probab. Relat. Top. 14 (2011), 465--516.

\bibitem[Has13]{Has13} T.\ Hasebe, \emph{Conditionally monotone independence II: Multiplicative convolutions and infinite divisibility}, Compl. Anal. Oper. Theory 7 (2013), 115--134.

 \bibitem[HS11]{HS11} T.\ Hasebe, H.\ Saigo, \emph{The monotone cumulants}, Ann. Inst. Henri Poincar\'{e} Probab. Stat. 
 47 (2011), 1160--1170.

 \bibitem[Has70]{Has70} W.K.\ Hastings, \emph{Monte Carlo Sampling Methods Using Markov Chains and Their Applications}, 
Biometrika 57 (1970), 97--109.

\bibitem[HS13]{HoSm} C.\ Hongler, S.\ Smirnov, \emph{The energy density in the planar {I}sing model},
  Acta Math., Vol. 211 (2013), 191--225.

\bibitem[HO07]{hora} A.\ Hora, N.\ Obata, \emph{ Quantum Probability and Spectral Analysis
of Graphs}, Theoretical and Mathematical Physics, Springer, 2007.  

\bibitem[HK18]{HK18} I.\ Hotta, M.\ Katori, \emph{Hydrodynamic Limit of Multiple SLE}, J. Stat. Phys. 171 (2018), 166--188. 


\bibitem[HS18]{HS18} I.\ Hotta, S.\ Schlei{\ss}inger, \emph{Problems related to slit-mappings}, arXiv:1811.12046.

\bibitem[HS20]{HS19} I.\ Hotta, S.\ Schlei{\ss}inger, \emph{Limits of radial multiple SLE and a Burgers-Loewner differential equation}, 
 	Journal of Theoretical Probability (2020). 


\bibitem[HSS]{HSS} I.\ Hotta, T.\ Sugawa, S.\ Schlei{\ss}inger, \emph{Nonlinear resolvents and decreasing Loewner chains}, preprint.

\bibitem [IO06]{io} D.\ Igarashi, N.\ Obata, 
\emph{Asymptotic spectral analysis of growing graphs: odd graphs and spidernets}, 
Banach Center Publications  73 (2006), 245--265. 

\bibitem[Jek20]{Jek17}D.\ Jekel, \emph{Operator-valued chordal Loewner chains and non-commutative probability},
Journal of Functional Analysis (2020).

\bibitem[KVV14]{KVV14} D.S.\ Kaliuzhnyi-Verbovetskyi, V.\ Vinnikov, \emph{Foundations of Free Noncommutative Function Theory}, 
  American Mathematical Society, 2014.

\bibitem[Kal02]{kall} O. Kallenberg, \emph{Foundations of modern probability}, second edition, Springer, 2002.

\bibitem[KM13]{KM13} N.-G.\ Kang, N.\ G.\ Makarov, \emph{Gaussian Free Field and Conformal Field Theory}, American Mathematical Society, 2013.

\bibitem[Kar19]{Karrila} A. Karrila, \emph{Multiple SLE type scaling limits: from local to global}, 
arXiv:1903.10354.

\bibitem[Kem03]{Kem03} J.\ Kempe, \emph{Quantum random walks: An introductory overview}, Contemporary Physics 44 (2003), 307--327. 

\bibitem[Kob00]{Kobe} S. Kobe, \emph{Ernst {I}sing 1900-1998}, Brazilian Journal of Physics, vol. 30 (2000), 649--653.

\bibitem[Koe07]{Koepf} W.\ Koepf, \emph{Bieberbach's conjecture, the de Branges and Weinstein functions and the Askey-Gasper inequality}, The Ramanujan Journal 13 (2007), 103--129.

\bibitem[Kuf47]{Kuf47} P.P.\ Kufarev, \emph{A remark on integrals of Loewner's equation}, Doklady Akad. Nauk SSSR57 (1947), 655--656.

\bibitem[Kup16]{Kup16}  A.\ Kupiainen, \emph{Quantum Fields and Probability}, arXiv:1611.05240.


\bibitem[LLN09]{LLN09} S.\ Lalley, G.\ Lawler, H.\ Narayanan, \emph{Geometric Interpretation of Half-Plane Capacity}, 
 Elect. Comm. in Probab. 14 (2009), 566--571.

\bibitem[Law05]{lawler05} G.~F. Lawler, \emph{Conformally invariant processes in the plane},  Mathematical Surveys and Monographs 114, American Mathematical Society, 2005.

\bibitem[Len11]{Len11} R.\ Lenczewski, \emph{Asymptotic properties of random matrices and pseudomatrices}, Advances in Mathematics 228 (2011), 2403--2440.

\bibitem[Len15]{Len15} R.\ Lenczewski, \emph{Limit Distributions of Gaussian Block Ensembles},  Acta Phys. Polon. B, Vol. 46 (2015), 1833--1850.

\bibitem[Len19]{Len19} R.\ Lenczewski, \emph{Conditionally monotone independence and the associated products of graphs}, 
Infin. Dimens. Anal. Quantum Probab. Relat. Top. 22, (2019).

\bibitem[LPW09]{LPW09} D.\ A.\ Levin, Y.\ Peres, E.\ L. Wilmer, \emph{Markov Chains and Mixing Times}, American Mathematical Society, 2009. 

\bibitem[Lin05]{Lind:2005}
J.~R. Lind, {\em A sharp condition for the {L}oewner equation to generate slits},
  Ann. Acad. Sci. Fenn. Math. 30 (2005), 143--158.
	
\bibitem[LMR10]{LMR10} J.\ Lind, D.E.\ Marshall,  S.\ Rohde, \emph{Collisions and spirals of Loewner traces}, 
Duke Math. J. 154 (2010), 527--573.

\bibitem[LR12]{LR12} J.\ Lind, S.\ Rohde, \emph{Space--filling curves and phases of the Loewner equation}, 
Indiana Univ. Math. J. 61 (2012), 2231--2249.

\bibitem[LQ19]{LQ19} 	Z.\ Ling, R.C.\ Qiu,  \emph{Spectrum Concentration in Deep Residual Learning: A Free Probability Approach}, IEEE Access 7 (2019), 105212--105223.
	
\bibitem[L{\"o}w23]{Loewner:1923} K.\ L{\"o}wner, \emph{Untersuchungen {\"u}ber schlichte konforme Abbildungen des
  Einheitskreises. I.}, Mathematische Annalen 89 (1923), 103--121.	
	
\bibitem[Lu97]{MR1455615} Y.G.\ Lu, \emph{An interacting free {F}ock space and the arcsine law}, Probab.  Math. Statist. 17 (1997), Acta Univ. Wratislav., 149--166.	
	
\bibitem[Maa92]{M92} H.\ Maassen, \emph{Addition of freely independent random variables}, J.Funct. Anal. 106 (1992), 409--438.

\bibitem[MR05]{MarshallRohde:2005}
D.\ E. Marshall, S.\ Rohde,  \emph{The {L}oewner differential equation and slit
  mappings}, J. Amer. Math. Soc. 18 (2005), 763--778.
	
\bibitem[McL03]{McL03} D.\ McLeish, Lecture notes on probability, University of Waterloo, Canada, 2003.
	
\bibitem[MRRTT53]{MRRTT53} N.\ Metropolis, A.\ Rosenbluth, M.\ Rosenbluth, A.\ Teller, E.\ Teller: \emph{Equation of State Calculations by Fast Computing Machines}, Journal of Chemical Physics 21 (1953), 1087--1092.
	
\bibitem[Mey95]{Mey95} P.A.\ Meyer, \emph{Quantum Probability for Probabilists}, Springer, 1995.


	
\bibitem[Moh82]{Moh82} B.\ Mohar, \emph{The spectrum of an infinite graph}, Linear Algebra Appl. 48 (1982), 245--256. 

\bibitem[MW89]{mw89} B.\ Mohar, W.\ Woess, \emph{A survey on spectra of infinite graphs},
   Bull. London Math. Soc. {\bf 21}(3) (1989), 209--234.
	
\bibitem[MP10]{MP10} P.\ M\"{o}rters, Y.\ Peres, \emph{Brownian motion},  Cambridge Series in Statistical  and  Probabilistic  Mathematics,  Cambridge  University  Press, 2010. 
	

	
\bibitem[Muk11]{Muk11} P. Mukhopadhyay, \emph{An Introduction to the Theory of Probability}, World Scientific Publishing Company, 2011. 

\bibitem[Mur96]{MR1467953} N.\ Muraki, \emph{A new example of noncommutative ``de {M}oivre-{L}aplace theorem.''}, 7th Japan-Russia symposium, Probability theory and mathematical statistics ({T}okyo, 1995), World Sci. Publ., River Edge, NJ (1996), 353--362.

\bibitem[Mur97]{MR1462227} N.\ Muraki, \emph{Noncommutative {B}rownian motion in monotone {F}ock space}, Comm. Math. Phys. 183 (1997), 557--570.

\bibitem[Mur00]{Mur00} N.\ Muraki, \emph{Monotonic convolution and monotone L\'{e}vy-Hin\v{c}in formula}, Preprint 2000.

\bibitem[Mur01a]{M01} N.\ Muraki, \emph{Towards ``monotonic probability''}, S\=urikaisekikenky\=usho
  K\=oky\=uroku (1186) (2001), in: Topics in information sciences and
  applied functional analysis, Kyoto, 2000, 28--35 (in Japanese).

\bibitem[Mur01b]{Mur01} N.\ Muraki, \emph{Monotonic independence, monotonic central limit theorem and monotonic law of small numbers}, 
 Infin. Dimens. Anal. Quantum Probab. Relat. Top. 4 (2001), 39--58.

\bibitem[Mur03]{MR2016316} N.\ Muraki, \emph{The five independences as natural products}, Infin. Dimens. Anal. Quantum Probab. Relat. Top. 6 (2003), 337--371.

\bibitem[MS17]{MS17}J.\ A.\ Mingo, R.\ Speicher, \emph{Free Probability and Random Matrices}, 2017. 

\bibitem[NS10]{NS10}
A.\ Nica, R.\ Speicher, \emph{Lectures on the Combinatorics of Free Probability}, Cambridge University Press, 2010.

\bibitem[Oba17]{Oba17} N. Obata, \emph{Spectral analysis of growing graphs 
(A quantum probability point of view)}, Springer, 2017.

\bibitem[{\O}ks03]{Oks03} B.\ {\O}ksendal, \emph{Stochastic Differential Equations: An Introduction with Applications}, Springer, 2003. 

\bibitem[Par08]{Par08} \'{E}.\ Pardoux, \emph{Markov Processes and Applications: Algorithms, Networks, Genome and Finance}, Wiley, 2008.


\bibitem[PSG17]{PSG17} J.\ Pennington, S.\ Schoenholz, S.\ Ganguli, \emph{Resurrecting the sigmoid in deep learning through dynamical isometry: theory and practice}, Advances in neural information processing systems (2017), 4785--4795.

\bibitem[Pom75]{P75}C.\ Pommerenke, \emph{Univalent functions}, Vandenhoeck \& Ruprecht,  1975.


\bibitem[Pom92]{MR1217706}
C.\ Pommerenke, \emph{Boundary behaviour of conformal maps}, Grundlehren der Mathematischen Wissenschaften, Springer, 1992.

\bibitem[Put05]{Put05} M.\ L.\ Puterman, \emph{Markov Decision Processes: Discrete Stochastic Dynamic Programming}, Wiley, 2005.
  
\bibitem[RS96]{RS96} S.\ Reich, D.\ Shoikhet, \emph{Generation theory for semigroups of holomorphic mappings in Banach spaces}, Abstr. Appl. Anal. 1 (1996), 1--44. 	

\bibitem[RS97]{RS97}  S.\ Reich, D.\ Shoikhet, \emph{Semigroups and generators on convex domains with the hyperbolic metric}, 
Atti della Accademia Nazionale dei Lincei. Classe di Scienze Fisiche, Matematiche e Naturali. Rendiconti Lincei. Matematica e Applicazioni, Serie 9, Vol. 8 (1997), 231--250.

\bibitem[Rud62]{Rud62} W.\ Rudin, \emph{Fourier Analysis on Groups}, Wiley, 1962.

\bibitem[Sat99]{Sat99} K.\ Sato, \emph{L\'{e}vy Processes and Infinitely Divisible Distributions}, Cambridge University Press,  1999.

\bibitem[Sch17]{monotone_schl} S.\ Schlei{\ss}inger, \emph{The Chordal Loewner Equation and Monotone Probability Theory}, Infin. Dimens. Anal. Quantum Probab. Relat. Top. 20 (2017), 1750016-1--1750016-17.

\bibitem[Sch18]{sch0} S.\ Schlei{\ss}inger, \emph{Loewner's Differential Equation and Spidernets}, 
Complex Analysis and Operator Theory 13 (2018), 3899--3921. 

\bibitem[Sch00]{MR1776084} O.\ Schramm, \emph{Scaling limits of loop-erased random walks and uniform
  spanning trees}, Israel J. Math. 118 (2000), 221--288.

\bibitem[SW05]{SW05} O.\ Schramm, D.\ B.\ Wilson, \emph{SLE coordinate changes}, New York J. Math. 11 (2005), 659--669.

\bibitem[Seg47]{seg47} I.\ E.\ Segal, \emph{Irreducible representations of operator algebras}, Bull. Amer. Math. Soc.  53 (1947), 73--88.

\bibitem[Sha93]{shapiro} J.\ H.\ Shapiro, \emph{Composition Operators and Classical Function Theory}, Springer, 1993.

\bibitem[She07]{She07} S.\ Sheffield, \emph{Gaussian free fields for mathematicians}, Probability Theory and Related Fields (2007) 139, 521--541.

\bibitem[SS11]{SS11} R.H.\ Shumway, D.S.\ Stoffer, \emph{Time Series Analysis and Its Applications}, Springer, 2011.

\bibitem[Ske04]{Ske04} M.\ Skeide, \emph{Independence and Product Systems}, in Recent Developments in Stochastic Analysis and Related Topics, World Scientific Publishing, 2004, 420--438.

\bibitem[Sok97]{Sok97} A.\ Sokal, \emph{Monte Carlo methods in statistical mechanics: foundations and new algorithms}, in: Functional integration, Springer, 1997, 131--192.

\bibitem[Spe90]{Spe90} R.\ Speicher, \emph{A New Example of `Independence' and `White Noise'}, Probab. Th. Rel. Fields 84 (1990), 141--159. 

\bibitem[Spe97]{Spe97} R.\ Speicher, \emph{On universal products}, Fields Inst. Commun. 12 (1997), 257--266. 

\bibitem[SW97]{SW97} R.\ Speicher, R.\ Woroudi, \emph{Boolean convolution}, in: Free Probability Theory,  D.\ Voiculescu, Fields Inst. Commun. 12, Amer. Math. Soc., 1997, 267--279.

\bibitem[SB18]{SB18}  S.\ Sutton, A.\ G.\ Barto, \emph{Reinforcement Learning: An Introduction}, MIT Press, 2018.

\bibitem[Tak08]{tak08} L.A.\ Takhtajan, \emph{Quantum mechanics for mathematicians}, Graduate Studies in Mathematics Volume 95, American Mathematical Society, 2008.

\bibitem[Tam98]{Tam98} A.\ Tamir, \emph{Applications of Markov Chains in Chemical Engineering}, Elsevier Science, 1998.

\bibitem[Tao12]{Tao12} T.\ Tao, \emph{Topics in Random Matrix Theory}, American Mathematical Society, 2012.

\bibitem[Voi86]{Voi86} D.\ Voiculescu, \emph{Addition of Certain Non-commuting Random Variables}, Journal of Functional Analysis 66 (1986), 323--346.

\bibitem[Voi91]{Voi91} D.\ Voiculescu, \emph{Limit laws for Random matrices and free products}, Inv. Math. 104 (1991), 201--220.

\bibitem[Voi97]{Voi97} D.\ Voiculescu, \emph{Free Probability Theory}, Fields Inst. Commun. 12, Amer. Math. Soc., 1997.

\bibitem[Wan14]{Wan14} J.-C.\ Wang, \emph{The central limit theorem for monotone convolution with applications to free L\'{e}vy processes and infinite ergodic theory}, Indiana University Mathematics Journal 63 (2014), 303--327.


\bibitem[Wig55]{Wig55} E.\ Wigner, \emph{Characteristic Vectors of Bordered Matrices with Infinite Dimensions}, Ann. of Math. 62 (1955), 548--564.

\bibitem[Wig58]{Wig58}E.\ Wigner, \emph{On the Distribution of the Roots of Certain Symmetric Matrices},  Ann. of Math. 67 (1958), 325--328. 

\bibitem[Wil79]{Wil63} R.L.\ Wilder, \emph{Topology of manifolds}, 
Reprint of 1963 edition, American Mathematical Society Colloquium Publications, 32, American Mathematical Society, 1979.

  
\bibitem[ZZ18]{ZZ18} H.\ Zhang, M.\ Zinsmeister, {\em  Local Analysis of Loewner Equation},
 	arXiv:1804.03410.

\end{thebibliography}
\end{document}